\documentclass[letter,11pt]{article}

\usepackage{tikz}
\usetikzlibrary {shapes.geometric}

\usepackage{mathtools}
\usepackage{subcaption}
\usepackage{enumerate}
\usepackage{amsmath}
\usepackage{amsfonts}
\usepackage{graphicx}
\usepackage{tcolorbox}
\usepackage{amssymb}
\usepackage{multirow}
\usepackage{geometry}
\usepackage{soul}
\usepackage{colortbl}
\usepackage{wrapfig}
\usepackage{algorithm}
\usepackage{algorithmicx}
\usepackage{algpseudocode}
\usepackage{amsthm}
\usepackage{todonotes}
\usepackage{mathtools}
\usepackage{mathrsfs}

\usepackage[pdftex, plainpages = false, pdfpagelabels, 
                 bookmarks=false,
                 bookmarksopen = true,
                 bookmarksnumbered = true,
                 breaklinks = true,
                 linktocpage,
                 pagebackref,
                 colorlinks = true,  
                 linkcolor = blue,
                 urlcolor  = blue,
                 citecolor = red,
                 anchorcolor = green,
                 hyperindex = true,
                 hyperfigures
                 ]{hyperref}

\usepackage{thmtools} 
\usepackage{thm-restate}
\usepackage{bm}

\newcommand{\ilog}[2]{{\sf ilog}(#1,#2)}

\newcommand{\threecoloring}{{\sc $3$-Coloring}}

\algtext*{EndWhile}
\algtext*{EndIf}
\algtext*{EndFor}

\newtheorem{lemma}{Lemma}
\newtheorem*{remark}{Remark}
\numberwithin{lemma}{section}
\newtheorem{theorem}[lemma]{Theorem}
\newtheorem{theorem*}[lemma]{Theorem*}
\newtheorem{corollary}[lemma]{Corollary}

\newtheorem{fact}[lemma]{Fact}
\newtheorem{claim}[lemma]{Claim}
\newtheorem{proposition}[lemma]{Proposition}

\title{Fine-Grained Bounds for Courcelle's Theorem}
\author{Daniel Lokshtanov\footnote{University of California, Santa Barbara, USA. Email: \texttt{daniello@ucsb.edu}} \and Fahad Panolan\footnote{University of Leeds, UK. Email: \texttt{F.panolan@leeds.ac.uk}} \and Saket Saurabh\footnote{Institute of Mathematical Sciences, Chennai, India. Email: \texttt{saket@imsc.res.in}} \and Jie Xue\footnote{New York University Shanghai, China. Email: \texttt{jiexue@nyu.edu}} \and Meirav Zehavi\footnote{Ben-Gurion University, Israel. Email: \texttt{meiravze@bgu.ac.il}}}
\date{}

\def\subparagraph{\paragraph}

\begin{document}


\maketitle

\begin{abstract}
Courcelle's theorem states that there exists an algorithm that takes as input a graph $G$ of treewidth at most $t$ and a MSO formula $\phi$, and determines whether $G$ satisfies $\phi$ in time $f(\phi,t) \cdot n$.
It is folklore that the the function $f$ contains a tower of exponentials whose height depends as a linear function of the number of quantifier alternations of the input formula $\phi$.
A classic reduction of Frick and Grohe shows that, assuming the Exponential Time Hypothesis (ETH), the linear growth of the height of the tower is unavoidable.
Nevertheless, there is still a huge gap between existing upper and lower bounds -- after all, there is quite a difference between a single exponential and a double exponential running time. 
In addition, this only gives us a very coarse understanding in the time complexity of Courcelle's theorem.
In this paper, we prove a fine-grained version of Courcelle's theorem with nearly ETH-tight dependence on the treewidth parameter $t$ and the quantifier structure of $\phi$ (specifically, the number of first order and second order variables in each quantifier alternation block).
\end{abstract}



\section{Introduction}
Courcelle's theorem~\cite{courcelle1990monadic,BoriePT92} is one of the most celebrated algorithmic meta-theorems, which shows that every graph property expressible in monadic second-order (MSO) logic can be checked in linear time on graphs of bounded treewidth.
Formally, it states that given an $n$-vertex graph $G$ and an MSO formula $\phi$, one can test whether $G$ satisfies $\phi$ in $f(\phi,t) \cdot n$ time for some (computable) function $f$, where $t=\mathbf{tw}(G)$ denotes the treewidth of $G$.

Due to the generality of the MSO logic and the importance of treewidth as a structural graph parameter, Courcelle's theorem has brought a profound impact on the theory of parameterized complexity.
Specifically, it implies that a large variety of NP-hard graph problems are fixed-parameter tractable (FPT) parameterized by treewidth.
In addition, for parameterized graph problems that can be defined using MSO formulas depending on the problem parameter $k$, Courcelle's theorem results in FPT algorithms parameterized by both treewidth and $k$.

While the running time of the algorithm in Courcelle's theorem is linear in $n$ (which is optimal), its dependency on $\phi$ and $t$ is rather intricate and less understood.
It was known~\cite{courcelle1990monadic,kneis2009practical,lampis2023first} that the function $f(\phi,t)$ in the bound is not elementary and contains a tower of exponentials whose height depends on $\phi$.
The seminal work of Frick and Grohe~\cite{frick2004complexity} proved that, assuming the ETH, having such a tower of exponentials in the time complexity is unavoidable even when $\phi$ is a first-order (FO) logic formula and $G$ is a tree.
These results, however, only provide us a very coarse understanding in what $f$ should look like in the worst case.
Therefore, a more ``fine-grained'' study for the function $f(\phi,t)$ in Courcelle's theorem turns out to be appealing.
While a lot of efforts have been made to understand the optimal time complexity for specific instances of MSO-expressible problems over years~\cite{DBLP:journals/talg/LokshtanovMS18,DBLP:journals/siamcomp/LokshtanovMS18,DBLP:journals/talg/CyganNPPRW22,DBLP:conf/stacs/HartmannM25,DBLP:conf/icalp/EsmerFMR24,DBLP:conf/soda/FockeMINSSW23}, little work focused on the general MSO testing problem.

Ideally, one wishes to give some concrete function $f(\phi,t)$ that can describe (either exactly or approximately), for \textit{every} $\phi$ and $t$, the minimum amount of time required to test property $\phi$ on graphs of treewidth $t$.
However, this is unfortunately impossible.
Indeed, it is not difficult to show the undecidability of the following problem under the assumption P$\neq$NP (we sketch a proof in Appendix~\ref{apx-undecidable}): given as input an MSO formula $\phi$, decide whether testing $\phi$ on graphs is polynomial-time solvable or not.
This hardness result implies that there is even no way to characterize the formulas $\phi$ for which the function $f(\phi,t)$ can be made polynomial in $t$.
As such, one cannot hope for any ``reasonable'' bound on $f(\phi,t)$ that is optimal for every individual $\phi$.

Given this frustrating fact, the next best thing one can do towards a fine-grained understanding in Courcelle's theorem is to establish bounds on $f(\phi,t)$ which, while not being optimal for individual formulas, are optimal (or near-optimal) in terms of certain structural parameter $S(\phi)$ of $\phi$ that can capture the ``complexity'' of $\phi$ as precisely as possible; here $S(\phi)$ can be a number or something more complicated.
More formally, we want to shoot for the best function $f(S,t)$, which gives the time required to test an MSO property $\phi$ with $S(\phi) = S$ on graphs of treewidth $t$.



Motivated by the above discussion, in this paper, we initiate a systematic study on the complexity of Courcelle's theorem, with the goal of understanding its dependency on $\phi$ and $t$ in a fine-grained way.
We prove (almost) matching upper and lower bounds for the complexity of Courcelle's theorem in terms of the \textit{quantifier structure} of $\phi$ (the detailed results will be discussed in the next section).
We expect this work to be a starting point of the long-term research towards thoroughly understanding the time complexity of Courcelle's theorem.

\paragraph{Other related work.}
We briefly summarize the existing work regarding algorithms and lower bounds on bounded-treewidth graphs.
A more detailed discussion can be found in Appendix~\ref{apx-literature}.

A large body of work~\cite{DBLP:journals/talg/LokshtanovMS18,DBLP:journals/siamcomp/LokshtanovMS18,DBLP:journals/talg/CyganNPPRW22,DBLP:conf/stacs/HartmannM25,DBLP:conf/icalp/EsmerFMR24,DBLP:conf/soda/FockeMINSSW23} focused on obtaining tight bounds for classic (MSO-expressible) NP-hard problems parameterized by treewidth.
Problems solvable in $2^{O(t)} \cdot n^{O(1)}$ time include \textsc{Vertex Cover}, \textsc{Dominating Set}, \textsc{$q$-Coloring}, \textsc{Max Cut}, \textsc{Hamiltonian Cycle}, \textsc{Steiner Tree} etc.
Problems solvable in $2^{O(t \log t)} \cdot n^{O(1)}$ time include \textsc{Cycle Packing}, \textsc{Chromatic Number}, etc.
The time bounds for these problems are all known to be tight.
Several works focused on \emph{classes} of problems. 
Prominent examples arise around \textsc{Dominating Set}, yielding optimal bounds for variants such as \textsc{$r$-Domination}~\cite{DBLP:conf/iwpec/BorradaileL16} and \textsc{$(\sigma,\rho)$-Domination}~\cite{DBLP:conf/soda/FockeMINSSW23}, among many others~\cite{DBLP:conf/stacs/HartmannM25,DBLP:conf/icalp/EsmerFMR24}.

Besides, there has been extensive work on identifying \emph{subclasses} of bounded-treewidth graphs on which every MSO/FO property is decidable with running time bounded by an \emph{elementary} function, in contrast to the non-elementary bounds implied by Courcelle’s theorem.
For example, Lampis~\cite{DBLP:journals/algorithmica/Lampis12} showed that MSO properties can be decided in double exponential time on graphs with bounded vertex cover number, and FO properties can be decided in single exponential time on graphs with bounded max-leaf number.
In a recent work, Lampis~\cite{DBLP:conf/icalp/Lampis23} proved that FO properties can be decided in elementary-function-bounded time on graphs with bounded \emph{pathwidth} (as opposed to treewidth).
\textcolor{blue}{Gajarský et al.~\cite{GajarskyLICS} characterised subgraph-closed graph classes for which the FO-model checking problem is fixed-parameter tractable with an elementary dependency on the formula size.} 
Gajarský and Hliněný~\cite{DBLP:journals/corr/abs-1204-5194} showed that in the universe of colored trees of fixed height, any MSO-expressible problem with \(r\) quantifiers admits a finite family of kernels whose size is bounded by an elementary function of \(r\) and the number of colors.

Kreutzer and Tazari~\cite{DBLP:conf/soda/KreutzerT10} showed that for graph classes with mild closure properties, the presence of graphs with sufficiently large treewidth (already polylogarithmic in \(n\)) precludes polynomial-time model checking for \(\mathsf{MSO}_2\) formulas. In this sense, bounded treewidth forms the effective boundary for tractable \(\mathsf{MSO}_2\) model checking (also see~\cite{DBLP:conf/lics/KreutzerT10}).

\subsection{Our results}

In order to discuss our results, we first need to define formally the ``quantifier structure'' of an MSO formula under consideration.
For simplicity of exposition, in this section, we only consider formulas in \textit{prenex normal form} (PNF), which requires all quantifiers to appear at the beginning of the formula.
The definition and our results apply to general MSO formulas\footnote{While every MSO formula can be modified to PNF, such a modification might increase the quantifier rank of the formula and thus makes the formula to have a more complex quantifier structure. As such, we do not make such a modification in our algorithms.} as well.

Consider an MSO formula $\phi$ in PNF, which consists of a sequence $\mathcal{Q}$ of quantifiers followed by a quantifier-free MSO formula on those quantified variables.
We can describe the quantifier structure of $\phi$ by considering the following three aspects of $\mathcal{Q}$.

\begin{itemize}
    \item \textbf{Quantifier alternations.}
    The number of quantifier alternations in an MSO formula turns out to be an important parameter to measure its complexity, which influences the height of the tower of exponentials in the running time of Courcelle's theorem~\cite{courcelle1990monadic,frick2004complexity}.
    As $\phi$ is in PNF, this parameter can be simply defined as the smallest integer $d \in \mathbb{N}$ such that one can partition the quantifier sequence $\mathcal{Q}$ into $d$ consecutive ``blocks'' $\mathcal{Q}_1,\dots,\mathcal{Q}_d$ each of which contains quantifiers of the same $(\exists,\forall)$-type.
    
    \item \textbf{Number of quantifiers.}
    Naturally, the number of quantifiers in the formula also captures how complex it is.
    Suppose $\mathcal{Q}$ is already partitioned into blocks $\mathcal{Q}_1,\dots,\mathcal{Q}_d$ according to the quantifier alternations (assume $\mathcal{Q}_1,\dots,\mathcal{Q}_d$ are sorted from left to right, or from outermost to innermost).
    Instead of simply considering the total number of quantifiers in $\mathcal{Q}$, we should consider the number of quantifiers in each individual block $\mathcal{Q}_i$.
    Note that the roles of $\mathcal{Q}_1,\dots,\mathcal{Q}_d$ in $\phi$ are not exchangeable, and as we will see later in our results, the numbers of quantifiers in $\mathcal{Q}_1,\dots,\mathcal{Q}_d$ indeed contribute to the time complexity in different ways.
    
    \item \textbf{Variable types.}
    In an MSO formula, there are two types of variables, i.e., vertex variables and set variables, which correspond to a single vertex and a set of vertices in the graph, respectively.
    An MSO formula with only vertex variables is just an FO formula.
    In many cases, checking FO graph properties is substantially easier than checking MSO graph properties.
    For example, testing a fixed FO formula on (general) graphs can always be done in polynomial time, while the problem of testing a fixed MSO formula can be NP-hard (e.g., $3$-\textsc{Coloring}).
    As such, when considering the quantifiers in $\phi$, we should distinguish the ones for vertex variables (called \textit{vertex quantifiers}) and the ones for set quantifiers (called \textit{set quantifiers}).
    For each block $\mathcal{Q}_i$, we use $k_i$ to denote the number of vertex quantifiers in $\mathcal{Q}_i$ and use $s_i$ to denote the number of set quantifiers in $\mathcal{Q}_i$.
    Note that the ordering of the $k_i+s_i$ quantifiers in $\mathcal{Q}_i$ does not matter, as all these quantifiers are of the same $(\exists,\forall)$-type.
\end{itemize}

Based on the above discussion, we can now naturally represent the quantifier structure of $\phi$ using the sequence $S = ((k_1,s_1),\dots,(k_d,s_d))$.
We call $\phi$ an \textit{$S$-MSO formula}.
Formally, an MSO formula in PNF is an $S$-MSO formula if its quantifier sequence can be partitioned into $d$ consecutive blocks $\mathcal{Q}_1,\dots,\mathcal{Q}_d$ (sorted from left to right) such that each block $\mathcal{Q}_i$ consists of $k_i$ vertex quantifiers and $s_i$ set quantifiers of the same $(\exists,\forall)$-type.
(One can further require the quantifiers in adjacent blocks to have different $(\exists,\forall)$-type, but this is not necessary.)
The notion of $S$-MSO formulas can be easily generalized to general MSO formulas.
The main focus of this paper is to understand the function $f(\phi,t)$ in Courcelle's theorem in terms of the sequence $S$ representing the quantifier structure of $\phi$.
We formulate this as the following parameterized problem.

\begin{tcolorbox}[colback=gray!5!white,colframe=gray!75!black]
        \textsc{MSO Testing} \hfill \textbf{Parameter:} $t \in \mathbb{N}$ and $S = ((k_1,s_1),\dots,(k_d,s_d))$
        \vspace{0.2cm} \\
        \textbf{Input:} A graph $G$ with $\mathbf{tw}(G) \leq t$ and an $S$-MSO formula $\phi$
        \vspace{0.1cm} \\
        \textbf{Goal:} Decide whether $G$ satisfies $\phi$ or not
\end{tcolorbox}

A particularly important special case of \textsc{MSO Testing} is the \textsc{FO Testing} problem, in which $\phi$ is an FO formula.
Since FO formulas are just an MSO formulas without set variables, we can also represent their quantifier structures using sequences.
Formally, for $S = (k_1,\dots,k_d)$, we define an $S$-FO formula as a $S^+$-MSO formula where $S^+ = ((k_1,0),\dots,(k_d,0))$.

\begin{tcolorbox}[colback=gray!5!white,colframe=gray!75!black]
        \textsc{FO Testing} \hfill \textbf{Parameter:} $t \in \mathbb{N}$ and $S = (k_1,\dots,k_d)$
        \vspace{0.2cm} \\
        \textbf{Input:} A graph $G$ with $\mathbf{tw}(G) \leq t$ and an $S$-FO formula $\phi$
        \vspace{0.1cm} \\
        \textbf{Goal:} Decide whether $G$ satisfies $\phi$ or not
\end{tcolorbox}

Our main results are (almost) matching upper and lower bounds for the complexity of solving \textsc{MSO/FO Testing} (and their variants/extensions).
Below we discuss these results in detail. 

\paragraph{Upper bounds.}
To present our algorithmic results, we need to first introduce some notations.
We define $\exp^{(0)}(x) = x$ and $\exp^{(i)}(x) = 2^{\exp^{(i-1)}(x)}$ for all integer $i \geq 1$.
In other words, $\exp^{(i)}(x)$ is a tower of exponentials of base $2$ and height $i$ with $x$ on top of it.
The notation $\hat{O}(\cdot)$ denotes the big-$O$ that hides subpolynomial factors, i.e., $\hat{O}(x) = x^{1+o(1)}$.
In other words, $\hat{O}(x)$ describes the bound that is \textit{almost} linear in $x$.
For a graph $G$ and a number $t \in \mathbb{N}$, we denote by $T_\mathsf{td}(G,t)$ the time required for computing a tree decomposition $G$ with width $t^{O(1)}$, provided that $\mathbf{tw}(G) \leq t$.
Our main algorithmic result is the following.

\begin{restatable}{theorem}{thmMSO} \label{thm-MSO}
    There exists an algorithm for \textsc{MSO Testing} that solves an instance $(G,t,S,\phi)$ with $|V(G)| = n$ and $S = ((k_1,s_1),\dots,(k_d,s_d))$ in time $T_\mathsf{td}(G,t)+$
    {\small\begin{equation*}
         f(d) \cdot \left(\sum_{i=1}^d \sum_{j=1}^{i-1} \exp^{(i)}(\hat{O}(s_j k_i)) + \sum_{i=1}^d 2^{O(s_d k_i)} + \sum_{i=1}^d \sum_{j=1}^i \exp^{(i)}(\hat{O}(t_j k_i)) + \sum_{i=1}^d \exp^{(i)}(\hat{O}(k_i \log t)) \right) \cdot (n+|\phi|^{O(1)})
    \end{equation*}}for a computable function $f$, where $t_j = \min\{k_j,t\}$ for $j \in [d]$ and $|\phi|$ is the description size of $\phi$.
\end{restatable}

\smallskip

Using the previous algorithms for computing tree decompositions, e.g.~\cite{bodlaender2016c,fomin2018fully,korhonen2023single}, we can take either $T_\mathsf{td}(G,t) = t^{O(1)} n \log n$ or $T_\mathsf{td}(G,t) = 2^{O(t)} n$.
When $d \geq 2$, if we set $T_\mathsf{td}(G,t) = 2^{O(t)} n$, then $T_\mathsf{td}(G,t)$ is dominated by the other part of the bound of Theorem~\ref{thm-MSO} and thus can be removed from the bound for free.
When $d = 1$, it results in an overhead of either $t^{O(1)} n \log n$ or $2^{O(t)} n$.
Theorem~\ref{thm-MSO} directly implies the following result for \textsc{FO Testing}.

\begin{corollary} \label{cor-FO}
    There exists an algorithm for \textsc{FO Testing} that solves an instance $(G,t,S,\phi)$ with $|V(G)| = n$ and $S = (k_1,\dots,k_d)$ in time $T_\mathsf{td}(G,t)+$
    {\small \begin{equation*}
         f(d) \cdot \left(\sum_{i=1}^d \sum_{j=1}^i \exp^{(i)}(\hat{O}(t_j k_i)) + \sum_{i=1}^d \exp^{(i)}(\hat{O}(k_i \log t)) \right) \cdot (n+|\phi|^{O(1)})
    \end{equation*}}for a computable function $f$, where $t_j = \min\{k_j,t\}$ for $j \in [d]$ and $|\phi|$ is the description size of $\phi$.
\end{corollary}

The bound in Theorem~\ref{thm-MSO} is complicated, and as we will see later, it is essentially the best one can hope for.
Before moving to the lower bound part, we briefly discuss the bound in Theorem~\ref{thm-MSO} and how it relies on the various parameters.

First, we consider the parameters $k_1,s_1,\dots,k_d,s_d$, while assuming the treewidth parameter $t$ is a constant.
In this case, the bound in Theorem~\ref{thm-MSO} becomes
\begin{equation*}
    f(d) \cdot \left(\sum_{i=1}^d \sum_{j=1}^{i-1} \exp^{(i)}(\hat{O}(s_j k_i)) + \sum_{i=1}^d 2^{O(s_d k_i)} \right) \cdot (n+|\phi|^{O(1)}),
\end{equation*}
which is the time complexity of the algorithm in Theorem~\ref{thm-MSO} when applied to \textsc{MSO Testing} on trees (or graphs whose treewidth is a constant).
If we further assume that $s_1,\dots,s_d$ are constant numbers, then the bound simply becomes $f(d) \cdot (\sum_{i=1}^d \exp^{(i)}(\hat{O}(k_i))) \cdot (n+|\phi|^{O(1)})$, which is also the time for solving \textsc{FO Testing} on bounded-treewidth graphs.
In other words, the dependency on $k_i$, i.e., the number of vertex quantifiers in the $i$-th block, forms an exponential tower of height $i$ (rather than $d$) with $k_i$ on top of it.
On the other hand, the dependency on $s_1,\dots,s_d$, i.e., the numbers of set quantifiers in the $d$ blocks, is much worse.
If $k_1,\dots,k_d$ are constant numbers, then the bound becomes $f(d) \cdot (\sum_{i=1}^{d-1} \exp^{(d)}(\hat{O}(s_i)) + 2^{\hat{O}(s_d)}) \cdot (n+|\phi|^{O(1)})$.
That says, all of the parameters $s_1,\dots,s_d$, except $s_d$, appear on top of the highest exponential towers, which are of height $d$.
Somewhat counterintuitively, however, the dependency on $s_d$ is single exponential.

Next, we consider the treewidth parameter $t$, while assuming $k_1,s_1,\dots,k_d,s_d$ are constant numbers.
In this case, the bound in Theorem~\ref{thm-MSO} becomes $f(d) \cdot \exp^{(d)}(\hat{O}(\log t))  \cdot (n+|\phi|^{O(1)})$.
In fact, a more careful analysis can give us an improved bound in this case, in which the tower is $\exp^{(d)}(O(\log t))$, i.e., $\exp^{(d-1)}(t^{O(1)})$; see Theorem~\ref{thm-MSOtw}.
Therefore, the dependency on $t$ forms an exponential tower of height $d-1$.
This implies, for example, that for a fixed formula $\phi$ where the number of quantifier alternations is $1$, the algorithm runs in polynomial time in $t$, and when the number of quantifier alternations is $2$ (such as \textsc{Independent Set} and \textsc{$3$-Coloring}), the running time is single exponential in $t$.
Figure~\ref{fig-dependency} gives an intuitive illustration for the dependency of our algorithm on the parameters, showing the level of the highest exponential tower on top of which each parameter appears (as a polynomial).


\begin{figure}
    \centering
    \includegraphics[height=5cm]{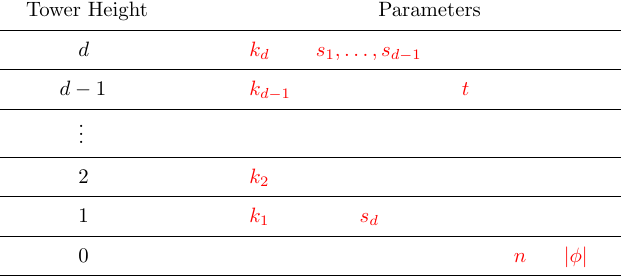}
    \caption{Dependency of the time complexity of our algorithms on various parameters.}
    \label{fig-dependency}
\end{figure}


\paragraph{Lower bounds.}
To complement our algorithmic results, we prove ETH-based lower bounds for \textsc{MSO Testing}, which demonstrates that the time complexity in Theorem~\ref{thm-MSO} is already tight, modulo the subpolynomial factors hidden in the $\hat{O}(\cdot)$-notation.
Specifically, our lower bounds imply that every tower of exponentials in the bound of Theorem~\ref{thm-MSO} is necessary (and the height of tower cannot be decreased).
All of our lower bounds hold even for the case where $d$ is a constant and $\phi$ is in PNF.
For convenience, we say a multivariate function $f(x_1,\dots,x_r)$ is \textit{independent} of the variable $x_i$ if $f(x_1,\dots,x_r) = g(x_1,\dots,x_{i-1},x_{i+1},\dots,x_r)$ for some function $g$.
Similarly, we can define functions independent of multiple variables.

We first consider the towers $\exp^{(i)}(\hat{O}(s_j k_i))$ and $2^{O(s_d k_i)}$.
This part is independent of the treewidth parameter $t$, and we can show the corresponding lower bounds even for the problem on \textit{trees}.
The results are presented in the following two theorems.


\begin{theorem}
\label{thm:intro-LW1}
    Let $d,i,j \in \mathbb{N}$ such that $j<i \leq d$.
    Assuming the ETH, if an algorithm solves \textsc{MSO Testing} on trees with $S = ((k_1,s_1),\dots,(k_d,s_d))$ in time $f(k_1,s_1,\dots,k_d,s_d) \cdot \exp^{(i)}(T(s_j, k_i)) \cdot$ $(n+|\phi|)^{O(1)}$ for a function $f(k_1,s_1,\dots,k_d,s_d)$ independent of $s_j$ and $k_i$, then $T(x,y) = \Omega(xy)$.
\end{theorem}

\begin{theorem}
\label{thm:intro-LW2}
    Let $d,i \in \mathbb{N}$ such that $i \leq d$.
    Assuming the ETH, if an algorithm solves \textsc{MSO Testing} on trees with $S = ((k_1,s_1),\dots,(k_d,s_d))$ in time $f(k_1,s_1,\dots,k_d,s_d) \cdot 2^{T(s_d, k_i)} \cdot (n+|\phi|)^{O(1)}$ for a function $f(k_1,s_1,\dots,k_d,s_d)$ independent of $s_d$ and $k_i$, then $T(x,y) = \Omega(xy)$.
\end{theorem}


Next, we consider the towers $\exp^{(i)}(\hat{O}(t_j k_i))$ and $\exp^{(i)}(\hat{O}(k_i \log t))$, where $t_j = \min\{k_j,t\}$.
This part is independent of the parameters $s_1,\dots,s_d$ describing the numbers of set quantifiers, and we can show the corresponding lower bounds even for \textsc{FO Testing}.

\begin{theorem}
\label{thm:intro-LW3}
    Let $d,i,j \in \mathbb{N}$ such that $j \leq i \leq d$.
    Assuming the ETH, if an algorithm solves \textsc{FO Testing} with $S = (k_1,\dots,k_d)$ in time $f(k_1,\dots,k_d,t) \cdot \exp^{(i)}(T(k_j,t,k_i)) \cdot (n+|\phi|)^{O(1)}$ for a function $f(k_1,\dots,k_d,t)$ that is independent of $t$, $k_j$, and $k_i$, then $T(x,y,z) = \Omega(\min\{x,y\} \cdot z)$.
\end{theorem}


\begin{theorem}
\label{thm:intro-LW4}
    Let $d,i \in \mathbb{N}$ such that $i \leq d$.
    Assuming the ETH, if an algorithm solves \textsc{FO Testing} with $S = (k_1,\dots,k_d)$ in time $f(k_1,\dots,k_d,t) \cdot \exp^{(i)}(T(k_i,t)) \cdot (n+|\phi|)^{O(1)}$ for a function $f(k_1,\dots,k_d,t)$ that is independent of $t$ and $k_i$, then $T(x,y) = \Omega(x \log y)$.
\end{theorem}

\paragraph{Extensions.}
Just as the original algorithm in Courcelle's theorem, our algorithm in Theorem~\ref{thm-MSO} can also be extended to problems in more general settings.

The first extension is to the MSO$_2$ logic, in which the variables can represent not only (sets of) vertices but also (sets of) edges.
Courcelle's theorem applies to MSO$_2$ logic as well.
Similarly, Theorem~\ref{thm-MSO} can be generalized to MSO$_2$ logic for free.
Indeed, one can easily reduce an \textsc{MSO$_2$ Testing} instance $(G,t,S,\phi)$ to an \textsc{MSO Testing} instance $(G',t',S,\phi')$ with $|V(G')| = O(t|V(G)|)$, $t' \leq t+1$, and $|\phi'| = |\phi|^{O(1)}$; see Section~\ref{sec-MSO2}.

The second extension is to the counting and optimization versions of \textsc{MSO Testing}.
Consider an MSO formula $\phi(x_1,\dots,x_k,X_1,\dots,X_s)$ with free vertex variables $x_1,\dots,x_k$ and free set variables $X_1,\dots,X_s$.
In the counting problem, our goal is to compute the number of satisfying assignments of $\phi$ in a graph $G$.
In the optimization variants, we are further given a weight function $w:V(G) \rightarrow \mathbb{R}$.
Define the weight of an assignment $(v_1,\dots,v_k,V_1,\dots,V_s)$ of $\phi$ to be $\sum_{i=1}^k w(v_i)+\sum_{i=1}^s \sum_{v \in V_i} w(v)$.
Then our goal is to compute the minimum (or maximum) weight satisfying assignment of $\phi$.

We formulate a problem, called \textsc{MSO Scoring}, using semi-fields, which simultaneously generalizes the two problems above.
Let $G$ be a graph and $\phi = \phi(x_1,\dots,x_k,X_1,\dots,X_s)$ be an MSO-formula.
Also, let $\mathbb{F}$ be a semi-field\footnote{A \textit{semi-field} is the same as a field except that the elements are not required to have additive inverses.} and $w:V(G) \rightarrow \mathbb{F}$ be a function.
Assume we are provided an oracle that can do additions/multiplications and find multiplicative inverses on $\mathbb{F}$ in constant time.
For each assignment $\alpha = (v_1,\dots,v_k,V_1,\dots,V_s)$ of $\phi$ where $v_1,\dots,v_k \in V(G)$ and $V_1,\dots,V_s \subseteq V(G)$, we write $w(\alpha) = (\prod_{i=1}^k w(v_i)) \cdot (\prod_{i=1}^s \prod_{v \in V_i} w(v))$.
Then the \textit{score} of $\phi$ on the vertex weighted graph $(G,w)$ is defined as $\mathsf{scr}(G,w,\phi) = \sum_{\alpha \in \mathcal{A}_\phi(G)} w(\alpha)$.
Here $\mathcal{A}_\phi(G)$ denotes the set of all satisfying assignments of $\phi$ in the graph $G$.
For a sequence $S = ((k_1,s_1),\dots,(k_d,s_d))$, an \textit{$S$-MSO$^*$} formula is defined as a $((k_2,s_2),\dots,(k_d,s_d))$-MSO formula with $k_1$ free vertex variables and $s_1$ free set variables.
Then \textsc{MSO Scoring} is defined as follows.

\begin{tcolorbox}[colback=gray!5!white,colframe=gray!75!black]
        \textsc{MSO Scoring} \hfill \textbf{Parameter:} $t \in \mathbb{N}$ and $S = ((k_1,s_1),\dots,(k_d,s_d))$
        \vspace{0.2cm} \\
        \textbf{Input:} A graph $G$ with $\mathbf{tw}(G) \leq t$, a function $w: V(G) \rightarrow \mathbb{F}$, and an $S$-MSO$^*$ formula $\phi$
        \vspace{0.1cm} \\
        \textbf{Goal:} Compute $\mathsf{scr}(G,w,\phi)$
\end{tcolorbox}

The intuition behind our definition of $S$-MSO$^*$ formulas is the following: testing an $S$-MSO formula $\phi$ on a graph $G$ can be reduced to counting satisfying assignments of an $S$-MSO$^*$ formula $\phi'$ on $G$, where $\phi'$ is obtained from $\phi$ by removing the first block of quantifiers and replacing the corresponding quantified variables with free variables.

It is easy to see that when $\mathbb{F} = \mathbb{R}$ with the normal addition and multiplication operators and $w(v) = 1$ for all $v \in V(G)$, $\mathsf{scr}(G,w,\phi)$ is just the number of satisfying assignments of $\phi$ and hence \textsc{MSO Scoring} generalizes the counting problem.
Also, when $\mathbb{F} = \mathbb{R} \cup \{-\infty,\infty\}$ with addition operator $\min\{\cdot,\cdot\}$ (resp., $\max\{\cdot,\cdot\}$) and multiplication operator that is the normal $+$ on real numbers, $\mathsf{scr}(G,w,\phi)$ is just the minimum (resp., maximum) weight of an satisfying assignment of $\phi$ and hence \textsc{MSO Scoring} generalizes the optimization problems.
We have the following theorem, which is a generalization of Theorem~\ref{thm-MSO}.

\begin{restatable}{theorem}{thmMSOS} \label{thm-MSO*}
    There exists an algorithm for \textsc{MSO Scoring} that solves an instance $(G,w,t,S,\phi)$ with $|V(G)| = n$ and $S = ((k_1,s_1),\dots,(k_d,s_d))$ in time $T_\mathsf{td}(G,t)+$
    {\small \begin{equation*}
         f(d) \cdot \left(\sum_{i=1}^d \sum_{j=1}^{i-1} \exp^{(i)}(\hat{O}(s_j k_i)) + \sum_{i=1}^d 2^{O(s_d k_i)} + \sum_{i=1}^d \sum_{j=1}^i \exp^{(i)}(\hat{O}(t_j k_i)) + \sum_{i=1}^d \exp^{(i)}(\hat{O}(k_i \log t)) \right) \cdot (n+|\phi|^{O(1)})
    \end{equation*}}for a computable function $f$, where $t_j = \min\{k_j,t\}$ for $j \in [d]$ and $|\phi|$ is the description size of $\phi$.
\end{restatable}

\paragraph{Our approaches.}
We briefly summarize the approaches we use to prove the main theorems.
A more detailed presentation is given in Section~\ref{sec-overview}.
To prove Theorem~\ref{thm-MSO}, the high-level framework of our algorithm is still the typical one: dynamic programming on a tree decomposition of $G$.
One of the main new insights is that when doing DP at a node $x$ of the tree decomposition, we exploit not only the small size of the bag of $x$, but also the small treewidth of the entire subgraph ``below $x$''.
Formally, let $(T,\beta)$ be a small-width tree decomposition of $G$.
For a node $x \in V(T)$, denote by $\gamma(x)$ the union of $\beta(y)$ for all nodes $y$ in the subtree of $T$ rooted at $x$.
In most DP algorithms on tree decompositions (in particular, the previous proofs of Courcelle's theorem), when computing the DP table at some node $x \in V(T)$, the algorithm actually no longer cares about whether the graph $G[\gamma(x)]$ has a small treewidth or not, and the time cost of this single step only depends on the size of $\beta(x)$ rather than the structure of $G[\gamma(x)]$.
In contrast, in our proof, we make heavy use of the bounded treewidth of the graphs $G[\gamma(x)]$ in the DP procedure, in order to reduce the size of the DP tables as well as the time for computing them.
To this end, we introduce a \textit{combinatorial} invariant of graphs, called \textit{$S$-signatures}, which characterizes the satisfiability of all $S$-MSO formulas on a graph.
Essentially, we prove the following nice properties of the signatures.
\begin{enumerate}[(i)]
    \item The (description) size of the $S$-signatures of \textit{bounded-treewidth} graphs is small: it only depends on $S$ and the treewidth parameter $t$, and satisfies the bound in Theorem~\ref{thm-MSO}.
    \item For each $x \in V(T)$, the $S$-signature of $G[\gamma(x)]$ can be computed given the $S$-signatures of $G[\gamma(y)]$ for all children $y$ of $x$, in time polynomial in the sizes of the signatures.
    In particular, one can efficiently compute the $S$-signature of $G$ by applying DP on $(T,\beta)$.
    \item Given the $S$-signature of $G$, one can test whether $G$ satisfies an $S$-MSO formula $\phi$ in time polynomial in the size of the signature and $|\phi|$.
\end{enumerate}
Our algorithm then uses property (ii) to compute the $S$-signature of $G$ and then uses property (iii) to test whether $G$ satisfies $\phi$ or not.
Finally, property (i) bounds the time complexity of the entire algorithm.
Among the three properties, the proof of property (i) is the most interesting, which requires a clever combination of the propeller decomposition technique~\cite{fomin2012planar,lokshtanov2018beating} and various structural properties of the signatures.
See Section~\ref{sec-overview1} for a more detailed discussion.

To prove the lower bounds, we build on the basic ideas in the reduction of Frick and Grohe~\cite{frick2004complexity}, and apply additional tricks to make the lower bound tight and more general.
The proof of Frick and Grohe implies (while not stated explicitly in~\cite{frick2004complexity}) that for any given $c \in \mathbb{N}$, one can reduce a problem with ETH lower bound $2^{\Omega(n)}$ to an \textsc{FO Testing} instance on trees with an $(k_1,\dots,k_d)$-FO formula $\phi$, where $d = 2c + O(1)$, $k_i = O(1)$ for all $i \in [d-1]$, and $\exp^{(c)}(k_d) = 2^{O(n)}$.
A drawback of this result is that it only gives us a lower bound roughly $\exp^{(d/2-O(1))}(\Omega(k_d))$, far away from the lower bound $\exp^{(d)}(\Omega(k_d))$ we want.
To achieve the desired bound, we need a much more careful reduction.
In particular, we give a more efficient way to encode numbers by trees in the sense that the (in)equality of two numbers can be checked using FO formulas with a much smaller number of quantifier alternations.
Furthermore, we construct new gadgets that allow us to obtain lower bounds regarding the numbers $s_1,\dots,s_d$ of set quantifiers and the treewidth parameter $t$, which are not considered in~\cite{frick2004complexity}.
Again, we provide more details in Section~\ref{sec-overview2}.

\paragraph{Organization.}
The rest of the paper is organized as follows.
In Section~\ref{sec-overview}, we overview the main ideas behind our proofs.
In Section~\ref{sec-pre}, we introduce the basic definitions and preliminaries required for our results.
Section~\ref{sec-ub} and Section~\ref{sec-lb} represent our proofs for the upper bounds and lower bounds, respectively.
Finally, Section~\ref{sec-conclusion} conclude the paper and poses some open problems.


\section{Technical overview} \label{sec-overview}
In this section, we provide an informal overview for the ideas used to prove our upper bounds and lower bounds.
We shall focus on the main insights, while omitting the details and calculations.

\subsection{Upper bounds} \label{sec-overview1}

As mentioned in the introduction, to prove Theorem~\ref{thm-MSO}, we need to introduce the notion of \textit{$S$-signatures}, which is a combinatorial invariant of graphs that characterizes the satisfiability of $S$-MSO formulas on a graph.
Formally, the \textit{satisfiability} of $S$-MSO formulas on a graph $G$ can be defined as a function $\mathsf{SAT}_{S,G}$ that maps each $S$-MSO formula $\phi$ to $\mathsf{True}$ if $G$ satisfies $\phi$ and to $\mathsf{False}$ if $G$ does not satisfy $\phi$.
If we use $\mathsf{sgn}_S(G)$ to denote the $S$-signature of $G$, we want the following condition to hold: for any graphs $G$ and $H$, $\mathsf{sgn}_S(G) = \mathsf{sgn}_S(H)$ iff $\mathsf{SAT}_{S,G} = \mathsf{SAT}_{S,H}$.

To get some intuition, consider a sequence $S = ((k_1,s_1),\dots,(k_d,s_d))$.
Every $S$-MSO formula (in PNF) can be written of the form $\phi = \mathsf{Q}x_1\dots\mathsf{Q}x_{k_1} \mathsf{Q}X_1\dots\mathsf{Q}X_{s_1}\ \phi'(x_1,\dots,x_{k_1},X_1,\dots,X_{s_1})$, where $\mathsf{Q} \in \{\exists,\forall\}$ and $\phi'$ is an $S'$-MSO formula (with free variables) for $S' = ((k_2,s_2),\dots,(k_d,s_d))$.
This recursive definition allows us to relate the satisfiability of $S$-MSO formulas to that of $S'$-MSO formulas as follows.
Let $\mathcal{A}_{G,k_1,s_1}$ be the set of sequences $(v_1,\dots,v_{k_1},V_1,\dots,V_{s_1})$ where $v_1,\dots,v_{k_1} \in V(G)$ and $V_1,\dots,V_{s_1} \subseteq V(G)$.
We view it as the set of assignments to the formulas of the form $\phi'(x_1,\dots,x_{k_1},X_1,\dots,X_{s_1})$ in $G$.
For each $A \in \mathcal{A}_{G,k_1,s_1}$, define $\mathsf{SAT}_{S',G,A}$ as the function that maps each $S'$-MSO formula $\phi'(x_1,\dots,x_{k_1},X_1,\dots,X_{s_1})$ to the value $\phi'(A) \in \{\mathsf{True},\mathsf{False}\}$ evaluated in $G$.
It turns out that $\mathsf{SAT}_{S,G} = \mathsf{SAT}_{S,H}$ iff both of the following are true:
\begin{itemize}
    \item for every $A \in \mathcal{A}_{G,k_1,s_1}$, there exists $B \in \mathcal{A}_{H,k_1,s_1}$ such that $\mathsf{SAT}_{S',G,A} = \mathsf{SAT}_{S',H,B}$,
    \item for every $B \in \mathcal{A}_{H,k_1,s_1}$, there exists $A \in \mathcal{A}_{G,k_1,s_1}$ such that $\mathsf{SAT}_{S',G,A} = \mathsf{SAT}_{S',H,B}$.
\end{itemize}
Equivalently, $\mathsf{SAT}_{S,G} = \mathsf{SAT}_{S,H}$ iff $\{\mathsf{SAT}_{S',G,A}: A \in \mathcal{A}_{G,k_1,s_1}\} = \{\mathsf{SAT}_{S',H,B}: B \in \mathcal{A}_{H,k_1,s_1}\}$.
Inspired by this nice relation, we obtain a natural idea to define the $S$-signatures: defining them \textit{inductively} based on the $S'$-signatures.
However, so far this idea does not quite work, since the satisfiability of $S'$-MSO formulas we use is already different from the original definition -- it takes into account free variables and assignments.

In order to make the idea work, we need to introduce a more general class of graphs, called \textit{labeled} and \textit{colored} graphs, which can ``encode'' assignments to free variables.
Let $G$ be a graph.
For $p \in \mathbb{N}_0$ (here $\mathbb{N}_0 = \{0,1,2,\dots\}$), a \textit{$p$-labeling} on $G$ is a function $\lambda: [p] \rightarrow V(G)$.
If $\lambda: [p] \rightarrow V(G)$ is a $p$-labeling on $G$ and $\lambda': [p'] \rightarrow V(G)$ is a $p'$-labeling on $G$, we define $\lambda \oplus \lambda': [p+p']  \rightarrow V(G)$ as
$(\lambda \oplus \lambda')(i) = \lambda(i)$ if $i \leq p$ and $(\lambda \oplus \lambda')(i) = \lambda'(i-p)$ if $i > p$, which is a $(p+p')$-labeling on $G$.
For a set $P$, a \textit{$P$-coloring} on $G$ is a function $\mu: V(G) \rightarrow P$.
If $\mu: V(G) \rightarrow P$ is a $P$-coloring on $G$ and $\mu': V(G) \rightarrow P'$ is a $P'$-coloring on $G$, then we define a function $\mu \otimes \mu': V(G) \rightarrow P \times P'$ as $(\mu \otimes \mu')(v) = (\mu(v),\mu'(v))$, which is a $(P \times P')$-coloring on $G$.
A \textit{$p$-labeled and $P$-colored graph} is a triple $(G,\lambda,\mu)$ where $G$ is a graph, $\lambda$ is a $p$-labeling on $G$, and $\mu$ is a $P$-coloring on $G$.
Two $p$-labeled and $P$-colored graphs $(G,\lambda,\mu)$ and $(G',\lambda',\mu')$ are \textit{isomorphic} if there exists an isomorphism $\pi: V(G) \rightarrow V(G')$ of $G$ and $G'$ such that $\lambda' = \pi \circ \lambda$ and $\mu = \mu' \circ \pi$.
Isomorphic labeled and colored graphs are viewed as the \textit{same} (or in other words, we only care about the isomorphic type of such graphs).
In particular, if $\mathcal{G}$ is a set of $p$-labeled and $P$-colored graphs, then different elements in $\mathcal{G}$ are always non-isomorphic.
If $(G,\lambda,\mu)$ is a $p$-labeled and $P$-colored graph and $A \subseteq V(G)$ is a subset, then we can naturally obtain a $|\lambda^{-1}(A)|$-labeled and $P$-colored graph $(G[A],\lambda_A,\mu_A)$ as follows.
The coloring $\mu_A$ is simply defined as $\mu_A = \mu_{|A}$.
To define $\lambda_A$, suppose $\lambda^{-1}(A) = \{x_1,\dots,x_r\}$ where $x_1 < \cdots < x_r$.
Then we define $\lambda_A: [|\lambda^{-1}(A)|] \rightarrow A$ by setting $\lambda_A(i) = \lambda(x_i)$.
We call $(G[A],\lambda_A,\mu_A)$ the \textit{restriction} of $(G,\lambda,\mu)$ to $A$.

Now we explain how to use labeled and colored graphs to encode assignments.
Again, consider a graph $G$ and some assignment $A = (v_1,\dots,v_k,V_1,\dots,V_s) \in \mathcal{A}_{G,k,s}$.
Naturally, the part $(v_1,\dots,v_k)$ for vertex variables can be represented as a $k$-labeling $\lambda: [k] \rightarrow V(G)$ where $\lambda(i) = v_i$.
Also, the part $(V_1,\dots,V_s)$ for set variables can be represented as a $\{0,1\}^s$-coloring $\mu: V(G) \rightarrow \{0,1\}^s$ where $\mu(v) = (\mathbf{1}_{v \in V_1},\dots,\mathbf{1}_{v \in V_s})$; here $\mathbf{1}_{v \in V_i} = 1$ if $v \in V_i$ and $\mathbf{1}_{v \in V_i} = 0$ if $v \notin V_i$.
Therefore, the $k$-labeled and $\{0,1\}^s$-colored graph $(G,\lambda,\mu)$ encodes the information of $A$.
To understand the intuition of the operators $\oplus$ and $\otimes$, consider a $((k_1,s_1),(k_2,s_2))$-MSO formula.
Suppose that the first block of quantifiers choose an assignment $A_1 \in \mathcal{A}_{G,k_1,s_1}$, and we already encode this assignment as above to obtain a $k_1$-labeled and $\{0,1\}^{s_1}$-colored graph $(G,\lambda_1,\mu_1)$.
Followed by this, the second block of quantifiers also choose an assignment $A_2 \in \mathcal{A}_{G,k_2,s_2}$.
We want to change the graph $(G,\lambda_1,\mu_1)$ so that it further encodes $A_2$.
Now we construct the $k_2$-labeling $\lambda_2$ and the $\{0,1\}^{s_2}$-coloring $\mu_2$ on $G$ corresponding to $A_2$.
Then we simply take $(G,\lambda_1 \oplus \lambda_2,\mu_1 \otimes \mu_2)$, which is just the desired graph that encodes the information of both $A_1$ and $A_2$.

In order to use labelings and colorings to replace the assignments, we also need to enhance the ability of MSO formulas a bit so that they can take into account the labels and colors.
Roughly speaking, an \textit{enhanced} MSO formula, which is designed for labeled and colored graphs, has two additional abilities.
First, it can refer to vertices with specific labels in the graph.
In other words, if $\lambda: [p] \rightarrow V(G)$ is the labeling of the graph, then the vertices $\lambda(1),\dots,\lambda(p)$ are viewed as arguments of the formula (and thus the formula can test the equality/adjacency among them and other quantified vertex variables).
Second, it can have atomic formulas which test whether a vertex has a specific color (in the set used for coloring the graph).
In other words, if $\mu: V(G) \rightarrow P$ is the coloring of the graph, then the formula can contain equations of the form $\mu(x) = a$, where $x$ is a vertex variable (possibly a labeled vertex) and $a \in P$.
Here we omit the formal definition of such formulas, as this concept is only introduced for intuitively understanding the definition of signatures (which is not needed in our actual proof).

Now we consider the satisfability of enhanced $S$-MSO formulas on a labeled and colored graph $(G,\lambda,\mu)$, which can be defined as a function $\mathsf{SAT}_{S,(G,\lambda,\mu)}$ that maps each enhanced $S$-MSO formula $\phi$ to $\mathsf{True}$ or $\mathsf{False}$ depending on whether $(G,\lambda,\mu)$ satisfies $\phi$ or not.
This generalized definition allows us to recursively characterize the satisfability of enhanced $S$-MSO formulas using the idea at the beginning of this section.
Suppose $S = ((k,s))+S'$.
Here $+$ denotes the concatenation operator for sequences.
Let $(G,\lambda,\mu)$ and $(H,\gamma,\tau)$ be two $p$-labeled and $P$-colored graphs.
Using the argument before, we can prove that $\mathsf{SAT}_{S,(G,\lambda,\mu)} = \mathsf{SAT}_{S,(H,\gamma,\tau)}$ iff
{\small \begin{equation*}
    \{\mathsf{SAT}_{S',(G,\lambda \oplus \lambda',\mu \otimes \mu')}: \lambda' \in \varLambda_{G,k} \text{ and } \mu' \in U_{G,\{0,1\}^s}\} = \{\mathsf{SAT}_{S',(H,\gamma \oplus \gamma',\tau \otimes \tau')}: \gamma' \in \varLambda_{H,k} \text{ and } \tau' \in U_{H,\{0,1\}^s}\},
\end{equation*}}where $\varLambda_{G,k}$ (resp., $\varLambda_{H,k}$) is the set of all $k$-labelings on $G$ (resp., $H$) and $U_{G,\{0,1\}^s}$ (resp., $U_{H,\{0,1\}^s}$) is the set of all $\{0,1\}^s$-coloring on $G$ (resp., $H$).
This nice characterization gives us a natural definition for $S$-signatures on labeled and colored graphs.
Suppose we have already define $S'$-signatures which characterize the satisfability of enhanced $S'$-MSO formulas on labeled and colored graphs.
Now we simply define the $S$-signature of a labeled and colored graph $(G,\lambda,\mu)$ as
\begin{equation*}
    \mathsf{sgn}_S(G,\lambda,\mu) = \{\mathsf{sgn}_{S'}(G,\lambda \oplus \lambda',\mu \oplus \mu'): \lambda' \in \varLambda_{G,k} \text{ and } \mu' \in U_{G,\{0,1\}^s}\}.
\end{equation*}
The characterization above guarantees that $S$-signatures characterize the satisfability of enhanced $S$-MSO formulas.
To complete the definition, we still need to define $S$-signatures for the base case, i.e., $S = ()$ is the empty sequence.
An enhanced $()$-MSO formula is quantifier-free and thus it can only test the equality/adjacency among the labeled vertices and the colors of the labeled vertices.
Therefore, on a graph $(G,\lambda,\mu)$, the induced subgraph $(G[\text{Im}(\lambda)],\lambda,\mu_{|\text{Im}(\lambda)})$ characterizes the satisfability of enhanced $()$-MSO formulas; here $\text{Im}(\lambda)$ is the image of $\lambda$ and we abuse the symbol $\lambda$ to denote the labeling on $G[\text{Im}(\lambda)]$ while its codomain is $V(G)$.
As such, we simply define $\mathsf{sgn}_{()}(G,\lambda,\mu) = (G[\text{Im}(\lambda)],\lambda,\mu_{|\text{Im}(\lambda)})$.
Finally, for a graph $G$, we can view it a labeled and colored graph by equipping it with the dummy labeling $-: \emptyset \rightarrow V(G)$ and the dummy coloring $-: V(G) \rightarrow \{0\}$ that maps every vertex $v \in V(G)$ to $0$, and define $\mathsf{sgn}_S(G) = \mathsf{sgn}_S(G,-,-)$.
One can verify that $\mathsf{sgn}_S(G)$ characterizes the satisfiability of (normal) $S$-MSO formulas on $G$.

Let $S = ((k_1,s_1),\dots,(k_d,s_d))$.
By construction, the $S$-signature of a $p$-labeled and $P$-colored graph is a $d$-layer nested set with $(p+\sum_{k=1}^d k_i)$-labeled and $(P \times \prod_{i=1}^d \{0,1\}^{s_i})$-colored graphs at the bottommost level.
We write $\lVert \mathsf{sgn}_S(G,\lambda,\mu) \rVert$ as the \textit{recursive size} of $\mathsf{sgn}_S(G,\lambda,\mu)$, which is defined as follows: if $\mathsf{sgn}_S(G,\lambda,\mu)$ is a set, then $\lVert \mathsf{sgn}_S(G,\lambda,\mu) \rVert = \sum_{x \in \mathsf{sgn}_S(G,\lambda,\mu)} \lVert x \rVert$; otherwise, $\lVert \mathsf{sgn}_S(G,\lambda,\mu) \rVert = 1$.
In fact, $\lVert \mathsf{sgn}_S(G,\lambda,\mu) \rVert$ is just the number of graphs at the bottommost level of the nested set $\mathsf{sgn}_S(G,\lambda,\mu)$.
The recursive size of an $S$-signature can be essentially viewed as the description size of the signature (i.e., the number of bits to encode the signature), modulo the description of each labeled and colored graph at the bottommost level which is anyway polynomial in the numbers in $S$.
Later we will sketch a proof that bounds the $S$-signatures of \textit{bounded-treewidth} graphs, which is a crucial part of our result.
Before this, we first briefly discuss how to do \textsc{MSO Testing} using signatures and how to compute the signatures by DP on tree decomposition.

\vspace{-0.2cm}

\paragraph{Testing MSO via signatures.}
In fact, from the construction of $S$-signatures, it is not difficult to see that given $\mathsf{sgn}_S(G,\lambda,\mu)$ and an (enhanced) $S$-MSO formula $\phi$, one can test whether $(G,\lambda,\mu)$ satisfies $\phi$ in $(\lVert \mathsf{sgn}_S(G,\lambda,\mu) \rVert+|\phi|)^{O(1)}$ time.
If $S = ()$, then this trivially holds.
Suppose $S = ((k,s))+S'$ and we already have an algorithm $\textsc{Test}_{S'}(\mathsf{sgn}_{S'}(H,\gamma,\tau),\phi')$ which returns $\mathsf{True}$ or $\mathsf{False}$ depending on whether $(H,\gamma,\tau)$ satisfies $\phi$ or not in $(\lVert \mathsf{sgn}_{S'}(H,\gamma,\tau) \rVert+|\phi'|)^{O(1)}$ time.
Then the algorithm $\textsc{Test}_S(\mathsf{sgn}_S(G,\lambda,\mu),\phi)$ essentially works as follows.
If the first block of quantifiers in $\phi$ are $\exists$-quantifiers, then we simply return $\bigvee_{x \in \mathsf{sgn}_S(G,\lambda,\mu)} \textsc{Test}_{S'}(x,\phi')$, where $\phi'$ is the part of $\phi$ after the first block of quantifiers.
On the other hand, if the first block of quantifiers in $\phi$ are $\forall$-quantifiers, then we simply return $\bigwedge_{x \in \mathsf{sgn}_S(G,\lambda,\mu)} \textsc{Test}_{S'}(x,\phi')$.

\vspace{-0.2cm}

\paragraph{Computing signatures by DP on tree decomposition.}
Let $(T,\beta)$ be a tree decomposition of $G$.
We want to compute $\mathsf{sgn}_S(G)$ by DP on $(T,\beta)$.
For a node $x \in V(T)$, denote by $T_x$ the subtree of $T$ rooted at $x$ and define $\gamma(x) = \bigcup_{y \in V(T_x)} \beta(y)$.
A natural idea is to compute, at each node $x \in V(T)$, the signature $\mathsf{sgn}_S(G[\gamma(x)])$, based on the signatures $\mathsf{sgn}_S(G[\gamma(y)])$ for children $y$ of $x$.
However, this does not directly work.
In fact, only having $\mathsf{sgn}_S(G[\gamma(x)])$ is not sufficient for the DP.
Instead, we have to view $G[\gamma(x)]$ as a \textit{boundaried} graph with boundary $\beta(x)$, and the $S$-signature computed for $G[\gamma(x)]$ should take into account the boundary vertices.
To this end, we again make use of coloring.
We arbitrarily choose an ordering $\sigma$ of $V(G)$.
For each node $x \in V(T)$, define a $[|\beta(x)|]_0$\footnote{Here the notation $[\cdot]_0$ is defined as $[n]_0 = \{0,1,\dots,n\}$.}-coloring $\mu_x: \gamma(x) \rightarrow [|\beta(x)|]_0$ on $G[\gamma(x)]$ that maps all vertices in $\gamma(x) \backslash \beta(x)$ to $0$ and maps the vertices in $\beta(x)$ bijectively to $[|\beta(x)|]$ following the ordering $\sigma$.
During the DP procedure, at each $x \in V(T)$, instead of computing $\mathsf{sgn}_S(G[\gamma(x)])$, we compute $\mathsf{sgn}_S(G[\gamma(x)],-,\mu_x)$.
With the help of the colorings $\mu_x$, the DP works and can be done efficiently, thanks to the following lemma.
Recall that a \textit{separation} of a graph $G$ is a pair $(A,B)$ with $A,B \subseteq V(G)$ such that $A \cup B = V(G)$ and there is no edge between $A \backslash B$ and $B \backslash A$ in $G$.
\begin{lemma}[informal] \label{lem-combine}
Let $(G,\lambda,\mu)$ be a $p$-labeled and $P$-colored graph and $(A,B)$ be a separation of $G$.
Also, let $\mu':V(G) \rightarrow [|A \cap B|]_0$ map all vertices in $V(G) \backslash (A \cap B)$ to $0$ and map $A \cap B$ bijectively to $[|A \cap B|]$.
Suppose $(G[A],\lambda_A,\mu_A)$ and $(G[B],\lambda_B,\mu_B)$ are the restrictions of $(G,\lambda,\mu \otimes \mu')$ to $A$ and $B$, respectively.
Then for any $S$, one can (efficiently) compute $\mathsf{sgn}_S(G,\lambda,\mu)$ by only knowing $\mathsf{sgn}_S(G[A],\lambda_A,\mu_A)$, $\mathsf{sgn}_S(G[B],\lambda_B,\mu_B)$, $\lambda^{-1}(A)$, and $\lambda^{-1}(B)$.
\end{lemma}
The above lemma is not only used for the computation of signatures.
Below we shall also use it to bound the signature sizes for bounded-treewidth graphs.

\vspace{-0.2cm}

\paragraph{Bounding the signature size.}
We now sketch our proof for bounding the recursive sizes of the signatures.
Although our final goal is to bound the recursive sizes $\lVert \mathsf{sgn}_S(G,\lambda,\mu) \rVert$ of the signatures, the main step here is to bound the sizes of the signatures \textit{as sets}, i.e., $|\mathsf{sgn}_S(G,\lambda,\mu)|$.
Indeed, due to the definition of the recursive size, once we can bound $|\mathsf{sgn}_S(G,\lambda,\mu)|$ for every $S$ and every $(G,\lambda,\mu)$, we can also bound $\lVert \mathsf{sgn}_S(G,\lambda,\mu) \rVert$ for every $S$ and every $(G,\lambda,\mu)$.

For simplicity, in this overview, we assume $t = O(1)$ and bound the size of the signatures using only the parameters in the sequence $S$.
For a set $\mathcal{G}$ of $p$-labeled and $P$-colored graphs, we define $\Delta_S(\mathcal{G}) = |\{\mathsf{sgn}_S(G,\lambda,\mu): (G,\lambda,\mu) \in \mathcal{G}\}|$, which is the number of different $S$-signatures the graphs in $\mathcal{G}$ have.
Denote by $\mathcal{G}_{p,P}$ the set of all $p$-labeled and $P$-colored graphs of treewidth at most $t$.

Consider a graph $(G,\lambda,\mu) \in \mathcal{G}_{p,P}$.
Let $S = ((k,s))+S'$ be a sequence of pairs of natural numbers.
By construction, the size of $\mathsf{sgn}_S(G,\lambda,\mu)$ is just to equal to the number of different signatures $\mathsf{sgn}_{S'}(G,\lambda \oplus \lambda',\mu \otimes \mu')$ we can obtain by choosing $\lambda' \in \varLambda_{G,k}$ and $\mu' \in U_{G,\{0,1\}^s}$.
Note that this number could be way smaller than the trivial bound $|\varLambda_{G,k}| \cdot |U_{G,\{0,1\}^s}|$, because different choices of $\lambda'$ and $\mu'$ may result in graphs $(G,\lambda \oplus \lambda',\mu \otimes \mu')$ with the same $S'$-signature.
To establish a good bound for this number, our plan is to associate with each choice $(\lambda',\mu') \in \varLambda_{G,k} \times U_{G,\{0,1\}^s}$ an ``object'' $\varGamma(\lambda',\mu')$ that satisfies the following properties:
\begin{enumerate}
    \item[\textbf{(P1)}] if $\varGamma(\lambda_1',\mu_1') = \varGamma(\lambda_2',\mu_2')$, then $\mathsf{sgn}_{S'}(G,\lambda \oplus \lambda_1',\mu \otimes \mu_1') = \mathsf{sgn}_{S'}(G,\lambda \oplus \lambda_2',\mu \otimes \mu_2')$,
    \item[\textbf{(P2)}] one can easily obtain a good upper bound for the number of different $\varGamma(\lambda',\mu')$.
\end{enumerate}
Property \textbf{(P1)} guarantees that $|\mathsf{sgn}_S(G,\lambda,\mu)|$ is at most the number of different $\varGamma(\lambda',\mu')$, and then property \textbf{(P2)} will allow us to find a bound for the latter.
To define the object $\varGamma(\lambda',\mu')$, we need the following decomposition lemma for bounded-treewidth graphs, which is known as propeller decomposition in~\cite{lokshtanov2018beating}.
We omit its proof in this overview.

\begin{lemma}[\cite{fomin2012planar,lokshtanov2018beating}] \label{lem-decompose-overview}
    Let $G$ be a graph with $\mathbf{tw}(G) = O(1)$.
    Then for any $R \subseteq V(G)$ with $|R| \leq r$, there exist $V_0,V_1,\dots,V_{r'} \subseteq V(G)$ where $r' = O(r)$ satisfying the following conditions:
    \begin{enumerate}[(i)]
        \item $V(G) = \bigcup_{i=0}^{r'} V_i$,
        \item $R \subseteq V_0$ and $|V_0| \leq O(r)$,
        \item $N_G(V_i \backslash V_0) \subseteq V_0 \cap V_i$ and $|V_0 \cap V_i| = O(1)$ for all $i \in [r']$.
    \end{enumerate}    
\end{lemma}

Consider a choice $(\lambda',\mu') \in \varLambda_{G,k} \times U_{G,\{0,1\}^s}$.
We apply the above lemma on the graph $G$ with $R = \text{Im}(\lambda \oplus \lambda')$ to obtain $V_0,V_1,\dots,V_{r'} \subseteq V(G)$ satisfying the three conditions.
Note that $|R| \leq p+k$ and thus $r' = O(p+k)$.
Condition (ii) of the lemma implies $\text{Im}(\lambda \oplus \lambda') \subseteq V_0$ and $|V_0| = O(r) = O(p+k)$.
Condition (iii) guarantees that there is no edge between $V_i \backslash V_0$ and $V_j \backslash V_0$ for any different $i,j \in [r']$.
For convenience, we assume without loss of generality that $|V_0 \cap V_1| = \cdots = |V_0 \cap V_{r'}| = z$; we have $z = O(1)$ by condition (iii).

For $i \in [r']_0$, let $(G[V_i],\lambda_i,\mu_i)$ be the restriction of $(G,\lambda \oplus \lambda',\mu \otimes \mu')$ to $V_i$.
Also, for $i \in [r']$, let $\pi_i: [z] \rightarrow V_0$ be a function that maps the numbers in $[z]$ bijectively to the vertices in $V_0 \cap V_i$, which is a $z$-labeling on $G[V_0]$, and let $\tau_i: V_i \rightarrow [z]_0$ be the function defined as $\mu_i(v) = \pi_i^{-1}(v)$ for $v \in V_0 \cap V_i$ and $\mu_i(v) = 0$ for $v \in V_i \backslash V_0$, which is a $[z]_0$-coloring on $G[V_i]$.
Now the key observation is that $\mathsf{sgn}_{S'}(G,\lambda \oplus \lambda',\mu \otimes \mu')$ is uniquely characterized by the graph $(G[V_0],\lambda_0 \oplus \pi_1 \oplus\cdots \oplus \pi_{r'},\mu_0)$ and the signatures $\mathsf{sgn}_{S'}(G[V_1],\lambda_1,\mu_1 \otimes \tau_1),\dots,\mathsf{sgn}_{S'}(G[V_{r'}],\lambda_{r'},\mu_{r'} \otimes \tau_{r'})$.

\begin{lemma} \label{lem-seq}
    One can compute $\mathsf{sgn}_{S'}(G,\lambda \oplus \lambda',\mu \otimes \mu')$, knowing only $(G[V_0],\lambda_0 \oplus \pi_1 \oplus \cdots \oplus \pi_{r'},\mu_0)$ and $\mathsf{sgn}_{S'}(G[V_1],\lambda_1,\mu_1 \otimes \tau_1),\dots,\mathsf{sgn}_{S'}(G[V_{r'}],\lambda_{r'},\mu_{r'} \otimes \tau_{r'})$.
\end{lemma}
\begin{proof}[Proof sketch.]
Let $(G_i,\lambda_i^*,\mu_i^*)$ be the restriction of $(G,\lambda \oplus \lambda',\mu \otimes \mu')$ to $V_0 \cup (\bigcup_{j=1}^i V_j)$.
Roughly speaking, the idea for computing $\mathsf{sgn}_{S'}(G,\lambda \oplus \lambda',\mu \otimes \mu')$ is to keep applying Lemma~\ref{lem-combine} to iteratively compute $\mathsf{sgn}_{S'}(G_i,\lambda_i^*,\mu_i^*)$ for $i = 0,1,\dots,r'$.
By construction, $(V(G_{i-1}),V_i)$ is a separation of $G_i$.
As such, Lemma~\ref{lem-combine} allows us to compute $\mathsf{sgn}_{S'}(G_i,\lambda_i^*,\mu_i^*)$ from $\mathsf{sgn}_{S'}(G_{i-1},\lambda_{i-1}^*,\mu_{i-1}^*)$ and the given signature $\mathsf{sgn}_{S'}(G[V_i],\lambda_i,\mu_i \otimes \tau_i)$, together with some other information that is encoded in the given graph $(G[V_0],\lambda_0 \oplus \pi_1 \oplus \cdots \oplus \pi_{r'},\mu_0)$.
Note that $(G_{r'},\lambda_{r'}^*,\mu_{r'}^*) = (G,\lambda \oplus \lambda',\mu \otimes \mu')$.
So we end up with the desired signature $\mathsf{sgn}_{S'}(G,\lambda \oplus \lambda',\mu \otimes \mu')$.

We remark that the actual proof of this lemma is more technical than our discussion above.
In fact, applying Lemma~\ref{lem-combine} does not enable us to compute $\mathsf{sgn}_{S'}(G_i,\lambda_i^*,\mu_i^*)$ from $\mathsf{sgn}_{S'}(G_{i-1},\lambda_{i-1}^*,\mu_{i-1}^*)$ and $\mathsf{sgn}_{S'}(G[V_i],\lambda_i,\mu_i \otimes \tau_i)$.
Therefore, in the actual proof, we need to compute the $S'$-signature of a graph more complicated than $(G_i,\lambda_i^*,\mu_i^*)$ in each iteration, in order to make the induction work.
For simplicity, we omit the details here.
\end{proof}

\noindent
Now we simply define $\varGamma(\lambda',\mu')$ to be the sequence
\begin{equation*}
    ((G[V_0],\lambda_0 \oplus \pi_1 \oplus \cdots \oplus \pi_{r'},\mu_0),\mathsf{sgn}_{S'}(G[V_1],\lambda_1,\mu_1 \otimes \tau_1),\dots,\mathsf{sgn}_{S'}(G[V_{r'}],\lambda_{r'},\mu_{r'} \otimes \tau_{r'})).
\end{equation*}
Lemma~\ref{lem-seq} directly implies property \textbf{(P1)} of $\varGamma(\lambda',\mu')$.
Next, we consider property \textbf{(P2)}, i.e., how to bound the number of different $\varGamma(\lambda',\mu')$.
We first observe some basic facts about $\varGamma(\lambda',\mu')$.
The first element in $\varGamma(\lambda',\mu')$ is a $(p+k+r'z)$-labeled and $(P \times \{0,1\}^s)$-colored graph of treewidth at most $t$, which has $O(p+k)$ vertices because $|V_0| = O(p+k)$.
For each $i \in [r']$, the graph $(G[V_i],\lambda_i,\mu_i \otimes \tau_i)$ is a $|\text{Im}(\lambda\oplus \lambda') \cap V_i|$-labeled and $(P \times \{0,1\}^s \times [z]_0)$-colored graph of treewidth at most $t$.
Since $\text{Im}(\lambda\oplus \lambda') = R \subseteq V_0$, we have $|\text{Im}(\lambda\oplus \lambda') \cap V_i| \leq |V_0 \cap V_i| = z$.
Thus, each of the remaining elements in $\varGamma(\lambda',\mu')$ is the $S'$-signature of a graph in $\mathcal{G}_{q,P \times \{0,1\}^s \times [z]_0}$ for some $q \in [z]$.

With these observations, we are ready to bound the number of different sequences $\varGamma(\lambda',\mu')$ for $(\lambda',\mu') \in \varLambda_{G,k} \times U_{G,\{0,1\}^s}$.
Since $r' = O(p+k)$, $t = O(1)$, and $z = O(1)$, it turns out that the number of $(p+k+r'z)$-labeled and $(P \times \{0,1\}^s)$-colored graphs with $O(p+k)$ vertices and treewidth at most $t$ is $(p+k)^{O(p+k)} \cdot (2^s|P|)^{O(p+k)}$ (up to isomorphism), which bounds the number of possible values for the first element in $\varGamma(\lambda',\mu')$.
For each of the remaining elements in $\varGamma(\lambda',\mu')$, the number of possible values is bounded by $\sum_{q=1}^z \Delta_{S'}(\mathcal{G}_{q,P \times \{0,1\}^s \times [z]_0})$, as it is the $S'$-signature of a graph in $\mathcal{G}_{q,P \times \{0,1\}^s \times [z]_0}$ for some number $q \in [z]$.
It is easy to see that $\sum_{j=1}^z \Delta_{S'}(\mathcal{G}_{j,P \times \{0,1\}^s \times [z]_0}) \leq z \cdot \Delta_{S'}(G_{z,P \times \{0,1\}^s \times [z]_0})$.
As such, there can be $(p+k)^{O(p+k)} \cdot (2^s|P|)^{O(p+k)} \cdot (\Delta_{S'}(\mathcal{G}_{z,P \times \{0,1\}^s \times [z]_0}))^{O(p+k)}$ different sequences $\varGamma(\lambda',\mu')$ in total,
which implies the bound
\begin{equation} \label{eq-recur1}
    |\mathsf{sgn}_S(G,\lambda,\mu)| \leq (p+k)^{O(p+k)} \cdot (2^s|P|)^{O(p+k)} \cdot (\Delta_{S'}(\mathcal{G}_{z,P \times \{0,1\}^s \times [z]_0}))^{O(p+k)}.
\end{equation}
To make use of this inequality, we need to further bound $\Delta_{S'}(\mathcal{G}_{z,P \times \{0,1\}^s \times [z]_0})$ on the right-hand side.
Interestingly, as we will see, our argument above also gives a \textit{recursive} bound for $\Delta$-values.

Recall our construction of the sequences $\varGamma(\lambda',\mu')$.
While we constructed $\varGamma(\lambda',\mu')$ with respect to the graph $(G,\lambda,\mu)$ and the choice $(\lambda',\mu') \in \varLambda_{G,k} \times U_{G,\{0,1\}^s}$, the sequence $\varGamma(\lambda',\mu')$ indeed \textit{only} depends on the $(p+k)$-labeled and $(P \times \{0,1\}^s)$-colored graph $(G,\lambda \oplus \lambda',\mu \otimes \mu')$.
In other words, using the same construction, we can associate with every graph in $\mathcal{G}_{p+k,P \times \{0,1\}^s}$ such a $\varGamma$-sequence that uniquely determines its $S'$-signature.
Therefore, the bound in Inequality~\ref{eq-recur1} applies to not only $\mathsf{sgn}_S(G,\lambda,\mu)$ but also its superset $K = \{\mathsf{sgn}_{S'}(H,\gamma,\tau):(H,\gamma,\tau) \in \mathcal{G}_{p+k,P \times \{0,1\}^s}\}$.
Note that $\{\mathsf{sgn}_S(G,\lambda,\mu): (G,\lambda,\mu) \in \mathcal{G}_{p,P}\} \subseteq 2^K$.
So we have the following recursive bound
\begin{equation} \label{eq-recur2}
    \Delta_S(\mathcal{G}_{p,P}) \leq 2^{|K|} \leq  2^{(p+k)^{O(p+k)} \cdot (2^s|P|)^{O(p+k)} \cdot (\Delta_{S'}(\mathcal{G}_{z,P \times \{0,1\}^s \times [z]_0}))^{O(p+k)}}.
\end{equation}

Now we can bound $|\mathsf{sgn}_S(G,\lambda,\mu)|$ by first applying Inequality~\ref{eq-recur1} and then repeatedly applying Inequality~\ref{eq-recur2}.
For the base case, we can show that $\Delta_{(k,s)}(\mathcal{G}_{p,P}) \leq 2^{(p+k)^{O(p+k)} \cdot |P|^{O(p+k)}}$, which is (surprisingly) independent of $s$.
We omit the proof of the base case and the calculation for working out the bound of $|\mathsf{sgn}_S(G,\lambda,\mu)|$.
But we can get some intuition about the bound just by checking Inequalities~\ref{eq-recur1} and~\ref{eq-recur2}.
Suppose $S = ((k_1,s_1),\dots,(k_d,s_d))$.
When applying Inequality~\ref{eq-recur1}, the parameter $k_1$ appears in the single exponential position, i.e., the top of an exponential tower of height $1$.
A nice property of this inequality is that the part $\Delta_{S'}(\mathcal{G}_{z,P \times \{0,1\}^s \times [z]_0})$ on the right-hand side, where $S' = ((k_2,s_2),\dots,(k_d,s_d))$ in this case, is \textit{independent} of $k_1$, and the number $z$ is a constant (only depending on $t$).
Therefore, while $\Delta_{S'}(\mathcal{G}_{z,P \times \{0,1\}^s \times [z]_0})$ might contain higher exponential towers, the parameter $k_1$ remains in the single exponential position.
When we further expand $\Delta_{S'}(\mathcal{G}_{z,P \times \{0,1\}^s \times [z]_0})$ using Inequality~\ref{eq-recur2}, the parameter $k_2$ appears in the double exponential position, i.e., the top of an exponential tower of height $2$.
Again, the recursive part we obtain is independent of $k_2$, and therefore $k_2$ remains in the double exponential position towards the end.
By keep expanding the $\Delta$-part in the bound using Inequality~\ref{eq-recur2}, we can see that each $k_i$ appears in the $i$-th exponential position.
In contrast, the parameters $s_1,\dots,s_{d-1}$ eventually all climb to the top of the exponential tower of height $d$ (due to our lower bounds, this is unavoidable).
This is because every time we expand the $\Delta$-part, the set $P$ brings the $s$-parameters in the previous levels to the next level.
The bound we obtain for $|\mathsf{sgn}_S(G,\lambda,\mu)|$ is independent of $s_d$, because the base case is independent of $s_d$.
When further using it to bound $\lVert \mathsf{sgn}_S(G,\lambda,\mu) \rVert$, however, $s_d$ will appear in the single exponential position, as the size of a $((k_d,s_d))$-signature is single exponential in $s_d$.

This completes the overview for bounding the recursive sizes of the $S$-signatures when $t = O(1)$.
The actual proof is substantially more involved since we need to consider the dependency on $t$ as well.
But the overview already covers most of the key insights in our proof.


\subsection{Lower bounds} \label{sec-overview2}


Proof Theorem~\ref{thm:intro-LW2} is by a simple reduction from {\sc $3$-CNF SAT} (see Section~\ref{subsec:lw} for details). We explain the technical overview of proofs of Theorems~\ref{thm:intro-LW1}, \ref{thm:intro-LW3}, and \ref{thm:intro-LW4}. We give reductions from  \threecoloring\ to prove Theorems~\ref{thm:intro-LW1}, \ref{thm:intro-LW3}, and \ref{thm:intro-LW4}, when $i\geq 2$. In \threecoloring, the objective is to test whether the given graph $G$ has a proper vertex coloring using three colors. Assuming ETH, no algorithm for  \threecoloring\  runs in time  $2^{o(n)} n^{O(1)}$, where $n$ is the number of vertices in the input graph. Our reduction algorithms will take an $n$ vertex graph $G$ as input, and outputs a graph $G'$ and formula $\phi$ such that $G'\models \phi$ if and only of $G$ is $3$-colorable. In addition to the binary edge relation ${\sf adj}$, we also use a finite number of unary label predicates in our formulas. We can eliminate these label predicates, but the use of them eases our explanation. For a label predicate ${\sf L}$, we write $\exists x\in {\sf L} \; \psi$ to denote $\exists x ({\sf L}(x)\wedge \psi)$ and $\forall x \in {\sf L} \; \psi$ to denote $\forall x ({\sf L}(x)\implies \psi)$.  For all the reductions the initial part of the construction of $G'$ is a tree (let us call it base tree). So, first we  explain the construction of the base tree. Along with the base tree, we define a formula which is an FO formula except that it contains a function ${\sf id}$ which returns a non-negative integer. So, let us call such a formula as FO+{\sf id} formula. Then, for different values of $i$ in Theorems~\ref{thm:intro-LW1}, \ref{thm:intro-LW3}, and \ref{thm:intro-LW4}, we explain how to replace ${\sf id}$ with a valid FO/MSO subformula. 

\medskip
\noindent
{\bf Construction of base tree.}
Let $G$ be an instance of \threecoloring, with vertex set $V(G)=\{1,2,\ldots,n\}$ and edge set $E(G)=\{e_1,\ldots,e_m\}$. 
Let $\alpha$ be a constant selected based on the runtime of the {\sc S-FO/MSO Testing} algorithm, under the assumption—made for contradiction—that our lower bound results do not hold.
We partition $V(G)$ into $\alpha$  groups $V_1,\ldots,V_{\alpha}$ such that for each $i\in [\alpha]$, $|V_i|\leq \lceil \frac{n}{\alpha}\rceil$. Let $\ell=3^{\lceil \frac{n}{\alpha}\rceil}$. For each $i\in [\alpha]$, there are at most $\ell$ proper $3$-colorings of $G[V_i]$. Let us call these $3$-colorings  $c_{i,1},\ldots c_{i,{\ell}_i}$.

Now we construct a tree $T_1$ rooted at a node $rt$ as follows. The root $rt$ has $\alpha+m$ children and we name them $U_1,\ldots,U_{\alpha}$ and $f_1,\ldots, f_m$. That is, each node $U_i$ corresponds to the vertex subset $V_i$ of $G$ and each node $f_i$ corresponds to the edge $e_i$ of $G$. Each node $f_i$ has two children corresponding to the endpoints of $e_i$. Let us name these nodes with $f_{i,a}$ and $f_{i,b}$, 
where $a$ and $b$ are the endpoints of $e_i$. 
Now, we explain the children of each $U_i$. Recall that $\{c_{i,1},\ldots,c_{i,\ell_i}\}$ is the set of all proper $3$-colorings of $G[V_i]$. The node $U_i$ has $\ell_i$ children, and they are named $C_{i,1},\ldots,C_{i,\ell_i}$. Each $C_{i,j}$ has $|V_i|$ children, and each of them corresponds to a vertex in $V_i$. 
That is, each $a\in V_i$, $C_{i,j}$ has a child node named $v_{i,j,a}$. Now, each $v_{i,j,a}$ has two children $id_{i,j,a}$ and $c_{i,j,a}$. 
See Figure~\ref{fig:Base-treeIntro} for an illustration.     

\begin{figure}[t]
    \centering
    
\begin{tikzpicture}[scale=0.9, ver/.style = {draw, circle}]

 \tikzset{tri1/.style={isosceles triangle, draw, inner sep=0pt,
     anchor=south, shape border rotate=90, isosceles triangle stretches}}

 \tikzset{tri2/.style={isosceles triangle, draw, inner sep=0pt,
     anchor=south, shape border rotate=270, isosceles triangle stretches}}


\node[ver] (r) at (0,0) {$rt$};

\node[ver] (f1) at (-3,2) {$f_{1}$}; 
\node at (-2,2) {$\cdots$};
\node[ver] (f2) at (-0.5,2) {$f_s$};
\node at (0.5,2) {$\cdots$};
\node[ver] (fm) at (2,2) {$f_m$};

\draw[-] (f1)--(r) -- (f2);
\draw[-] (r) -- (fm);

\node[ver] (es1) at (-2,4) {$f_{s,a}$};
\node[ver] (es2) at (1,4) {$f_{s,b}$};

\draw[-] (es1)--(f2) -- (es2);

\node[ver] (1) at (-4,-2) {$U_1$}; 
\node at (-2.5,-2) {$\cdots$};
\node[ver] (i) at (-1,-2) {$U_i$};
\node at (0.5,-2) {$\cdots$};
\node[ver] (p) at (2.5,-2) {$U_{\alpha}$};
\draw[-] (1)--(r) -- (i);
\draw[-] (r) -- (p);

\node[ver] (d1) at (-3,-4) {$C_{i,1}$};
\node at (-2,-4) {$\cdots$};
\node[ver] (d2) at (-1,-4) {$C_{i,j}$};
\node at (0,-4) {$\cdots$};
\node[ver] (dl) at (1,-4) {$C_{i,\ell_i}$};

\draw[-] (d1)--(i) -- (d2);
\draw[-] (i) -- (dl);

\node[ver] (dija) at (-1,-6) {$v_{i,j,a}$};
\node at (0,-6) {$\cdots$};
\node at (-2,-6) {$\cdots$};

\draw[-] (dija)--(d2);

\node[ver] (dj1) at (-2,-8) {$id_{i,j,a}$};
\node[ver] (dj2) at (0,-8) {$c_{i,j,a}$};

\draw[-] (dj1)--(dija) -- (dj2);




\end{tikzpicture} 

    \caption{Illustration of construction of base tree $T_1$}
    \label{fig:Base-treeIntro}
\end{figure}
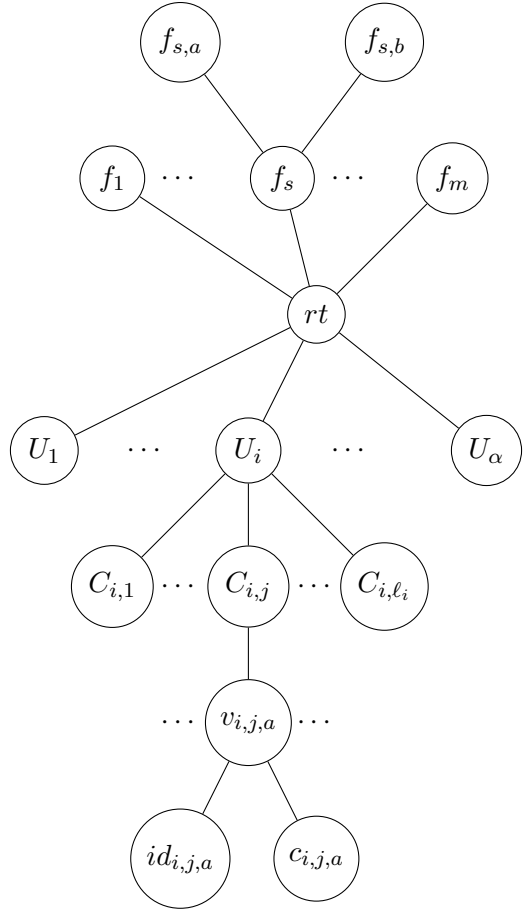

%

In the formal proof, we define nine label predicates. However, for the purpose of this overview, we present only the following three.
For each $q\in [3]$, 
$${\sf Q}_q=\{c_{i,j,a} ~:~a \textit{ is colored with $q$ in the proper coloring $c_{i,j}$ of $G[V_i]$}\}.$$
Now, we define the function ${\sf id}$ on nodes corresponding to the vertices in $G$ as follows. 

$$
  {\sf id}(x) =     \left\{ 
  \begin{array}{rcl}
  b & \mbox{if} & x=f_{s,b} \\
  a & \mbox{if} & x=id_{i,j,a}
  \end{array}\right.
$$

Now we define an FO+${\sf id}$ formula. We want to encode the statement that ``there exist nodes $C_{1,j_1}, C_{2,j_2},\ldots, C_{\alpha,j_{\alpha}}$ that correspond to proper $3$-colorings of $G[V_1], \ldots G[V_{\alpha}]$, respectively,  such that for any node $f_s$, the endpoints of the edge corresponding to $f_s$ should get different colors according to the selected proper $3$-colorings of $G[V_1], \ldots G[V_{\alpha}]$''. 
This can be encoded as follows, using variable names that match the node names for clarity. 

\begin{eqnarray}
 \psi &\equiv& \exists C_{1,j_1} \; \exists C_{2,j_2}\; \ldots \; \exists C_{\alpha,j_{\alpha}} \nonumber\\
 && \forall f_s \; \forall f_{s,a}\; \forall f_{s,b} \; \forall v_{i,j_i,a}\; \forall v_{i',j_{i'},b} \; \forall {id_{i,j_i,a}} \forall {id_{i',j_{i'},b}}  \; 
 \forall {c_{i,j_i,a}}, {c_{i',j_{i'},b}}   \in {\sf Q}
 _1\cup {\sf Q}_2 \cup {\sf Q}_3 \nonumber\\
 && \psi_{{\sf evalid}} \wedge( \psi_{{\sf uvalid}} \Rightarrow ( \psi_{{\sf id}}\Rightarrow \psi_{{\sf color}} ))\nonumber\\
 &\equiv& \exists C_{1,j_1} \; \exists C_{2,j_2}\; \ldots \; \exists C_{\alpha,j_{\alpha}} \nonumber\\
 && \forall f_s \; \forall f_{s,a}\; \forall f_{s,b} \; \forall v_{i,j_i,a}\; \forall v_{i',j_{i'},b} \; \forall {id_{i,j_i,a}} \forall {id_{i',j_{i'},b}}  \; 
 \forall {c_{i,j_i,a}}, {c_{i',j_{i'},b}}   \in {\sf Q}
 _1\cup {\sf Q}_2 \cup {\sf Q}_3 \nonumber\\
 && \psi_{{\sf evalid}} \wedge ( \neg \psi_{{\sf uvalid}} \vee \neg\psi_{{\sf id}}\vee \psi_{{\sf color}}) \label{eqn:intro:psi} 
\end{eqnarray}

We explain the meaning of each subformulas in $\psi$ below. 

\begin{itemize}
    \item $\psi_{{\sf evalid}}$ ensures that  $C_{1,j_1}, C_{2,j_2},\ldots, C_{\alpha,j_{\alpha}}$ are nodes in $T_1$ that correspond to proper $3$-colorings of $G[V_1], \ldots G[V_{\alpha}]$, respectively. 
    \item $\psi_{{\sf uvalid}}$ ensures that all the variables with universal quantifiers are selected appropriately. For example, $f_s$ corresponds to an edge in $G$, $f_{s,a}$ and $f_{s,b}$ are children of $f_s$. Also, $v_{i,j_i,a}$ and $v_{i,j_{i'},b}$ are nodes corresponding to vertices in $G$ and each connected to a vertex in $\{C_{1,j_1}, C_{2,j_2},\ldots, C_{\alpha,j_{\alpha}}\}$. Moreover, $id_{i,j_i,a}$ and $c_{i,j_i,a}$ are children of $v_{i,j_i,a}$ and $id_{i',j_{i'},b}$ and $c_{i',j_{i'},b}$ are children of $v_{i',j_{i'},b}$
    
    \item $\psi_{\sf id}$ is true if and only if ${\sf id}(f_{s,a})={\sf id}(id_{i,j_i,a})$ and ${\sf id}(f_{s,b})={\sf id}(id_{i',j_{i'},b})$. 
    \item $\psi_{\sf color}$ is true if and only if $c_{i,j_i,a}$ and $c_{i,j_{i'},b}$ are different colors. 
\end{itemize}

One can prove that $T_1\models \psi$ if and only of $G$ is $3$-colorable. But, $\psi$ is a $(\alpha,9)$-FO+{\sf id} formula. We design various methods to formulate $\psi_{\sf id}$. 

\medskip
\noindent
{\bf Encoding ids.} The formula $\psi$ constructed above is an $(\alpha,9)$-FO+{\sf id} formula that contains a subformula $\neg \psi_{\sf id}$. Here,  

\begin{eqnarray*}
\neg \psi_{\sf id}&\equiv& \neg( {\sf id}(id_{i,j_i,a})={\sf id}(f_{s,a}) \wedge {\sf id}(id_{i',j_{i'},b})={\sf id}(f_{s,b}))\\
&\equiv& {\sf id}(z_1)\neq{\sf id}(y_1) \vee {\sf id}(z_2)\neq {\sf id}(y_2),
\end{eqnarray*}
where we substituted $z_1=id_{i,j_i,a}, y_1=f_{s,a}, z_2=id_{i',j_{i'},b}$, and $y_2=f_{s,b}$, for convenience. We encode (in)equality of ids using FO/MSO formulas leading to different cases of lower bound results

\medskip
\noindent
{\bf $\log n$-length FO identifier test.}
Notice that, since $V(G)=[n]$, ${\sf id}(x)$ for any node $x\in V(T_1)$, belongs to $[n]$ (if ${\sf id}(x)$ is defined on $x$). 
We need an FO/MSO formula to represent ${\sf id}(z)={\sf id}(y)$ for two variables $z$ and $y$. Towards that we create two label predicates ${\bf 1}$ and ${\bf 0}$. 
Let $k$ be the smallest integer such that $k\geq \log n$. 
For any node $x$ in $T_1$ such that ${\sf id}(x)$ is defined, we do the following. 
Let $b_1,\ldots,b_{k}$ be the binary representation of the number ${\sf id}(x)$. Let $\pi_x$ be a path $a_1,\ldots,a_{k},x$ on $k+1$ vertices. 
Now, we replace node $x$ with path $\pi_x$.  For each $i\in [k]$, $a_i\in {\bf 0}$ if $b_i=0$ and $a_i\in {\bf 1}$ otherwise. The tree constructed as explained above is $T_2$. 
Now, for two nodes $z$ and $y$ in $T$, ${\sf id}(z)={\sf id}(y)$ can be encoded as 

\begin{eqnarray}
\exists a_1,a_1'\in {\bf O}\cup {\bf 1} \ldots \exists a_{k},a_{k}' \in {\bf O}\cup {\bf 1}
&&
\left(\bigwedge_{i\in [k] }
(a_i \in {\bf 0} \Leftrightarrow a'_i \in {\bf 0})\right) \nonumber \\
&& \wedge  {\sf path}(a_1,\ldots,a_{k},z)
\wedge
{\sf path}(a'_1,\ldots,a'_{k},y) \label{enq:intro:LB1}
\end{eqnarray}

where, ${\sf path}(w_1,\ldots,w_q) \equiv (\bigwedge_{i\in [q-1]} {\sf adj}(w_i,w_{i+1}))$. Notice that the number of quantifiers in the above formula is $2\log n$. Now, by substituting (\ref{enq:intro:LB1}) in $\psi$ we get an $(\alpha,9+2\log n)$-FO formula. Thus, any algorithm of running time $\exp^{(2)}(o(k_2))$ for testing $T_2\models \psi$, where $\psi$ is an $(O(1),O(\log n))$-FO formula leads to a $2^{o(n)}$ time algorithm for \threecoloring, a contradiction to ETH. This is a proof overview for Theorem~\ref{thm:intro-LW1} for $i=2$ and $s_j=0$ for all $j$. 

Next we explain how to  use set variables to reduce the length of the path $k$ above as follows. 
Let $k$ and $s$ be two positive integers such that $k\cdot s\geq \log n$. Now  let $b_1,\ldots,b_{k}$ be the base $2^{s}$ representation of the number ${\sf id}(x)$. As before, we have a path $\pi_x=a_1,\ldots,a_k,x$ of length $k+1$ that replaces the node $x$. 
Then, define ${\sf id}(a_{j'})=b_{j'}$ for all $j'$. 

Since $2^{k \cdot s} \geq  n$, each number in $[n]$ can be uniquely represented as above. Now,  we use $s$ set variables $W_1,\ldots,W_s$ to encode ids. Let  $b''_1,\ldots,b''_s$ be the binary representation of ${\sf id}(a_r)$ and $b'_1,\ldots,b'_s$ be the binary representation of ${\sf id}(a'_r)$. Suppose we are able to force the set variables in such a way that for all $j'\in [s]$, $a_{r}\in W_{j'}$ if and only if $b''_{j'}=1$ and $a'_{r}\in W_{j'}$ if and only if $b'_j=1$. Under that condition, ${\sf id}(a_r) = {\sf id}(a'_{r})$ can be encoded as 

\begin{equation*}
 \bigwedge_{j'\in [s]} (a_r\in W_{j'}) \Leftrightarrow (a'_r\in W_{j'})
\end{equation*}

To satisfy the additional condition mentioned above, we create a formula $\psi_{\sf set}(W_1,\ldots W_s)$ such that $\psi_{\sf set}(W_1,\ldots W_s)$ is true if and only if the following is true. For any node $a$, let $b_1,\ldots,b_s$ be the binary representation of ${\sf id}(a)$. Then, the formula $\psi_{\sf set}(W_1,\ldots W_s)$ is true if and only if for any vertex $a$ and ${j'}\in [s]$, $a\in W_{j'}$ only when $b_{j'}=1$. This formula is a $O(1)$-MSO formula with the only set variables are the free variables $W_1,\ldots,W_s$, and all the vertex variables are quantified with universal quantifiers. See Theorem~\ref{thm:low-MSO:s>0:i=2} for more details.  
Then, ${\sf id}(z)={\sf id}(y)$ can be encoded as 

\begin{eqnarray}
\exists W_1, \ldots, \exists W_s, \exists a_1,a_1' \ldots \exists a_{k},a_{k}' 
&&
\left(\bigwedge_{r\in [k] }
\bigwedge_{j'\in [s]} (a_r\in W_{j'}) \Leftrightarrow (a'_r\in W_{j'})
\right) \wedge \psi_{\sf set}(W_1,\ldots W_s)
\nonumber \\
&& \wedge  {\sf path}(a_1,\ldots,a_{k},z)
\wedge
{\sf path}(a'_1,\ldots,a'_{k},y) \label{enq:intro:LB2}
\end{eqnarray}

The new id test leads to Theorem~\ref{thm:intro-LW1} for $i=2$.

\medskip
\noindent
{\bf $\log\log n$-length FO identifier test.}
Recall that we need an FO/MSO formula to represent ${\sf id}(z)={\sf id}(y)$ for two variables $z$ and $y$ where ${\sf id}(z),{\sf id}(y)\in [n]$. 
Let $k$ and $k'$ be the smallest integers such that $k\geq \log n$ and $k'\geq \log \log n$. 
For any node $t \in V(T_1)$ such that ${\sf id}(t)$ is defined, let $b_1\ldots b_{k}$ be the binary representation of the number ${\sf id}(t)$. 
Now we replace $t$ with a subtree as shown in Figure~\ref{fig:treeid1:Intro}. 
Here, each $a_{j'}$ represents $b_{j'}$ in the following way. We set ${\sf id}(d_{j'})=j'$, $c_{j'}\in {\bf O}$ if $b_{j'}=0$ and $c_{j'}\in {\bf 1}$ if $b_{j'}=1$. Here, ${\bf O}$ and ${\bf 1}$ are label predicates to represent whether the bits corresponding to the vertices is 0 and 1, respectively.  In other words, ${\sf id}(d_{j'})$ represents the position of $b_{j'}$ in the binary representation $b_1\ldots b_{k}$ and $c_{j'}$ denotes the value of the bit $b_{j'}$. The crucial observation is that for each $j'\in [k]$, ${\sf id}(d_{j'})$ is a positive integer less than or equal to $\lceil \log n \rceil$. 

\begin{figure}
    \centering
\begin{tikzpicture}[
  grow=down,
  level 1/.style={sibling distance=3cm},
  level 2/.style={sibling distance=2cm},
  every node/.style={circle,draw},
  edge from parent/.style={draw,-}
]

\node {$t$}
  child {node {$a_1$}
    child {node {$d_1$}}
    child {node {$c_1$}}
  }
  child {node {$a_2$}
    child {node {$d_2$}}
    child {node {$c_2$}}
  }
  child {node[draw=none] {$\ldots$} edge from parent[draw=none]
    child {node[draw=none] {$\ldots$} edge from parent[draw=none]}
  }
    child {node {$a_k$}
    child {node {$d_k$}}
    child {node {$c_k$}}
  };

\end{tikzpicture}
  \caption{Subtree to replace the node $t$}
\label{fig:treeid1:Intro}
\end{figure}
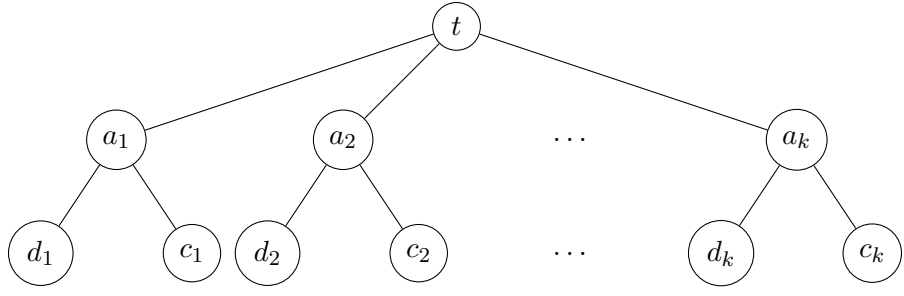

Informally, for two nodes $z$ and $y$ in $T_1$, ${\sf id}(z)={\sf id}(y)$ is true if and only if for any child $a$ of $z$ and any child $a'$ of $y$ if the {\em id of the left child of} $a$ and the {\em id of the left child of} $a'$ are equal, then the {\em corresponding bits (encoded in the right child of $a$ and $a'$)} are same. This can be encoded as follows.    

\begin{eqnarray}
{\sf id}(z)={\sf id}(y) &\equiv& \forall a,a', \forall d,d' \forall c,c' \;\;
\phi_{{\sf valid}} \Rightarrow  ( \phi_{{\sf id}}\Rightarrow \phi_{{\sf bit}} )\nonumber\\
&\equiv& \forall a,a', \forall d,d' \forall c,c' \;\;
(\neg \phi_{{\sf valid}} \vee  \neg \phi_{{\sf id}}\vee \phi_{{\sf bit}} )\label{eqn:uniloglog:Intro}
\end{eqnarray}

Here, $\phi_{{\sf valid}}$, $\phi_{{\sf id}}$ and $\phi_{{\sf bit}}$ are defined below.

\begin{eqnarray*}
\phi_{\sf id} &\equiv& {\sf id}(d)={\sf id}(d')\\
\phi_{\sf bit} &\equiv& c \in {\bf 0} \Leftrightarrow c'\in {\bf 0} \\
\phi_{\sf valid} &\equiv& {\sf adj}(a,z) \wedge {\sf adj}(a',y) \wedge {\sf adj}(d,a) \wedge {\sf adj}(c,a) \wedge {\sf adj}(d',a') \wedge {\sf adj}(c',a')    
\end{eqnarray*}

The formula $\phi_{\sf valid}$ is true if and only if $a$ and $a'$ are children of $z$ and $y$, respectively, $d$ and $c$ are children of $a$, and $d'$ and $c'$ are children of $a'$. The formula $\phi_{\sf bit}$ is true if and only if they both {\em encode} the same bit. Clearly, the formula in (\ref{eqn:uniloglog:Intro}) is not an FO formula, but an FO+${\sf id}$ formula with the value of ${\sf id}(d)$ and ${\sf id}(d')$ are positive integers less than or equal to $\lceil \log n \rceil$. So, we apply the $\log n'$-FO identifier test where $n'=\log n$ and get an $(O(\log\log n))$-FO formula to represent ${\sf id}(d)={\sf id}(d')$ where all the quantifiers are existential. By substituting $(O(\log\log n))$-FO formula for ${\sf id}(d)={\sf id}(d')$ in (\ref{eqn:uniloglog:Intro}) we get a 
$(O(\log\log n))$-FO formula where all quantifiers are universal, because of the negation symbol before $\phi_{\sf id}$.   
Now, if we use this formula to test equality of two identifiers in $\psi$, we get Theorem~\ref{thm:intro-LW1} for  $i=3$. 

We would like to mention that if we apply the same strategy recursively, we can prove  Theorem~\ref{thm:intro-LW1} for  $i>3$. 
In other words, we design a $(k_1,\ldots,k_{i-1})$-FO+{\sf id} formula, where ${\sf id}(x)\leq \log^{(i-2)}n$ and $k_r=O(1)$ for all $r\in [i-1]$. Finally, we apply the $(\log \log n')$-FO/MSO identifier test to get a required formula as the output of the reduction algorithm, where $n'\leq \log^{(i-2)}n$. For the case of $i=1$, we give a simple reduction from {\sc $3$-CNF SAT} (see Section~\ref{subsec:reductions} for details).

\medskip
\noindent
{\bf Proof overview of Theorems~\ref{thm:intro-LW3} and \ref{thm:intro-LW4}.}
Let us discuss the case when $i=2$. For the case when $i>2$, the approach is similar to the case of Theorem~\ref{thm:intro-LW1} along with the ideas used below for $i=2$ here. First,  we construct the base tree $T_1$ and an $(\alpha,9)$-FO+{\sf id} formula $\psi$ as mentioned in (\ref{eqn:intro:psi}). Recall the formula (\ref{enq:intro:LB2}) created for id test. Here, we have $s$ set variables $W_1,\ldots,W_s$ and for any node $a_r$ its id is denoted using inclusion of it in the sets $W_1,\ldots,W_s$. Here, notice that ${\sf id}(a_r)\in [2^s-1]$ and $2^{s\cdot k}\geq n$. In Theorem~\ref{thm:intro-LW3}, instead  of using set variables, we add $s$ new nodes $\{w_1,\ldots,w_s\}$.  Recall the role of set variables.  For two node $a_r$ and $a'_r$ ${\sf id}(a_r)={\sf id}(a'_r)$ is encoded by 

$$\bigwedge_{i'\in [s]} (a_r\in W_{i'}) \Leftrightarrow (a'_r\in W_{i'}).$$

Now, to get rid of set variables, we add edges between $\{w_1,\ldots,w_s\}$ and nodes in $T$ for which ${\sf id}$ is defined as follows. Let $x$ be a node and $b_1,\ldots,b_s$ be the binary representation of ${\sf id}(x)$. 
Then, $x$ is adjacent to $t_r$ if and only if $b_r=1$.  Clearly, the treewidth of the new graph is $s$. Then, we can replace (\ref{enq:intro:LB2}) with the following formula. 

\begin{eqnarray}
\exists a_1,a_1' \ldots \exists a_{k},a_{k}' 
&&
\left(\bigwedge_{r\in [k] }
\bigwedge_{i'\in [s]} ({\sf adj}(a_r,w_{i'})) \Leftrightarrow ({\sf adj}(a'_r,w_{i'})
\right) 
\nonumber \\
&& \wedge  {\sf path}(a_1,\ldots,a_{k},z)
\wedge
{\sf path}(a'_1,\ldots,a'_{k},y) \label{enq:intro:LB3}
\end{eqnarray}

Now, by substituting (\ref{enq:intro:LB3}) in $\psi$, we get Theorem~\ref{thm:intro-LW3} for $i=2$.  
Next, we explain the idea for Theorems~\ref{thm:intro-LW4}. In this case we add $t=2^s$ vertices to get rid of set variables. Let $w_0,\ldots,w_{t-1}$ be the new vertices added. Then, a node  $x$ is adjacent to $w_r$ if and only if ${\sf id}(x)=r$.  
Clearly, the treewidth of the new graph is at most $t$. 
Then, we can replace (\ref{enq:intro:LB2}) with the following formula. 

\begin{eqnarray}
\exists a_1,a_1' \ldots \exists a_{k},a_{k}' 
&&
\left(\bigwedge_{r\in [k] }
\bigwedge_{i'\in [t]} ({\sf adj}(a_r,w_{i'})) \Leftrightarrow ({\sf adj}(a'_r,w_{i'})
\right) 
\nonumber \\
&& \wedge  {\sf path}(a_1,\ldots,a_{k},z)
\wedge
{\sf path}(a'_1,\ldots,a'_{k},y) \label{enq:intro:LB4}
\end{eqnarray}

Now, by substituting (\ref{enq:intro:LB4}) in $\psi$, we get Theorem~\ref{thm:intro-LW4} for $i=2$. 

\section{Preliminaries} \label{sec-pre}
\paragraph{Basic notations.}
We write $\mathbb{N} = \{1,2,\dots\}$ as the set of natural numbers, and write $\mathbb{N}_0 = \mathbb{N} \cup \{0\}$.
For a number $n \in \mathbb{N}$, we write $[n] = \{1,\dots,n\}$ and $[n]_0 = \{0,1,\dots,n\}$.
Let $f:A \rightarrow B$ be a function.
We denote by $\text{Im}(f) = \{b \in B: b = f(a) \text{ for some } a \in A\}$ the \textit{image} of $f$.
For a subset $A' \subseteq A$, $f_{|A'}: A' \rightarrow B$ is the restriction of $f$ to $A'$.
For a subset $B' \subseteq B$, we define $f^{|B'}: f^{-1}(B') \rightarrow B'$ simply as $f^{|B'}(a) = f(a)$ for all $a \in f^{-1}(B')$.
Let $G$ be a graph.
The notation $V(G)$ denotes the set of vertices of $G$ and the notation $E(G)$ denotes the set of edges of $G$.
For $V \subseteq V(G)$, we use $N_G(V)$ to denote the set of vertices in $V(G) \backslash V$ that are neighboring to $V$, and define $N_G[V] = V \cup N_G(V)$.
A \textit{separation} of a graph $G$ is a pair $(A,B)$ with $A,B \subseteq V(G)$ such that $A \cup B = V(G)$ and there is no edge between $A \backslash B$ and $B \backslash A$ in $G$.

\paragraph{Tree decompostion and treewidth.}
A \textit{tree decomposition} of a graph $G$ is a pair $(T,\beta)$ where $T$ is a tree and $\beta: T \rightarrow 2^{V(G)}$ maps each node $x \in T$ to a set $\beta(x) \subseteq V(G)$ called the \textit{bag} of $x$ such that
\textbf{(i)} $\bigcup_{x \in T} \beta(x) = V(G)$, \textbf{(ii)} for any $v \in V(G)$, the nodes $x \in V(T)$ with $v \in \beta(x)$ induce a subtree in $T$, and \textbf{(iii)} for any edge $(u,v) \in E(G)$, there exists $x \in V(T)$ with $u,v \in \beta(x)$.
The \textit{width} of $(T,\beta)$ is $\max_{x \in T} |\beta(x)| - 1$.
The \textit{treewidth} of a graph $G$, denoted by $\mathbf{tw}(G)$, is the minimum width of a tree decomposition of $G$.
A tree decomposition $(T,\beta)$ is \textit{nice} if $T$ is a rooted binary tree and the following hold: (i) for each node $x \in V(T)$ with two children $y_1$ and $y_2$, $\beta(x) = \beta(y_1) \cup \beta(y_2)$ (in this case $x$ is called a \textit{join} node), (ii) for each node $x \in V(T)$ with one child $y$, either $\beta(x) = \beta(y) \cup \{v\}$ for some $v \in V(G) \backslash \beta(y)$ (in this case $x$ is called an \textit{introduce} node) or $\beta(x) = \beta(y) \backslash \{v\}$ for some $v \in \beta(y)$ (in this case $x$ is called a \textit{forget} node).
\begin{lemma}[\cite{CyganFKLMPPS15}]
    Every graph $G$ admits a nice tree decomposition of width $\mathbf{tw}(G)$.
\end{lemma}

\begin{lemma}[\cite{fomin2012planar,lokshtanov2018beating}] \label{lem-decomp}
    Let $t \in \mathbb{N}$ and $G$ be a graph with $\mathbf{tw}(G) \leq t$.
    Then for any $R \subseteq V(G)$ with $|R| \leq r$, there exist $V_0,V_1,\dots,V_{6r} \subseteq V(G)$ satisfying the following conditions:
    \begin{enumerate}[(i)]
        \item $V(G) = \bigcup_{i=0}^{6r} V_i$,
        \item $R \subseteq V_0$ and $|V_0| \leq O(rt)$,
        \item $N_G(V_i \backslash V_0) \subseteq V_0 \cap V_i$ and $|V_0 \cap V_i| = 4t$ for all $i \in [6r]$.
    \end{enumerate}    
\end{lemma}
\begin{proof}
Consider a nice tree decomposition $(T,\beta)$ of $G$ with width $t$.
For a subset $X \subseteq V(T)$, we write $\beta(X) = \bigcup_{x \in X} \beta(x)$ for convenience.
For each vertex $v \in R$, we pick an arbitrary node $x_v \in V(T)$ such that $v \in \beta(x_v)$.
Let $X \subseteq V(T)$ be the minimal subset such that $x_v \in X$ for all $v \in R$ and $T[X]$ is connected, or equivalently, $X$ consists of the nodes on the paths in $T$ between vertices in $\{x_v: v \in R\}$.
By construction, $T[X]$ is a tree whose leaves are in $\{x_v: v \in R\}$.
Define $X^* \subseteq X$ as the subset consisting of the nodes in $\{x_v: v \in R\}$ and all nodes with degree at least $3$ in $T[X]$.
Since the number of leaves of $T[X]$ is at most $|R| \leq r$, we have $|X^*| \leq 2r$.
As $(T,\beta)$ is a nice tree decomposition, the maximum degree of $T$ is $3$ and thus $T - X^*$ has at most $3|X^*| \leq 6r$ connected components.
Let $C_1,\dots,C_m$ be the connected components of $T - X^*$.
We observe that each $C_i$ is adjacent to at most two nodes in $X^*$.
Indeed, if $C_i$ is adjacent to three nodes $x_1,x_2,x_3 \in X^*$, then there exists $x \in C_i$ such that $x_1,x_2,x_3$ belong to different connected components of $T-x$, which implies $x \in X$.
It also follows that $x \in X^*$, as the degree of $x$ in $T[X]$ is at least $3$ because of $x_1,x_2,x_3$.
This contradicts the fact that $x \in C_i$.
Thus, there are at most two nodes in $X^*$ adjacent to $C_i$ in $T$.

Define $V_0 = \beta(X^*)$, $V_i = N_G[\beta(C_i) \backslash \beta(X^*)]$ for $i \in [m]$, and $V_i = \emptyset$ for $i \in \{m+1,\dots,6r\}$.
Clearly, $R \subseteq V_0$ and we have $|V_0| \leq t|X^*| = O(rt)$.
The fact $V(G) = \bigcup_{i=0}^{6r} V_i$ follows from that $V(G) = \beta(X^*) \cup (\bigcup_{i=1}^m \beta(C_i))$.
To see the other properties, consider an index $i \in [6r]$.
If $i > m$, then $N_G(V_i \backslash V_0) = \emptyset$ and $V_0 \cap V_i = \emptyset$.
Otherwise, $V_i = N_G[\beta(C_i) \backslash \beta(X^*)]$.
Let $U_i = \beta(C_i) \backslash \beta(X^*)$.
Also, let $X_i^* \subseteq X^*$ consist of the (at most two) nodes in $X^*$ adjacent to $C_i$ in $T$.
We observe that $N_G(U_i) \subseteq \beta(X_i^*) \subseteq V_0$.
Consider a vertex $v \in N_G(U_i)$ and let $u \in U_i$ be a neighbor of $v$.
Then there exists a node $x \in V(T)$ with $u,v \in \beta(x)$.
If $x \in V(T) \backslash C_i$, then $u \in \beta(X_i^*)$, because in this case $u \in \beta(C_i) \cap \beta(V(T) \backslash C_i)$ and the property of a tree decomposition guarantees that the bag of some vertex in $X_i^*$ contains $u$.
This contradicts the fact that $u \in U_i$.
Thus, $x \in C_i$.
Since $v \notin U_i$, there exists a node $y \in V(T) \backslash C_i$ with $v \in \beta(y)$.
Note that the path between $x$ and $y$ in $T$ intersects $X_i^*$.
This implies $v \in \beta(X_i^*)$, as $v \in \beta(x)$ and $v \in \beta(y)$.
As such, $N_G(U_i) \subseteq \beta(X_i^*) \subseteq V_0$.
By construction, $N_G(U_i) \subseteq V_i$ and hence $N_G(U_i) \subseteq V_0 \cap V_i$.
Furthermore, as $V_0 \cap U_i = \emptyset$, we have $V_0 \cap V_i = V_0 \cap N_G[U_i] = V_0 \cap N_G(U_i) = N_G(U_i)$, which implies 
\begin{equation*}
    |V_0 \cap V_i| = |N_G(U_i)| \leq |\beta(X_i^*)| \leq (t+1)\cdot |X_i^*| \leq 2(t+1) \leq 4t.
\end{equation*}
To further have $|V_0 \cap V_i| = 4t$, we can arbitrarily add some vertices in $V_0$ to $V_i$ (note that by doing this, the other properties preserve).
\end{proof}

\paragraph{FO and MSO logic on graphs.}
An \textit{MSO formula} is a quantified logic formula with vertex variables (representing vertices of a graph) and set variables (representing sets of vertices of a graph).
It is built by combining three types of {\em atomic formulas}: (i) equality for vertex variables $u,v$, denoted by $u = v$, which is true iff $u$ and $v$ represent the same vertex, and (ii) adjacency for vertex variables $u,v$, denoted $\mathsf{adj}(u,v)$, which is true if $u$ and $v$ represent two vertices that are adjacent, and (iii) membership of a vertex variable $v$ in a set variable $V$, denoted by $v \in V$, which is true if $v$ is contained in $V$.
Atomic formulas are combined using $\wedge, \vee, \neg$ and quantifiers $\exists,\forall$.
An \textit{FO formula} is just an MSO formula without set variables and type (iii) atomic formulas.

For a sequence $((k_1,s_1),\dots,(k_d,s_d))$ of pairs of natural numbers, we define a class of MSO formulas, called \textit{$((k_1,s_1),\dots,(k_d,s_d))$-MSO}, inductively on $d$ as follows.
An MSO formula is a \textit{$()$-MSO} if it is quantifier-free.
Suppose now $((k_2,s_2),\dots,(k_d,s_d))$-MSO has been defined.
Then $((k_1,s_1),\dots,(k_d,s_d))$-MSO is defined as the MSO formulas which can be generated using the following rules.
\begin{itemize}
    \item If $\psi(u_1,\dots,u_{c+k+s})$ is a $((k_2,s_2),\dots,(k_d,s_d))$-MSO with $c+k+s$ free variables for some $c \in \mathbb{N}$, $k \in [k_1]_0$, and $s \in [s_1]_0$, where $u_{c+1},\dots,u_{c+k}$ are vertex variables and $u_{c+k+1},\dots,u_{c+k+s}$ are set variables, then for $\mathsf{Q} \in \{\exists,\forall\}$, 
    \begin{equation*}
        \phi(v_1,\dots,v_c) = \mathsf{Q} x_1 \dots \mathsf{Q} x_k\  \mathsf{Q} Y_1 \dots \mathsf{Q} Y_s \ \psi(v_1,\dots,v_c,x_1,\dots,x_k,Y_1,\dots,Y_s)
    \end{equation*}
    is a $((k_1,s_1),\dots,(k_d,s_d))$-MSO.
    \item An MSO formula obtained by combining $((k_1,s_1),\dots,(k_d,s_d))$-MSOs using logic operators $\neg,\wedge,\vee$ is a $((k_1,s_1),\dots,(k_d,s_d))$-MSO.
\end{itemize}
For $S = ((k_1,s_1),\dots,(k_d,s_d))$, an \textit{$S$-MSO$^*$} is defined as a $((k_2,s_2),\dots,(k_d,s_d))$-MSO with $k_1$ free vertex variables and $s_1$ free set variables.
For a sequence $S = (k_1,\dots,k_d)$ of natural numbers, an \textit{$S$-FO} is defined as an $S^+$-MSO for $S^+ = ((k_1,0),\dots,(k_d,0))$, which is an FO formula

For a sequence $(k_1,\dots,k_d)$ of natural numbers, we define a class of FO formulas, called \textit{$(k_1,\dots,k_d)$-FOs}, inductively on $d$ as follows.
An FO formula is a \textit{$()$-FO} if it is quantifier-free.
Suppose now $(k_2,\dots,k_d)$-FOs have been defined.
Then $(k_1,\dots,k_d)$-FOs are the FO formulas which can be generated using the following rules.
\begin{itemize}
    \item If $\psi(u_1,\dots,u_{c+k})$ is a $(k_2,\dots,k_r)$-FO with $c+k$ free variables for some $c \in \mathbb{N}$ and $k \in [k_1]_0$, then for $\mathsf{Q} \in \{\exists,\forall\}$, 
    \begin{equation*}
        \phi(v_1,\dots,v_c) = \mathsf{Q} x_1 \dots \mathsf{Q} x_k\ \psi(v_1,\dots,v_c,x_1,\dots,x_k)
    \end{equation*}
    is a $(k_1,\dots,k_r)$-FO.
    \item An FO formula obtained by combining $(k_1,\dots,k_d)$-FOs using logic operators $\neg,\wedge,\vee$ is a $(k_1,\dots,k_d)$-FO.
\end{itemize}

Next, we define a similar notion for MSO formulas.

\paragraph{Labeled and colored graphs.}
Let $G$ be a graph.
For a number $p \in \mathbb{N}_0$, a \textit{$p$-labeling function} (or \textit{$p$-labeling} for short) on $G$ is a function $\lambda: [p] \rightarrow V(G)$.
If $\lambda: [p] \rightarrow V(G)$ is a $p$-labeling on $G$ and $\lambda': [p'] \rightarrow V(G)$ is a $p'$-labeling on $G$, then we define a function $\lambda \oplus \lambda': [p+p']  \rightarrow V(G)$ as
\begin{equation*}
    (\lambda \oplus \lambda')(i) = \left\{
    \begin{array}{ll}
        \lambda(i) & \text{if } i \leq p, \\
        \lambda'(i-p) & \text{if } i>p.
    \end{array}
    \right.
\end{equation*}
Note that $\lambda \oplus \lambda'$ is a $(p+p')$-labeling on $G$.
For a function $f: X \rightarrow V(G)$ where $X \subseteq \mathbb{Z}$, we construct a $|X|$-labeling on $G$, $\mathsf{sort}(f): [|X|] \rightarrow V(G)$, as follows.
Suppose $X = \{a_1,\dots,a_{|X|}\}$ where $a_1 < \cdots < a_{|X|}$.
Then simply define $\mathsf{sort}(f)(i) = f(a_i)$ for all $i \in [|X|]$.

For a set $P$, a \textit{$P$-coloring} on $G$ is a function $\mu: V(G) \rightarrow P$.
If $\mu: V(G) \rightarrow P$ is a $P$-coloring on $G$ and $\mu': V(G) \rightarrow P'$ is a $P'$-coloring on $G$, then we define a function $\mu \otimes \mu': V(G) \rightarrow P \times P'$ as $(\mu \otimes \mu')(v) = (\mu(v),\mu'(v))$, which is a $(P \times P')$-coloring on $G$.

A \textit{$p$-labeled and $P$-colored graph} is a triple $(G,\lambda,\mu)$ where $G$ is a graph, $\lambda$ is a $p$-labeling on $G$, and $\mu$ is a $P$-coloring on $G$.
Every graph $G$ can be viewed as a labeled and colored graph by equipping it with the dummy labeling $-: \emptyset \rightarrow V(G)$ and the dummy coloring $-: V(G) \rightarrow \{0\}$ that maps every vertex $v \in V(G)$ to $0$.
We always use $-$ to denote the dummy labeling/coloring.
Two $p$-labeled and $P$-colored graphs $(G,\lambda,\mu)$ and $(G',\lambda',\mu')$ are \textit{isomorphic} if there exists an isomorphism $\pi: V(G) \rightarrow V(G')$ of $G$ and $G'$ such that $\lambda' = \pi \circ \lambda$ and $\mu = \mu' \circ \pi$.
Throughout this paper, isomorphic labeled and colored graphs are viewed as the \textit{same} (or in other words, we only care about the isomorphic type of such graphs).
In particular, if $\mathcal{G}$ is a set of $p$-labeled and $P$-colored graphs, then different elements in $\mathcal{G}$ are always non-isomorphic.
The \textit{core} of a $p$-labeled and $P$-colored graph $(G,\lambda,\mu)$, denoted by $\mathsf{core} (G,\lambda,\mu)$, is another $p$-labeled and $P$-colored graph obtained by restricting $(G,\lambda,\mu)$ to $\text{Im}(\lambda)$, i.e.,
$\mathsf{core} (G,\lambda,\mu) = (G[\text{Im}(\lambda)],\lambda_0,\mu_{|\text{Im}(\lambda)})$, where $\lambda_0: [p] \rightarrow \text{Im}(\lambda)$ is simply defined as $\lambda_0(v) = \lambda(v)$ for all $v \in \text{Im}(\lambda)$.

\begin{fact} \label{fact-numLC}
    Up to isomorphism, the number of $p$-labeled and $P$-colored graphs $(G,\lambda,\mu)$ with $|V(G)| \leq n$ and $\mathbf{tw}(G) \leq t$ is at most $n^{O(\min\{n,t\} \cdot n + p)} \cdot |P|^{n}$.
\end{fact}
\begin{proof}
If $G$ is an $n$-vertex (unlabeled) graph of $n$ vertices and treewidth at most $t$, then $|E(G)| \leq tn$.
Thus, the number of $n$-vertex (unlabeled) graphs with treewidth at most $t$ is bounded by $n^{2tn}$ up to isomorphism.
Furthermore, there can be $n^c$ different $c$-labelings for an $n$-vertex graph.
So there are at most $n^{2tn} \cdot n^c = n^{2tn+c}$ different $n$-vertex $c$-labeled graphs with treewidth at most $t$.
Considering the graphs with $i$ vertices for $i = 1,\dots,n$, we finally obtain the desired bound $n^{2tn+c+1}$.
\end{proof}

\begin{fact} \label{fact-num}
    Let $\sigma$ be a $p$-labeled and $P$-colored graph, and $q,r,t \in \mathbb{N}_0$.
    Then up to isomorphism, the number of $(p+q)$-labeled and $P$-colored graphs $(G,\lambda,\mu)$ with treewidth at most $t$ satisfying $|V(G) \backslash \textnormal{Im}(\lambda_{|[p]})| \leq r$ and $\mathsf{core}(G,\lambda_{|[p]},\mu) = \sigma$ is at most $(p+r)^{O(\min\{r,t\} \cdot (p+r)+q)} \cdot |P|^{O(r)}$.
\end{fact}
\begin{proof}
If $\mathsf{core}(\sigma) \neq \sigma$, then there is no graph $(G,\lambda,\mu)$ satisfying the condition $\mathsf{core}(G,\lambda_{|[p]},\mu) = \sigma$.
As such, assume $\mathsf{core}(\sigma) \neq \sigma$, which implies that $\sigma$ contains at most $p$ vertices.
Let $\sigma = (H,\gamma,\tau)$.
A graph $(G,\lambda,\mu)$ under consideration can be obtained from $\sigma$ by adding a set $V'$ of (at most) $r$ vertices, adding a set $E$ of edges between the vertices in $V(H)$ and the vertices in $V'$, adding a set $E'$ of edges among the vertices in $V'$, and specifying two functions $\lambda': [q] \rightarrow V(H) \cup V'$ and $\mu': V' \rightarrow P$.
Then we simply set $V(G) = V(G) \cup V'$, $E(G) = E(H) \cup E \cup E'$, $\lambda = \gamma \oplus \lambda'$, and
\begin{equation*}
    \mu(v) = \left\{
    \begin{array}{ll}
        \tau(v) & \text{if } v \in V(H), \\
        \mu'(v) & \text{if } v \in V'.
    \end{array}
    \right.
\end{equation*}
As we need to guarantee that $\mathbf{tw}(G) \le t$, we must have $|E| \leq \min\{(pr,(p+r)t\} \leq \min\{r,t\} \cdot (p+r)$ and $|E'| \leq \min\{r^2,rt\} \leq \min\{r,t\} \cdot (p+r)$.
Therefore, the number of choices for $E$ and $E'$ is bounded by $(p+r)^{O(\min\{r,t\} \cdot (p+r))}$.
The number of choices for $\lambda'$ is $(p+r)^q$, and the number of choices for $\mu'$ is at most $|P|^r$.
So the bound $(p+r)^{O(\min\{r,t\} \cdot (p+r)+q)} \cdot |P|^{O(r)}$ follows.
\end{proof}


\paragraph{Exponential towers and almost linear functions.}
For a number $x \in \mathbb{R}$, we define $\exp^{(0)}(x) = x$ and $\exp^{(i)}(x) = 2^{\exp^{(i-1)}(x)}$ for all integers $i \geq 1$.
More generally, we can consider exponential towers with any base $b \in \mathbb{R}$.
Formally, we define $\exp_b^{(0)}(x) = x$ and $\exp_b^{(i)}(x) = b^{\exp_b^{(i-1)}(x)}$ for all integers $i \geq 1$.
We observe the following basic properties of exponential towers.

\begin{fact} \label{fact-basechange1}
    Let $m \in \mathbb{N}$, $b \geq 2$, and $a_1,\dots,a_m \in \mathbb{N}$.
    Then
    \begin{equation*}
        b^{\sum_{i=1}^m b^{a_i}} \leq \sum_{i=1}^m \exp_{bm}^{(2)}(a_i).
    \end{equation*}
\end{fact}
\begin{proof}
Assume $a_1 \geq \dots \geq a_m$ without loss of generality.
We have $b^{\sum_{i=1}^m b^{a_i}} = b^{m \cdot b^{a_1}} \leq \exp_{bm}^{(2)}(a_1)$.
Thus, the inequality $b^{\sum_{i=1}^m b^{a_i}} \leq \sum_{i=1}^m \exp_{bm}^{(2)}(a_i)$ follows.
\end{proof}

\begin{fact} \label{fact-basechange2}
    Let $b \geq 2$ and $r \in \mathbb{N}$.
    Then there exists a number $c > 0$ only depending on $b$ and $r$ such that $\exp_b^{(r)} (x) \leq c \cdot \exp^{(r)} (x \log (x+1))$ for all $x \in \mathbb{R}_{\geq 0}$.
\end{fact}
\begin{proof}
We first show that $\exp_b^{(r)} (x) \leq \exp^{(r)} (b(x+1))$ for all $x \in \mathbb{R}_{\geq 0}$.
This can be done by induction on $r$.
When $r = 0$, the inequality trivially holds.
Suppose it holds for $r-1$.
Then 
\begin{equation*}
    \exp_b^{(r)} (x) = \exp_b^{(r-1)} (b^x) \leq \exp^{(r-1)} (b(b^x+1)) \leq \exp^{(r)} (\log b \cdot (x+1) + 1) \leq \exp^{(r)} (b(x+1)).
\end{equation*}
Note that $b(x+1) \leq \max\{x \log (x+1),8^{b}\}$, since if $2b \leq \log (x+1)$ and $x \geq 1$, then $b(x+1) \leq 2bx \leq x \log (x+1)$, otherwise $b(x+1) \leq b\cdot 2^{2b}+ 2b \leq 8^b$.
Thus, by setting $c = \exp^{(r)} (3b)$, we have $\exp^{(r)} (b(x+1)) \leq c \cdot \exp^{(r)} (x \log (x+1))$.
\end{proof}

A function $f: \mathbb{R}_{\geq 0} \rightarrow \mathbb{R}_{\geq 0}$ is \textit{almost linear} if $f(x) = x^{1+o(1)}$, or more precisely, for any constant $\alpha > 1$, there exists a number $x_0 \in \mathbb{R}_{\geq 0}$ such that $f(x) < x^\alpha$ for all $x \geq x_0$.
We say an almost linear function $f$ is \textit{well-behaved} if there exists a constant $c > 0$ such that $f(cx) \leq cf(x)$ for all $x \in \mathbb{R}_{\geq 0}$.
Note that if $f$ and $f'$ are both well-behaved almost linear functions, then $f+f'$ is also a well-behaved almost linear function.
We have the following observation.
\begin{fact} \label{fact-decompose}
    There exists a well-behaved almost linear function $f: \mathbb{R}_{\geq 0} \rightarrow \mathbb{R}_{\geq 0}$ such that for any numbers $r \in \mathbb{N}$, $a_1,\dots,a_r,z \in \mathbb{R}_{\geq 0}$, and $b \geq 2$, we have
    \begin{equation*}
        \left( \sum_{i=1}^r \exp_b^{(2)}(a_i) \right)^z \leq r^r b^{b^{2 \log b}}\left( b^{f(z)} + \sum_{i=1}^r \exp_b^{(2)}(f(a_i)) \right).
    \end{equation*}
\end{fact}
\begin{proof}
Define $f(x) = 2^{\log x + (\log (2x+4)/ \log\log (2x+4))} + 100$ for all $x \in \mathbb{R}_{\geq 0}$.
Clearly, $f$ is almost linear and well-behaved.
We observe two basic properties of $f$: for all $x \in \mathbb{R}_{\geq 0}$,
\begin{enumerate}[(i)]
    \item $f(x) \geq 2 x \log x \geq 2 x \log_b x$,
    \item $(f(x)-x) / \log(f(x)-x) > 2x+2$.
\end{enumerate}

To see $f$ satisfies the desired property, assume without loss of generality that $a_1 \geq \cdots \geq a_r$.
Then we have the inequality
\begin{equation*}
    \left(\sum_{i=1}^r \exp_b^{(2)}(a_i)\right)^z \leq r^z (\exp_b^{(2)}(a_1))^z = r^z \cdot b^{b^{a_1}z}.
\end{equation*}

We claim that $b^{b^{a_1}z} \leq b^{b^b}(b^{f(z)} + b^{b^{f(a_1)}})/z^z$.
Note that if this is true, then we are done, because $r^z \cdot b^{b^{a_1}z} \leq (\frac{r}{z})^z b^{b^b} \cdot (b^{f(z)} + b^{b^{f(a_1)}}) \leq r^r b^{b^b} \cdot (b^{f(z)} + b^{b^{f(a_1)}}) \leq r^r b^{b^b} \cdot (b^{f(z)} + \sum_{i=1}^r b^{b^{f(a_i)}})$, which implies the desired inequality.
If $b^{f(a_1)} \geq b^{a_1} z + z \log_b z$, then $b^{b^{a_1}z} \leq b^{b^{f(a_1)}}/z^z \leq (b^{f(z)} + b^{b^{f(a_1)}})/z^z$.
So we only need to consider the case $b^{f(a_1)} < b^{a_1} z + z \log_b z$.
We shall show $f(z) + b^{2 \log b} - z \log_b z \geq b^{a_1}z$ and therefore $b^{b^{a_1}z} \leq b^{b^{2 \log b}} (b^{f(z)}/z^z) \leq b^{b^{2 \log b}} (b^{f(z)} + b^{b^{f(a_1)}})/z^z$.
Consider the following two cases.

\begin{itemize}
    \item \textbf{Case 1: $\log_b z \geq b^{a_1}$.}
    In this case, we have $b^{a_1} z \leq z \log_b z$.
    Property (i) of $f$ above then implies $f(z) - z \log_b z \geq z \log_b z \geq b^{a_1}z$.
    \item \textbf{Case 2: $\log_b z < b^{a_1}$.}
    In this case, we have $b^{f(a_1)} < b^{a_1} z + z \log_b z < 2z \cdot b^{a_1}$, which implies $\log_b (2z) > f(a_1)-a_1$ and therefore
    \begin{equation*}
        \frac{\log_b (2z+4)}{\log \log_b (2z+4)} \geq \frac{\log_b (2z)}{\log \log_b (2z)} > \frac{f(a_1)-a_1}{\log(f(a_1)-a_1)} > 2 a_1+2,
    \end{equation*}
    where the last $>$ follows from property (ii) of $f$.
    If $\log\log (2z+4) \geq 2 \log\log b$, then $\frac{\log_b (2z+4)}{\log \log (2z+4)} \geq \frac{\log_b (2z+4)}{2 \log \log_b (2z+4)} > a_1+1$.
    It follows that
    \begin{equation*}
        f(z) - z \log_b z \geq f(z)/2 = (z/2) \cdot b^{\log_b (2z+4)/ \log\log (2z+4)} \geq (z/2) \cdot b^{a_1+1} \geq b^{a_1} z.
    \end{equation*}
    On the other hand, if $\log\log (2z+4) < 2 \log\log b$, then it holds that $b^{2 \log b} \geq z^2$.
    The inequality $b^{f(a_1)} < b^{a_1} z + z \log_b z$ implies $z > b^{a_1}$ and thus $b^{2 \log b} \geq b^{a_1} z$.
    As $f(z) > z \log_b z$ by property (i) of $f$, we have $f(z) + b^{2 \log b} - z \log_b z \geq b^{a_1}z$.
\end{itemize}
In both cases, we have $f(z) + b^{2 \log b} - z \log_b z \geq b^{a_1}z$.
This completes the proof.
\end{proof}

\section{Upper bounds} \label{sec-ub}

We define the notion of \textit{$S$-signatures} on labeled and colored graphs for a sequence $S$ of pairs of natural numbers, inductively on the length of $S$.
The base case is of course $|S| = 0$, i.e., $S = ()$.
The \textit{$()$-signature} of a $p$-labeled and $P$-colored graph $(G,\lambda,\mu)$ is just defined as its core, i.e., $\mathsf{sgn}_{()}(G,\lambda,\mu) = \mathsf{core}(G,\lambda,\mu)$.
Consider a sequence $S$ of $\ell+1$ natural numbers and suppose $S = ((k,s))+S'$ where $k,s \in \mathbb{N}$ and $S'$ is a sequence of $\ell$ pairs of natural numbers.
The \textit{$S$-signature} of a labeled and colored graph $(G,\lambda,\mu)$ is defined as 
\begin{equation*}
\mathsf{sgn}_S(G,\lambda,\mu) = \{\mathsf{sgn}_{S'}(G,\lambda \oplus \lambda',\mu \otimes \mu') :\lambda' \in \varLambda_{G,k} \text{ and } \mu' \in U_{G,\{0,1\}^s}\},    
\end{equation*}
where $\varLambda_{G,k}$ is the family of all $k$-labelings on $G$ and $U_{G,Q}$ is the family of all $Q$-colorings on $G$.
For convenience, we also define the \textit{$S$-signature} of a (unlabeled and uncolored) graph $G$, denoted by $\mathsf{sgn}_S(G)$, simply as the $S$-signature of $(G,-,-)$.
The \textit{recursive size} of a signature $\mathsf{sgn}_S(G,\lambda,\mu)$, denoted by $\lVert \mathsf{sgn}_S(G,\lambda,\mu) \rVert$, is defined as follows: if $\mathsf{sgn}_S(G,\lambda,\mu)$ is a set (the case $|S| \geq 1$), then $\lVert \mathsf{sgn}_S(G,\lambda,\mu) \rVert = \sum_{x \in \mathsf{sgn}_S(G,\lambda,\mu)} \lVert x \rVert$; otherwise $\mathsf{sgn}_S(G,\lambda,\mu)$ is a labeled and colored graph (the case $|S|=0$), and we define $\lVert \mathsf{sgn}_S(G,\lambda,\mu) \rVert = 1$.

\subsection{Algorithms via signatures} \label{sec-algo}

In this section, we propose two algorithms.
The first algorithm efficiently computes the $S$-signature of a given graph $G$ by dynamic programming on a tree decomposition of $G$.
The second algorithm efficiently tests whether a graph $G$ satisfies an $S$-MSO $\phi$, given $\mathsf{sgn}_S(G)$.
We start by proving the following important lemma, which will be used in the first algorithm and also used for bounding the signature size in Section~\ref{sec-size}.

\begin{lemma}\label{lem-sgncomp}
    Let $(G,\lambda,\mu)$ be a $p$-labeled $P$-colored graph, and $S = ((k_1,s_1),\dots,(k_d,s_d))$.
    \begin{enumerate}[(i)]
        \item Let $\lambda': [q] \rightarrow V(G)$ be a $q$-coloring on $G$ and $f: [p] \rightarrow [q]$ be a map such that $\lambda = \lambda' \circ f$.
        Then given only $f$ and $\mathsf{sgn}_S(G,\lambda',\mu)$, one can compute the signature $\mathsf{sgn}_S(G,\lambda,\mu)$ in time
        \begin{equation*}
            \left(\sum_{i=1}^d (k_i+s_i) + \lVert \mathsf{sgn}_S(G,\lambda',\mu) \rVert\right)^{O(1)}.
        \end{equation*}
        \item Let $\mu': V(G) \rightarrow Q$ be a $Q$-coloring on $G$ and $f: Q \rightarrow P$ be a map such that $\mu = f \circ \mu'$.
        Then given only $f$ and $\mathsf{sgn}_S(G,\lambda,\mu')$, one can compute the signature $\mathsf{sgn}_S(G,\lambda,\mu)$ in time
        \begin{equation*}
            \left(\sum_{i=1}^d (k_i+s_i) + \lVert \mathsf{sgn}_S(G,\lambda,\mu') \rVert\right)^{O(1)}.
        \end{equation*}        
        \item Let $(A,B)$ be a separation of $G$ with $|A \cap B| = z$ and $\mu':V(G) \rightarrow [z]_0$ be a function that maps all vertices in $V(G) \backslash (A \cap B)$ to $0$ and maps $A \cap B$ bijectively to $[z]$.
        Denote $\lambda_A=\mathsf{sort}(\lambda^{|A})$, $\lambda_B=\mathsf{sort}(\lambda^{|B})$, $\mu_A=(\mu \otimes \mu')_{|A}$ and $\mu_B=(\mu \otimes \mu')_{|B}$.
        Then given only $S$, $\lambda^{-1}(A)$, $\lambda^{-1}(B)$, $\mathsf{sgn}_S(G[A],\lambda_A,\mu_A)$, and $\mathsf{sgn}_S(G[B],\lambda_B,\mu_B)$, one can compute $\mathsf{sgn}_S(G,\lambda,\mu)$ in time
        \begin{equation*}
            2^{\sum_{i=1}^dk_i}\cdot\left(\sum_{i=1}^d (k_i+s_i) + \lVert \mathsf{sgn}_S(G[A],\lambda_A,\mu_A) \rVert + \lVert \mathsf{sgn}_S(G[B],\lambda_B,\mu_B) \rVert\right)^{O(1)}.
        \end{equation*}
    \end{enumerate}
\end{lemma}

\begin{proof}
The proof is by induction on $|S|=d$. To prove item (iii), we will need to prove a more general item, termed (iii*), defined as follows:

Let $(A,B)$ be a separation of $G$ with $|A \cap B| = z$ and $\mu':V(G) \rightarrow [z]_0$ be a function that maps all vertices in $V(G) \backslash (A \cap B)$ to $0$ and maps $A \cap B$ bijectively to $[z]$. Let $I_A,I_B\subseteq [p]$ such that $I_A\cup I_B=[p]$ and $|I_A|+|I_B|-|I_A\cap I_B|=p$.  Additionally, let $J_A,J_B\subseteq [p]$ such that $|I_A|=|J_A|$ and $|I_B|=|J_B|$. Let $\lambda_A: [p_A]\rightarrow A$ and $\lambda_B: [p_B]\rightarrow B$ be a $p_A$-labeling  of $G[A]$ and a $p_B$-labeling of $G[B]$ whose domains contain $J_A$ and $J_B$, respectively.. Additionally, let $\mu_A: A\rightarrow P$ and $\mu_B: B\rightarrow P$ be two $P$-colorings of $G[A]$ and $G[B]$, respectively.

We say that a $p$-labeling $\lambda^*$ of $G$ is {\em $(I_A,I_B,J_A,J_B,\lambda_A,\lambda_B)$-compatible} if ${\lambda^*}^{-1}(A)=I_A$, ${\lambda^*}^{-1}(B)=I_B$, $\mathsf{sort}({\lambda_A}_{|J_A})=\mathsf{sort}(\lambda^*_{|I_A})$, $\mathsf{sort}({\lambda_B}_{|J_B})=\mathsf{sort}(\lambda^*_{|I_B})$. Notice that if there exists a $p$-labeling that is  $(I_A,I_B,\lambda_A,\lambda_B)$-compatible, then it is unique.

Similarly, we say that $P$-coloring $\mu^*$ of $G$ is {\em $(\mu_A,\mu_B)$-compatible} if $\mu_A=\mu^*_{|A}$ and $\mu_B=\mu^*_{|B}$. Notice that if there exists a $P$-coloring that is  $(\mu_A,\mu_B)$-compatible, then it is unique.
    
        Then given only $S$, $I_A,I_B,J_A,J_B,$, $\mathsf{sgn}_S(G[A],\lambda_A,\mu_A)$, and $\mathsf{sgn}_S(G[B],\lambda_B,\mu_B)$, one can either:
        \begin{enumerate}
        \item Determine that there does not exist a $p$-labeling that is  $(I_A,I_B,J_A,J_B,\lambda_A,\lambda_B)$-compatible or there does not exist a $P$-coloring that is $(\mu_A,\mu_B)$-compatible (or both).
        \item For $\lambda^*$ and $\mu^*$ being the  $p$-labeling that is  $(I_A,I_B,J_A,J_B,\lambda_A,\lambda_B)$-compatible and the $P$-coloring that is $(\mu_A,\mu_B)$-compatible,  compute the signature $\mathsf{sgn}_S(G,\lambda^*,\mu^*)$.
        \end{enumerate}
        Furthermore, this is to be done in time
        \begin{equation*}
            2^{\sum_{i=1}^dk_i}\cdot\left(\sum_{i=1}^d (k_i+s_i) + \lVert \mathsf{sgn}_S(G[A],\lambda_A,\mu_A) \rVert + \lVert \mathsf{sgn}_S(G[B],\lambda_B,\mu_B) \rVert\right)^{O(1)}.
            \end{equation*}

\medskip
\noindent{\bf Basis ($S=\emptyset$) for (i):} Let us denote $(H,\alpha',\beta)=\mathsf{sgn}_S(G,\lambda',\mu)$. Define $\alpha = \alpha'\circ f$.  Then, $(H,\alpha,\beta)=\mathsf{sgn}_S(G,\lambda,\mu)$.

\medskip
\noindent{\bf Basis ($S=\emptyset$) for (ii):} Let us denote $(H,\alpha,\beta')=\mathsf{sgn}_S(G,\lambda,\mu')$. Define $\beta= f\circ \beta'$.  Then, $(H,\alpha,\beta)=\mathsf{sgn}_S(G,\lambda,\mu)$.

\medskip
\noindent{\bf Basis ($S=\emptyset$) for (iii*):} Let us denote  $(H_A,\alpha_A,\beta_A)=\mathsf{sgn}_S(G[A],\lambda_A,\mu_A)$, and $(H_B,\alpha_B,\beta_B)=\mathsf{sgn}_S(G[B],\lambda_B,\mu_B)$. 

First, notice that $z=|\{i\in\mathbb{N}: \exists X$ s.t.~$(X,i)\in\mathsf{image}(\beta_A)\}|$, and $p=|I_A|+|I_B|-|I_A\cap I_B|$. Thus, we  suppose to have $z$ and $p$ at hand.  Due to the renaming of vertices, we must observe that two vertices $u$ in $H_A$ and $v$ in $H_B$ should be ``treated'' as the same vertex (i.e., they are different names of the same vertex in $A\cap B$) if and only if: there exist $i\in [z],X$ for which $\beta_A(u)=(X,i)$ and $\beta_B(v)=(X,i)$. Accordingly, we define a function $\mathsf{retrace}: V(H_B)\rightarrow V(H_A)$ as
follows: for every $u\in V(H_B)$, if $\beta_B(u)=(X,i)$ for some (unique) $i\in[z]$, then
$\mathsf{retrace}(u)=\beta_A(X,i)$,
and otherwise  $\mathsf{retrace}(u)=u$. Moreover, let $M_A$ be the set of vertices in the image of $\mathsf{retrace}$ that do not belong to $B$, and let $M_B$ be the set of vertices in $B$ that $\mathsf{retrace}$ does not map to themselves.

With this in mind, we reconstruct a function $\widetilde{\lambda}^*$ on $[p]$ that is equivalent to $\lambda^*$ (if it exists) up to the renaming of the vertices as follows. For every $i\in[p]$:
\begin{itemize}
\item If $i\in I_A\cap I_B$, then verify that $\mathsf{sort}({\alpha_A}_{|J_A})(j)=\mathsf{retrace}(\mathsf{sort}({\alpha_B}|_{|J_B})(\ell))$,
 where $j$ is the number of integers in $I_A$ that are smaller or equal to $i$, and $\ell$ is the number of integers in $I_B$ that are smaller or equal to $i$. If this is not true, then return that there does not exist a $p$-labeling that is  $(I_A,I_B,J_A,J_B,\lambda_A,\lambda_B)$-compatible. Else (when this is true), then $\widetilde{\lambda}^*(i)=\mathsf{sort}({\alpha_A}_{|J_A})(j)$.
\item Else if $i\in I_A\setminus I_B$, then 
$\widetilde{\lambda}^*(i)=\mathsf{sort}({\alpha_A}_{|J_A})(j)$, where $j$ is the number of integers in $I_A$ that are smaller or equal to $i$.  If $\alpha_A(j)\in M_A$, then return that there does not exist a $p$-labeling that is  $(I_A,I_B,J_A,J_B,\lambda_A,\lambda_B)$-compatible.
\item Else (i.e., $i\in I_B\setminus I_A$), $\widetilde{\lambda}^*(i)=\mathsf{retrace}(\mathsf{sort}({\alpha_B}_{|J_B})(j))$, where $j$ is the number of integers in $I_B$ that are smaller or equal to $i$. If $\mathsf{sort}({\alpha_B}_{|J_B})(j)\in M_B$, then return that there does not exist a $p$-labeling that is  $(I_A,I_B,J_A,J_B,\lambda_A,\lambda_B)$-compatible.
\end{itemize}
So far, observe that if the above tests have passed, then there indeed exists a $p$-labeling that is  $(I_A,I_B,J_A,J_B,\lambda_A,\lambda_B)$-compatible, and we denote it by $\lambda^*$.

Now, we compute $(H,\alpha,\beta)$ as follows:
\begin{itemize}
\item $V(H)=\widetilde{\lambda}^*([p])$.
\item $E(H)=\{\{u,v\}\in E(H_A): u,v\in V(H_A)\}\cup\{\{\mathsf{retrace}(u),\mathsf{retrace}(v)\}: \{u,v\}\in E(H_B)\}$. Since $(A,B)$ is a separation, two vertices are adjacent in $G$ if and only if they are adjacent in at least one of $G[A]$ and $G[B]$. Hence, by the definition of $\mathsf{sgn}$, two vertices are adjacent in $G[\lambda^*([p])]$ if and only if the corresponding vertices are adjacent in at least one of $H_A$ and $H_B$. Thus, we derive that $E(H)$ that we defined satisfies $E(H)=E(G[\widetilde{\lambda}^*([p]])$ up to the remaining of the vertices.
\item $\alpha=\widetilde{\lambda}^*$.
\item We construct a function $\beta$ on $V(H)$ as follows. For every $v\in V(H)$:
	\begin{itemize}
	\item If $v\in M_A$: Verify that $\beta_A(v)=\beta_B(\mathsf{retrace}^{-1}(v))$. If this is not true, then return that there does not exist a $P$-coloring that is  $(\mu_A,\mu_B)$-compatible. Else, $\beta(v)=X$ where $X$ is the element such that $\beta_A(v)=(X,i)$ for some $i$.
	\item  Else if $v\in V(H_A)\setminus M_A$: $\beta(v)=X$ where $X$ is the element such that $\beta_A(v)=(X,i)$ for some $i$.
	\item  Else (when  $v\in V(H)\setminus V(H_A)$): $\beta(v)=X$ where $X$ is the element such that $\beta_B(v)=(X,i)$ for some $i$.
	\end{itemize}
	Observe that if the above tests have passed, then there indeed exists a $P$-coloring that is  $(\mu_A,\mu_B)$-compatible, and we denote it by $\mu^*$.
\end{itemize}

Overall, we have $(H,\alpha,\beta)$ is isomorphic to $\mathsf{sgn}_\emptyset(G,\lambda^*,\mu^*)$. Clearly, the time complexity to perform the computation above can be naively bounded by a polynomial in $\lVert \mathsf{sgn}_S(G[A],\lambda_A,\mu_A) \rVert + \lVert \mathsf{sgn}_S(G[B],\lambda_B,\mu_B) \rVert$. This completes the proof of the basis.

\bigskip

Next, we suppose correctness for all sequences of pairs of natural numbers whose length is at most $d-1$. Let $\widehat{S}=((k_1,s_1),\ldots,(k_{d-1},s_{d-1}))$.

\medskip
\noindent{\bf Step ($|S|=d\geq 1$) for (i):}  We initialize $\cal C$ to be empty, and insert elements as follows. For all  $X\in \mathsf{sgn}_S(G,\lambda',\mu)$: Insert into $\cal C$ the element returned by  the application of the algorithm from the inductive hypothesis on   $X$ and $f$. Correctness follows directly from the definition of $\mathsf{sgn}$, and the running time bound is immediate from the definition of recursive size.

\medskip
\noindent{\bf Step ($|S|=d\geq 1$) for (ii):}  We initialize $\cal C$ to be empty, and insert elements as follows. For all  $X\in \mathsf{sgn}_S(G,\lambda,\mu')$: Insert into $\cal C$ the element returned by  the application of the algorithm from the inductive hypothesis on   $X$ and $f$. Correctness follows directly from the definition of $\mathsf{sgn}$, and the running time bound is immediate from the definition of recursive size.

\medskip
\noindent{\bf Step ($|S|=d\geq 1$) for (iii*):}  We initialize ${\cal C}$ to be empty, and insert elements as follows:
\begin{itemize}
\item For all $X\in \mathsf{sgn}_S(G[A],\lambda_A,\mu_A)$ and $Y\in \mathsf{sgn}_S(G[B],\lambda_B,\mu_B)$:
\begin{itemize}
    \item For all $\widehat{I}_A,\widehat{I}_B\subseteq [k_d]$ such that $\widehat{I}_A\cup \widehat{I}_B=[k_d]$ and $|\widehat{I}_A|+ |\widehat{I}_B|-|\widehat{I}_A\cap \widehat{I}_B|=k_d$:
    \begin{itemize}
    \item Denote $I^*_A=I_A\cup \{p+i:i\in \widehat{I}_A\}$, $I^*_B=I_B\cup \{p+i:i\in \widehat{I}_B\}$, $J^*_A=J_A\cup\{|\mathsf{domain}(\lambda_A)|+i:i\in\widehat{I}_A\}$ and $J^*_B=J_B\cup\{|\mathsf{domain}(\lambda_B)|+i:i\in\widehat{I}_B\}$.
        \item Use the inductive hypothesis on $\widehat{S},I^*_A, I^*_B,J^*_A,J^*_B,X,Y$  to compute $Z$. If invalidity is determined, then continue to the next iteration (or stop, if all iterations have already been completed). Else, we have $Z$ at hand, and insert it  into ${\cal C}$.
    \end{itemize}
\end{itemize}
\end{itemize}

If $\cal C$ is empty, then this is the indication that $\lambda^*$ or $\mu^*$ (or both) does not exist; then, we return this information. Else, the proof proceeds as follows, showing that ${\cal C}=\mathsf{sgn}_S(G,\lambda^*,\mu^*)$. 

First, we will argue that ${\cal C}\subseteq \mathsf{sgn}_S(G,\lambda^*,\mu^*)$. 
For this purpose, consider some iteration of the algorithm, corresponding to some $\widehat{I}_A,\widehat{I}_B,X,Y$. Then, by the definition of $\mathsf{sgn}$, $X$ is equal to $\mathsf{sgn}_{\widehat{S}}(G[A],\lambda_A\oplus \lambda'_A,\mu_A\otimes \mu'_A)$ for some (possibly several equivalent under isomorphism, but we pick one arbitrarily) $\lambda'_A\in A_{G[A],k_d}$ and $\mu'_A\in U_{G[A],\{0,1\}^{s_d}}$. (Notice that the knowledge of these $\lambda'_A,\mu'_A$ is not available to the algorithm, but can only be used in the analysis.) Similarly, $Y$ is equal to $\mathsf{sgn}_{\widehat{S}}(G[B],\lambda_B\oplus \lambda'_B,\mu_B\otimes \mu'_B)$ for some $\lambda'_B\in A_{G[B],k_d}$ and $\mu'_B\in U_{G[B],\{0,1\}^{s_d}}$. Now, we define a function $\widehat{\lambda}\in A_{G,k_d}$ as follows.  For all $i\in[k_d]$:
\begin{itemize}
            \item If $i\in \widehat{I}_A$: 
$\widehat{\lambda}(i)=\lambda'_A(j)$.
            \item Otherwise (i.e., $i\in \widehat{I}_B\setminus \widehat{I}_A$): $\widehat{\lambda}(i)=\lambda'_B(j)$.
\end{itemize}
Additionally, we define a function $\widehat{\mu}\in U_{G,\{0,1\}^{s_d}}$ as follows.  For all $v\in V(G)$:
\begin{itemize}
            \item If $v\in A$: $\widehat{\mu}(v)=\mu'_A(v)$.
            \item Otherwise (i.e., $v\in B\setminus A$): $\widehat{\mu}(v)=\mu'_B(v)$.
\end{itemize}

The inductive hypothesis yields that $Z=\mathsf{sgn}_{\widehat{S}}(G,\lambda^*\oplus \widehat{\lambda}, \mu^*\otimes \widehat{\mu})$. Since $\widehat{\lambda}\in A_{G,k_d}$ and $\widehat{\mu}\in U_{G,\{0,1\}^{s_d}}$, we conclude that $Z\in \mathsf{sgn}_S(G,\lambda^*,\mu^*)$. This completes the proof that ${\cal C}\subseteq \mathsf{sgn}_S(G,\lambda^*,\mu^*)$.

Second, we will argue that ${\cal C}\supseteq \mathsf{sgn}_S(G,\lambda_D)$. For this purpose, consider some element in $\mathsf{sgn}_S(G,\lambda^*,\mu^*)$. So, this element is of the form $\mathsf{sgn}_{\widehat{S}}(G,\lambda^*\oplus\widehat{\lambda},\mu^*\otimes\widehat{\mu})$ for some  (possibly several equivalent under isomorphism, but we pick one of each arbitrarily) $\widehat{\lambda}\in A_{G,k_d}$ and $\widehat{\mu}\in U_{G,\{0,1\}^{s_d}}$. 
Define:
\begin{itemize}
\item  $\widehat{I}_A=\{i-p+|\mathsf{domain}(\lambda_A)|:i\in\widehat{\lambda}^{-1}(A)\}$ and
$\widehat{I}_B=\{i-p+|\mathsf{domain}(\lambda_B)|:i\in\widehat{\lambda}^{-1}(B)\}$.
\item $\lambda'_A\in A_{G[A],k_d}$: For all $i\in[k_d]$: If $i\in\widehat{A}_I$, then $\lambda'_A(i)=\widehat{\lambda}(i)$; Otherwise, let $\lambda'_A(i)$ equal some arbitrary vertex in $A$.
\item $\lambda'_B\in A_{G[B],k_d}$: For all $i\in[k_d]$: If $i\in\widehat{B}_I$, then $\lambda'_B(i)=\widehat{\lambda}(i)$; Otherwise, let $\lambda'_B(i)$ equal some arbitrary vertex in $B$.
\item $\mu'_A=\mu^*_{|A}$ and $\mu'_B=\mu^* \otimes_{|B}$. Then, $\mu'_A\in U_{G[A,\{0,1\}^{s_d}}$ and $\mu'_B\in U_{G[B],\{0,1\}^{s_d}}$.
\end{itemize} 

Thus, the definition of $\mathsf{sgn}$ directly yields that $\mathsf{sgn}_{\widehat{S}}(G[A],\lambda_A\oplus\lambda'_A,\mu_A\otimes\mu'_A)\in\mathsf{sgn}_S(G[A],\lambda_A,\mu_A)$ and $\mathsf{sgn}_{\widehat{S}}(G[B],\lambda_B\oplus\lambda'_B,\mu_B\otimes\mu'_B)\in\mathsf{sgn}_S(G[B],\lambda_B,\mu_B)$. Hence, the algorithm will perform an iteration with $X=\mathsf{sgn}_{\widehat{S}}(G[A],\lambda_A\oplus\lambda'_A,\mu_A\otimes\mu'_A)$ and $Y=\mathsf{sgn}_{\widehat{S}}(G[B],\lambda_B\oplus\lambda'_B,\mu_B\otimes\mu'_B)$, and an inner-iteration with $\widehat{I}_A,\widehat{I}_B$. In particular, by the correctness of the inductive hypothesis, in this iteration we reconstruct $Z=\mathsf{sgn}_{\widehat{S}}(G,\lambda^*\oplus \widehat{\lambda}, \mu^*\otimes \widehat{\mu})$. Thus, we conclude that $\mathsf{sgn}_{\widehat{S}}(G,\lambda^*\oplus \widehat{\lambda}, \mu^*\otimes \widehat{\mu})\in{\cal C}$. This completes the proof of ${\cal C}\supseteq \mathsf{sgn}_S(G,\lambda^*,\mu^*)$.

Notice that we perform $2^{O(k_d)}\cdot |\mathsf{sgn}_S(G[A],\lambda_A,\mu_A)|\cdot|\mathsf{sgn}_S(G[B],\lambda_B,\mu_B)|$ iterations (where $|\mathsf{sgn}_S(G[A],\lambda_A,\mu_A)|$ and $|\mathsf{sgn}_S(G[B],\lambda_B,\mu_B)|$  denote the sizes of the corresponding sets, not their recursive sizes). So, the bound on the time complexity of the algorithm follows from the inductive hypothesis and the definition of the recursive size of $\mathsf{sgn}$. This completes the proof.
\end{proof}

To give the first algorithm, we still need an algorithm that computes $S$-signatures by brute-force, which is presented in the following lemma.

\begin{lemma}\label{lem:bruteForce1}
    Let $(G,-,\mu)$ be a $0$-labeled and $[|V(G)|]$-colored graph.
    Let $S=((k_1,s_1),\ldots,(k_d,s_d))$ be a sequence of pairs of natural numbers.
    Then we can compute $\mathsf{sgn}_S(G,-,\mu)$ in time $f(d)\cdot\left(\sum_{i=1}^d2^{O(k_i\cdot\log |V(G)|)}+\sum_{i=1}^{d-1}2^{O(s_i\cdot |V(G)|)}+ \sum_{i=1}^d2^{O(s_d\cdot k_i)}\right)$ for some computable function $f$ of $d$.
\end{lemma}

\begin{proof}
Following the definition of signature, we compute it using brue-force as follows. For $i=1,2,\ldots,d-1$: first, we consider all options to pick $k_i$ vertices (with repetition) from $V(G)$, thereby spending $|V(G)|^{O(k_i)}$ time, and then consider all options to color all vertices using colors from $\{0,1\}^{s_i}$, thereby spending $2^{O(|V(G)|\cdot s_i)}$ time. The only difference is when $i=d$. Then, again we consider all options to pick $k_i$ vertices (with repetition) from $V(G)$, thereby spending $|V(G)|^{O(k_d)}$ time, and then, if $\sum_{i=1}^dk_i<|V(G)|$, then since all vertices that were labeled are already known at this point (and unlabeled vertices do not appear in the signature), we can consider all colorings only of the labeled vertices rather than all vertices, thereby spending  $2^{O((\sum_{i=1}^dk_i)\cdot s_d)}$ time. Notice that by multiplying the entire formula by $f(d)$, we can put the sum $\sum_{i=1}^d$  before corresponding expressions rather than let it be in the exponents, e.g., $2^{O((\sum_{i=1}^dk_i)\cdot s_d)}\leq g(d)\cdot \sum_{i=1}^d2^{O(s_d\cdot k_i)}$ for some computable function $g$ of $d$.

Notice that the order of $\sum_{i=1}^d2^{O(k_i\log |V(G)|)}\cdot\sum_{i=1}^{d-1}2^{O(s_i\cdot |V(G)|)}\cdot\sum_{i=1}^d2^{O(s_d\cdot k_i)}$  is bounded by the dominant sum (since there is an $O$-notation in the exponent), and hence by
 $\sum_{i=1}^d2^{O(k_i\log |V(G)|)}+\sum_{i=1}^{d-1}2^{O(s_i\cdot |V(G)|)}+ \sum_{i=1}^d2^{O(s_d\cdot k_i)}$; thus, the bound follows.
\end{proof}

Now we are ready to give our first algorithm, which is the following theorem.

\begin{theorem}\label{thm:twcomputation}
    Given a sequence $S=((k_1,s_1),\dots,(k_d,s_d))$ of pairs of natural numbers and a graph $G$ with a nice tree decomposition $(T,\beta)$ of width $t$, one can compute $\mathsf{sgn}_S(G)$ in time 
    \begin{equation*}
        f(d)\cdot \left(\sum_{i=1}^d2^{O(k_i\log t)}+\sum_{i=1}^{d-1}2^{O(s_i t)}+ \sum_{i=1}^{d}2^{O(s_d k_i)}\right)\cdot|V(T)| \cdot  \left(\max_{\mu \in U_{G,[t+1]_0}} \lVert \mathsf{sgn}_S(G,-,\mu) \rVert \right)^{O(1)},
    \end{equation*}
    for some computable function $f$ of $d$.
\end{theorem}

\begin{proof}
Let $<$  be an arbitrary order on $V(G)$. 
For every $x\in V(T)$, we will fix $\mu_x:\beta(x)\rightarrow [|\beta(x)|]$ as the following bijective function: for all $v\in\beta(x)$, $\mu_x(v)$ equals the number of vertices smaller or equal to $v$ among those in $\beta(x)$ w.r.t.~$<$. Additionally, let $\mu^*_x:\gamma(x)\rightarrow [|\beta(x)|]_0$ be the extension of $\mu$ that assigns $0$ to every vertex in $\gamma(x)\setminus\beta(x)$. (We remind that $\gamma(x)$ denotes the union of $\beta(x')$ over every $x'$ that is either $x$ or a descendant of $x$.)
Let $\mathsf{emp}=\{\cdots\{(-,-,-)\}\cdots\}$, where each pair $\{,\}$ appears $d=|S|$ times and $(-,-,-)$ is the empty graph with the empty labeling and the empty coloring.

Now, we describe the algorithm.
\begin{enumerate}
\item Initialize a dynamic programming table $M$ with an entry $M[x]$ for every $x\in V(T)$. (The purpose of every entry $M[x]$ will be to store $\mathsf{sgn}_S(G[\gamma(x)],-,\mu^*_x)$.)
\item For every $x\in V(T)$ in preorder, consider the following cases:
\begin{enumerate}
\item {\bf $x$ is a Leaf node:} Assign $M[x]=\mathsf{emp}$. 

\item {\bf $x$ is an Introduce node:} Let $s$ be the child of $x$. Let $\{v\}=\beta(x)\setminus\beta(s)$. Let $i_v$ be the number of vertices in $\beta(x)$ that are smaller or equal to $v$ w.r.t.~$<$. Compute $X=\mathsf{sgn}_S(G[\beta(x)],-,\mu_x)$ using the algorithm from Lemma~\ref{lem:bruteForce1}.


Let $f_X:[|\beta(x)|]_0\rightarrow [|\beta(x)|]_0\times [|\beta(s)|]_0$ be defined as follows. For all $i\in [|\beta(x)|]_0$:
\begin{itemize}
\item  If $i<i_v$, then $f_X(i)=(i,i)$.
\item If $i=i_v$, then $f_X(i)=(i,0)$.
\item If $i>i_v$, then $f_X(i)=(i,i-1)$.
\end{itemize}
Call the algorithm from item (ii) of Lemma~\ref{lem-sgncomp} with $X$ and $f_X$, and let $X'$ denote the result.

Additionally, denote $Y=M[s]$. Let $f_Y:[|\beta(s)|]_0\rightarrow [|\beta(x)|]_0\times [|\beta(s)|]_0$ be defined as follows. For all $i\in [|\beta(s)|]_0$: 
\begin{itemize}
\item  If $i<i_v$, then $f_Y(i)=(i,i)$.
\item If $i\geq i_v$, then $f_Y(i)=(i+1,i)$.
\end{itemize}
Call the algorithm from item (ii) of Lemma~\ref{lem-sgncomp} with $Y$ and $f_Y$, and let $Y'$ denote the result.

Now, call the algorithm from item (iii) of Lemma~\ref{lem-sgncomp} with $I_A=I_B=J_A=J_B=\emptyset$, $X'$ and $Y'$, and let $Z$ denote the result. Assign $M[x]=Z$.

\item {\bf $x$ is a Forget node:} Let $s$ be the child of $x$. Let $\{v\}=\beta(s)\setminus\beta(x)$. Let $i_v$ be the number of vertices in $\beta(x)$ that are smaller or equal to $v$ w.r.t.~$<$. Denote $X=M[s]$.

Let $f_X:[|\beta(s)|]_0\rightarrow [|\beta(x)|]_0$ be defined as follows. For all $i\in [|\beta(s)|]_0$:
\begin{itemize}
\item  If $i<i_v$, then $f_X(i)=i$.
\item If $i=i_v$, then $f_X(i)=0$.
\item If $i>i_v$, then $f_X(i)=i-1$.
\end{itemize}
Call the algorithm from item (ii) of Lemma~\ref{lem-sgncomp} with $X$ and $f_X$, and let $X'$ denote the result. Assign $M[x]=X'$.

\item {\bf $x$ is a Join node:} Let $h,s$ be the children of $t$.  Denote $X=M[h]$ and $Y=M[s]$. Let $f:[|\beta(x)|]_0\rightarrow [|\beta(x)|]_0\times [|\beta(x)|]_0$ be defined as follows. For all $i\in [|\beta(s)|]_0$:  $(i)=(i,i)$.

Call the algorithm from item (ii) of Lemma~\ref{lem-sgncomp} with $X$ and $f$, and let $X'$ denote the result. Additionally, call the algorithm from item (ii) of Lemma~\ref{lem-sgncomp} with $Y$ and $f$, and let $Y'$ denote the result.

Now, call the algorithm from item (iii) of Lemma \ref{lem-sgncomp} with $I_A=I_B=J_A=J_B=\emptyset$, $X'$ and $Y'$, and let $Z$ denote the result. Assign $M[x]=Z$.
\end{enumerate}

\item Denote $X=M[\mathsf{root}(T)]$. For every graph in the ``innermost'' sets of $X$, every vertex is assigned a color that is a tuple that begins with $0$ (since $\mu^*_{\mathsf{root}(T)}$ is the function from $V(G)$ to $\{0\}$). Remove these $0$'s.
Let $X'$ denote the result, and return $X'$.
\end{enumerate}

We use induction on the order of computation of $M$ to show that for every $x\in V(T)$, $M[x]=\mathsf{sgn}_S(G[\gamma(x)],-,\mu^*_x)$. Recall that $\gamma(\mathsf{root}(T))=V(G)$ and $\beta(\mathsf{root}(T))=\emptyset$. Thus, since we remove the $0$'s assigned $\mu^*_x$, proving the inductive claim will prove that the output of the algorithm is indeed $\mathsf{sgn}_S(G)$.

In the basis, where $x$ is a Leaf node and therefore $\gamma(x)=\emptyset$, correctness follows directly from the definition of $\mathsf{sgn}$.

Now, we prove correctness for some non-leaf $x\in V(T)$, while we assume correctness for its child(ren). We have the following cases:
\begin{itemize}
    \item $x$ is an Introduce node: Consider the graph $(G[\gamma(x)],-,\mu^*_x)$ together with the separation $(\beta(x),\gamma(s))$. Note that $\beta(x)\cap\gamma(s)=\{v\}$. From the inductive hypothesis, we have $M[s]=\mathsf{sgn}_S(G[\gamma(s)],\lambda_s)$. Then, Lemma~\ref{lem-sgncomp} directly implies that $M[x]=\mathsf{sgn}_S(G[\gamma(x)],-,\mu^*_x)$. We remark that the application of this lemma is valid due to our definitions of $f_X$ and $f_Y$: particularly, we update the coloring corresponding to $M[s]$ to be a restriction of $\mu^*_x$ rather than $\mu^*_s$, and we ``append'' $\mu_s$, thereby enabling to pick $\mu'=\mu_s$.
    
    \item $x$ is a Forget node: From the inductive hypothesis, we have $M[s]=\mathsf{sgn}_S(G[\gamma(s)],-,\mu^*_s)$. Thus, we only need to update the coloring from $\mu^*_s$ to $\mu_x$. This is done using $f_X$ and the application Lemma~\ref{lem-sgncomp}, which directly implies that $M[x]=\mathsf{sgn}_S(G[\gamma(x)],-,\mu^*_x)$.
    
    \item $x$ is a Join node: Consider the graph $(G[\gamma(x)],\mu^*_x)$ together with the separation $(\gamma(h),\gamma(s))$. From the inductive hypothesis,  we have that $M[h]=\mathsf{sgn}_S(G[\gamma(h)],-,\mu^*_h)$ and $M[s]=\mathsf{sgn}_S(G[\gamma(s)],-,\mu^*_s)$. Then, Lemma~\ref{lem-sgncomp} directly implies that $M[x]=\mathsf{sgn}_S(G[\gamma(x)],-,\mu^*_x)$.  We remark that the application of this lemma is valid due to our definition of $f$: particularly, we ``append'' $\mu_x$, thereby enabling to pick $\mu'=\mu_x$.
\end{itemize}
This completes the proof of the correctness of the algorithm. 

For the time complexity analysis, we observe that for any node $x\in V(T)$,  each application of Lemma~\ref{lem-sgncomp} is done in time 
$2^{\sum_{i=0}^d k_i}\cdot (\sum_{i=1}^d(k_i+s_i))^{O(1)}\cdot(\sum_{s\in\mathsf{children}(x)}||\mathsf{sgn}_S(G[\gamma(s)],-,\mu^*_s)||)^{O(1)}$.
In addition, if $x$ is an Introduce node, we perform the additional computation of Lemma~\ref{lem:bruteForce1}, whose time complexity is bounded by $f(d)\cdot\left(\sum_{i=1}^d2^{O(k_i\cdot\log t)}+\sum_{i=1}^{d-1}2^{O(s_i\cdot t)}+ \sum_{i=1}^d2^{O(s_d\cdot k_i)}\right)$  for some computable function $f$ of $d$ (since it is applied on a graph with at most $t+1$ vertices).
Since $||\mathsf{sgn}_S(G[\gamma(x)],-,\mu^*_x)||\leq \max_{\mu \in U_{G,[t+1]_0}} \lVert \mathsf{sgn}_S(G,-,\mu) \rVert$ for all nodes $x \in V(T)$, the bound claimed in the lemma follows. 
\end{proof}

Next, we present our algorithm for testing MSO formulas using signatures.

\begin{lemma}\label{lem:testSat}
Let $S$ be a sequence of pairs of natural numbers.
Given the $S$-signature $\mathsf{sgn}_S(G)$ of a graph $G$ and an $S$-MSO $\phi$, we can check if $G$ satisfies $\phi$ in $(||\mathsf{sgn}_S(G)||+|\phi|)^{O(1)}$ time.
\end{lemma}

\begin{proof}
To prove this lemma, we need to prove a more general claim: Let $S=((k_1,s_1),\ldots,(k_d,s_d))$ be a sequence of pairs of natural numbers. Suppose that we are given the $S$-signature $\mathsf{sgn}_S(G,\lambda,\mu)$ of a $p$-labeled $\{0,1\}^q$-colored graph $G$.
Now, let $I=((a_1,b_1),\ldots,(a_r,b_r)), J=((\widetilde{a}_1,\widetilde{b}_1),\ldots,(\widetilde{a}_r,\widetilde{b}_r))$ be two other sequences of pairs of nonnegative integers so that $\sum_{i=1}^r\widetilde{a}_i=p$ and $\sum_{i=1}^r\widetilde{b}_i=q$, and for every $i\in[r]$, $a_i\leq\widetilde{a}_i$ and $b_i\leq\widetilde{b}_i$.

 Let  $\phi(v_1,\ldots,v_{p'},U_1,\ldots,U_{q'})$ be an $S$-formula with $p'$ free vertex variables, $v_1,\ldots,v_{p'}$, and $q'$ free vertex set variables, $U_1,\ldots,U_{q'}$, where $\sum_{i=1}^ra_i=p'$ and $\sum_{i=1}^rb_i=q'$.  For all $i\in[q]$, let $A_i$ denote the set of vertices $v$ in $G$ such that the $i$-th coordinate of $\mu(v)$ is $1$.
Now, we say that $(G,\lambda,\mu)$ satisfies  $\phi$ w.r.t.~$I,J$ if $(G,V^*,U^*)$ satisfies $\phi$, where:
\begin{itemize}
\item $V^*=\lambda(1),\ldots,\lambda(a_1),\lambda(\widetilde{a}_1+1),\ldots,\lambda(\widetilde{a}_1+a_2),\cdots,\lambda(\sum_{i=1}^{r-1}\widetilde{a}_i+1),\ldots\lambda(\sum_{i=1}^{r-1}\widetilde{a}_i+a_r)$.
\item $U^*=A_1,\ldots,A_{b_1},A_{\widetilde{b}_1+1},\ldots,A_{\widetilde{b}_1+b_2},\cdots,A_{\sum_{i=1}^{r-1}\widetilde{b}_i+1},\ldots,A_{\sum_{i=1}^{r-1}\widetilde{b}_i+b_r}$.
\end{itemize}
Then, the claim is that we can check whether $(G,\lambda,\mu)$ satisfies $\phi$ w.r.t.~$I,J$ in $(||\mathsf{sgn}_S(G,\lambda,\mu)||+|\phi|)^{O(1)}$ time (given $\mathsf{sgn}_S(G,\lambda,\mu)$).

The proof is by induction on $|S|=d$.

\medskip
\noindent{\bf Basis ($S=\emptyset$):} Denote $(H,\alpha,\beta)=\mathsf{sgn}_\emptyset(G,\lambda,\mu)$. By the definition of signature, $(H,\alpha,\beta)$ is the core of $(G,\lambda,\mu)$, which immediately yields that $(G,\lambda,\mu)$ satisfies $\phi$ w.r.t.~$I,J$ if and only if $(H,\alpha,\beta)$ satisfies $\phi$ w.r.t.~$I,J$. So, we trivially check whether $(H,\alpha,\beta)$ satisfies $\phi$ w.r.t.~$I,J$ in time $(||\mathsf{sgn}_S(G,\lambda,\mu)||+|\phi|)^{O(1)}$, and return the answer.

\medskip
\noindent{\bf Step ($|S|=d\geq 1$):}  Here, we suppose correctness for all sequences of pairs of natural numbers whose length is at most $d-1$. Denote $S= (k_1,s_1)+\widehat{S}$, i.e., $k=k_1$ and $s=s_1$. Without loss of generality, we can suppose that $\phi$ excludes occurrences of $\forall$, since they can be replaced by occurrences of $\exists$ and $\neg$ (i.e., $ A (\forall x_1,\ldots\forall x_t\forall X_1,\ldots\forall X_s B$ is equivalent to $A \neg(\exists x_1\ldots\exists x_t\exists X_1\ldots\exists X_s\neg B)$). 
We will use induction on $\ell$, being the number of occurrences of $\neg,\wedge,$ and $\vee$, in $\phi$.
But first, we consider the following case:
\begin{itemize}
    \item $\phi(v_1,\dots,v_{p'},U_1,\ldots,U_{q'}) = \exists x_1 \cdots \exists x_i\exists Y_1 \cdots \exists Y_j \psi(v_1,\dots,v_{p'},x_1,\dots,x_i,U_1,\ldots,U_{q'},Y_1,\ldots,Y_j)$, where $\psi$ is an $\widehat{S}$-formula with $p'+i$ free vertex variables for some $i\in[k]_0$ and $q'+j$ free vertex set variables for some $j\in[s]_0$ (and $i+j\geq 1$). Our algorithm works as follows. For every element $X\in\mathsf{sgn}_S(G,\lambda,\mu)$, we use the algorithm guaranteed by the inductive hypothesis (w.r.t.~$d$) on $X$, $\psi$, $I'$ being  $I$ to which we append the pair $(i,j)$, and $J'$ being  $J$ to which we append the pair $(k,s)$ and return its answer. The definition of the recursive size of a signature directly yields the required time bound.

    For the forward direction of correctness, suppose that $(G,\lambda,\mu)$ satisfies $\phi$ w.r.t.~$I,J$. So, there exist  functions $\lambda'\in A_{G,k}$
     and $\mu'\in U_{G,\{0,1\}^s}$ such that $(G,\lambda\oplus\lambda',\mu\otimes\mu')$ satisfies $\psi$ w.r.t.~$I'$. 
     Notice that $\mathsf{sgn}_{\widehat{S}}(G,\lambda\oplus\lambda',\mu\otimes\mu')\in\mathsf{sgn}_S(G,\lambda,\mu)$. 
     So, there exists an iteration where the algorithm considers $X=\mathsf{sgn}_{\widehat{S}}(G,\lambda\oplus\lambda',\mu\otimes\mu')$. By the inductive hypothesis, the algorithm then returns yes.

    For the reverse direction of correctness, suppose that the algorithm returned yes. Let $X\in\mathsf{sgn}_S(G,\lambda,\mu)$ correspond to the iteration where the algorithm returned yes. Then, $X=\mathsf{sgn}_{\widehat{S}}(G,\lambda\oplus\lambda',\mu\otimes\mu')$ for some $\lambda'\in A_{G,k}$ and $\mu'\in U_{G,\{0,1\}^s}$. By the inductive hypothesis, $(G,\lambda\oplus\lambda',\mu\otimes\mu')$ satisfies $\psi$ w.r.t.~$I',J'$. Thus, by assigning $x_1=\lambda'(1),\ldots,x_i=\lambda'(i)$ and $Y_{j'}$, for all $j'\in[j]$, to be the set of vertices such that the $j'$-th coordinate assigned to them by $\mu'$ is $1$, we see that $(G,\lambda,\mu)$ satisfies $\phi$ w.r.t.~$I,J$.
\end{itemize}
Now, we continue with the induction on $\ell$.
Since $d\geq 1$, the base case where $\ell=0$ is already covered by the above (since, when $d\geq1$ and $\ell=0$, the formula must begin with $\exists$). 

Next, suppose $\ell\geq 1$, and assume correctness when we have at most $\ell-1$ occurrences of $\neg,\wedge,$ and $\vee$ (but the same number of free variables). 

According to the definition of $S$-formulas and since $\varphi$ excludes occurrences of $\forall$, it remains to consider the following cases:
\begin{itemize}    
    \item $\phi(v_1,\dots,v_{p'},U_1,\ldots,U_{q'}) = \neg \psi(v_1,\dots,v_{p'},U_1,\ldots,U_{q'})$, where $\psi$ is an $S$-formula with $p'$ free vertex variables and $q'$ free vertex set variables. Our algorithm simply calls the algorithm guaranteed by the inductive hypothesis (w.r.t.~$\ell$) on~$\mathsf{sgn}_S(G,\lambda,\mu),\psi,I,J$,  and returns the opposite answer. Correctness and time complexity requirements are immediate.

    \item Either $\phi(v_1,\dots,v_{p'},U_1,\ldots,U_{q'})=\psi(v_1,\dots,v_{p'},U_1,\ldots,U_{q'}) \wedge \varphi(v_1,\dots,v_{p'},U_1,\ldots,U_{q'})$ or $\phi(v_1,\dots,v_{p'},U_1,\ldots,U_{q'})=\psi(v_1,\dots,v_{p'},U_1,\ldots,U_{q'}) \vee \varphi(v_1,\dots,v_{p'},U_1,\ldots,U_{q'})$, where $\psi$ and $\varphi$ are both $S$-formulas with $p'$ free vertex variables and $q'$ free vertex set variables.
    
    Let us consider only the case where $\phi(v_1,\dots,v_{p'},U_1,\ldots,U_{q'})=\psi(v_1,\dots,v_{p'},U_1,\ldots,U_{q'}) \wedge \varphi(v_1,\dots,v_{p'},U_1,\ldots,U_{q'})$, since the proof for the other case is similar. Our algorithm simply calls the algorithm guaranteed by the inductive hypothesis (w.r.t.~$\ell$) twice: one time w.r.t.~$\mathsf{sgn}_S(G,\lambda,\mu),\psi,I$, and another time w.r.t.~$\mathsf{sgn}_S(G,\lambda,\mu),\varphi,I,J$;  the algorithm returns yes if and only if both calls returned yes. Again, correctness and time complexity requirements are immediate.
\end{itemize}
This completes the proof.
\end{proof}

We say two graphs $G$ and $G'$ are \textit{$S$-equivalent} if for every $S$-MSO $\phi$, $G$ satisfies $\phi$ iff $G'$ satisfies $\phi$.
The above lemma implies that if two graphs have the same $S$-signatures, then they are $S$-equivalent.
In fact, the inverse of this statement is also true.
We show that the $S$-equivalence is characterized by the $S$-signatures.

\begin{lemma}\label{lem-equiv}
    Let $G$ and $G'$ be two graphs, and $S$ be a sequence of pairs of natural numbers.
    Then $G$ and $G'$ are $S$-equivalent if and only if $\mathsf{sgn}_S(G) = \mathsf{sgn}_S(G')$.
\end{lemma}
\begin{proof}
The ``if'' direction directly follows from Lemma~\ref{lem:testSat}.
Thus, it suffices to show the ``only if'' direction. Towards the proof of this direction, we reuse the definition of a $p$-labeled $\{0,1\}^q$-colored graph $(G,\lambda,\mu)$ satisfying $\phi$ w.r.t.~$I,J$ from the proof of Lemma~\ref{lem:testSat}. Now, we say that two $p$-labeled $\{0,1\}^q$-colored graphs $(G,\lambda,\mu)$ and $(G',\lambda',\mu')$ are $S$-equivalent iff for every $S$-MSO $\phi$ with $p'\leq p$ free vertex variables and $q'\leq q$ set variables, and every compatible $I,J$, ($G,\lambda,\mu)$ satisfies $\phi$ w.r.t.~$I,J$ iff $(G',\lambda',\mu')$ does.
 Here, compatibility naturally means that $I=((a_1,b_1),\ldots,(a_r,b_r)), J=((\widetilde{a}_1,\widetilde{b}_1),\ldots,(\widetilde{a}_r,\widetilde{b}_r))$ are two sequences of pairs of nonnegative integers, for some $r$, so that $\sum_{i=1}^r\widetilde{a}_i=p, \sum_{i=1}^ra_i=p'$, $\sum_{i=1}^r\widetilde{b}_i=q$, $\sum_{i=1}^rb_i=q'$, and for every $i\in[r]$, $a_i\leq\widetilde{a}_i$ and $b_i\leq\widetilde{b}_i$. However, in what follows, we will only be use the supposition when $I=J,p'=p$ and $q'=q$, and then we do not need to mention $I,J$ explicitly.

So, we aim to prove the more general claim: if $(G,\lambda,\mu)$ and $(G',\lambda',\mu')$ are $S$-equivalent, then then have the same $S$-signatures. Next, suppose that $(G,\lambda,\mu)$ and $(G',\lambda',\mu')$ are $S$-equivalent. 

The proof is by induction on $|S|=d$. For all $i\in[q]$, let $A_i$ denote the set of vertices $v$ in $G$ such that the $i$-th coordinate of $\mu(v)$ is $1$.

\medskip
\noindent{\bf Basis ($S=\emptyset$):} Denote $(H,\alpha,\beta)=\mathsf{sgn}_\emptyset(G,\lambda,\mu)$ and $(H',\alpha',\beta')=\mathsf{sgn}_\emptyset(G',\lambda',\mu')$.

 Consider the following $\emptyset$-formula:
\[\begin{array}{ll}
\bigskip
\phi(v_1,\dots,v_{p},U_1,\ldots,U_{q}) &= \left(\bigwedge_{(i,j): (\alpha(i),\alpha(j))\in E(H)} \mathbf{adj}(v_i,v_j)\right) \wedge \left(\bigwedge_{(i,j): (\alpha(i),\alpha(j))\notin E(H)} \neg\mathbf{adj}(v_i,v_j)\right)\\

\bigskip
& \wedge\left(\bigwedge_{(i,j): \alpha(i)\neq\alpha(j)} v_i=v_j\right) \wedge \left(\bigwedge_{(i,j): \alpha(i)\neq\alpha(j)} v_i\neq v_j\right)\\

& \wedge \left(\bigwedge_{i,j: \alpha(i)\in A_j}v_i\in U_j\right)\wedge \left(\bigwedge_{i,j: \alpha(i)\notin A_j}v_i\notin U_j\right).
\end{array}\]

In particular, the formula is defined so that $(H,\alpha,\beta)$ satisfies it, and any $p$-labeled $\{0,1\}^q$-colored graph that satisfies it must be isomorphic to $(H,\alpha)$. Now, by the definition of $\mathsf{sgn}$, specifically the derivation of $(H,\alpha,\beta)$ from $(G,\lambda,\mu)$, it is immediate that $(G,\lambda,\mu)$ satisfies $\phi$ as well. So, since $(G,\lambda,\mu)$ and $(G',\lambda',\mu')$ are $\emptyset$-equivalent, it follows that $(G',\lambda')$ satisfies $\phi$. Then, again by the definition of $\mathsf{sgn}$, it is immediate that $(H',\alpha',\beta')$ satisfies $\phi$ as well. Hence, $(H',\alpha',\beta')$ is isomorphic to $(H,\alpha,\beta)$, that is, $\mathsf{sgn}_\emptyset(G,\lambda,\mu) = \mathsf{sgn}_\emptyset(G',\lambda',\mu')$.

\medskip
\noindent{\bf Step ($|S|=r\geq 1$):}  Here, we suppose correctness for all sequences of natural numbers whose length is at most $r-1$. Denote $S= (k,s)+\widehat{S}$.

We will only show the containment $\mathsf{sgn}_S(G,\lambda,\mu)\subseteq \mathsf{sgn}_S(G',\lambda',\mu')$, since the proof of the other containment $\mathsf{sgn}_S(G',\lambda',\mu')\subseteq \mathsf{sgn}_S(G,\lambda,\mu)$ is symmetric. For this purpose, consider some arbitrary $X\in \mathsf{sgn}_S(G,\lambda,\mu)$. Thus, $X=\mathsf{sgn}_{\widehat{S}}(G,\lambda\oplus\widehat{\lambda},\mu\otimes\widehat{\mu})$ for some
$\widehat{\lambda}\in A_{G,k}$
     and $\widehat{\mu}\in U_{G,\{0,1\}^s}$.
We define an $S$-formula $\phi(v_1,\ldots,v_p,U_1,\ldots,U_q)$ as follows:

\[\phi(v_1,\ldots,v_p,U_1,\ldots,U_q)=\exists x_1\cdots\exists x_k\exists Y_1\cdots\exists Y_s \psi(v_1,\ldots,v_p,x_1,\ldots,x_k,U_1,\ldots,U_q,Y_1,\ldots,Y_s),\]

where $\psi(v_1,\ldots,v_p,x_1,\ldots,x_k,U_1,\ldots,U_q,Y_1,\ldots,Y_s)$ is the $\widehat{S}$-formula defined as follows:
\[\begin{array}{l}
\medskip
\psi(v_1,\ldots,v_p,x_1,\ldots,x_k,U_1,\ldots,U_q,Y_1,\ldots,Y_s)=\\
\medskip\left(\bigwedge_{\varphi\in UNSAT}\neg\varphi(v_1,\ldots,v_p,x_1,\ldots,x_k,U_1,\ldots,U_q,Y_1,\ldots,Y_s)\right)\wedge\\
\left(\bigwedge_{\varphi\in SAT}\varphi(v_1,\ldots,v_p,x_1,\ldots,x_k,U_1,\ldots,U_q,Y_1,\ldots,Y_s)\right),
\end{array}\]

where $UNSAT$ is the collection of all $\widehat{S}$-formulas with $p+k$ free vertex variables and $q+s$ vertex set variables that are not satisfied by $(G,\lambda\oplus\widehat{\lambda},\mu\otimes\widehat{\mu})$, and $SAT$ is the collection of all such formulas that are satisfied by $(G,\lambda\oplus\widehat{\lambda},\mu\otimes\widehat{\mu})$.

Observe that, by the definition of $\psi$, $(G,\lambda\oplus\widehat{\lambda},\mu\otimes\widehat{\mu})$ satisfies it.
So, by assigning $x_1=\widehat{\lambda}(1),x_2=\widehat{\lambda}(2),\ldots,x_k=\widehat{\lambda}(k)$, and $Y_i$, for all $i\in[s]$ be the set of vertices whose $i$-th coordinate in $\widehat{\mu}$ is $1$, we see that $(G,\lambda,\mu)$ satisfies $\phi$. So, since $(G,\lambda,\mu)$ and $(G',\lambda',\mu')$ are $S$-equivalent, we get that $(G',\lambda',\mu')$ satisfies $\phi$ as well. Thus, let $w$ be an assignment of vertices and vertex sets of $G'$ to the variables $x_1,\ldots,x_k,Y_1,\ldots,Y_s$ that witnesses this.
Accordingly, define $\widehat{\lambda}': [k]\rightarrow V(G')$ so that for all $i\in[k]$, $\widehat{\lambda}'(i)=w(x_{i})$. Additionally, define $\widehat{\mu}': V(G')\rightarrow \{0,1\}^s$ so that for all $v\in V(G')$ and $i\in[s]$, the $i$-th coordinate of $\widehat{\mu}'(v)$ is $1$ iff $v\in w(Y_i)$.
By the definition of $\widehat{\lambda}'$ and $\widehat{\mu}'$, it immediately follows that $(G',\lambda'\oplus\widehat{\lambda}',\mu'\otimes\widehat{\mu}')$ satisfies $\psi$.
However, by the definition of $\psi$, this implies that $(G,\lambda\oplus\widehat{\lambda},\mu\otimes\widehat{\mu})$ and $(G',\lambda'\oplus\widehat{\lambda}',\mu'\otimes\widehat{\mu}')$ are $\widehat{S}$-equivalent. In turn, by the inductive hypothesis, which can be applied since $|\widehat{S}|\leq d-1$, we derive that $\mathsf{sgn}_{\widehat{S}}(G,\lambda\oplus\widehat{\lambda},\mu\otimes\widehat{\mu})=\mathsf{sgn}_{\widehat{S}}(G',\lambda'\oplus\widehat{\lambda}',\mu'\otimes\widehat{\mu}')$ (under isomorphism). By the definition of $\mathsf{sgn}_S(G',\lambda',\mu')$, it includes $\mathsf{sgn}_{\widehat{S}}(G',\lambda'\oplus\widehat{\lambda}',\mu'\otimes\widehat{\mu}')$, and therefore it includes $\mathsf{sgn}_{\widehat{S}}(G,\lambda\oplus\widehat{\lambda},\mu\otimes\widehat{\mu})$. Since the choice of $X$ was arbitrary, we conclude that $\mathsf{sgn}_S(G,\lambda,\mu)\subseteq \mathsf{sgn}_S(G',\lambda',\mu')$.
\end{proof}

\subsection{Extension to the MSO scoring problem} \label{sec-ext}

To extend our results in the previous section to \textsc{MSO Scoring}, we need to introduce a new concept called \textit{evaluation maps}.
As before, let $S$ be a sequence of pairs of natural numbers with $|S| \geq 1$.
Suppose $S = ((k,s)+S')$.
Consider a $p$-labeled and $P$-colored graph $(G,\lambda,\mu)$, and a weight function $w: V(G) \rightarrow \mathbb{F}$ where $\mathbb{F}$ is a semi-field.
Recall that the $S$-signature of $(G,\lambda,\mu)$ is defined as
\begin{equation*}
\mathsf{sgn}_S(G,\lambda,\mu) = \{\mathsf{sgn}_{S'}(G,\lambda \oplus \lambda',\mu \otimes \mu') :\lambda' \in \varLambda_{G,k} \text{ and } \mu' \in U_{G,\{0,1\}^s}\}.
\end{equation*}
We now view $\mathsf{sgn}_S(G,\lambda,\mu)$ as a set (in which each element is the $S'$-signature of some graph).
For each element $x \in \mathsf{sgn}_S(G,\lambda,\mu)$, let $\varGamma_x$ be the set consisting of all pairs $(\lambda',\mu') \in \varLambda_{G,k} \times U_{G,\{0,1\}^s}$ satisfying $\mathsf{sgn}_{S'}(G,\lambda \oplus \lambda',\mu \otimes \mu') = x$.
Then the \textit{$(S,w)$-evaluation map} of $(G,\lambda,\mu)$ is a function $f: \mathsf{sgn}_S(G,\lambda,\mu) \rightarrow \mathbb{F}$ defined as
\begin{equation*}
    f(x) = \sum_{(\lambda',\mu') \in \varGamma_x} \left(\prod_{i=1}^k w(\lambda'(i)) \cdot \prod_{i=1}^s \prod_{v \in V(G)} (w(v))^{\mu_i'(v)}\right),
\end{equation*}
where $\mu_i'(v) \in \{0,1\}$ denotes the $i$-th bit of $\mu'(v) \in \{0,1\}^s$.

We first state two modified versions of item (iii) of Lemma~\ref{lem-sgncomp}. Towards this, suppose we are given $p'\in [p]_0$ and $q'$ such that $P$ is of the form $\{0,1\}^{q'}\times P'$ for some $P'$. Then, we define the {\em simple $(p',q',w)$-evaluation map} of $(G,\lambda,\mu)$ as $g: \mathsf{sgn}_S(G,\lambda,\mu)\rightarrow \mathbb{F}$ where $g(x)=\prod_{i=1}^{p'} w(\lambda(i)) \cdot \prod_{i=1}^{q'} \prod_{v \in V(G)} (w(v))^{\mu_i(v)}$.

\begin{lemma}\label{lem-sgncomp-new1}
    Let $(G,\lambda,\mu)$ be a $p$-labeled $P$-colored graph, $w: V(G)\rightarrow\mathbb{F}$, and $S = ((k_1,s_1),\dots,(k_d,s_d))$. Suppose we are given $p'\in [p]_0$ and $q'\in [s]_0$ such that $P$ is of the form $\{0,1\}^{q'}\times P'$ for some $P'$. Let $(A,B)$ be a separation of $G$ with $|A \cap B| = z$ and $\mu':V(G) \rightarrow [z]_0$ be a function that maps all vertices in $V(G) \backslash (A \cap B)$ to $0$ and maps $A \cap B$ bijectively to $[z]$.
        Denote $\lambda_A=\mathsf{sort}(\lambda^{|A})$, $\lambda_B=\mathsf{sort}(\lambda^{|B})$, $\mu_A=(\mu \otimes \mu')_{|A}$ and $\mu_B=(\mu \otimes \mu')_{|B}$.
        Then given only $S$, $\lambda^{-1}(A)$, $\lambda^{-1}(B)$, and the simple $(p',q',w)$-evaluation maps of $(G[A],\lambda_A,\mu_A)$ and $(G[B],\lambda_B,\mu_B)$, one can compute the simple $(p',q',w)$-evaluation map of $(G,\lambda,\mu)$ in time
        \begin{equation*}
            2^{\sum_{i=1}^dk_i}\cdot\left(\sum_{i=1}^d (k_i+s_i) + \lVert \mathsf{sgn}_S(G[A],\lambda_A,\mu_A) \rVert + \lVert \mathsf{sgn}_S(G[B],\lambda_B,\mu_B) \rVert\right)^{O(1)}.
        \end{equation*}
\end{lemma}

\begin{proof}[Proof sketch]
To avoid copy-past of text, we will only describe the changes in the proof of Lemma~\ref{lem-sgncomp}. Again, we prove item (iii*) instead of (iii), where, now, the basic $(p',q',w)$-evaluation maps of $(G[A],\lambda_A,\mu_A)$ and $(G[B],\lambda_B,\mu_B)$ are given, and we are to compute the simple $(p',q',w)$-evaluation map of $(G,\lambda^*,\mu^*)$.

In the basis, to obtain $g(H,\alpha,\beta)$, we compute $g(H_A,\alpha_A,\beta_A)\cdot g(H_B,\alpha_B,\beta_B)$ and divide it by the multiplication of all of the following terms:
\begin{itemize}
\item For all $i\in J_A\cap J_B$, $w(\alpha_A(i))$ (since each of them is double-counted).
\item For all $i\in [p]\setminus J_A$, $w(\alpha_A(i))$ (since each of them is irrelevant).
\item For all $i\in[p]\setminus J_B$, $w(\alpha_B(i))$ (since each of them is irrelevant).
\end{itemize}

For the inductive step, where $|S|\geq 1$, the only  difference is that we use the inductive hypothesis of this lemma rather than of Lemma~\ref{lem-sgncomp}, and thereby obtain the  simple evaluation map. Notice that the same signature $Z$ might be computed multiple times, but in all of them, $g(Z)$ is the same, and hence there is no conflict.
\end{proof}

\begin{lemma}\label{lem-sgncomp-new2}
    Let $(G,\lambda,\mu)$ be a $p$-labeled and $P$-colored graph, $w: V(G)\rightarrow\mathbb{F}$, and $S = ((k_1,s_1),\dots,(k_d,s_d))$ with $|S|\geq 1$. Let $(A,B)$ be a separation of $G$ with $|A \cap B| = z$ and $\mu':V(G) \rightarrow [z]_0$ be a function that maps all vertices in $V(G) \backslash (A \cap B)$ to $0$ and maps $A \cap B$ bijectively to $[z]$.
        Denote $\lambda_A=\mathsf{sort}(\lambda^{|A})$, $\lambda_B=\mathsf{sort}(\lambda^{|B})$, $\mu_A=(\mu \otimes \mu')_{|A}$ and $\mu_B=(\mu \otimes \mu')_{|B}$.
        Then given only $S$, $\lambda^{-1}(A)$, $\lambda^{-1}(B)$, and the $(S,w)$-evaluation maps of $(G[A],\lambda_A,\mu_A)$ and $(G[B],\lambda_B,\mu_B)$, one can compute the $(S,w)$-evaluation map of $(G,\lambda,\mu)$ in time
        \begin{equation*}
            2^{\sum_{i=1}^dk_i}\cdot\left(\sum_{i=1}^d (k_i+s_i) + \lVert \mathsf{sgn}_S(G[A],\lambda_A,\mu_A) \rVert + \lVert \mathsf{sgn}_S(G[B],\lambda_B,\mu_B) \rVert\right)^{O(1)}.
        \end{equation*}
\end{lemma}

\begin{proof}[Proof sketch]
To avoid copy-past of text, we will only describe the changes in the proof of Lemma~\ref{lem-sgncomp}. Here, we do not use induction. In particular, we directly move to the arguments given in the step, where $|S|\geq 1$. Then, to compute $Z$, we apply the algorithm from Lemma~\ref{lem-sgncomp-new1} in order to obtain $g(Z)$. Besides this, the critical difference is that when we insert $Z$ into $\cal C$, if it already exists, then we add $g(Z)$ to the existing $f(Z)$, and else we just assign $f(Z)=g(Z)$.
\end{proof}

With the help of Lemma~\ref{lem-sgncomp-new2}, we can prove a generalization of Theorem~\ref{thm:twcomputation}.

\begin{theorem}\label{thm:new-twcomputation}
    Given a sequence $S=((k_1,s_1),\dots,(k_d,s_d))$ of pairs of natural numbers with $|S|\geq 1$, a graph $G$ with a nice tree decomposition $(T,\beta)$ of width $t$, and a function $w:V(G) \rightarrow \mathbb{F}$, one can compute the $(S,w)$-evaluation map of $G$ in time 
    \begin{equation*}
        f(d)\cdot \left(\sum_{i=1}^d2^{O(k_i\log t)}+\sum_{i=1}^{d-1}2^{O(s_i t)}+ \sum_{i=1}^{d}2^{O(s_d k_i)}\right)\cdot|V(T)| \cdot  \left(\max_{\mu \in U_{G,[t+1]_0}} \lVert \mathsf{sgn}_S(G,-,\mu) \rVert \right)^{O(1)},
    \end{equation*}
    for some computable function $f$ of $d$.
\end{theorem}

\begin{proof}[Proof sketch]
To avoid copy-past of text, we will only describe the changes in the proof of Theorem~\ref{thm:twcomputation}.  Now, the purpose of each table entry $M[x]$, $x\in V(T)$, will be to store the $(S,w)$-evaluation map of $(G[\gamma(x)],-,\mu^*_x)$ rather than just its signature.

First, in case $x$ is a leaf node, we assign the weight $0$.

Second, in case $x$ is an Introduce node, the first change is to apply a generalized version of Lemma~\ref{lem:bruteForce1}, to compute the $(S,w)$-evaluation map of $X$ rather than just $X$: the brute-force computation in the proof can be trivially modified to account for weights. Then, instead of applying item (iii) of Lemma~\ref{lem-sgncomp}, we apply Lemma~\ref{lem-sgncomp-new2}, and store the result in $M[x]$.

In case $x$ is a Forget node, the proof remains the same (keeping the same evaluation map).

Lastly, in case $x$ is a Join node, the only difference is that  instead of applying item (iii) of Lemma~\ref{lem-sgncomp}, we apply Lemma~\ref{lem-sgncomp-new2}, and store the result in $M[x]$.
\end{proof}

Next, we prove the counterpart of Lemma~\ref{lem:testSat}.

\begin{lemma} \label{lem:new-testSat}
    Let $S$ be a sequence of pairs of natural numbers with $|S|\geq 1$.
    Given the $S$-signature $\mathsf{sgn}_S(G)$ of a graph $G$, the $(S,w)$-evaluation map of $G$ for a function $w:V(G) \rightarrow \mathbb{F}$, and an $S$-MSO$^*$ $\phi$, we can compute $\mathsf{scr}(G,w,\phi)$ in $(||\mathsf{sgn}_S(G)||+|\phi|)^{O(1)}$ time.
\end{lemma}

\begin{proof}[Proof sketch]
Let $S=(k,s)+S'$. Recall that by definition, for an assignment $\alpha = (v_1,\dots,v_k,V_1,\dots,V_s)$, we write $w(\alpha) = (\prod_{i=1}^k w(v_1)) \cdot (\prod_{i=1}^s \prod_{v \in V_i} w(v))$, and that  $\mathsf{scr}(G,w,\phi) = \sum_{\alpha \in \mathcal{A}_\phi(G)} w(\alpha)$, where $\mathcal{A}_\phi(G)$ denotes the set of all satisfying assignments of $\phi$ on  $G$.

Also recall that for every $x\in\mathsf{sgn}_S(G)$,  $f(x) = \displaystyle{\sum_{(\lambda,\mu) \in \varGamma_x} \left(\prod_{i=1}^k w(\lambda(i)) \cdot \prod_{i=1}^s \prod_{v \in V(G)} (w(v))^{\mu_i(v)}\right)}$,  where $\varGamma_x$ be the set consisting of all pairs $(\lambda,\mu) \in \varLambda_{G,k} \times U_{G,\{0,1\}^s}$ satisfying $\mathsf{sgn}_{S'}(G, \lambda,\mu) = x$.


Now, notice that the proof of Lemma~\ref{lem:testSat} implies that, in time
$(||\mathsf{sgn}_S(G)||+|\phi|)^{O(1)}$, we can obtain all the answers for the following for all $x\in\mathsf{sgn}_S(G)$: Does there exist $(\lambda,\mu) \in \varGamma_x$ so that $(G,\lambda,\mu)$ satisfies $\phi$? Furthermore, the proof shows that if there exists such $(\lambda,\mu) \in \varGamma_x$, then, in fact, all  $(\lambda',\mu') \in \varGamma_x$ are such that $(G,\lambda',\mu')$ satisfies $\phi$. Let $X_{\mathsf{yes}}$ denote the set of $x\in\mathsf{sgn}_S(G)$ such that the above answer is positive.

For all $x\in X_{\mathsf{yes}}$, let $M_x$ denote the set of all assignments $\alpha = (v_1,\dots,v_k,V_1,\dots,V_s)\in \mathcal{A}_\phi(G)$ such that for all $i\in[k]$, $v_i=\lambda(i)$, and for all $i\in[s]$, $V_i$ is the set of all vertices whose $i$-coordinate in $\mu$ is $1$,  for some $(\lambda,\mu) \in \varGamma_x$. Then, by the definitions of the weight of an assignment and of $f$, we get that $f(x)=\sum_{\alpha\in M_x}w(\alpha)$. Notice that for distinct $x,y\in X_{\mathsf{yes}}$, $M_x\cap M_y=\emptyset$, and that $\bigcup_{x\in X_{\mathsf{yes}}}M_x=\mathcal{A}_\phi(G)$.

Thus, algorithmically, we derive that it suffices to compute $f(G)$ (using Theorem~\ref{thm:new-twcomputation}) and then $X_{\mathsf{yes}}$ by iterating through every signature in $\mathsf{sign}_S(G)$ and using the arguments for Lemma~\ref{lem:testSat}, eventually returning $\sum_{x\in X_{\mathsf{yes}}}f(x)$ as $\mathsf{scr}(G,w,\phi)$.
Due to the running times stated in Theorem~\ref{thm:new-twcomputation} and (the proof of) Lemma~\ref{lem:testSat}, the proof is complete.
\end{proof}

\subsection{Signature size of bounded-treewidth graphs} \label{sec-size}
In this section, we prove upper bounds for the recursive size of signatures of bounded-treewidth graphs.
Specifically, we are going to prove the following theorem.
\begin{theorem} \label{thm-sgnsize}
    For a sequence $S = ((k_1,s_1),\dots,(k_d,s_d))$ of pairs of natural numbers and a $p$-labeled and $P$-colored graph $(G,\lambda,\mu)$ with $\mathbf{tw}(G) \leq t$, $\lVert \mathsf{sgn}_S(G,\lambda,\mu) \rVert$ is at most
    {\small \begin{equation*}
        f(d) \cdot \left(\sum_{i=1}^d \sum_{j=0}^{i-1} \exp^{(i)}(\hat{O}(s_j k_i)) + \sum_{i=0}^d \sum_{j=0}^i \exp^{(i)}(\hat{O}(t_j k_i)) + \sum_{i=1}^d \exp^{(i)}(\hat{O}(k_i \log t)) + \sum_{i=1}^d 2^{O(s_d k_i)} \right)
    \end{equation*}}
    for some computable function $f$, where $k_0 = p$, $s_0 = \log |P|$, and $t_j = \min\{k_j,t\}$.
\end{theorem}

To prove the theorem, our goal is to bound $\lVert \mathsf{sgn}_S(G,\lambda,\mu) \rVert$ for every labeled and colored graph $(G,\lambda,\mu)$ and every sequence $S$ of pairs of natural numbers.
Instead of bounding the overall size $\lVert \cdot \rVert$ of the signatures directly, we first bound the size of the signatures \textit{as sets}, i.e., $|\mathsf{sgn}_S(G,\lambda,\mu)|$.

We shall apply induction on the length of $S$.
To this end, we need a stronger induction statement that bounds not only the size of the signatures but also bounds the number of different signatures a certain class of graphs can have.
We begin with introducing some notations.
For a class $\mathcal{G}$ of labeled and colored graphs, consider the natural equivalence relation $\sim$ on $\mathcal{G}$ defined as $(G,\lambda,\mu) \sim (G',\lambda',\mu')$ if $\mathsf{core}(G,\lambda,\mu)$ and $\mathsf{core}(G',\lambda',\mu')$ are isomorphic.
Let $\varPhi(\mathcal{G})$ be the set of equivalence classes of $\sim$.
Then for a sequence $S$ of pairs of natural numbers, we write $\delta_S(\mathcal{G}) = \sup_{(G,\lambda,\mu) \in \mathcal{G}} |\mathsf{sgn}_S(G,\lambda,\mu)|$ and 
\begin{equation*}
    \Delta_S(\mathcal{G}) = \sup_{\mathcal{G}' \in \varPhi(\mathcal{G})} |\{\mathsf{sgn}_S(G,\lambda,\mu):(G,\lambda,\mu) \in \mathcal{G}'\}|.
\end{equation*}
In other words, $\delta_S(\mathcal{G})$ is the maximum size of the $S$-signature of a graph in $\mathcal{G}$ and $\Delta_S(\mathcal{G})$ is the maximum number of different $S$-signatures the graphs in an equivalence class $\mathcal{G}' \in \varPhi(\mathcal{G})$ can have.
Denote by $\mathcal{G}_{p,P,t}$ the set of all $p$-labeled and $P$-colored graphs with treewidth at most $t$.
We shall bound $\delta_S(\mathcal{G}_{p,P,t})$ and $\Delta_S(\mathcal{G}_{p,P,t})$, by induction on $|S|$.

\subsubsection{Characterization sequences and recurrences}
Since the bounds are rather complicated, before explicitly stating and proving them, we first present the key ingredient of our induction argument, which are certain recursive relations for the numbers $\delta_S(\cdot)$ and $\Delta_S(\cdot)$.
To establish the recurrences, we need to introduce an important notion, called \textit{characterization sequences}, and show its relation to the signatures.

Consider a sequence $S$ of pairs of natural numbers.
For each graph $(G,\lambda,\mu) \in \mathcal{G}_{p,P,t}$, we associate it with a sequence $\varGamma_S(G,\lambda,\mu)$ defined as follows.
Let $R = \text{Im}(\lambda) \subseteq V(G)$.
We have $|R| \leq p$.
We apply Lemma~\ref{lem-decomp} with $R$ to obtain the sets $V_0,V_1,\dots,V_{6p} \subseteq V(G)$ satisfying the three conditions in the lemma.
The image of $\lambda$, which is $R$, is contained in $V_0$, by (ii) of Lemma~\ref{lem-decomp}.
So $\lambda^{|V_0}$ is a $p$-labeling on $G[V_0]$.
For $i \in [6p]$, let $\pi_i: [4t] \rightarrow V_0$ be a function that maps the numbers in $[4t]$ bijectively to the vertices in $V_0 \cap V_i$ and let $\mu_i: V_i \rightarrow P \times [4t]_0$ be the function defined as $\mu_i(v) = (\mu(v),\pi_i^{-1}(v))$ for $v \in V_0 \cap V_i$ and $\mu_i(v) = (\mu(v),0)$ for $v \in V_i \backslash V_0$.
Note that $\pi_i$ and $\mu_i$ are well-defined since $|V_0 \cap V_i| = 4t$ by (iii) of Lemma~\ref{lem-decomp}.
By construction, each $\pi_i$ is a $4t$-labeling on $G[V_0]$ and each $\mu_i$ is a $(P \times [4t]_0)$-coloring on $G[V_i]$.
Next, for $i \in [6p]$, let $X_i = \text{Im}(\lambda) \cap V_i \subseteq V_0 \cap V_i$ and define a function $\gamma_i: \pi_i^{-1}(X_i) \rightarrow V_i$ as $\gamma_i(a) = \pi_i(a)$ for all $a \in \pi_i^{-1}(X_i)$, i.e., $\gamma_i$ is obtained from $\pi_i^{|X_i}$ by change its codomain to $V_i$.
Finally, let $\lambda_i = \mathsf{sort}(\gamma_i)$, which is a $|X_i|$-labeling on $G[V_i]$.
Now we are ready to define the sequence
\begin{equation*}
    \varGamma_S(G,\lambda,\mu) = ((G[V_0],\lambda^{|V_0} \oplus \pi_1 \oplus \cdots \oplus \pi_{6p},\mu_{|V_0}),\mathsf{sgn}_S(G[V_1],\lambda_1,\mu_1),\dots,\mathsf{sgn}_S(G[V_{6p}],\lambda_{6p},\mu_{6p})).
\end{equation*}
We call $\varGamma_S(G,\lambda,\mu)$ the \textit{$S$-characterization sequence}\footnote{Our construction of the $S$-characterization sequence is not necessarily unique, because the decomposition in Lemma~\ref{lem-decomp} and the choices of $\pi_1,\dots,\pi_{6p}$ are not unique. However, we can fix a canonical choice of $\varGamma_S(G,\lambda,\mu)$ for each $(G,\lambda,\mu)$.} of $(G,\lambda,\mu)$.
We observe some simple facts about this sequence.
Clearly, the length of $\varGamma_S(G,\lambda,\mu)$ is $6p+1$.
The first element of $\varGamma_S(G,\lambda,\mu)$, i.e., $(G[V_0],\lambda^{|V_0} \oplus \pi_1 \oplus \cdots \oplus \pi_{6p},\mu_{|V_0})$, is a $(p+24pt)$-labeled and $P$-colored graph (of treewidth at most $t$), i.e., an element in $\mathcal{G}_{p+24pt,P,t}$.
The graph $(G[V_i],\lambda_i,\mu_i)$ corresponding to the $(i+1)$-th element of $\varGamma_S(G,\lambda,\mu)$, a $|X_i|$-labeled and $(P \times [4t]_0)$-colored graph.
Note that 
\begin{equation*}
    |X_i| \leq \min\{p,|V_0 \cap V_i|\} =\min\{p,4t\} = 4 \cdot\min\{p,t\}.
\end{equation*}
We also notice that both $|X_i|$ and $\mathsf{core}(G[V_i],\lambda_i,\mu_i)$ are uniquely determined by the first element $(G[V_0],\lambda^{|V_0} \oplus \pi_1 \oplus \cdots \oplus \pi_{6p},\mu_{|V_0})$ of $\varGamma_S(G,\lambda,\mu)$.
Indeed, from $(G[V_0],\lambda^{|V_0} \oplus \pi_1 \oplus \cdots \oplus \pi_{6p},\mu_{|V_0})$, we can retrieve the maps $\lambda^{|V_0}$ and $\pi_i$, which together determine $X_i$ (and thus $|X_i|$), simply because $X_i = \text{Im}(\lambda^{|V_0}) \cap \text{Im}(\pi_i)$.
Furthermore, $(\mu_i)_{|X_i}$ is determined by $\mu_{|V_0}$, $\pi_i$, and $X_i$.
Since $\mathsf{core}(G[V_i],\lambda_i,\mu_i) = (G[X_i],\lambda_i^{|X_i},(\mu_i)_{|X_i})$ by construction, it is then determined by $(G[V_0],\lambda^{|V_0} \oplus \pi_1 \oplus \cdots \oplus \pi_{6p},\mu_{|V_0})$.
This observation will be used later in the proof of Lemma~\ref{lem-recurrence}.

A key property of the $S$-characterization sequences is that they determines the $S$-signatures, i.e., $\mathsf{sgn}_S(G,\lambda,\mu)$ is uniquely determined by $\varGamma_S(G,\lambda,\mu)$.
\begin{lemma} \label{lem-character} 
    Let $(G,\lambda,\mu) \in \mathcal{G}_{p,P,t}$.
    One can compute $\mathsf{sgn}_S(G,\lambda,\mu)$ knowing only $\varGamma_S(G,\lambda,\mu)$.
\end{lemma}
\begin{proof}
Recall that how we constructed $\varGamma_S(G,\lambda,\mu)$.
Let $V_0,V_1,\dots,V_{6p}$, $\pi_1,\dots,\pi_{6p}$, $\mu_1,\dots,\mu_{6p}$, $X_1,\dots,X_{6p}$, $\gamma_1,\dots,\gamma_{6p}$, and $\lambda_1,\dots,\lambda_{6p}$ be as in the construction of $\varGamma_S(G,\lambda,\mu)$.
For $i \in [6p]$, denote by $\hat{\pi}_i:V(G) \rightarrow [4t]_0$ the function where $\hat{\pi}_i(v) = \pi_i^{-1}(v)$ for all $v \in V_0 \cap V_i$ and $\pi_i'(v) = 0$ otherwise. Also, write $G_i = G[V_0 \cup (\bigcup_{j=1}^i V_j)]$ for $i \in [6p]$. 

We show how to compute $\mathsf{sgn}_S(G,\lambda,\mu)$ if we are \textit{only} given $\varGamma_S(G,\lambda,\mu)$.
Specifically, we are going to prove that, for all $i \in [6p]_0$, $\mathsf{sgn}_S(G_i,\lambda^{|V(G_i)},(\mu \otimes \hat{\pi}_{6p} \otimes \hat{\pi}_{6p-1} \otimes \cdots \otimes \hat{\pi}_{i+1})_{|V(G_i)})$ can be computed using $\varGamma_S(G,\lambda,\mu)$.
We apply induction on $i$.
Consider the base case $i = 0$.
As $G_0 = G[V_0]$, what we want to compute is $\mathsf{sgn}_S(G[V_0],\lambda^{|V_0},(\mu \otimes \hat{\pi}_{6p} \otimes \hat{\pi}_{6p-1} \otimes \cdots \otimes \hat{\pi}_1)_{|V_0})$.
By construction, the first element of $\varGamma_S(G,\lambda,\mu)$ is $(G[V_0],\lambda^{|V_0} \oplus \pi_1 \oplus \cdots \oplus \pi_{6p},\mu_{|V_0})$, from which we can obtain the functions $\mu_{|V_0},\pi_1,\dots,\pi_{6p}$.
Knowing $\mu_{|V_0},\pi_1,\dots,\pi_{6p}$, we can further compute the function $(\mu \otimes \hat{\pi}_{6p} \otimes \hat{\pi}_{6p-1} \otimes \cdots \otimes \hat{\pi}_1)_{|V_0}$.
Therefore, from the first element of $\varGamma_S(G,\lambda,\mu)$, we can obtain the graph $(G[V_0],\lambda^{|V_0},(\mu \otimes \hat{\pi}_{6p} \otimes \hat{\pi}_{6p-1} \otimes \cdots \otimes \hat{\pi}_1)_{|V_0})$ itself and thus its $S$-signature.
Now suppose we have $\mathsf{sgn}_S(G_{i-1},\lambda^{|V(G_{i-1})},(\mu \otimes \hat{\pi}_{6p} \otimes \hat{\pi}_{6p-1} \otimes \cdots \otimes \hat{\pi}_i)_{|V(G_{i-1})})$, and we show how to use it to compute $\mathsf{sgn}_S(G_i,\lambda^{|V(G_i)},(\mu \otimes \hat{\pi}_{6p} \otimes \hat{\pi}_{6p-1} \otimes \cdots \otimes \hat{\pi}_{i+1})_{|V(G_i)})$.

For convenience, write $\tau = \mu \otimes \hat{\pi}_{6p} \otimes \hat{\pi}_{6p-1} \otimes \cdots \otimes \hat{\pi}_{i+1}$.
Then what we want to compute is just $\mathsf{sgn}_S(G_i,\lambda^{|V(G_i)},\tau_{|V(G_i)})$.
We can obtain $V_0 \cap V_1,\dots,V_0 \cap V_{6p}$ from $\varGamma_S(G,\lambda,\mu)$, as they are just $\text{Im}(\pi_1),\dots,\text{Im}(\pi_{6p})$.
For each $j \in [i]$, construct a function $f_j: [4t]_0 \rightarrow [4t]_0$ as
\begin{equation*}
    f_j(a) = \left\{
    \begin{array}{ll}
        \pi_j^{-1}(\pi_i(a)) & \text{if } a \neq 0 \text{ and } \pi_i(a) \in V_0 \cap V_j, \\
        0 & \text{otherwise}.
    \end{array}
    \right.
\end{equation*}
Then define $f: P \times [4t]_0 \rightarrow P \times [4t]_0^{i}$ as $f(a,x) = (a,f_1(x),\dots,f_i(x))$.
Now it is easy to verify that $(\tau \otimes \hat{\pi}_i)_{|V_i} = f \circ \mu_i$.
Thus, by (ii) of Lemma~\ref{lem-sgncomp}, we can compute $\mathsf{sgn}_S(G[V_i],\lambda_i,(\tau \otimes \hat{\pi}_i)_{|V_i})$ using only $f$ and $\mathsf{sgn}_S(G[V_i],\lambda_i,\mu_i)$.
Next, we show how to further compute $\mathsf{sgn}_S(G[V_i],\mathsf{sort}(\lambda^{|V_i}),(\tau \otimes \hat{\pi}_i)_{|V_i})$.
It suffices to construct a function $f$ (from what we know) satisfying that $\mathsf{sort}(\lambda^{|V_i}) = \lambda_i \circ f$, and then applies (i) of Lemma~\ref{lem-sgncomp}.
The domain of $\lambda^{|V_i}$ is $\lambda^{-1}(V_i) = \lambda^{-1}(V_0 \cap V_i) \subseteq [p]$, which can be known from $\varGamma_S(G,\lambda,\mu)$ as it just consists of the numbers $a \in [p]$ such that $\lambda(a) \in V_0 \cap V_i$.
Thus, we can construct the bijection $f': [|\lambda^{-1}(V_i)|] \rightarrow \lambda^{-1}(V_i)$ such that $\mathsf{sort}(\lambda^{|V_i}) = \lambda^{|V_i} \circ f'$.
Recall that the domain of $\gamma_i$ is $\pi_i^{-1}(X_i)$ and $\lambda_i = \mathsf{sort}(\gamma_i)$.
Note that $X_i$ can be known from $\varGamma_S(G,\lambda,\mu)$, since it is just $\text{Im}(\lambda^{|V_0}) \cap (V_0 \cap V_i)$.
Thus, we can construct the bijection $f'': [|X_i|] \rightarrow \pi_i^{-1}(X_i)$ such that $\gamma_i = \mathsf{sort}(\gamma_i) \circ f'' = \lambda_i \circ f''$.
Furthermore, it is easy to verify that $\lambda^{|V_i} = \gamma_i \circ (\hat{\pi}_i \circ \lambda^{|V_0})^{| \pi_i^{-1}(X_i)}$, and $(\hat{\pi}_i \circ \lambda^{|V_0})^{| \pi_i^{-1}(X_i)}$ can also be known from $\varGamma_S(G,\lambda,\mu)$ simply because $\hat{\pi}_i \circ \lambda^{|V_0} =(\hat{\pi}_i)_{|V_0} \circ \lambda^{|V_0}$.
Setting $f = f'' \circ (\hat{\pi}_i \circ \lambda^{|V_0})^{| \pi_i^{-1}(X_i)} \circ f'$, we have $\mathsf{sort}(\lambda^{|V_i}) = \lambda_i \circ f$.
Applying (i) of Lemma~\ref{lem-sgncomp}, we can then compute $\mathsf{sgn}_S(G[V_i],\mathsf{sort}(\lambda^{|V_i}),(\tau \otimes \hat{\pi}_i)_{|V_i})$ using only $f$ and $\mathsf{sgn}_S(G[V_i],\lambda_i,(\tau \otimes \hat{\pi}_i)_{|V_i})$.
Finally, using $\mathsf{sgn}_S(G[V_i],\lambda_i,(\tau \otimes \hat{\pi}_i)_{|V_i})$, we compute $\mathsf{sgn}_S(G_i,\lambda^{|V(G_i)},\tau_{|V(G_i)})$ as follows.
By our induction hypothesis, we already have $\mathsf{sgn}_S(G_{i-1},\lambda^{|V(G_{i-1})},(\tau \otimes \hat{\pi}_i)_{|V(G_{i-1})})$ in hand.
For convenience, let us write $A = V(G_{i-1}) = V_0 \cup (\bigcup_{j=1}^{i-1} V_j)$ and $B = V_i$.
Then $(A,B)$ is a separation of $G_i$ by (iii) of Lemma~\ref{lem-decomp}.
Observe that $(\lambda^{|V(G_i)})^{-1}(A) = [p]$ and $(\lambda^{|V(G_i)})^{-1}(B) = \lambda^{-1}(V_0 \cap V_i)$.
Furthermore, $\mathsf{sort}((\lambda^{|V(G_i)})^{|A}) = \lambda^{|V(G_{i-1})}$ and $\mathsf{sort}((\lambda^{|V(G_i)})^{|B}) = \mathsf{sort}(\lambda^{|V_i})$.
Therefore, by (iii) of Lemma~\ref{lem-sgncomp}, given $[p]$, $\lambda^{-1}(V_0 \cap V_i)$, $\mathsf{sgn}_S(G_{i-1},\lambda^{|V(G_{i-1})},(\tau \otimes \hat{\pi}_i)_{|V(G_{i-1})})$ and $\mathsf{sgn}_S(G[V_i],\mathsf{sort}(\lambda^{|V_i}),(\tau \otimes \hat{\pi}_i)_{|V_i})$, we can compute $\mathsf{sgn}_S(G_i,\lambda^{|V(G_i)},\tau_{|V(G_i)})$.
This completes our induction argument.

Now we see $\mathsf{sgn}_S(G_i,\lambda^{|V(G_i)},(\mu \otimes \hat{\pi}_{6p} \otimes \hat{\pi}_{6p-1} \otimes \cdots \otimes \hat{\pi}_{i+1})_{|V(G_i)})$ can be computed using only $\varGamma_S(G,\lambda,\mu)$, for all $i \in [6p]$.
Note that $(G,\lambda,\mu) = (G_i,\lambda^{|V(G_i)},(\mu \otimes \hat{\pi}_{6p} \otimes \hat{\pi}_{6p-1} \otimes \cdots \otimes \hat{\pi}_{i+1})_{|V(G_i)})$ for $i = 6p$.
Thus, we can compute $\mathsf{sgn}_S(G,\lambda,\mu)$ given only $\varGamma_S(G,\lambda,\mu)$.
\end{proof}

The above lemma implies that if two graphs in $\mathcal{G}_{p,P,t}$ have the same $S$-characterization sequences, then they have the same $S$-signatures.
Based on this, we can establish recursive relations for $\delta_S$-values and $\Delta_S$-values for $|S| \geq 2$.
\begin{lemma} \label{lem-recurrence}
    Let $S$ be a sequence of pairs of natural numbers.
    Suppose $|S| \geq 2$ and $S = ((k,s))+S'$.
    Then we have the bounds
    \begin{enumerate}[(i)]
        \item $\delta_S(\mathcal{G}_{p,P,t}) \leq ((p+k)t)^{O((p+k)t^2)} \cdot (2^s|P|)^{O((p+k)t)} \cdot (\Delta_{S'}(\mathcal{G}_{p',P',t}))^{6(p+k)}$,
        \item $\Delta_S(\mathcal{G}_{p,P,t}) \leq 2^{((p+k)t)^{O((p+k)t^2)} \cdot (2^s|P|)^{O((p+k)t)} \cdot (\Delta_{S'}(\mathcal{G}_{p',P',t}))^{6(p+k)}}$,
    \end{enumerate}
    where $p' = \min\{p+k,4t\}$ and $P' = P \times \{0,1\}^s\times[4t]_0$.
\end{lemma}
\begin{proof}
Define a set $K = \{\mathsf{sgn}_{S'}(G',\lambda',\mu'): (G',\lambda',\mu') \in \mathcal{G}_{p+k,P\times \{0,1\}^s,t}\}$.
By construction, we have $\mathsf{sgn}_S(G,\lambda,\mu) \subseteq K$ for all $(G,\lambda,\mu) \in \mathcal{G}_{p,P,t}$.
Thus, $\{\mathsf{sgn}_S(G,\lambda,\mu): (G,\lambda,\mu) \in \mathcal{G}_{p,P,t}\} \subseteq 2^K$.
It then follows that $\delta_S(\mathcal{G}_{p,P,t}) \leq |K|$ and $\Delta_S(\mathcal{G}_{p,P,t}) \leq |\{\mathsf{sgn}_S(G,\lambda,\mu): (G,\lambda,\mu) \in \mathcal{G}_{p,P,t}\}| \leq 2^{|K|}$.
So it suffices to bound $|K|$.
By Lemma~\ref{lem-character}, $|K|$ is bounded by the number of different $S'$-characterization sequences the graphs in $\mathcal{G}_{p+k,P\times \{0,1\}^s,t}$ have.

In the $S'$-characterization sequence of a graph in $\mathcal{G}_{p+k,P\times \{0,1\}^s,t}$, the first element is a $((p+k)+24(p+k)t)$-labeled and $(P\times \{0,1\}^s)$-colored graph, which has $O((p+k)t)$ vertices and treewidth at most $t$.
By Fact~\ref{fact-numLC}, the number of such graphs is bounded by $((p+k)t)^{O((p+k)t^2)} \cdot (2^s|P|)^{O((p+k)t)}$.
Fixing such a graph $(H,\sigma,\tau) \in \mathcal{G}_{(p+k)+24(p+k)t,P\times \{0,1\}^s,t}$, we count the number of different $S'$-characterization sequences starting with $(H,\sigma,\tau)$ the graphs in $\mathcal{G}_{p+k,P\times \{0,1\}^s,t}$ have.
Consider an $S'$-characterization sequence whose first element is $(H,\eta,\tau)$.
Denote by $e_1,\dots,e_{6(p+k)}$ be the remaining elements in the sequence.
By construction, each of $e_i$ is the $S'$-signature of some $q_i$-labeled and $P'$-colored graph $(H_i,\eta_i,\tau_i)$ where $q_i \leq \min\{p+k,4t\} = p'$.
As observed before, both $q_i$ and $\mathsf{core}(H_i,\eta_i,\tau_i)$ are determined by $(H,\eta,\tau)$; we use $\sigma_i$ to denote this core.
Therefore, the number of possible choices for $e_i$ is bounded by
\begin{equation*}
    |\{\mathsf{sgn}_{S'}(H',\eta',\tau'): (H',\eta',\tau') \in \mathcal{G}_{q_i,P',t} \text{ and } \mathsf{core}(H',\eta',\tau') = \sigma_i\}| \leq \Delta_{S'}(\mathcal{G}_{q_i,P',t}) \leq \Delta_{S'}(\mathcal{G}_{p',P',t}).
\end{equation*}
So the total number of choices for $e_1,\dots,e_{6(p+k)}$ is at most $(\Delta_{S'}(\mathcal{G}_{p',P',t}))^{6(p+k)}$.
It follows that
\begin{equation*}
    |K| \leq ((p+k)t)^{O((p+k)t^2)} \cdot (2^s|P|)^{O((p+k)t)} \cdot (\Delta_{S'}(\mathcal{G}_{p',P',t}))^{6(p+k)}.
\end{equation*}
As such, the bounds for $\delta_S(\mathcal{G}_{p,P,t})$ and $\Delta_S(\mathcal{G}_{p,P,t})$ follows.
\end{proof}

\begin{corollary} \label{cor-recurrence}
    Let $S = ((k,s))+S'$ be a sequence of pairs of natural numbers.
    There exists a well-behaved almost linear function $f: \mathbb{N} \rightarrow \mathbb{N}$ such that
    \begin{enumerate}[(i)]
        \item {\small $\delta_S(\mathcal{G}_{p,P,t}) \leq 2^{f(p)} + 2^{f(k)} + \exp^{(2)}(f(\log t)) + \exp^{(2)}(f(s)) + \exp^{(2)}(f(\log|P|)) + (\Delta_{S'}(\mathcal{G}_{p',P',t}))^{7(p+k)}$},
        \item $\Delta_S(\mathcal{G}_{p,P,t}) \leq 2^{2^{f(p)} + 2^{f(k)} + \exp^{(2)}(f(\log t)) + \exp^{(2)}(f(s)) + \exp^{(2)}(f(\log|P|)) + (\Delta_{S'}(\mathcal{G}_{p',P',t}))^{7(p+k)}}$,
    \end{enumerate}
    where $p' = \min\{p+k,4t\}$ and $P' = P \times \{0,1\}^s\times[4t]_0$.
\end{corollary}
\begin{proof}
We modify the bounds in Lemma~\ref{lem-recurrence} as follows.
Observe that
\begin{align*}
    & ((p+k)t)^{O((p+k)t^2)} \cdot (2^s|P|)^{O((p+k)t)} \cdot (\Delta_{S'}(\mathcal{G}_{p',P',t}))^{6(p+k)} \\
    =\ & ((p+k)t)^{O((p+k)t^2)} + (2^s|P|)^{O((p+k)t)} + (\Delta_{S'}(\mathcal{G}_{p',P',t}))^{7(p+k)} \\
    =\ & 2^{O((p+k)t^3 \log (p+k))} + 2^{O((s+\log |P|) \cdot (p+k)t)} + (\Delta_{S'}(\mathcal{G}_{p',P',t}))^{7(p+k)} \\
    =\ & 2^{O((p+k)t^3 \log (p+k))} + 2^{2^{O(s + \log |P|)}} + 2^{O((p+k)t^2 \log(p+k))} + (\Delta_{S'}(\mathcal{G}_{p',P',t}))^{7(p+k)},
\end{align*}
where the last $=$ follows from the fact $(s+\log |P|) \cdot (p+k)t \leq \max\{2^{s+\log|P|},(p+k)t^2 \log(p+k)\}$.
Fact~\ref{fact-decompose} implies $2^{O((p+k)t^3 \log (p+k))} = 2^{\hat{O}(p+k)} + 2^{2^{\hat{O}(\log t)}} = 2^{\hat{O}(p)} + 2^{\hat{O}(k)} + 2^{2^{\hat{O}(\log t)}}$.
Therefore,
\begin{align*}
    & 2^{O((p+k)t^3 \log (p+k))} + 2^{2^{O(s + \log |P|)}} + 2^{O((p+k)t^2 \log(p+k))} + (\Delta_{S'}(\mathcal{G}_{p',P',t}))^{7(p+k)} \\
    =\ & 2^{O((p+k)t^3 \log (p+k))} + 2^{2^{O(s + \log |P|)}} + (\Delta_{S'}(\mathcal{G}_{p',P',t}))^{7(p+k)} \\
    =\ & 2^{\hat{O}(p+k)} + 2^{2^{\hat{O}(\log t)}} + 2^{2^{\hat{O}(s)}} + 2^{2^{\hat{O}(\log |P|)}} + (\Delta_{S'}(\mathcal{G}_{p',P',t}))^{7(p+k)}.
\end{align*}
As such, the desired function $f$ exists.
\end{proof}

\subsubsection{Bounding $\delta$-values and $\Delta$-values}

Using the recurrences in the last section, we can now establish bounds for $\delta_S(\mathcal{G}_{p,P,t})$ and $\Delta_S(\mathcal{G}_{p,P,t})$ by induction on $|S|$.
Specifically, we shall consider $|S| = 1$ as the base case and apply the recurrences for $|S| \geq 2$.
Suppose $S = ((k,s))$.
For each $\sigma \in \mathcal{G}_{p,P,t}$, let us define $\mathcal{H}_{\sigma,k}$ as the set of all $(p+k)$-labeled and $P$-colored graphs $(H,\gamma,\tau)$ with treewidth at most $t$ satisfying that $|V(H) \backslash \text{Im}(\gamma_{|[p]})| \leq k$ and $\mathsf{core}(H,\gamma_{|[p]},\tau) = \sigma$.
By Fact~\ref{fact-num}, we have $|\mathcal{H}_{\sigma,k}| = (p+k)^{O(\min\{k,t\} \cdot (p+k))} \cdot |P|^{O(k)}$.
Now consider a graph $(G,\lambda,\mu) \in \mathcal{G}_{p,P,t}$.
The following observation relates $\mathsf{sgn}_S(G,\lambda,\mu)$ with the set $\mathcal{H}_\sigma$ for $\sigma = \mathsf{core}(G,\lambda,\mu)$, and also bounds $|\mathsf{sgn}_S(G,\lambda,\mu)|$.

\begin{lemma} \label{lem-base}
    Let $S = ((k,s))$.
    Then for every graph $(G,\lambda,\mu) \in \mathcal{G}_{p,P,t}$, there exists a unique subset $\mathcal{H}' \subseteq \mathcal{H}_{\sigma,k}$, where $\sigma = \mathsf{core}(G,\lambda,\mu)$, such that 
    \begin{equation*}
        \mathsf{sgn}_S(G,\lambda,\mu) = \left\{(H,\gamma,\tau \otimes \tau'): (H,\gamma,\tau) \in \mathcal{H}' \textnormal{ and } \tau' \in U_{H,\{0,1\}^s}\right\}.
    \end{equation*}
\end{lemma}
\begin{proof}
Let $\mathcal{H}'= \{(G[\text{Im}(\lambda\oplus \lambda')],\lambda\oplus \lambda',\mu_{|\text{Im}(\lambda\oplus \lambda')}): \lambda' \in \varLambda_{G,k}\}$.
Clearly, $\mathcal{H}' \subseteq \mathcal{H}_{\sigma,k}$.
We prove that $\mathcal{H}'$ satisfies the desired condition.
By definition, 
\begin{align*}
    \mathsf{sgn}_S(G,\lambda,\mu) &= \left\{G[\text{Im}(\lambda\oplus \lambda')],\lambda\oplus \lambda',(\mu \otimes \mu')_{|\text{Im}(\lambda\oplus \lambda')}): \lambda' \in \varLambda_{G,k} \text{ and } \mu' \in U_{G,\{0,1\}^s}\right\} \\
    &= \{(G[\text{Im}(\lambda\oplus \lambda')],\lambda\oplus \lambda',\mu_{|\text{Im}(\lambda\oplus \lambda')} \otimes \tau'): \lambda' \in \varLambda_{G,k} \text{ and } \tau' \in U_{G[\text{Im}(\lambda\oplus \lambda')],\{0,1\}^s}\} \\
    &= \{(H,\gamma,\tau \otimes \tau'): (H,\gamma,\tau) \in \mathcal{H}' \text{ and } \tau' \in U_{H,\{0,1\}^s}\},
\end{align*}
which completes the proof.
\end{proof}

\begin{corollary} \label{cor-base}
    There exists a well-behaved almost linear function $f: \mathbb{N} \rightarrow \mathbb{N}$ such that
    \begin{enumerate}[(i)]
        \item $\delta_{((k,s))}(\mathcal{G}_{p,P,t}) \leq 8^{sk} + 8^{sp} + 2^{f(k \log |P|)} + 2^{f(\min\{k,t\} \cdot k)} + 2^{f(\min\{k,t\} \cdot p)}$,
        \item $\Delta_{((k,s))}(\mathcal{G}_{p,P,t}) \leq \exp^{(2)}(f(k\log|P|)) + \exp^{(2)}(f(\min\{k,t\} \cdot k)) + \exp^{(2)}(f(\min\{k,t\} \cdot p))$.
    \end{enumerate}
\end{corollary}
\begin{proof}
Lemma~\ref{lem-base} implies that $|\mathsf{sgn}_{((k,s))}(G,\lambda,\mu)| \leq |\mathcal{H}_{\sigma,k}| \cdot |U_{H,\{0,1\}^s}|$ for all $(G,\lambda,\mu) \in \mathcal{G}_{p,P,t}$.
Therefore, we have $\delta_{((k,s))}(\mathcal{G}_{p,P,t}) \leq |\mathcal{H}_{\sigma,k}| \cdot |U_{H,\{0,1\}^s}| = (p+k)^{O(\min\{k,t\} \cdot (p+k))} \cdot |P|^{O(k)} \cdot 2^{s(p+k)}$, because the graphs in $\mathcal{H}_{\sigma,k}$ have at most $p+k$ vertices and thus $|U_{H,\{0,1\}^s}| \leq 2^{s(p+k)}$.
As such, $\delta_{((k,s))}(\mathcal{G}_{p,P,t})$ is bounded by
\begin{equation*}
    8^{sk} + 8^{sp} + 2^{O(k \log |P|)} + 2^{\hat{O}(\min\{k,t\} \cdot k)} + 2^{\hat{O}(\min\{k,t\} \cdot p)}.
\end{equation*}
Lemma~\ref{lem-base} also implies that the $S$-signature of every graph in $\mathcal{G}_{p,P,t}$ with core $\sigma$ is uniquely determined by a subset $\mathcal{H}' \subseteq \mathcal{H}_{\sigma,k}$.
Therefore, the number of different $S$-signatures such graphs can have is bounded by the number of different subsets of $\mathcal{H}_{\sigma,k}$, i.e., $2^{|\mathcal{H}_{\sigma,k}|}$.
Since $|\mathcal{H}_{\sigma,k}| = (p+k)^{O(\min\{k,t\} \cdot (p+k))} \cdot |P|^{O(k)}$, we have
\begin{equation*}
    2^{|\mathcal{H}_{\sigma,k}|} = \exp^{(2)}(O(k\log|P|)) + \exp^{(2)}(\hat{O}(\min\{k,t\} \cdot k)) + \exp^{(2)}(\hat{O}(\min\{k,t\} \cdot p)).
\end{equation*}
As such, the desired function $f$ exists.
\end{proof}

With the recurrences in Corollary~\ref{cor-recurrence} and the base cases in Corollary~\ref{cor-base} in hand, to eventually prove the bound in Theorem~\ref{thm-sgnsize}, the remaining work is just tedious and involved calculation (though it has to be done very carefully).
Let $f:\mathbb{N} \rightarrow \mathbb{N}$ be a well-behaved almost linear function that satisfy the conditions in Corollary~\ref{cor-recurrence} and Corollary~\ref{cor-base}.
Note that such a function exists, as we can simply take the sum of two well-behaved almost linear functions that satisfy the conditions in the two corollaries individually.
Let $c > 0$ be a number satisfying $f(ax) \leq a^c f(x)$ for all $a,x \in \mathbb{N}$ and $f(x) \leq x^c$ for all integers $x \geq 2$.
Define $f_1 = f$ and $f_i = f \circ f$ for all integers $i \geq 2$.

Next, we introduce the numbers we are going to use as the upper bounds for the $\delta$-values and $\Delta$-values.
For $S = ((k_1,s_1),\dots,(k_r,s_r))$, integers $p,t,b \in \mathbb{N}$, and a set $P$, we define
{\small \begin{equation*}
        \psi_{S,p,P,t,b} = \sum_{i=1}^r \sum_{j=0}^{i-1} \exp_{b}^{(i)}(f_i(s_jk_i)) + \sum_{i=1}^r \sum_{j=0}^i \exp_{b}^{(i)}(f_i(t_j k_i)) + \sum_{i=1}^r \exp_{b}^{(i)}(f_i(k_i \log t)) + b^{f_1(t_1p)},
\end{equation*}
\begin{equation*}
        \varPsi_{S,p,P,t,b} = \sum_{i=1}^r \sum_{j=0}^{i-1} \exp_{b}^{(i+1)}(f_i(s_jk_i)) + \sum_{i=1}^r \sum_{j=0}^i \exp_{b}^{(i+1)}(f_i(t_j k_i)) + \sum_{i=1}^r \exp_{b}^{(i+1)}(f_i(k_i \log t)) + \exp_{b}^{(2)}(f_1(t_1p)),
\end{equation*}}

\smallskip

\noindent
where $s_0 = \log |P|$, $t_0 = \min\{p,t\}$, and $t_j = \min\{k_j,t\}$ for $j \in [r]$.
Furthermore, we write $b_1 = 2$ and $b_r = (r+1)^{2(r+1)} \cdot \exp_{b_{r-1}}^{(2)}(2 \log b_{r-1} + 10c^3)$, for all integers $r \geq 2$.
The bounds we want to prove are presented below.

\begin{itemize}
    \item \textbf{[$\delta$-Bound$_r$]:} $\delta_S(\mathcal{G}_{p,P,t}) \leq \psi_{S,p,P,t,b_r}$ for any $S$ with $|S| = r$ and any $p,P,t$.
    \item \textbf{[$\Delta$-Bound$_r$]:} $\Delta_S(\mathcal{G}_{p,P,t}) \leq \psi_{S,p,P,t,b_r}$ for any $S$ with $|S| = r$ and any $p,P,t$.
\end{itemize}

We observe that \textbf{$\Delta$-Bound$_1$} holds because of (ii) of Corollary~\ref{cor-base}.
(However, \textbf{$\delta$-Bound$_1$} does not hold. But this does not matter as \textbf{$\delta$-Bound$_1$} will not be used in our proof.)
The following observation then implies that \textbf{$\delta$-Bound$_r$} and \textbf{$\Delta$-Bound$_r$} hold for all integers $r \geq 2$.

\begin{lemma}
    For $r \geq 2$, if \textnormal{\bf $\Delta$-Bound$_{r-1}$} holds, then \textnormal{\bf $\delta$-Bound$_r$} and \textnormal{\bf $\Delta$-Bound$_r$} also hold.
\end{lemma}
\begin{proof}
Let $S = ((k_1,s_1),\dots,(k_r,s_r))$ and $S' = ((k_2,s_2),\dots,(k_r,s_r))$.
Also, let $p,t \in \mathbb{N}$ and $P$ be a set.
To establish \textbf{$\delta$-Bound$_r$} and \textbf{$\Delta$-Bound$_r$}, we need to bound $\delta_S(\mathcal{G}_{p,P,t})$ and $\Delta_S(\mathcal{G}_{p,P,t})$, both of which are related to $(\Delta_{S'}(\mathcal{G}_{p',P',t}))^{7(p+k_1)}$, where $p' = \min\{p+k_1,4t\}$ and $P' = P \times \{0,1\}^{s_1} \times [4t]_0$, by our recurrences.
So we first need to bound $(\Delta_{S'}(\mathcal{G}_{p',P',t}))^{7(p+k_1)}$.

By assumption, we have \textbf{$\Delta$-Bound$_{r-1}$}, which implies that $\Delta_{S'}(\mathcal{G}_{p',P',t}) \leq \varPsi_{S',p',P',t,b_{r-1}}$.
Thus, $(\Delta_{S'}(\mathcal{G}_{p',P',t}))^{7(p+k_1)} \leq \varPsi_{S',p',P',t,b_{r-1}}^{7(p+k_1)}$.
By definition, $\varPsi_{S',p',P',t,b_{r-1}}$ is the sum of $m = (r-1)(r+2)+1$ towers of exponentials with base $b_{r-1}$ each of which has height at least 2.
That is, $\varPsi_{S',p',P',t,b_{r-1}}$ can be written in the form of $\sum_{i=1}^m \exp_{b_{r-1}}^{(2)} (a_i)$.
This allows us to apply Fact~\ref{fact-decompose} to deduce
\begin{equation*}
    \varPsi_{S',p',P',t,b_{r-1}}^{7(p+k_1)} \leq m^{m} \exp_{b_{r-1}}^{(2)}(2 \log b_{r-1}) \cdot \left(b_{r-1}^{f(7(p+k_1))} + \sum_{i=1}^{r(r-1)} \exp_{b_{r-1}}^{(2)} (f(a_i)) \right).
\end{equation*}
For convenience, write $X = \sum_{i=1}^m \exp_{b_{r-1}}^{(2)} (f(a_i))$.
Let $b = (b_{r-1}^c)^{8^{c^2}}$.
The terms in $X$ one-to-one correspond to the terms in $\varPsi_{S',p',P',t,b_{r-1}}$.
Next, we consider these terms one by one and bound each of them by one or multiple terms in $\psi_{S,p,P,t,b}$.

\begin{itemize}
    \item For $i \in \{3,\dots,d\}$ and $j \in \{2,\dots,i-1\}$, $\varPsi_{S',p',P',t,b_{r-1}}$ contains a term $$\exp_{b_{r-1}}^{(i)}(f_{i-1}(s_j k_i)) = \exp_{b_{r-1}}^{(i)}(f_i(s_j k_i))).$$
    This term corresponds to $\exp_{b_{r-1}}^{(2)}(f(\exp_{b_{r-1}}^{(i-2)} (f_i(s_j k_i))))$ in $X$, which is then bounded by $\exp_{b_{r-1}^c}^{(i)}(f_i(s_j k_i))$ and hence bounded by $\exp_b^{(i)}(f_i(s_j k_i))$.
    \item For $i \in \{3,\dots,d\}$, $\varPsi_{S',p',P',t,b_{r-1}}$ contains a term $$\exp_{b_{r-1}}^{(i)}(f_{i-1}(k_i \log|P'|)) = \exp_{b_{r-1}}^{(i)}(f_i(k_i\log|P|+s_1 k_i+ k_i\log (4t+1))).$$
    This term corresponds to $\exp_{b_{r-1}}^{(2)}(f(\exp_{b_{r-1}}^{(i-2)} (f_i(k_i \log|P|+s_1 k_i+ k_i\log (4t+1)))))$ in $X$,
    which is bounded by $\exp_{b_{r-1}^c}^{(i)}(f_i(k_i\log|P|+s_1 k_i+ k_i\log (4t+1)))$.
    As $f_i = f \circ f$, we have
    \begin{align*}
        & f_i(k_i\log|P|+s_1 k_i+ k_i\log (4t+1)) \\
        \leq\ & f_i(k_i\log|P|+2s_1 k_i+ 3k_i\log (4t+1)) \\
        \leq\ & 6^c \cdot \max\{f_i(k_i\log|P|),f_i(s_1k_i),f_i(k_i \log t)\}.
    \end{align*}
    Therefore, it follows that
    \begin{align*}
        & \exp_{b_{r-1}^c}^{(i)}(f_i(k_i\log|P|+s_1 k_i+ k_i\log (4t+1))) \\
        \leq\ &  \exp_{b_{r-1}^c}^{(i)}(6^c \cdot f_i(k_i\log|P|)) + \exp_{b_{r-1}^c}^{(i)}(6^c \cdot f_i(s_1 k_i)) + \exp_{b_{r-1}^c}^{(i)}(6^c \cdot  f_i(k_i \log t)) \\
        \leq\ & \exp_b^{(i)}(f_i(k_i\log|P|)) + \exp_b^{(i)}(f_i(s_1 k_i)) + \exp_b^{(i)}(f_i(k_i \log t)).
    \end{align*}
    Furthermore, $\varPsi_{S',p',P',t,b_{r-1}}$ also contains the term $\exp_{b_{r-1}}^{(2)}(f_1(k_2 \log|P'|))$.
    It corresponds to $\exp_{b_{r-1}}^{(2)}(f(f_1(k_2 \log|P'|))) = \exp_{b_{r-1}}^{(2)}(f_2(k_2 \log|P'|)))$ in $X$.
    By the argument above, we have
    \begin{align*}
        & \exp_{b_{r-1}}^{(2)}(f_2(k_2 \log|P'|))) \\
        =\ & \exp_{b_{r-1}}^{(2)}(f_2(k_2\log|P|+s_1 k_2+ k_2\log (4t+1))) \\
        \leq\ & \exp_b^{(2)}(f_2(k_2\log|P|)) + \exp_b^{(2)}(f_2(s_1 k_2)) + \exp_b^{(2)}(f_2(k_2 \log t)).
    \end{align*}
    \item For $i,j \in \{3,\dots,d\}$ with $j \leq i$, $\varPsi_{S',p',P',t,b_{r-1}}$ contains a term $$\exp_{b_{r-1}}^{(i)}(f_{i-1}(\min\{k_j,t\} \cdot k_i)) = \exp_{b_{r-1}}^{(i)}(f_i(\min\{k_j,t\} \cdot k_i))).$$
    This term corresponds to $\exp_{b_{r-1}}^{(2)}(f(\exp_{b_{r-1}}^{(i-2)} (f_i(\min\{k_j,t\} \cdot k_i))))$ in $X$, which is then bounded by $\exp_{b_{r-1}^c}^{(i)}(f_i(\min\{k_j,t\} \cdot k_i))$ and hence bounded by $\exp_{b}^{(i)}(f_i(\min\{k_j,t\} \cdot k_i))$.
    Also, $\varPsi_{S',p',P',t,b_{r-1}}$ contains the terms $\exp_{b_{r-1}}^{(2)}(f_1(\min\{k_2,t\} \cdot k_2))$ and $\exp_{b_{r-1}}^{(2)}(f_1(\min\{p',t\} \cdot k_2))$.
    These two terms correspond respectively to $\exp_{b_{r-1}}^{(2)}(f(f_1(\min\{p',t\} \cdot k_2))) = \exp_{b_{r-1}}^{(2)}(f_2(\min\{p',t\} \cdot k_2)))$ and $\exp_{b_{r-1}}^{(2)}(f(f_1(\min\{k_2,t\} \cdot k_2))) = \exp_{b_{r-1}}^{(2)}(f_2(\min\{k_2,t\} \cdot k_2)))$ in $X$.
    We have
    \begin{align*}
        & \exp_{b_{r-1}}^{(2)}(f_2(\min\{p',t\} \cdot k_2)) \\
        =\ & \exp_{b_{r-1}}^{(2)}(f_2(\min\{p+k,t\} \cdot k_2)) \\
        \leq\ & \exp_{b_{r-1}}^{(2)}(f_2(\min\{p,t\} \cdot k_2+\min\{k,t\} \cdot k_2)) \\
        \leq\ & \exp_b^{(2)}(f_2(\min\{p,t\} \cdot k_2))+\exp_b^{(2)}(f_2(\min\{k,t\} \cdot k_2)),
    \end{align*}
    where the last $\leq$ follows from the fact that $f_2(x+y) \leq 2^{c^2} \cdot \max\{f_2(x),f_2(y)\}$.
    The other term $\exp_{b_{r-1}}^{(2)}(f_2(\min\{k_2,t\} \cdot k_2)))$ is clearly bounded by $\exp_b^{(2)}(f_2(\min\{k_2,t\} \cdot k_2)))$.
    
    \item For $i \in \{3,\dots,d\}$, $\varPsi_{S',p',P',t,b_{r-1}}$ contains a term $\exp_{b_{r-1}}^{(i)}(f_{i-1}(k_i \log t)) = \exp_{b_{r-1}}^{(i)}(f_i(k_i \log t)))$.
    This term corresponds to $\exp_{b_{r-1}}^{(2)}(f(\exp_{b_{r-1}}^{(i-2)} (f_i(k_i \log t))))$ in $X$, which is then bounded by $\exp_{b_{r-1}^c}^{(i)}(f_i(k_i \log t))$ and hence bounded by $\exp_{b}^{(i)}(f_i(k_i \log t))$.
    Also, $\varPsi_{S',p',P',t,b_{r-1}}$ contains the term $\exp_{b_{r-1}}^{(2)}(f_1(k_2 \log t))$.
    It corresponds to $\exp_{b_{r-1}}^{(2)}(f(f_1(k_2 \log t))) = \exp_{b_{r-1}}^{(2)}(f_2(k_2 \log t))$ in $X$, which is bounded by $\exp_b^{(2)}(f_2(k_2 \log t))$.

    \item Finally, $\varPsi_{S',p',P',t,b_{r-1}}$ contains the term 
    $$\exp_{b_{r-1}}^{(2)}(f_1(\min\{k_2,t\} \cdot p')) = \exp_{b_{r-1}}^{(2)}(f_1(\min\{k_2,t\} \cdot \min\{p+k_1,4t\})).$$
    It corresponds to $\exp_{b_{r-1}}^{(2)}(f(f_1(\min\{k_2,t\} \cdot \min\{p+k_1,4t\})))$ in $X$.
    We have
    \begin{align*}
        & \exp_{b_{r-1}}^{(2)}(f(f_1(\min\{k_2,t\} \cdot \min\{p+k_1,4t\}))) \\
        =\ & \exp_{b_{r-1}}^{(2)}(f_2(\min\{k_2,t\} \cdot \min\{p+k_1,4t\})) \\
        \leq\ & \exp_{b_{r-1}}^{(2)}(f_2(\min\{p+k_1,4t\} \cdot k_2)) \\
        \leq\ & \exp_{b_{r-1}}^{(2)}(f_2(4 \cdot \min\{p,t\} \cdot k_2 + 4 \cdot \min\{k_1,t\} \cdot k_2))
    \end{align*}
    Using the fact that $f_2(4x+4y) = 8^{c^2} \cdot \max\{f_2(x),f_2(y)\}$, we further have
    \begin{align*}
        & \exp_{b_{r-1}}^{(2)}(f_2(4 \cdot \min\{p,t\} \cdot k_2 + 4 \cdot \min\{k_1,t\} \cdot k_2)) \\
        \leq\ & \exp_b^{(2)}(f_2(\min\{p,t\} \cdot k_2)) + \exp_b^{(2)}(f_2(\min\{k_1,t\} \cdot k_2)).
    \end{align*}
\end{itemize}
Now we see that each term in $X$ is bounded by one or multiple terms in $\varPsi_{S,p,P,t,b}$, and is thus at most $\varPsi_{S,p,P,t,b}$.
It follows that $X \leq m \cdot \psi_{S,p,P,t,b} \leq r^2 \cdot \psi_{S,p,P,t,b}$.
Since 
$$b_{r-1}^{f(7(p+k_1))} \leq b_{r-1}^{7cf(p)+7cf(k_1)} \leq b_{r-1}^{14cf(p)}+b_{r-1}^{14cf(k_1)} \leq b^{f_1(p)}+b^{f_1(k_1)} \leq \psi_{S,p,P,t,b},$$
we have $\varPsi_{S',p',P',t,b_{r-1}}^{7(p+k_1)} \leq (m^m \exp_{b_{r-1}}^{(2)}(2 \log b_{r-1})+1) \cdot \psi_{S,p,P,t,b}$, which implies
\begin{align*}
    & 2^{f(p)} + 2^{f(k)} + \exp^{(2)}(f(\log t)) + \exp^{(2)}(f(s)) + \exp^{(2)}(f(\log|P|)) + (\Delta_{S'}(\mathcal{G}_{p',P',t}))^{7(p+k)} \\
    \leq\ & \psi_{S,p,P,t,2} + (\Delta_{S'}(\mathcal{G}_{p',P',t}))^{7(p+k)} \\
    \leq\ &\psi_{S,p,P,t,2} + (m^m \exp_{b_{r-1}}^{(2)}(2 \log b_{r-1})+1) \cdot \psi_{S,p,P,t,b} \\
    \leq\ & (m^m \exp_{b_{r-1}}^{(2)}(2 \log b_{r-1})+2) \cdot \psi_{S,p,P,t,b} \\
    \leq\ &\psi_{S,p,P,t,b'},
\end{align*}
where $b' = (m^m \exp_{b_{r-1}}^{(2)}(2 \log b_{r-1})+2) b$.
Then Corollary~\ref{cor-recurrence} directly implies $\delta_S(\mathcal{G}_{p,P,t}) \leq \psi_{S,p,P,t,b'}$ and $\Delta_S(\mathcal{G}_{p,P,t}) \leq 2^{\psi_{S,p,P,t,b'}}$.
Recall that $\psi_{S,p,P,t,b'}$ consists of $r(r+2)+1$ terms each of which is a tower of exponentials with base $b'$.
Thus, $2^{\psi_{S,p,P,t,b'}} \leq \varPsi_{S,p,P,t,(r(r+2)+1)b'}$ by Fact~\ref{fact-basechange1}.
Note that
\begin{align*}
    b_r &= (r+1)^{2(r+1)} \cdot \exp_{b_{r-1}}^{(2)}(2 \log b_{r-1} + 10c^3) \\
    &\geq (r+1)^{2(r+1)} \cdot \exp_{b_{r-1}}^{(2)}(2 \log b_{r-1}) \cdot (b_{r-1}^c)^{8^{c^2}} \\
    &\geq (r(r+2)+1)\cdot ((r-1)(r+3)+1)^{(r-1)(r+3)+1} \cdot \exp_{b_{r-1}}^{(2)}(2 \log b_{r-1}) \cdot b \\
    &= (r(r+2)+1)\cdot m^m \exp_{b_{r-1}}^{(2)}(2 \log b_{r-1}) \cdot b \\
    &\geq (r(r+2)+1)\cdot b'.
\end{align*}
Therefore, we directly have $\delta_S(\mathcal{G}_{p,P,t}) \leq \psi_{S,p,P,t,b_r}$ and $\Delta_S(\mathcal{G}_{p,P,t}) \leq \varPsi_{S,p,P,t,b_r}$, which gives us the desired \textbf{$\delta$-Bound$_r$} and \textbf{$\Delta$-Bound$_r$}.
\end{proof}



\subsubsection{Recursive size of the signatures}
By the discussion in the previous section, for all $(G,\lambda,\mu) \in \mathcal{G}_{p,P,t}$ and all sequences $S$ of $d$ pairs of natural numbers, we have $|\mathsf{sgn}_S(G,\lambda,\mu)| \leq \delta_S(\mathcal{G}_{p,P,t}) \leq \psi_{S,p,P,t,b_d}$.
Based on this, we can now bound the recursive size of the signatures.

Let $S = ((k_1,s_1),\dots,(k_d,s_d))$ and $(G,\lambda,\mu) \in \mathcal{G}_{p,P,t}$.
By construction, $\mathsf{sgn}_S(G,\lambda,\mu)$ is a nested set with $d$ levels.
The $z$-th level of $\mathsf{sgn}_S(G,\lambda,\mu)$ consists of $S_z$-signatures of $p_z$-labeled and $P_z$-colored graphs, where $S_z = ((k_{z+1},s_{z+1}),\dots,(k_d,s_d))$, $p_z = p+ \sum_{j=1}^z k_j$, and $P_z = P \times \prod_{j=1}^z \{0,1\}^{s_j}$.
It follows that $\lVert \mathsf{sgn}_S(G,\lambda,\mu) \rVert \leq \prod_{z=0}^{d-1} \delta_{S_z}(\mathcal{G}_{p_z,P_z,t}) \leq (\prod_{z=0}^{d-2} \psi_{S_z,p_z,P_z,t,b_{d-z}}) \cdot \delta_{S_{d-1}}(\mathcal{G}_{p_{d-1},P_{d-1},t})$.

\begin{lemma}
    $\psi_{S_z,p_z,P_z,t,b_{d-z}} \leq \psi_{S,p,P,t,b_d}$ for all $z \in [d-1]_0$.
\end{lemma}
\begin{proof}
When $z = 0$, the inequality trivially holds.
So assume $z \geq 1$.
Observe that for most of the exponential towers in $\psi_{S_z,p_z,P_z,t,b_{d-z}}$, there is a corresponding tower in $\psi_{S,p,P,t,b_{d-z}}$ with larger height and base, except the terms 
\begin{enumerate}[(i)]
    \item $b_{d-z}^{f_1((\log|P|+s_1+\cdots+s_z) \cdot k_i)}$ for $i \in \{z+1,\dots,d\}$,
    \item $b_{d-z}^{f_1(\min\{p+k_1+\cdots+k_z,t\} \cdot k_i)}$ for $i \in \{z+1,\dots,d\}$,
    \item $b_{d-z}^{f_1(\min\{k_{z+1},t\} \cdot (p+k_1+\cdots+k_z))}$.
\end{enumerate}
Recall that $f_1 = f$ and we have the number $c$ such that $f(ax) \leq a^c f(x)$ for all $a,x \in \mathbb{N}$.
Therefore, $f_1((\log|P|+s_1+\cdots+s_z) \cdot k_i) \leq (z+1)^c \cdot \max\{f_1(k_i \log|P|),f_1(s_1k_i),\dots,f_1(s_zk_i)\}$, which implies
\begin{align*}
    \sum_{i=z+1}^d b_{d-z}^{f_1((|P|+s_1+\cdots+s_z) \cdot k_i)} &\leq (z+1)^c \cdot \left(\sum_{i=z+1}^d b_{d-j}^{f_1(k_i \log|P|)} + \sum_{i=z+1}^d \sum_{j=1}^z b_{d-j}^{f_1(s_jk_i)} \right) \\
    & \leq (z+1)^c \cdot \psi_{S,p,P,t,b_{d-z}} \leq \psi_{S,p,P,t,(z+1)^c b_{d-z}}.
\end{align*}
Similarly, we can show $\sum_{i=z+1}^d b_{d-z}^{f_1(\min\{p+k_1+\cdots+k_z,t\} \cdot k_i)} \leq \psi_{S,p,P,t,(z+1)^c b_{d-z}}$ as well.
Finally, we consider $b_{d-z}^{f_1(\min\{k_{z+1},t\} \cdot (p+k_1+\cdots+k_z))}$.
Because $\min\{k_{z+1},t\} \cdot x \leq \min\{k_{z+1},t\} \cdot k_{z+1} + \min\{x,t\} \cdot x$,
\begin{align*}
    & f_1(\min\{k_{z+1},t\} \cdot (p+k_1+\cdots+k_z)) \\
    \leq\ & f_1\left(\min\{p,t\} \cdot p+ \sum_{j=1}^z \min\{k_j,t\} \cdot k_j + (z+1) \cdot \min\{k_{z+1},t\} \cdot k_{z+1} \right) \\
    \leq\ & (2z+2)^c \cdot \max\{f_1(\min\{p,t\} \cdot p),f_1(\min\{k_1,t\} \cdot k_1),\dots,f_1(\min\{k_{z+1},t\} \cdot k_{z+1})\},
\end{align*}
which implies $b_{d-z}^{f_1(\min\{k_{z+1},t\} \cdot (p+k_1+\cdots+k_z))} \leq (2z+2)^c \cdot \sum_{i=1}^{z+1} b_{d-z}^{f_1(\min\{k_i,t\} \cdot k_i)} \leq \psi_{S,p,P,t,(2z+2)^c b_{d-z}}$.

To summarize, the sum of the terms in (i), (ii), and (iii) above is bounded by $3 \psi_{S,p,P,t,(2z+2)^c b_{d-z}}$, and the sum of the other terms in $\psi_{S_z,p_z,P_z,t,b_{d-z}}$ is bounded by $\psi_{S,p,P,t,b_{d-z}}$.
Therefore, we have $\psi_{S_z,p_z,P_z,t,b_{d-z}} \leq 4 \psi_{S,p,P,t,(2z+2)^c b_{d-z}}$.
Since $b_d \geq 4(2z+2)^c b_{d-z}$ for all $z \geq 1$ by our construction, we eventually have $\psi_{S_z,p_z,P_z,t,b_{d-z}} \leq \psi_{S,p,P,t,b_d}$.
\end{proof}

The above lemma implies $\lVert \mathsf{sgn}_S(G,\lambda,\mu) \rVert \leq \psi_{S,p,P,t,b_d}^{d-2} \cdot \delta_{S_{d-1}}(\mathcal{G}_{p_{d-1},P_{d-1},t})$.
According to (i) of Corollary~\ref{cor-base}, we have
\begin{align*}
    \delta_{S_{d-1}}(\mathcal{G}_{p_{d-1},P_{d-1},t}) &\leq 8^{f(s_d k_d)} + 8^{f(s_d p_{d-1})} + 2^{f(k_d \log |P_{d-1}|)} + 2^{f(\min\{k_d,t\} \cdot k_d)} + 2^{f(\min\{k_d,t\} \cdot p_{d-1})} \\
    & \leq 8^{f(s_d k_d)} + 8^{f(s_d p_{d-1})} + \psi_{S_{d-1},p_{d-1},P_{d-1},t,2}.
\end{align*}
and thus $\lVert \mathsf{sgn}_S(G,\lambda,\mu) \rVert \leq \psi_{S,p,P,t,b_d}^{O(d)} + (8^d)^{s_d p} + \sum_{i=1}^d (8^d)^{s_d k_i}$.
As $\psi_{S,p,P,t,b_d}$ is the sum of $O(r^2)$ exponential towers with base $b_d$, we have $\psi_{S,p,P,t,b_d}^{O(d)} = \psi_{S,p,P,t,b_d^{O(d)}}$.
Since $b_d$ only depends on $d$, Fact~\ref{fact-basechange2} implies that $\lVert \mathsf{sgn}_S(G,\lambda,\mu) \rVert \leq f(d) \cdot (\psi_{S,p,P,t,2} +  2^{\hat{O}(s_d p)} + \sum_{i=1}^d 2^{\hat{O}(s_d k_i)})$ for some computable function $f$.
Note that $f(d) \cdot (\psi_{S,p,P,t,2} + 2^{\hat{O}(s_d p)} + \sum_{i=1}^d 2^{\hat{O}(s_d k_i)})$ is bounded by the expression of Theorem~\ref{thm-sgnsize}.
To verify this, we see that $2^{\hat{O}(s_d p)} + \sum_{i=1}^d 2^{\hat{O}(s_d k_i)}$ appears in the expression.
The term $2^{f_1(\min\{k_1,t\} \cdot p)}$ in $\psi_{S,p,P,t,2}$ is bounded by $2^{\hat{O}(\min\{p,t\} \cdot p)} + 2^{\hat{O}(\min\{k_1,t\} \cdot k_1)}$, where the latter appears in the expression, and the sum of all the other terms in $\psi_{S,p,P,t,2}$ is clearly bounded by the expression.
Therefore, Theorem~\ref{thm-sgnsize} follows.

\begin{remark}
From our proof of Theorem~\ref{thm-sgnsize}, one can easily verify that $f(d) = \exp^{(d+o(d))}(0)$.  
\end{remark}

In the rest of this section, we consider the case where $S$ is a \textit{fixed} sequence of pairs of natural numbers, $p \in \mathbb{N}_0$ is a \textit{fixed} integer, and $P$ is a \textit{fixed} set.
We shall show that in this case, $\lVert \mathsf{sgn}_S(G,\lambda,\mu) \rVert = \exp^{(d-1)}(t^{O(1)})$ for any $(G,\lambda,\mu) \in \mathcal{G}_{p,P,t}$.

We first assume $S$ and $p$ are fixed, while $P$ is not.
Suppose $d = |S|$.
We apply induction on $d$ to show that $\delta_S(\mathcal{G}_{p,P,t}) = \exp^{(d)}(O(\log |P| + \log t))$ and $\Delta_S(\mathcal{G}_{p,P,t}) = \exp^{(d+1)}(O(\log |P| + \log t))$; here the constant hidden in $O(\cdot)$ depends on $S$ (and thus $d$) and $p$.
The base case $d=1$ follows from (the proof of) Corollary~\ref{cor-base}: while the bounds in the corollary has the almost linear function $f$, the proof actually implies that $f(\cdot)$ can be replaced with $O(\cdot)$ for the terms depending on $|P|$.
Suppose the bounds hold for $d-1$.
Consider a (fixed) sequence $S = ((k,s))+S'$ of length $d$ where the first element.
Lemma~\ref{lem-recurrence} implies that $\delta_S(\mathcal{G}_{p,P,t}) = 2^{(\log |P|+t)^{O(1)}} \cdot (\Delta_{S'}(\mathcal{G}_{p',P',t}))^{O(1)}$, where $p' = p+k$ and $P' = P \times \{0,1\}^s \times [4t]_0$.
By our induction hypothesis,
\begin{equation*}
    \Delta_{S'}(\mathcal{G}_{p',P',t}) = \exp^{(d)}(O(\log |P'| + \log t)) = \exp^{(d)}(O(\log |P| + \log t)).
\end{equation*}
Since $2^{(\log |P|+t)^{O(1)}} \leq 2^{2^{O(\log |P| + \log t)}}$, we have $\delta_S(\mathcal{G}_{p,P,t}) = \exp^{(d)}(O(\log |P| + \log t))$.
In the same way, we can show $\Delta_S(\mathcal{G}_{p,P,t}) = \exp^{(d+1)}(O(\log |P| + \log t))$ as well.
As a result, if $P$ is also fixed, then $\delta_S(\mathcal{G}_{p,P,t}) = \exp^{(d-1)}(t^{O(1)})$ and $\Delta_S(\mathcal{G}_{p,P,t}) = \exp^{(d)}(O(t^{O(1)}))$.

Let $S = ((k_1,s_1),\dots,(k_d,s_d))$ be a fixed sequence.
As before, let $S_z = ((k_{z+1},s_{z+1}),\dots,(k_d,s_d))$, $p_z = p+ \sum_{j=1}^z k_j$, and $P_z = P \times \prod_{j=1}^z \{0,1\}^{s_j}$ for $z \in [d-1]_0$.
For a graph $(G,\lambda,\mu) \in \mathcal{G}_{p,P,t}$, we have $\lVert \mathsf{sgn}_S(G,\lambda,\mu) \rVert = \prod_{z=0}^{d-1} \delta_{S_z}(\mathcal{G}_{p_z,P_z,t})$.
As $p_z = O(1)$ and $|P_z| = O(1)$ for all $z \in [d-1]_0$ by assumption, $\delta_{S_z}(\mathcal{G}_{p_z,P_z,t}) = \exp^{(d-z-1)}(t^{O(1)})$.
It follows that $\lVert \mathsf{sgn}_S(G,\lambda,\mu) \rVert = \exp^{(d-1)}(t^{O(1)})$.

\begin{theorem} \label{thm-sgnsize'}
    Let $S$ a fixed sequence of $d$ pairs of natural numbers, $p \in \mathbb{N}_0$ be a fixed number, and  $P$ be a fixed set.
    Then for any $p$-labeled and $P$-colored graph $(G,\lambda,\mu)$ with $\mathbf{tw}(G) \leq t$, we have $\lVert \mathsf{sgn}_S(G,\lambda,\mu) \rVert = \exp^{(d-1)}(t^{O(1)})$.
\end{theorem}

\subsection{Putting everything together}
Plugging in the signature bounds proved in the previous section with the algorithms in Sections~\ref{sec-algo} and~\ref{sec-ext}, we are able to prove our main algorithmic results.

\thmMSO*
\begin{proof}
Let $(G,t,S,\phi)$ be an instance of \textsc{MSO Testing} where $S = ((k_1,s_1),\dots,(k_d,s_d))$.
We first compute a nice tree decomposition of $G$ with width $t' = t^{O(1)}$ in $T_\mathsf{td}(G,t)$ time.
Then we apply Theorem~\ref{thm:twcomputation} to compute $\mathsf{sgn}(G)$ and further apply Lemma~\ref{lem:testSat} to test whether $G$ satisfies $\phi$ or not.
By Theorem~\ref{thm-sgnsize}, for any $\mu \in U_{G,[t'+1]_0}$, $\lVert \mathsf{sgn}_S(G,-,\mu) \rVert^{O(1)}$ is bounded by
{\small \begin{equation*}
        f(d) \cdot \left(\sum_{i=1}^d \sum_{j=1}^{i-1} \exp^{(i)}(\hat{O}(s_j k_i)) + \sum_{i=1}^d \sum_{j=1}^i \exp^{(i)}(\hat{O}(t_j k_i)) + \sum_{i=1}^d \exp^{(i)}(\hat{O}(k_i \log t)) + \sum_{i=1}^d 2^{O(s_d k_i)} \right)
\end{equation*}}for some computable function $f$, where $t_j = \min\{k_j,t\}$.
This also bounds the part $\sum_{i=1}^d 2^{O(k_i\log t)}+\sum_{i=1}^{d-1} 2^{O(s_i t)}+ \sum_{i=1}^d 2^{O(s_d k_i)}$ in Theorem~\ref{thm:twcomputation}, simply because $2^{O(s_it)} = 2^{\hat{O}(s_i)} + 2^{2^{\hat{O}(\log t)}}$ by Fact~\ref{fact-decompose}.
Therefore, the total running time of the algorithm follows the bound in the theorem.
\end{proof}

\thmMSOS*
\begin{proof}
The proof is the same as that of Theorem~\ref{thm-MSO}, with Theorem~\ref{thm:twcomputation} and Lemma~\ref{lem:testSat} replaced with Theorem~\ref{thm:new-twcomputation} and Lemma~\ref{lem:new-testSat}, respectively.
\end{proof}

Using the bound for the signature size in Theorem~\ref{thm-sgnsize'}, we can obtain the bounds for \textsc{MSO Testing} and \textsc{MSO Scoring} with a fixed $S$.
\begin{theorem} \label{thm-MSOtw}
    Let $S$ be a fixed sequence of $d$ pairs of natural numbers.
    \begin{itemize}
        \item Given an $n$-vertex graph $G$ with $\mathbf{tw}(G) \leq t$ and an $S$-MSO formula $\phi$, one can test whether $G$ satisfies $\phi$ in $T_\mathsf{td}(G,t) +\exp^{(d-1)}(t^{O(1)}) \cdot (n+|\phi|^{O(1)})$ time.
        \item Given an $n$-vertex graph $G$ with $\mathbf{tw}(G) \leq t$, a weight function $w:V(G) \rightarrow \mathbb{F}$ where $\mathbb{F}$ is a semi-field in which additions and multiplications can be done in $O(1)$ time, and an $S$-MSO* formula $\phi$, one can compute $\mathsf{scr}(G,w,\phi)$ in $T_\mathsf{td}(G,t) +\exp^{(d-1)}(t^{O(1)}) \cdot (n+|\phi|^{O(1)})$ time.
    \end{itemize}
\end{theorem}

\subsection{Generalization to MSO$_2$ logic} \label{sec-MSO2}

All of our results in the previous section can be generalized to MSO$_2$ logic for free.
Indeed, we observe that one can reduce from testing MSO$_2$ properties to testing MSO properties.

\begin{theorem}
    Given an $n$-vertex graph $G$ of treewidth $t$ and an $S$-MSO$_2$ formula $\phi$, one can construct in $O(tn)+|\phi|^{O(1)}$ time a graph $G'$ and an $S$-MSO formula $\phi'$ such that $|V(G')| = O(tn)$, $|\phi'| = |\phi|^{O(1)}$, and $G$ satisfies $\phi$ iff $G'$ satisfies $\phi'$.
\end{theorem}
\begin{proof}
The graph $G'$ is obtained by modifying $G$ as follows.
For each edge $e = (u,v) \in E(G)$, we add a new vertex $v_e'$ together with two edges two edges $(u,v_e')$ and $(v,v_e')$.
Also, we add a self-loop on $v_e'$.
Let $G'$ be the resulting graph.
Clearly, $|V(G')| = |V(G)|+|E(G)| = O(tn)$ and $\mathbf{tw}(G') \leq \mathbf{tw}(G)+1$.
Let $V_1(G') \subseteq V(G')$ consists of the vertices that are originally in $G$ and $V_2(G')  \subseteq V(G')$ consists of the new vertices $v_e'$.
Let $\pi:V(G) \cup E(G) \rightarrow V(G')$ be the bijective function that maps each vertex $v \in V(G)$ to its corresponding vertex in $V_1(G')$ and maps each edge $e \in E(G)$ to $v_e' \in V_2(G')$.

The construction of $\phi'$ is more complicated.
We need to apply induction on $|S|$.
For $d \in \mathbb{N}_0$, we show that for every sequence $S$ of $d$ pairs of natural numbers and every $S$-MSO$_2$ formula $\phi$ (with free variables), one can construct an $S$-MSO formula $\phi'$ satisfying the following conditions.
\begin{enumerate}[(i)]
    \item The free variables of $\phi$ one-to-one correspond to the free variables of $\phi'$.
    If $x$ is a free vertex/edge variable of $\phi$, then it corresponds to a free vertex variable of $\phi'$.
    If $X$ is a free vertex/edge set variable of $\phi$, then it corresponds to a free (vertex) set variable of $\phi'$.
    Let $\mathcal{X}_v'$, $\mathcal{X}_e'$, $\mathcal{X}_V'$, $\mathcal{X}_E'$ denote the sets of free variables of $\phi'$ corresponding to vertex variables, edge variables, vertex set variables, and edge set variables of $\phi$, respectively.
    \item For convenience, for an assignment $A$ of a formula and a free variable $x$ of the formula, we denote by $A(x)$ the value of $x$ in $A$.
    We say an assignment $A'$ of $\phi'$ in $G'$ is \textit{valid} if for every $x' \in \mathcal{X}_v'$ (resp., $x' \in \mathcal{X}_e'$), $A'(x') \in V_1(G')$ (resp., $A'(x') \in V_2(G')$).
    If $A'$ is a valid assignment of $\phi'$ in $G'$, then it naturally corresponds to an assignment $A$ of $\phi$ in $G$:
    \begin{itemize}
        \item if $x$ is a free vertex/edge variable of $\phi$ corresponding to $x' \in \mathcal{X}_v'$, let $A(x) = \pi^{-1}(A'(x'))$;
        \item if $X$ is a free vertex set (resp., edge set) variable of $\phi$ corresponding to $X' \in \mathcal{X}_V'$, let $A(X) = \pi^{-1}(A'(X') \cap V_1(G'))$ (resp., $A(X) = \pi^{-1}(A'(X') \cap V_2(G'))$).
    \end{itemize}
    We require $\phi'$ to satisfy the following condition: for every valid assignment $A'$ of $\phi'$ in $G'$ and its corresponding assignment $A$ of $\phi$ in $G$, $\phi'(A') = \phi(A)$.
\end{enumerate}

When $d = 0$, $\phi$ is quantifier-free and the construction of $\phi'$ is the following.
The free variables of $\phi'$ is one-to-one corresponding to those of $\phi$ as required in condition (i).
We replace each variable in $\phi$ with its corresponding variable of $\phi'$, and change each atomic formula $\mathsf{inc}(\cdot,\cdot)$ in $\phi$ to $\mathsf{adj}(\cdot,\cdot)$.
Let $\phi'$ be the resulting formula.
By checking the atomic formula, it is easy to verify that condition (ii) holds for $\phi'$.
Suppose that for every $S$ of length $d-1$ and every $S$-MSO$_2$ formula $\phi$ (with free variables), we can construct the corresponding $S$-MSO formula $\phi'$ satisfying the two conditions.
Consider a sequence $S = ((k,s))+S'$ with $|S| = d$ and an $S$-MSO$_2$ formula $\phi$.
Without loss of generality, we can assume $\phi(F) = \mathsf{Q} x_1\dots \mathsf{Q} x_k \mathsf{Q} X_1,\dots,\mathsf{Q} X_s\ \phi_0(F,x_1,\dots,x_k,X_1,\dots,X_s)$ where $F$ denotes the free variables of $\phi$, $\mathsf{Q} \in \{\exists,\forall\}$, $x_1,\dots,x_k$ are vertex/edge variables, $X_1,\dots,X_s$ are vertex/edge set variables, and $\phi_0$ is an $S'$-MSO$_2$ formula.
As $|S'| = d-1$, $\phi_0$ has a corresponding $S'$-MSO formula $\phi_0'$ satisfying the two conditions.
Let $I = \{i \in [k]: x_i \text{ is a vertex variable}\}$ and $J = \{j \in [s]: X_j \text{ is a vertex set variable}\}$.
If $\mathsf{Q}=\exists$, then we define 
{\small \begin{equation*}
    \phi'(F') = \exists x_1'\dots \exists x_k' \exists X_1',\dots,\exists X_s'\ \left(\bigwedge_{i \in I} \neg\mathsf{adj}(x_i',x_i')\right)\wedge\left(\bigwedge_{i \in [k] \backslash I} \mathsf{adj}(x_i',x_i')\right)\wedge \phi_0'(F',x_1',\dots,x_k',X_1',\dots,X_s'),
\end{equation*}}where $x_1',\dots,x_k'$ are vertex variables and $X_1',\dots,X_s'$ are (vertex) set variables.
Here the variables in $F'$ one-to-one correspond to the variables in $F$ as in condition (i), and the argument $x_i'$ (resp., $X_j'$) of $\phi_0'$ corresponds to the argument $x_i$ (resp., $X_j$) of $\phi_0$.
Clearly, $\phi'$ is an $S$-MSO formula.
We claim that $\phi'$ satisfies condition (ii).
Let $A'$ be an assignment of $\phi'$ and $A$ be its corresponding assignment of $\phi$.
Assume $\phi'(A') = \mathsf{True}$.
There there should exist $x_1',\dots,x_k' \in V(G')$ and $X_1',\dots,X_s' \subseteq V(G')$ that makes the part of $\phi'(A')$ after the sequence of quantifiers true.
The fact that $(\bigwedge_{i \in I} \neg\mathsf{adj}(x_i',x_i')) \wedge (\bigwedge_{i \in [k] \backslash I} \mathsf{adj}(x_i',x_i')) = \mathsf{True}$ and the construction of $G'$ guarantees that $x_i' \in V_1(G)$ if $i \in I$ and $x_i' \in V_2(G)$ if $i \notin I$.
Set $x_i = \pi^{-1}(x_i')$ for $i \in [k]$, $X_j = \pi^{-1}(X_j' \cap V_1(G'))$ for $j \in J$, and $X_j = \pi^{-1}(X_j' \cap V_2(G'))$ for $j \in [s] \backslash J$.
We have $x_i \in V(G)$ if $i \in I$ and $x_i \in E(G)$ if $i \in I$.
Furthermore, since $\phi_0'(A',x_1',\dots,x_k',X_1',\dots,X_s') = \mathsf{True}$, we have $\phi_0(A,x_1,\dots,x_k,X_1,\dots,X_s) = \mathsf{True}$ by our induction hypothesis.
It follows that $\phi(A) = \mathsf{True}$.
Next, assume $\phi(A) = \mathsf{True}$.
Then there should exist $x_i \in V(G)$ for $i \in I$, $x_i \in E(G)$ for $i \in [k] \backslash I$, $X_j \subseteq V(G)$ for $j \in J$, and $X_j \subseteq E(G)$ for $j \in [s] \backslash J$ such that $\phi_0(A,x_1,\dots,x_k,X_1,\dots,X_s) = \mathsf{True}$.
Set $x_i' = \pi(x_i)$ for $i \in [k]$ and $X_i' = \pi(X_i)$ for $j \in [s]$.
Clearly, $\neg\mathsf{adj}(x_i',x_i')$ for all $i \in I$ and $\mathsf{adj}(x_i',x_i')$ for all $i \in [k] \backslash I$.
Furthermore, the construction of $X_1',\dots,X_s'$ satisfies that $X_j = \pi^{-1}(X_j' \cap V_1(G'))$ for $j \in J$ and $X_j = \pi^{-1}(X_j' \cap V_2(G'))$ for $j \in [s] \backslash J$.
Therefore, $\phi_0'(A',x_1',\dots,x_k',X_1',\dots,X_s') = \mathsf{True}$ by our induction hypothesis, which further implies that $\phi'(A') = \mathsf{True}$.
So $\phi'$ satisfies condition (ii).
On the other hand, if $\mathsf{Q} = \forall$, then we define
{\small \begin{equation*}
    \phi'(F') = \forall x_1'\dots \forall x_k' \forall X_1',\dots,\forall X_s'\ \left(\bigwedge_{i \in I} \neg\mathsf{adj}(x_i',x_i')\right)\wedge\left(\bigwedge_{i \in [k] \backslash I} \mathsf{adj}(x_i',x_i')\right) \rightarrow \phi_0'(F',x_1',\dots,x_k',X_1',\dots,X_s').
\end{equation*}}A similar argument as above shows that in this case, $\phi'$ satisfies condition (ii) as well.
Therefore, our induction works.
This completes the construction of $\phi'$.
One can easily verify that the time for the construction is $|\phi|^{O(1)}$.
\end{proof}

\section{Lower bounds} \label{sec-lb}

In this section we prove our lower bound results. That is, we prove Theorems~\ref{thm:intro-LW1},\ref{thm:intro-LW2},\ref{thm:intro-LW3}, and \ref{thm:intro-LW4}. In Section~\ref{subsec:reductions}, we give reductions that helps us to prove the lower bound results, and finally in Section~\ref{subsec:lw} we prove Theorems~\ref{thm:intro-LW1}, \ref{thm:intro-LW2}, \ref{thm:intro-LW3}, and \ref{thm:intro-LW4}.

\subsection{Reductions}
\label{subsec:reductions}

\paragraph*{Notations} For two positive integers $n$ and $i$, $\log^{(i)}n=\log^{(i-1)} \log  (n)$, where $\log^{(1)}n=\log n$. For two integers $n$ and $d$, $\ilog{n}{d}=\prod_{i=1}^d\log^{(i)} n$. For any $d$, 
$\ilog{n}{d}=(\log n)^{1+o(1)}$. An MSO formula is in {\em prenex normal form} (PNF) if it is of the
form $Q_1x_1 \ldots Q_{\ell} x_{\ell} B$ where $Q_i\in \{\forall,\exists\}$ and $B$ is a quantifier free MSO formula. We can apply the following rules to convert an MSO formula to one in PNF. We use $\phi_1 \equiv \phi_2$ to denote two formulas are logically equivalent.  

\begin{proposition}
\label{prop:PNF:Rules}
Let $\phi$ and $\psi$ be two MSO formulas and $x$ be a variable such that $x$ is a free variable in $\phi$ and $x$ is not a variable in $\psi$. Then, 
\begin{itemize}
    \item $(Q x \;\phi) \wedge \psi \equiv Q x (\phi \wedge \psi)$, where $Q\in \{\forall,\exists\}$. 
    \item $(Q x\; \phi) \vee \psi \equiv Q x (\phi \vee \psi)$, where $Q\in \{\forall,\exists\}$.  
    \item $\neg \forall x \; \phi\equiv \exists x\; \neg \phi$ and  $\neg \exists x \; \phi\equiv \forall x\; \neg \phi$
    \item $(\forall x\;\phi )\rightarrow \psi \equiv \exists x(\phi \rightarrow \psi )$ 
    \item $(\exists x\;\phi )\rightarrow \psi \equiv \forall x(\phi \rightarrow \psi )$ 
\item $(\psi \rightarrow \forall x\;\phi ) \equiv \forall x(\psi \rightarrow \phi )$ 
    \item $(\psi\rightarrow \exists x\;\phi )\equiv \exists x(\psi \rightarrow \phi)$ 
\end{itemize}
\end{proposition}

A graph $G$ is a relational structure with universe $V(G)$ over a finite signature consisting of binary edge relation ${\sf adj}$ and finite number of unary label predicates. In Section~\ref{subsec:labelpredicates}, we explain how to remove label predicates of the formula we construct in the reduction. 
For a label predicate $P$, we write $\exists x\in P \; \psi$ to denote $\exists x (P(x)\wedge \psi)$ and $\forall x \in P \; \psi$ to denote $\forall x (P(x)\implies \psi)$.  

All our reductions that are used to prove Theorems~\ref{thm:intro-LW1}, 
\ref{thm:intro-LW3}, and \ref{thm:intro-LW4} for $i\geq 2$ 
are from \threecoloring. That is, given a graph $G$ (an instance of \threecoloring), we need to construct a tree or treewidth bounded graph with finite number of label predicates and a logic formula such that the formula is true if and only if $G$ is $3$-colorable. These reductions are stated in the three theorems below.

\begin{theorem}
\label{thm:low-MSO}
    Let $d,i,j \in \mathbb{N}$ such that $j<i \leq d$.
    Given an $n$-vertex graph $G$ and integers $k \geq 2$, $s \geq 0$, $\alpha \geq 1$ such that $n \leq \exp^{(i-1)}(\frac{k-9}{2}\cdot \max\{s,1\})$, one can compute in 
    $O(3^{n/\alpha})$ 
    time a tree $T$ and an $S$-MSO formula $\phi$ for $S = ((k_1,s_1),\dots,(k_d,s_d))$ satisfying the following.
    \begin{enumerate}[(i)]
        \item $G$ is a $3$-colorable if and only if $T$ satisfies $\phi$.
        \item $|V(T)| = O(3^{n/\alpha})$ and 
        $|\phi| = O(2^{s}+k+\alpha^2)$. 
        \item $k_i = k$, $k_1 = O(\alpha)$, $k_{i'} = O(1)$ for all $i' \in \{2,\dots,i-1\}$, and $k_{i'} = 0$ for all $i' \in [d] \backslash [i]$.
        \item $s_j = s$ and $s_{j'} = 0$ for all $j' \in [d] \backslash \{j\}$.
        \item The number of label predicates used in $\phi$ is $O(1)$
    \end{enumerate}
\end{theorem}

\begin{theorem}
\label{thm:low-MSO-TW1}
    Let $d,i,j \in \mathbb{N}$ such that $j \leq i \leq d$ and $i \geq 2$.
    Given an $n$-vertex graph $G$ and integers $k \geq 2$, $k' \geq 2$, $t \geq 1$, $\alpha \geq 1$ such that $n \leq \exp^{(i-1)}(\min\{k,t\}\cdot \frac{k'-9}{2})$, one can compute in $O(3^{n/\alpha})$ time a graph $G'$ and an $S$-FO formula $\phi$ for $S = (k_1,\dots,k_d)$ satisfying the following.
    \begin{enumerate}[(i)]
        \item $G$ is a $3$-colorable if and only if $G'$ satisfies $\phi$.
        \item $|V(G')| = O(3^{n/\alpha})$, $\mathbf{tw}(G') \leq t$, and 
        $|\phi| = O(2^{\min\{k,t\}} +k'+\alpha^2)$.
        \item $k_{i'} = O(1)$ for all $i' \in [i] \backslash \{1,j,i\}$, $k_{i'} = 0$ for all $i' \in [d] \backslash [i]$, and
        \begin{itemize}
            \item if $1 = j < i$, then $k_1 = k_j = O(\alpha+k)$ and $k_i = k'$;
            \item if $1 < j < i$, then $k_1 = O(\alpha)$, $k_j = k$, and $k_i = k'$;
            \item if $1 < j = i$, then $k_1 = O(\alpha)$, $k_j = k_i = k+k'$.
        \end{itemize}
        \item The number of label predicates used in $\phi$ is $O(1)$.
    \end{enumerate}
\end{theorem}

\begin{theorem}
\label{thm:low-MSO-TW22}
    Let $d,i \in \mathbb{N}$ such that $2 \leq i \leq d$.
    Given an $n$-vertex graph $G$ and integers $k \geq 2$, $t \geq 1$, $\alpha \geq 1$ such that $n \leq \exp^{(i-1)}(\frac{k-6}{3} \log t)$, one can compute in $O(3^{n/\alpha})$ time a graph $G'$ and an $S$-FO formula $\phi$ for $S = (k_1,\dots,k_d)$ satisfying the following.
    \begin{enumerate}[(i)]
        \item $G$ is a $3$-colorable if and only if $G'$ satisfies $\phi$.
        \item $|V(G')| = O(3^{n/\alpha})$, $\mathbf{tw}(G') \leq t$, and $|\phi| = O(\alpha^{2}+k)$.
        \item $k_i = k$, $k_1 = O(\alpha)$, $k_{i'} = O(1)$ for all $i' \in \{2,\dots,i-1\}$, and $k_{i'} = 0$ for all $i' \in [d] \backslash [i]$.
        \item The number of label predicates used in $\phi$ is $O(1)$. 
    \end{enumerate}
\end{theorem}

 For all the reductions the initial part of the construction of the tree or treewidth bounded graph is a tree (let as call it as base tree). So, first we  explain the construction of the base tree and for each of the above theorems, we modify the base tree. Along with the base tree, we also define finite number of label predicates and define a formula which is an FO formula except that it contains a function ${\sf id}$ which returns a non-negative integer. So, let us call such a formula as FO+{\sf id} formula. Then, for different theorems above, we explain how to replace ${\sf id}$ with a valid subformula in the logic. 

\begin{lemma}[Construction of base tree]
\label{lem:base-red}
There is an algorithm that given an $n$-vertex graph $G$ and an integer $\alpha>1$, runs in time $O(3^{n/\alpha}+m)$ and outputs a tree $T$ and an {\em $(2\alpha+1,9)$-FO+${\sf id}$} formula $\psi$ satisfying the following, where $m=|E(G)|$.  

    \begin{enumerate}[(i)]
        \item $G$ is a $3$-colorable if and only if $T$ satisfies $\psi$.
        \item $|V(T)| = O(3^{n/\alpha}+m)$ and $|\psi| =  O(\alpha^2)$. 
        \item The number of label predicates specified in the formula is $9$. 
    \end{enumerate}
\end{lemma}

\begin{proof}
Let $G$ be the given $n$-vertex graph. 
Let $V(G)=\{1,2,\ldots,n\}$ and $E(G)=\{e_1,\ldots,e_m\}$. Now we partition $V(G)$ into $\alpha$ groups $V_1,\ldots,V_{\alpha}$ such that for each $i\in [\alpha]$, $|V_i|\leq \lceil \frac{n}{\alpha}\rceil$. Let $\ell=3^{\lceil \frac{n}{\alpha}\rceil}$. For each $i\in [\alpha]$, there are at most $\ell$ proper $3$-colorings of $G[V_i]$. Let us call these $3$-colorings  $c_{i,1},\ldots c_{i,{\ell}_i}$. 

\smallskip 
\medskip
\noindent
{\bf Construction.}
Now we construct a tree $T$ rooted at a node $rt$ as follows. The root $rt$ has $\alpha+m$ children and we name them $U_1,\ldots,U_{\alpha}$ and $f_1,\ldots, f_m$. That is, each node $U_i$ corresponds to the vertex subset $V_i$ of $G$ and each node $f_i$ corresponds to the edge $e_i$ of $G$. Each node $f_i$ has two children corresponding to the endpoints of $e_i$. Let us name these nodes with $f_{i,a}$ and $f_{i,b}$, 
where $a$ and $b$ are the endpoints of $e_i$. 
Now, we explain the children of each $U_i$. Recall that $\{c_{i,1},\ldots,c_{i,\ell_i}\}$ is the set of all proper $3$-colorings of $G[V_i]$. The node $U_i$ has $\ell_i$ children, and they are named $C_{i,1},\ldots,C_{i,\ell_i}$. Each $C_{i,j}$ has $|V_i|$ children, and each of them corresponds to a vertex in $V_i$. 
That is, each $a\in V_i$, $C_{i,j}$ has a child node named $v_{i,j,a}$. Now, each $v_{i,j,a}$ has two children $id_{i,j,a}$ and $c_{i,j,a}$. 
See Figure~\ref{fig:Base-tree} for an illustration.

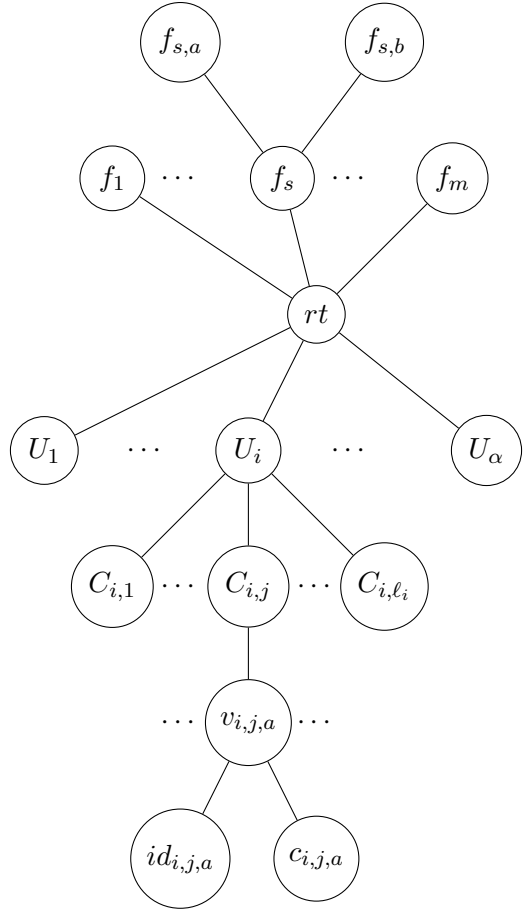
\begin{figure}[t]
    \centering
    
\begin{tikzpicture}[scale=0.9, ver/.style = {draw, circle}]

 \tikzset{tri1/.style={isosceles triangle, draw, inner sep=0pt,
     anchor=south, shape border rotate=90, isosceles triangle stretches}}

 \tikzset{tri2/.style={isosceles triangle, draw, inner sep=0pt,
     anchor=south, shape border rotate=270, isosceles triangle stretches}}


\node[ver] (r) at (0,0) {$rt$};

\node[ver] (f1) at (-3,2) {$f_{1}$}; 
\node at (-2,2) {$\cdots$};
\node[ver] (f2) at (-0.5,2) {$f_s$};
\node at (0.5,2) {$\cdots$};
\node[ver] (fm) at (2,2) {$f_m$};

\draw[-] (f1)--(r) -- (f2);
\draw[-] (r) -- (fm);

\node[ver] (es1) at (-2,4) {$f_{s,a}$};
\node[ver] (es2) at (1,4) {$f_{s,b}$};

\draw[-] (es1)--(f2) -- (es2);

\node[ver] (1) at (-4,-2) {$U_1$}; 
\node at (-2.5,-2) {$\cdots$};
\node[ver] (i) at (-1,-2) {$U_i$};
\node at (0.5,-2) {$\cdots$};
\node[ver] (p) at (2.5,-2) {$U_{\alpha}$};
\draw[-] (1)--(r) -- (i);
\draw[-] (r) -- (p);

\node[ver] (d1) at (-3,-4) {$C_{i,1}$};
\node at (-2,-4) {$\cdots$};
\node[ver] (d2) at (-1,-4) {$C_{i,j}$};
\node at (0,-4) {$\cdots$};
\node[ver] (dl) at (1,-4) {$C_{i,\ell_i}$};

\draw[-] (d1)--(i) -- (d2);
\draw[-] (i) -- (dl);

\node[ver] (dija) at (-1,-6) {$v_{i,j,a}$};
\node at (0,-6) {$\cdots$};
\node at (-2,-6) {$\cdots$};

\draw[-] (dija)--(d2);

\node[ver] (dj1) at (-2,-8) {$id_{i,j,a}$};
\node[ver] (dj2) at (0,-8) {$c_{i,j,a}$};

\draw[-] (dj1)--(dija) -- (dj2);




\end{tikzpicture} 

    \caption{Illustration of construction of tree $T$ in Lemma~\ref{lem:base-red}}
    \label{fig:Base-tree}
\end{figure}

Next, we define $9$ unary label predicates on $T$.

\begin{itemize}
    \item ${\sf R}=\{rt\}$. 
    \item ${\sf EN}=\{f_1,\ldots,f_m\}$
    \item ${\sf EP}$ is the set of children of nodes in ${\sf EN}$. 
    \item  ${\sf P}$ is the union of the set of children of $U_i$ over all $i$. That is, ${\sf P}=\{C_{i,j}~:~i\in [\alpha], j\in [\ell_i]\}$.  
    \item ${\sf N}=\{v_{i,j,a}~:~i\in [\alpha],j\in [\ell_i],a\in V_i\}$
    \item ${\sf I}=\{id_{i,j,a}~:~i\in [\alpha], j\in [\ell_i],a\in V_i\}$
    \item For each $q\in [3]$, ${\sf Q}_q=\{c_{i,j,a} ~:~a \textit{ is colored with $q$ in the proper coloring $c_{i,j}$ of $G[V_i]$}\}$. 
\end{itemize}

Next we define the function ${\sf id}$ on ${\sf EP\cup I}$. 

$$
  {\sf id}(x) =     \left\{ 
  \begin{array}{rcl}
  b & \mbox{if} & x=f_{s,b}\in {\sf EP} \\
  a & \mbox{if} & x=id_{i,j,a}\in {\sf I}
  \end{array}\right.
$$

Now we define an FO+${\sf id}$ formula. We want to encode the statement that ``there exist nodes $x_1, x_2,\ldots, x_{\alpha}$ that correspond to proper $3$-colorings of $G[V_1],\ldots G[V_{\alpha}]$, respectively, such that for any node $z \in {\sf EN}$, the endpoints of the edge corresponding to $z$ should get different colors according to the selected colorings of $G[V_1],\ldots G[V_{\alpha}]$''. This can be encoded as follows.

\begin{eqnarray}
 \psi \equiv&& \exists r \in {\sf R} \exists x_1\in {\sf P} \; \exists x_2 \in {\sf P}\; \ldots \; \exists x_{\alpha}\in {\sf P} \exists p_1 \ldots \exists p_{\alpha}  \nonumber\\ 
 && \forall z \in {\sf EN}\; \forall z_1, z_2 \in {\sf EP}\; \forall v_1, v_2 \in {\sf N} \; \forall y_1, y_2 \in {\sf I} \; \forall c_1,c_2 \in {\sf Q}
 _1\cup {\sf Q}_2 \cup {\sf Q}_3\nonumber\\
 && \psi_{{\sf evalid}} \wedge( \psi_{{\sf uvalid}} \Rightarrow ( \psi_{{\sf id}}\Rightarrow \psi_{{\sf color}} ))\nonumber\\
  &\equiv& \exists r \in {\sf R} \exists x_1\in {\sf P} \; \exists x_2 \in {\sf P}\; \ldots \; \exists x_{\alpha}\in {\sf P} \exists p_1 \ldots \exists p_{\alpha}   \nonumber\\ 
 && \forall z \in {\sf EN}\; \forall z_1, z_2 \in {\sf EP}\; \forall v_1, v_2 \in {\sf N} \; \forall y_1, y_2 \in {\sf I} \; \forall c_1,c_2 \in {\sf Q}
 _1\cup {\sf Q}_2 \cup {\sf Q}_3 \nonumber\\
 && \psi_{{\sf evalid}} \wedge ( \neg \psi_{{\sf uvalid}} \vee \neg\psi_{{\sf id}}\vee \psi_{{\sf color}}) \label{eqn:psi1}
\end{eqnarray}

Here, $\psi_{\sf evalid}$ and $\psi_{\sf uvalid}$ ensure that the selected vertices are {\em valid}. Before explaining it, let us define  
$\psi_{{\sf id}}$ and $\psi_{{\sf color}}$. 
The formula
$\psi_{{\sf id}}$ is the formula $ {\sf id}(z_1)={\sf id}(y_1)\wedge {\sf id}(z_2)={\sf id}(y_2)$. 
Next, we explain the formula $\psi_{{\sf color}}$.

$$\psi_{{\sf color}}\equiv \neg (\bigvee_{i\in [3]} (c_1\in {\sf Q}_i \wedge c_2\in {\sf Q}_i)  )$$

That is, $\psi_{{\sf color}}$ is true if and only if $c_1$ and $c_2$ are two different colors. 
To define $\psi_{\sf uvalid}$, let us first define ${\sf con}(X,Y)$ for two sets of variables $X$ and $Y$. ${\sf con}(X,Y)$ is true if and only if every vertex $x\in X$ is adjacent to at least one vertex in~$Y$. Formally, $${\sf con}(X,Y)=\bigwedge_{x\in X}\left(\bigvee_{y\in Y} {\sf adj(x,y)}\right).$$

Now, $\psi_{{\sf uvalid}}$ is defined below. 

$$\psi_{{\sf uvalid}}\equiv {\sf con}(\{z_1,z_2\},\{z\})\wedge {\sf con}(\{v_1,v_2\},\{x_1,\ldots,x_{\alpha}\}) \wedge {\sf con}(\{y_1,c_1\},\{v_1\})\wedge {\sf con}(\{y_2,c_2\},\{v_2\}).$$

We also need $x_1, x_2,\ldots, x_{\alpha}$ to correspond to appropriate $3$-colorings, one for each $G[V_i]$, where $i\in [\alpha]$. We encode it using a formula $\psi_{\sf evalid}$ as follows. 

$$\psi_{\sf evalid} \equiv \left(\bigwedge_{1\leq i < j \leq [\alpha]} p_i\neq p_j \right) \wedge \left( \bigwedge_{i\in \alpha} {\sf adj}(x_i,p_i)\wedge {\sf adj}(p_i,r)\right)$$

It is easy to see that $\psi$ is an $(2\alpha+1,9)$-FO+{\sf id} formula and 
$|\psi|=O(\alpha^2)$. Notice that the size of $T$ is $O(3^{n/\alpha}+m)$ and it can be constructed in time $O(3^{n/\alpha}+m)$. 
Now prove the correctness of the reduction.  
\begin{claim}
$G$ has a proper $3$-coloring if and only if $T\models \psi$.     
\end{claim}

\begin{proof}
 Suppose $G$ has a proper $3$-coloring. 
 Let $c_{1,j_1},\ldots,c_{\alpha,j_{\alpha}}$ be proper $3$-colorings of $G[V_1],\ldots,G[V_{\alpha}]$, respectively, such that their union is a proper $3$-coloring of $G$. Now, substitute $x_i=C_{i,j_i}$ and $p_i=U_i$ for all $i\in [\alpha]$. Notice that $\psi_{\sf evalid}$ is true for the above assignment. Now, fix a choice of values to $z,z_1,z_2,v_1,v_2,y_1,y_2,c_1,c_2$ such that the unary label predicates mentioned in the quantifier prefix hold. 
 If $\psi_{{\sf uvalid}}$
 is false, then $\psi$ is true. Now, assume that $\psi_{\sf valid}$ is true. 
 This implies that 
 \begin{itemize}
     \item[$(i)$] $z_1$ and $z_2$ are adjacent to $z$ and $z$ corresponds to an edge $e\in E(G)$, 
     \item[$(ii)$] $v_i$ is adjacent to a vertex in $\{C_{1,j_1},\ldots,C_{p,j_p}\}$ for all $i\in \{1,2\}$, and 
    \item[$(iii)$] $y_i$ and $c_i$ are adjacent to $v_i$ for all $i\in \{1,2\}$. 
 \end{itemize}
 Now, we need to prove that 
 $\psi_{{\sf id}}\Rightarrow \psi_{{\sf color}}$
 is true. Notice that $z_1$ and $z_2$ correspond to endpoints of an edge $e$ in $G$ (see item $(i)$). Suppose $\psi_{{\sf id}}$ is true. 
 That is,  ${\sf id}(z_1)={\sf id}(y_1)$ and ${\sf id}(z_2)={\sf id}(y_2)$, then clearly the vertices corresponding to $y_i$ and $z_i$ in $G$ are the same. Moreover, we know that ${\sf id}(z_1)$ and ${\sf id}(z_2)$ are the endpoints of the edge $e$ in $G$. This implies that colors of ${\sf id}(z_1)$ and ${\sf id}(z_2)$ by the $3$-coloring should be different. The colors of these endpoints are encoded in $c_1$ and $c_2$ because ${\sf id}(z_1)={\sf id}(y_1)$ and ${\sf id}(z_2)={\sf id}(y_2)$. This implies that 
 $\psi_{\sf color}$ is true. 
 Hence, $\psi_{{\sf id}}\Rightarrow \psi_{{\sf color}}$
 is true.

 Now, we prove the reverse direction. Suppose $G$ is not $3$-colorable. Then we need to prove that $\neg \psi$ is true. Notice that 

\begin{eqnarray*}
\neg \psi \equiv&& \forall r\in {\sf R} \; \forall x_1\in {\sf P} \; \ldots \;\forall x_{\alpha}\in {\sf P}\; \forall p_1 \;  \ldots \; \forall p_{\alpha} \\ 
&& \exists z \in {\sf EN}\; \exists z_1, z_2 \in {\sf EP}\; \exists v_1, v_2 \in {\sf N}\; \exists y_1, y_2 \in {\sf I}\; \exists c_1,c_2 \in {\sf Q}_1\cup {\sf Q}_2 \cup {\sf Q}_3\\
 && \neg \psi_{{\sf evalid}} \vee ( \psi_{{\sf evalid}} \wedge \psi_{{\sf id}}\wedge \neg \psi_{{\sf color}} )
\end{eqnarray*}
Fix $r\in {\sf R}$ and $x_1,\ldots, x_{\alpha}\in {\sf P}$. If $\psi_{\sf evalid}$ is false for the above assignments, then $\neg \psi$ is true and we are done. So now assume that $\psi_{\sf evalid}$ is true. This implies that $x_1, x_2,\ldots, x_{\alpha}$ correspond to $3$-colorings, one for each $G[V_i]$, where $i\in [\alpha]$. Without loss of generality assume that each $x_i$ corresponds to a proper $3$-coloring of $G[V_i]$. Since $G$ is not $3$-colorable, there is an edge $e=\{a,b\}$ such that the endpoints of $e$ are colored using the same color by the union $\pi$ of above colorings. Let $f_s$ be the node in $T$ corresponding to $e$. Let $z=f_s, z_1=f_{s,a}$, and $z_2=f_{s,b}$. Let $i_1,i_2\in [\alpha]$ such that $a\in V_{i_1}$ and $b\in V_{i_2}$. Then, there exist $r_1$ and $r_2$ such that $v_{i_1,r_1,a}$ is a child of $x_{i_1}$ and $v_{i_2,r_2,b}$ is a child of $x_{i_2}$. Now we set $v_1=v_{i_1,r_1,a}$
and $v_2=v_{i_2,r_2,b}$. Let $y_1=id_{i_1,r_1,a}, \; y_2=id_{i_2,r_2,b}$, 
$c_1=c_{i_1,r_1,a}$, and $c_2=c_{i_2,r_2,b}$. Now it is easy to verify that $\psi_{\sf uvalid}$ and $\psi_{\sf id}$ are true. Since both $a$ and $b$ got the same color by $\pi$. This implies that $\psi_{\sf color}$ is false and hence $\neg \psi_{\sf color}$ is true. 
This implies that $\neg \psi$ is true. This completes the proof of the claim and the lemma. 
\end{proof}
\end{proof}

Notice that in the formula $\psi$ created in Lemma~\ref{lem:base-red}, the sub formulas $\psi_{\sf evalid}, \psi_{\sf uvalid}$ and $\psi_{\sf color}$ are quantifier free FO-formulas. On the other hand $\psi_{{\sf id}} \equiv {\sf id}(z_1)={\sf id}(y_1)\wedge {\sf id}(z_2)={\sf id}(y_2)$. So we need to encode ${\sf id}(z_1)={\sf id}(y_1)$ for any two nodes $z_1$ and $y_1$ in ${\sf EP \cup I}$, 
using FO or MSO formulas. All lower bound theorems are based on how we encode testing whether two {\sf id}s are same. Next, we explain various methods for this task. 

{\bf For the rest of this section $T$ is the graph constructed in Lemma~\ref{lem:base-red}}.   

\begin{lemma}[$\log n$-length FO identifier test]
\label{lem:L1-ex-IDtest}
Let $n$ be a positive integer and for any leaf node $x$, ${\sf id}(x) \in [n]$ ( if ${\sf id}$ is defined on $x$). 
One can construct a tree $T'$ (which is super graph of $T$) and $(2\log n)$-FO formula $\phi_{\exists}$ on two free variables 
with the following specifications. 
\begin{enumerate}[(a)]
    \item For any two nodes $z$ and $y$ in $T$, $T'\models \phi_{\exists}(z,y)$ if and only if ${\sf id}(z)={\sf id}(y)$
    \item All the quantifiers in $\phi_{\exists}(z,y)$ are existential quantifiers and $|\phi_{\exists}(z,y)|=O(\log n)$
    \item The number of label predicates used in $\phi_{\exists}(z,y)$ is two. 
    \item $V(T')=O(|V(T)|\cdot \log n)$
\end{enumerate}
\end{lemma}

\begin{proof}
We need an FO formula to represent ${\sf id}(z)={\sf id}(y)$ for two variables $z$ and $y$. First, we create two label predicates ${\bf 1}$ and ${\bf 0}$. 
Let $k$ be the smallest integer such that $k\geq \log n$. 
For any node $t$ in $T$ such that ${\sf id}(t)$ is defined, we do the following. 
We know that ${\sf id}(t)\in [n]$. 
Thus, ${\sf id}(t)$ can be represented in binary using $k$ bits. Let $\pi_t$ be a path $a_1,\ldots,a_{k},t$ on $k+1$ vertices. Now, we replace node $t$ with path $\pi_t$. Let $b_1,\ldots,b_{k}$ be the binary representation of the number ${\sf id}(t)$. For each $i\in [k]$, $a_i\in {\bf 0}$ if $b_i=0$ and $a_i\in {\bf 1}$ otherwise. The tree constructed as explained above is $T'$. Clearly,  $|V(T')|=O(|V(T)|\cdot \log n)$. Now, for two leaf nodes $z$ and $y$ in $T$, ${\sf id}(z)={\sf id}(y)$ can be encoded as 

\begin{eqnarray*}
\exists a_1,a_1'\in {\bf O}\cup {\bf 1} \ldots \exists a_{k},a_{k}' \in {\bf O}\cup {\bf 1}
&&
\left(\bigwedge_{i\in [k] }
(a_i \in {\bf 0} \Leftrightarrow a'_i \in {\bf 0})\right) \\
&& \wedge  {\sf path}(a_1,\ldots,a_{k},z)
\wedge
{\sf path}(a'_1,\ldots,a'_{k},y) 
\end{eqnarray*}

where, ${\sf path}(w_1,\ldots,w_q) \equiv (\bigwedge_{i\in [q-1]} {\sf adj}(w_i,w_{i+1}))$. Notice that the number of quantifiers in the above formula is $2\log n$.     
\end{proof}

\begin{lemma}[$\log\log n$-length FO identifier test]
\label{lem:L1-uni-IDtest}
Let $n$ be a positive integer and 
for any leaf node $x$, ${\sf id}(x) \in [n]$ ( if ${\sf id}$ is defined on $x$). 
One can construct a tree $T'$ (which is super graph of $T$) and $(6+2\log\log n)$-FO formula $\phi_{\forall}$ on two free variables 
with the following specifications. 
\begin{enumerate}[(a)]
    \item For any two nodes $z$ and $y$ in $T$, $T'\models \phi_{\forall}(z,y)$ if and only if ${\sf id}(z)={\sf id}(y)$
    \item All the quantifiers in $\phi_{\forall}(z,y)$ are universal quantifiers and $|\phi_{\forall}(z,y)|=O(\log\log n)$
    \item The number of label predicates used in $\phi_{\forall}(z,y)$ is $6$. 
    \item $V(T')=O(|V(T)|\cdot \log n \cdot \log \log n)=O(|V(T)|\cdot \ilog{n}{2})$
\end{enumerate}
\end{lemma}

\begin{proof}

Like in Lemma~\ref{lem:L1-ex-IDtest}, we need an FO formula to represent ${\sf id}(z)={\sf id}(y)$ for two variables $z$ and $y$. First, we create four label predicates ${\bf 1}$, ${\bf 0}$, ${\sf L}_1$, and ${\sf L}_2$. 
Let $k$ and $k'$ be the smallest integers such that $k\geq \log n$ and $k'\geq \log \log n$. 
For any node $t \in V(T)$ such that ${\sf id}(t)$ is defined, 
we do the following. We know that ${\sf id}(t)\in [n]$. 
Thus, ${\sf id}(t)$ can be represented in binary using $k$ bits. 
Let $b_1\ldots b_{k}$ be the binary representation of the number ${\sf id}(t)$. 
Now we replace $t$ with a subtree as shown in Figure~\ref{fig:treeid1}. That is, $t$ has $k$ children $a_1,\ldots,a_k$ and each $a_i$ has two children $d_i$ and $c_i$. Here, each $a_i$ represents $b_i$ in the following way. We set ${\sf id}(d_i)=i$, $c_i\in {\bf O}$ if $b_i=0$ and $c_i\in {\bf 1}$ if $b_i=1$. In other words, ${\sf id}(d_i)$ represents the position of $b_i$ in the binary representation $b_1\ldots b_{k}$ and $c_i$ denotes the value of the bit $b_i$. The crucial observation is that for each $i\in [k]$, ${\sf id}(d_i)$ is a positive integer less than or equal to $\lceil \log n \rceil$. Now we add all $a_i$ to ${\sf L}_1$ and all $d_i$ to ${\sf L}_2$. Let us name the tree constructed as $T_1$. It is easy to see that $|V(T_1)|=O(|V(T)| \log n)$.


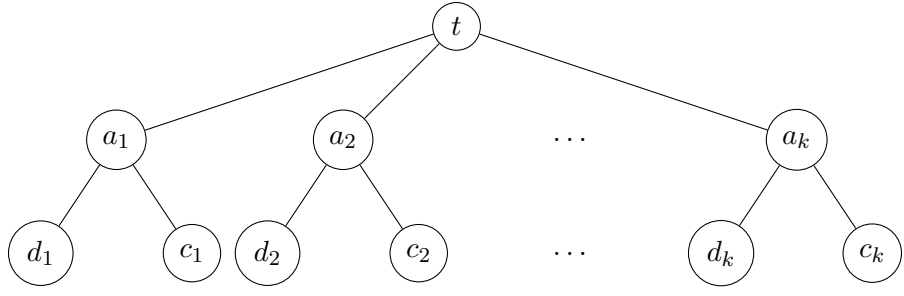
\begin{figure}
    \centering
\begin{tikzpicture}[
  grow=down,
  level 1/.style={sibling distance=3cm},
  level 2/.style={sibling distance=2cm},
  every node/.style={circle,draw},
  edge from parent/.style={draw,-}
]

\node {$t$}
  child {node {$a_1$}
    child {node {$d_1$}}
    child {node {$c_1$}}
  }
  child {node {$a_2$}
    child {node {$d_2$}}
    child {node {$c_2$}}
  }
  child {node[draw=none] {$\ldots$} edge from parent[draw=none]
    child {node[draw=none] {$\ldots$} edge from parent[draw=none]}
  }
    child {node {$a_k$}
    child {node {$d_k$}}
    child {node {$c_k$}}
  };

\end{tikzpicture}
  \caption{Subtree to replace the node $t$}
\label{fig:treeid1}
\end{figure}

Informally, for two leaf nodes $z$ and $y$ in $T$, ${\sf id}(z)={\sf id}(y)$ is true if and only if for any child $a$ of $z$ and any child $a'$ of $y$ if the {\em id of the left child of} $a$ and the {\em id of the left child of} $a'$ are equal, then the {\em corresponding bits (encoded in the right child of $a$ and $a'$)} are same. This can be encoded as follows.    

\begin{eqnarray}
{\sf id}(z)={\sf id}(y) &\equiv& \forall a,a' \in {\sf L}_1, \forall d,d' \in {\sf L}_2 \forall c,c' \in {\bf 0}\cup {\bf 1} \;\;
\phi_{{\sf valid}} \Rightarrow  ( \phi_{{\sf id}}\Rightarrow \phi_{{\sf bit}} )\nonumber\\
&\equiv& \forall a,a' \in {\sf L}_1, \forall d,d' \in {\sf L}_2 \forall c,c' \in {\bf 0}\cup {\bf 1} \;\;
(\neg \phi_{{\sf valid}} \vee  \neg \phi_{{\sf id}}\vee \phi_{{\sf bit}} )\label{eqn:uniloglog}
\end{eqnarray}

Here, $\phi_{{\sf valid}}$, $\phi_{{\sf id}}$ and $\phi_{{\sf bit}}$ are defined below.

\begin{eqnarray*}
\phi_{\sf id} &\equiv& {\sf id}(d)={\sf id}(d')\\
\phi_{\sf bit} &\equiv& c \in {\bf 0} \Leftrightarrow c'\in {\bf 0} \\
\phi_{\sf valid} &\equiv& ({\sf con}(\{d,c\},\{a\})\wedge ({\sf con}(\{d',c'\},\{a'\}) \wedge {\sf adj}(a,z) \wedge {\sf adj}(a',y)     
\end{eqnarray*}

The formula $\phi_{\sf valid}$ is true if and only if $a$ and $a'$ are children of $z$ and $y$, respectively, $d$ and $c$ are children of $a$, and $d'$ and $c'$ are children of $a'$. The formula $\phi_{\sf bit}$ is true if and only if they both {\em encode} the same bit. Clearly, the formula in (\ref{eqn:uniloglog}) is not an FO formula, but an FO+${\sf id}$ formula with the value of ${\sf id}(d)$ and ${\sf id}(d')$ are positive integers less than or equal to $\lceil \log n \rceil$. So, we apply Lemma~\ref{lem:L1-ex-IDtest} on $T_1$ where $n$ is replaced with $\log n$. Thus, we get a tree $T'$ (which is super graph of $T_1$) and a $(\log\log n)$-FO formula $\phi_{\exists}$ on two free variables 
with the following specifications. 
\begin{enumerate}[(a)]
    \item For any two leaf nodes $d$ and $d'$ in $T_1$, $T'\models \phi_{\exists}(d,d')$ if and only if ${\sf id}(d)={\sf id}(d')$. 
    \item All the quantifiers in $\phi_{\exists}(d,d')$ are existential quantifiers and $|\phi_{\exists}(d,d')|=O(\log \log n)$
    \item The number of label predicates used in $\phi_{\exists}(z,y)$ is two. 
    \item $V(T')=O(|V(T_1)|\cdot \log\log n)$
\end{enumerate}

Thus, we can replace $\phi_{\sf id}$ in the formula in (\ref{eqn:uniloglog}) with $\phi_{\exists}(d,d')$. Notice that there is negation before $\phi_{\sf id}$ in the formula in (\ref{eqn:uniloglog}). The formula $\neg \phi_{\exists}(d,d')$ is of the form $\neg \exists a_1 \ldots \exists a_{k} \phi'$ where $\phi'$ is quantifier free and $k=2\log \log n$. This implies that 
$$
\neg \phi_{\exists}(d,d') \equiv 
\forall a_1 \ldots \forall a_{k} \neg \phi'.
$$

Thus, by replacing $\phi_{\sf id}$ in the formula in (\ref{eqn:uniloglog}) with $\phi_{\exists}(d,d')$, we get $(6+2\log \log n)$-FO formula $\phi_{\forall}$ where all the quantifiers are universal quantifiers. By Lemma~\ref{lem:L1-ex-IDtest} and (\ref{eqn:uniloglog}), we get that the length of $\phi_{\forall}$ is $O(\log \log n)$. Property (d) above and the fact that  $|V(T_1)|=O(|V(T)|\cdot \log n)$ implies that $|V(T')|=O(|V(T)|\cdot \log n \cdot \log \log n)$. 
Initially, we created four label predicates and Lemma~\ref{lem:L1-ex-IDtest} introduced two label predicates. Thus total number of label predicates created is $6$. This completes the proof of the lemma.  
\end{proof}

Notice that the proof of Lemma~\ref{lem:L1-uni-IDtest} has two steps. In the first step, we created an $O(1)$ length FO+${\sf id}$ formula (see (\ref{eqn:uniloglog})) and in the second step we applied Lemma~\ref{lem:L1-ex-IDtest} on the formula obtained in step 1, but now the {\sf id} values are at most $\log n$. In fact we can apply step 1 multiple times to get an $(c_1,\ldots,c_{d})$-FO+${\sf id}$ formula with {\sf id} values at most $\log^{(d)}n$. Towards getting that, let us formalize the outcome of step 1 of the proof of Lemma~\ref{lem:L1-uni-IDtest} in the following lemma.     

\begin{lemma}[Lower the identifier value]
\label{lem:L1-uni-IDtestStep1}
Let $n$ be a positive integer and 
for any leaf node $x$, ${\sf id}(x) \in [n]$ (if ${\sf id}$ is defined on $x$). 
One can construct a tree $T'$ (which is super graph of $T$) and $(6)$-FO+{\sf id} formula $\phi_{\forall}$ on two free variables 
with the following specifications. 
\begin{enumerate}[(a)]
    \item $\phi_{\forall}(z,y)$ is of the form $\forall x_1, \ldots, \forall x_6 \;\; (\phi_1(z,y,x_1,\ldots,x_6) \vee \neg ({\sf id}(x_1)={\sf id}(x_2)))$, where $\phi_1$ is a quantifier free formula of length $O(1)$ and ${\sf id}(x_1), {\sf id}(x_2) \in [\lceil \log n \rceil]$. 
    \item For any two nodes $z$ and $y$ in $T$, $T'\models \phi_{\forall}(z,y)$ if and only if ${\sf id}(z)={\sf id}(y)$. 
    \item The number of label predicates used in $\phi_{\forall}(z,y)$ is $4$. 
    \item $V(T')=O(|V(T)|\cdot \log n)$
\end{enumerate}
\end{lemma}

\begin{proof}[Proof sketch]
Follow the proof of Lemma~\ref{lem:L1-uni-IDtest} until the construction of formula in (\ref{eqn:uniloglog}). The formula in the right hand side of the equation in (\ref{eqn:uniloglog}) is the required formula where $\phi_1(z,y)$ is $\neg \phi_{\sf valid} \vee \phi_{\sf bit}$.     
\end{proof}

Now, by applying Lemma~\ref{lem:L1-uni-IDtestStep1} $d$ times, we get the following. 

\begin{lemma}[Lower the identifier value iteratively]
\label{lem:ID-level-d}
Let $n$ and $d$ be two positive integers and 
for any leaf node $x$, ${\sf id}(x) \in [n]$ (if ${\sf id}$ is defined on $x$). 
One can construct a tree $T'$ (which is super graph of $T$) and $(c_1,\ldots,c_d)$-FO+{\sf id} formula $\phi$ on two free variables  
with the following specifications. 
\begin{enumerate}[(a)]
    \item $c_i=6$ for all $i\in [d]$. 
    \item If $d$ is odd, then $\phi(z,y)$ is of the form
    \begin{eqnarray*}
     \forall x_{1,1}, \ldots, \forall x_{1,6} && ( \phi_1(z,y, x_{1,1},\ldots,x_{1,6}) \vee \\
     (\exists x_{2,1}\ldots,\exists x_{2,6} && ( \phi_2(x_{1,1},x_{1,2},x_{2,1},\ldots,x_{2,6}) \wedge \\
     (\forall x_{3,1}, \ldots, \forall x_{3,6} &&  (\phi_3(x_{2,1},x_{2,2},x_{3,1},\ldots,x_{3,6})\vee \\ 
     \vdots && \\
     (\forall x_{d,1}, \ldots, \forall x_{d,6} && 
     (\phi_d(x_{d-1,1},x_{d-1,2},x_{d,1},\ldots,x_{d,6})\vee 
     \neg ({\sf id}(x_{d,1})={\sf id}(x_{d,2})))))))))
    \end{eqnarray*}
    where ${\sf id}(x_{d,1}), {\sf id}(x_{d,2}) \in [\lceil \log^{(d)} n \rceil],$ and $\phi_1,\ldots,\phi_d$ are quantifier free formulas. 
    \item If $d$ is even, then $\phi(z,y)$ is of the form
    \begin{eqnarray*}
     \forall x_{1,1}, \ldots, \forall x_{1,6} && ( \phi_1(z,y, x_{1,1},\ldots,x_{1,6}) \vee \\
     (\exists x_{2,1}\ldots,\exists x_{2,6} && ( \phi_2(x_{1,1},x_{1,2},x_{2,1},\ldots,x_{2,6}) \wedge \\
     (\forall x_{3,1}, \ldots, \forall x_{3,6} &&  (\phi_3(x_{2,1},x_{2,2},x_{3,1},\ldots,x_{3,6})\vee \\ 
     \vdots && \\
     (\exists x_{d,1}, \ldots, \exists x_{d,6} && 
     (\phi_d(x_{d-1,1},x_{d-1,2},x_{d,1},\ldots,x_{d,6})\wedge 
      ({\sf id}(x_{d,1})={\sf id}(x_{d,2})))))))))
    \end{eqnarray*}
where ${\sf id}(x_{d,1}), {\sf id}(x_{d,2}) \in [\lceil \log^{(d)} n \rceil]$ , and $\phi_1,\ldots,\phi_d$ are quantifier free formulas. 
    \item For any two nodes $z$ and $y$ in $T$, $T'\models \phi(z,y)$ if and only if ${\sf id}(z)={\sf id}(y)$
    \item The number of label predicates used in $\phi_{\forall}(z,y)$ is $4d$. 
    \item $V(T')=O(|V(T)|\cdot \ilog{n}{d})$
\end{enumerate}
\end{lemma}


\begin{lemma}[$\log^{(d+1)} n$-length FO identifier test]
\label{lem:ID-level-d-odd}
Let $d\geq 1$ be an odd number. 
Let $n$  be a positive integer and 
for any leaf node $x$, ${\sf id}(x) \in [n]$ (if ${\sf id}$ is defined on $x$). 
One can construct a tree $T'$ (which is super graph of $T$) and $(c_1,\ldots,c_d)$-FO formula $\phi$ on two free variables  
with the following specifications. 
\begin{enumerate}[(a)]
    \item $c_i=6$ for all $i\in [d-1]$ and $c_d= 6+2\log^{(d+1)}n$. 
    \item $\phi(z,y)$ is of the form $$\forall x_{1,1}, \ldots, \forall x_{1,6} \exists x_{2,1}\ldots,\exists x_{2,6}, \ldots,  \forall x_{d,1}\ldots, \forall x_{d,c_d} \;\; \phi_1(z,y)$$ where $\phi_1$ is a quantifier free formula.
    \item $|\phi|= O(d+\log^{(d+1)} n)$
    \item For any two nodes $z$ and $y$ in $T$, $T'\models \phi(z,y)$ if and only if ${\sf id}(z)={\sf id}(y)$. 
    \item The number of label predicates used in $\phi(z,y)$ is $4d+2$. 
    \item $V(T')=O(|V(T)|\cdot 
    \ilog{n}{d+1})$. 
\end{enumerate}
\end{lemma}

\begin{proof}
When $d=1$, Lemma~\ref{lem:L1-uni-IDtest} implies the lemma. Now cosnider the case when $d>1$.  
Apply Lemma~\ref{lem:ID-level-d} and obtain a tree $T_1$ and $(c_1,\ldots,c_d)$-FO+{\sf id} formula $\phi$ with specifications mentioned in the lemma. Since $d$ is an odd number, condition (b) in Lemma~\ref{lem:ID-level-d} implies that 
for any two leaf nodes $z,y$ in $T$,  $\phi(z,y)$ is of the form 


    \begin{eqnarray}
     \forall x_{1,1}, \ldots, \forall x_{1,6} && ( \phi_1(z,y, x_{1,1},\ldots,x_{1,6}) \vee \nonumber\\
     (\exists x_{2,1}\ldots,\exists x_{2,6} && ( \phi_2(x_{1,1},x_{1,2},x_{2,1},\ldots,x_{2,6}) \wedge \nonumber\\
     (\forall x_{3,1}, \ldots, \forall x_{3,6} &&  (\phi_3(x_{2,1},x_{2,2},x_{3,1},\ldots,x_{3,6})\vee \nonumber\\ 
     \vdots && \nonumber\\
     (\forall x_{d,1}, \ldots, \forall x_{d,6} && 
     (\phi_d(x_{d-1,1},x_{d-1,2},x_{d,1},\ldots,x_{d,6})\vee 
     \neg ({\sf id}(x_{d,1})={\sf id}(x_{d,2}))))))))) \label{eqn:oddD}
    \end{eqnarray}
    where ${\sf id}(x_{d,1}), {\sf id}(x_{d,2}) \in [\lceil \log^{(d)} n \rceil]$, and $\phi_1,\ldots,\phi_d$ are quantifier free formulas.

Now we want to replace 
$({\sf id}(x_{d,1})={\sf id}(x_{d,2}))$ with an FO formula. So, we apply Lemma~\ref{lem:L1-ex-IDtest}, 
and obtain a tree $T'$ and $(2\log^{(d+1)} n)$-FO formula $\phi_{\exists}$ on two free variables 
with the following specifications. 
\begin{enumerate}[(i)]
    \item For any two  nodes $x_1$ and $x_2$ in $T_1$, $T'\models \phi_{\exists}(x_1,x_2)$ if and only if ${\sf id}(x_1)={\sf id}(x_2)$. 
    \item All the quantifiers in $\phi_{\exists}(x_1,x_2)$ are existential quantifiers and $|\phi_{\exists}(x_1,x_2)|=O(\log^{(d+1)} n)$. 
    \item The number of label predicates used in $\phi_{\exists}(x_1,x_2)$ is two. 
    \item $V(T')=O(|V(T_1)|\cdot \log^{(d+1)}n )$
\end{enumerate}

Thus, we can replace $\neg ({\sf id}(x_{d,1})={\sf id}(x_{d,2})))$ in the formula in (\ref{eqn:oddD}) with $\neg \phi_{\exists}(x_{d,1},x_{d,2})$. 
 The formula $\neg \phi_{\exists}(x_{d,1},x_{d,2})$ is of the form $\neg \exists a_1 \ldots \exists a_{k} \phi_2$ where $\phi_2$ is quantifier free and $k=2\log^{(d+1)} n$. This implies that 
$$
\neg \phi_{\exists}(x_{d,1},x_{d,2}) \equiv 
\forall a_1 \ldots \forall a_{k} \neg \phi_2.
$$

Thus, the new formula $\phi$ obtained from (\ref{eqn:oddD}), by replacing 
$\neg ({\sf id}(x_{d,1})={\sf id}(x_{d,2})))$ with $\neg \phi_{\exists}(x_{d,1},x_{d,2})$, has the form specified in property (b) of the Lemma, after converting to prenex normal form using Proposition~\ref{prop:PNF:Rules}
The property (b) of Lemma~\ref{lem:ID-level-d} and item (ii) above implies that $|\phi|=O(d+\log^{(d+1)}n)$. Property 
(e) of Lemma~\ref{lem:ID-level-d}
and item (iii) above implies that the number of label predicates in $\phi$ is $4d+2$. Property 
(f) of Lemma~\ref{lem:ID-level-d}
and item (iv) above implies the condition (f) of the lemma. 
This completes the proof. 
\end{proof}

\begin{lemma}[$\log^{(d+1)} n$-length FO identifier test]
\label{lem:ID-level-d-even}
Let $d\geq 2$ be an even number. 
Let $n$  be a positive integer and 
for any leaf node $x$, ${\sf id}(x) \in [n]$ (if ${\sf id}$ is defined on $x$). 
One can construct a tree $T'$ (which is super graph of $T$) and $(c_1,\ldots,c_d)$-FO formula $\phi$ on two free variables  
with the following specifications. 
\begin{enumerate}[(a)]
    \item $c_i=6$ for all $i\in [d-1]$ and $c_d= 6+2\log^{(d+1)}n$. 
    \item $\phi(z,y)$ is of the form $$\forall x_{1,1}, \ldots, \forall x_{1,6} \exists x_{2,1}\ldots,\exists x_{2,6} \forall x_{3,1}\ldots, \forall x_{3,6} \ldots \exists x_{d,1}\ldots, \exists x_{d,c_d} \;\; \phi_1(z,y)$$ where $\phi_1$ is a quantifier free formula.
    \item $|\phi|= O(d+\log^{(d+1)} n)$.
    \item For any two leaf nodes $z$ and $y$ in $T$, $T'\models \phi(z,y)$ if and only if ${\sf id}(z)={\sf id}(y)$. 
    \item The number of label predicates used in $\phi(z,y)$ is $4d+2$. 
    \item $V(T')=O(|V(T)|\cdot 
    \ilog{n}{d+1} )$
\end{enumerate}
\end{lemma}

\begin{proof}
To prove this lemma we apply Lemma~\ref{lem:ID-level-d} followed by Lemma~\ref{lem:L1-uni-IDtest}. Recall that $d\geq 2$ is an even number. Apply Lemma~\ref{lem:ID-level-d} and obtain a tree $T_1$ and $(c_1,\ldots,c_d)$-FO+{\sf id} formula $\phi$ with specifications mentioned in the lemma. Since $d$ is an even number, condition (c) in Lemma~\ref{lem:ID-level-d} implies that 
for any two leaf nodes $z,y$ in $T$,  $\phi(z,y)$ is of the form 


    \begin{eqnarray}
     \forall x_{1,1}, \ldots, \forall x_{1,6} && ( \phi_1(z,y, x_{1,1},\ldots,x_{1,6}) \vee \nonumber \\
     (\exists x_{2,1}\ldots,\exists x_{2,6} && ( \phi_2(x_{1,1},x_{1,2},x_{2,1},\ldots,x_{2,6}) \wedge \nonumber\\
     (\forall x_{3,1}, \ldots, \forall x_{3,6} &&  (\phi_3(x_{2,1},x_{2,2},x_{3,1},\ldots,x_{3,6})\vee \nonumber\\ 
     \vdots && \nonumber\\
     (\exists x_{d,1}, \ldots, \exists x_{d,6} && 
     (\phi_d(x_{d-1,1},x_{d-1,2},x_{d,1},\ldots,x_{d,6})\wedge 
      ({\sf id}(x_{d,1})={\sf id}(x_{d,2}))))))))) \label{eqn:evenD}
    \end{eqnarray}
where ${\sf id}(x_{d,1}), {\sf id}(x_{d,2}) \in [\lceil \log^{(d)} n \rceil]$ , and $\phi_1,\ldots,\phi_d$ are quantifier free formulas. 

Now we want to replace 
$({\sf id}(x_{d,1})={\sf id}(x_{d,2}))$ with an FO formula.  So, we apply Lemma~\ref{lem:L1-ex-IDtest}, 
and obtain a tree $T'$ and $(2\log^{(d+1)} n)$-FO formula $\phi_{\exists}$ on two free variables 
with the following specifications. 
\begin{enumerate}[(i)]
    \item For any two leaf nodes $x_1$ and $x_2$ in $T_1$, $T'\models \phi_{\exists}(x_1,x_2)$ if and only if ${\sf id}(x_1)={\sf id}(x_2)$. 
    \item All the quantifiers in $\phi_{\exists}(x_1,x_2)$ are existential quantifiers and $|\phi_{\exists}(x_1,x_2)|=O(\log^{(d+1)} n)$. 
    \item The number of label predicates used in $\phi_{\exists}(x_1,x_2)$ is two. 
    \item $V(T')=O(|V(T_1)|\cdot \log^{(d+1)}n )$
\end{enumerate}

Thus, we can replace ${\sf id}(x_{d,1})={\sf id}(x_{d,2})$ in the formula in (\ref{eqn:evenD}) with $\phi_{\exists}(x_{d,1},x_{d,2})$. 
This new formula will be of the form specified in item (b) of the lemma, after converting to the prenex normal form. 
The arguments for the proof of other properties in the lemma are similar to the arguments in the proof of Lemma~\ref{lem:ID-level-d-odd}. This completes the proof.
\end{proof}
%

Now, we prove Theorem~\ref{thm:low-MSO}, when $s=0$. That is, the simplified version of Theorem~\ref{thm:low-MSO} in this case is the following.

\begin{theorem}
\label{thm:low-FO-Tree}
    Let $d,i,j \in \mathbb{N}$ such that $j<i \leq d$.
    Given an $n$-vertex graph $G$ and integers $k \geq 2$, $s \geq 0$, $\alpha \geq 1$ such that $n \leq \exp^{(i-1)}(\frac{k-9}{2})$, one can compute in $O(3^{n/\alpha})$ time a tree $T$ and an $S$-FO formula $\phi$ for $S = (k_1,\dots,k_d)$ satisfying the following.
    \begin{enumerate}[(i)]
        \item $G$ is a $3$-colorable if and only if $T$ satisfies $\phi$.
        \item $|V(T)| = O((3^{n/\alpha}+m)\ilog{n}{i-1})$ and $|\phi| = O(\alpha^2 +i+ \log^{(i-1)}n) $.
        \item $k_i = k$, $k_1 = O(\alpha)$, $k_{i'} = O(1)$ for all $i' \in \{2,\dots,i-1\}$, and $k_{i'} = 0$ for all $i' \in [d] \backslash [i]$.
        \item The number of label predicates used in $\phi$ is $
        4i+3$. 
    \end{enumerate}
\end{theorem}

\begin{proof}
Apply Lemma~\ref{lem:base-red}
and obtain a tree $T$ and an $(2\alpha+1,9)$-FO+${\sf id}$ formula $\psi$. Notice that Figure~\ref{fig:Base-tree} is an illustration of $T$. The FO+${
\sf id 
}$ formula $\psi$ is given below. 

\begin{eqnarray}
\psi  &\equiv&\exists r\in {\sf R} \exists x_1\in {\sf P} \; \exists x_2 \in {\sf P}\; \ldots \; \exists x_{\alpha}\in {\sf P} \; \exists p_1 \; \ldots \; \exists p_{\alpha}  \nonumber\\ 
 && \forall z \in {\sf EN}\; \forall z_1, z_2 \in {\sf EP}\; \forall v_1, v_2 \in {\sf N} \; \forall y_1, y_2 \in {\sf I} \; \forall c_1,c_2 \in {\sf Q}
 _1\cup {\sf Q}_2 \cup {\sf Q}_3 \nonumber\\
 && \psi_{{\sf evalid}}\wedge (\neg \psi_{{\sf uvalid}} \vee \neg\psi_{{\sf id}}\vee \psi_{{\sf color}}) \label{eqn:FO1}
\end{eqnarray}
All the label predicates in the formula is explained in the proof of Lemma~\ref{lem:base-red}. In the formula $\psi$, the subformulas $\psi_{{\sf evalid}}, \psi_{{\sf uvalid}}$, and $\psi_{{\sf color}}$ are quantifier free FO formulas. But $\psi_{{\sf id}}$ contains the function sysmbol ${\sf id
}$. Recall that 

$$\psi_{{\sf id}}\equiv {\sf id}(z_1)={\sf id}(y_1)\wedge {\sf id}(z_2)={\sf id}(y_2).$$

Now we explain how to get rid of the function ${\sf id}$. We have two cases, one is when $i=2$ and the other is when $i>2$. 

\smallskip
\noindent
{\bf Case 1: $i=2$.} In order to replace $\psi_{\sf id}$, we apply Lemma~\ref{lem:L1-ex-IDtest} 
and obtain a tree $T'$ (which is super graph of $T$) and $(2\log n)$-FO formula $\phi_{\exists}$ on two free variables 
with the following specifications. 
\begin{enumerate}[(a)]
    \item For any two leaf nodes $z$ and $y$ in $T$, $T'\models \phi_{\exists}(z,y)$ if and only if ${\sf id}(z)={\sf id}(y)$. 
    \item All the quantifiers in $\phi_{\exists}(z,y)$ are existential quantifiers and $|\phi_{\exists}(z,y)|=O(\log n)$. 
    \item The number of label predicates used in $\phi_{\exists}(z,y)$ is two. 
    \item $V(T')=O(|V(T)|\cdot \log n)$
\end{enumerate}

Now we substitute ${\sf id}(z_1)={\sf id}(y_1)\wedge {\sf id}(z_2)={\sf id}(y_2)$ with $\phi_{\exists}(z_1,y_1) \wedge \phi_{\exists}(z_2,y_2)$ and let $\psi'$ be the resulting formula. Equation (\ref{eqn:FO1}) and the definition of $\phi_{\exists}$ 
implies that $\psi'$ is a $(2\alpha+1,9+2\log n)$-FO formula. Property (i) of Lemma~\ref{lem:base-red} and item (a) above implies that 
\begin{itemize}
    \item[(i)] $G$ is a $3$-colorable if and only if $T'$ satisfies $\psi'$.
\end{itemize}
Property (ii) of Lemma~\ref{lem:base-red} and condition (d) above implies that 
\begin{itemize}
    \item[(ii)] $|V(T')|=O((3^{n/\alpha}+m)\log n)$ and $|\psi'|=O(\alpha^2+\log n)$. 
\end{itemize}

Property (iii) of Lemma~\ref{lem:base-red} and condition (c) above implies that 

\begin{itemize}
    \item[(iii)] Number of label predicates used is $11$.  
\end{itemize}
So, $T'$ and $\psi'$ are the required output of the theorem in this case.

\medskip
\noindent
{\bf Case 2: $i>2$. }
In this case, to replace $\psi_{\sf id}$, we apply Lemma~\ref{lem:ID-level-d-even} when $i$ is even and apply Lemma~\ref{lem:ID-level-d-odd} when $i$ is odd with $d=i-2$, and obtain a tree $T'$ (which is super graph of $T$) and a $(c_1,\ldots,c_{i-2})$-FO formula $\phi$ on two free variables  
with the following specifications. 
\begin{enumerate}[(a)]
    \item $c_q=6$ for all $q\in [i-3]$ and $c_{i-2}= 6+2\log^{(i-1)}n$. 
    \item $\phi(z,y)$ is of the form $$\forall x_{1,1}, \ldots, \forall x_{1,6} \exists x_{2,1}\ldots,\exists x_{2,6}, \ldots, Q x_{i-2,1}\ldots, Q x_{i-2,c_{i-2}} \;\; \phi_1(z,y)$$ where $\phi_1$ is a quantifier free formula, and $Q=\exists$ if $i$ is odd and $Q=\forall$ otherwise.
    \item $|\phi|= O(i+\log^{(i-1)} n)$. 
    \item For any two leaf nodes $z$ and $y$ in $T$, $T'\models \phi_{\forall}(z,y)$ if and only if ${\sf id}(z)={\sf id}(y)$. 
    \item The number of label predicates used in $\phi(z,y)$ is $4i-6$. 
    \item $V(T')=O(|V(T)|\cdot 
    \ilog{n}{i})$. 
\end{enumerate}

Now we substitute ${\sf id}(z_1)={\sf id}(y_1)\wedge {\sf id}(z_2)={\sf id}(y_2)$ with $\phi(z_1,y_1) \wedge \phi(z_2,y_2)$ and let $\psi'$ be the resulting formula. Equation (\ref{eqn:FO1}) and the definition of $\phi$ 
implies that $\psi'$ is a $(2\alpha+1,9, k_3,\ldots,k_{i-1},6+2\log^{(i-1)} n)$-FO formula. Property (i) of Lemma~\ref{lem:base-red} and item (d) above implies that 
\begin{itemize}
    \item[(i)] $G$ is a $3$-colorable if and only if $T'$ satisfies $\psi'$.
\end{itemize}
Property (ii) of Lemma~\ref{lem:base-red} and items (c) and (f) above implies that 
\begin{itemize}
    \item[(ii)] $|V(T')|=O((3^{n/\alpha}+m)\ilog{n}{i})$ and $|\psi'|=O(\alpha^2+i+\log^{(i-1)} n)$. 
\end{itemize}

Property (iii) of Lemma~\ref{lem:base-red} and item (e) above implies that 

\begin{itemize}
    \item[(iii)] Number of label predicates used in $\psi'$ is $4i+3$.  
\end{itemize}
So, $T'$ and $\psi'$ are the required output of the theorem in this case. This completes the proof of the theorem.     
\end{proof}

Now, we prove Theorem~\ref{thm:low-MSO}, when $s>0$ and $i=2$.

\begin{theorem}
\label{thm:low-MSO:s>0:i=2}
    Let $d\in \mathbb{N}$.
    Given an $n$-vertex graph $G$ and integers $k \geq 2$, $s \geq 1$, $\alpha \geq 1$ such that $n \leq 2^{(\frac{k-14}{4}\cdot s)}$, one can compute in $O(3^{n/\alpha})$ time a tree $T$ and an $S$-MSO formula $\psi$ for $S = ((k_1,s_1),\dots,(k_d,s_d))$ satisfying the following.

    \begin{enumerate}[(i)]
        \item $G$ is a $3$-colorable if and only if $T$ satisfies $\psi$.
        \item $|V(T)| = O((3^{n/\alpha}+m)2^s\log k)$
        and $|\psi| = O(\alpha^2 +k+2^s) $.
        \item $k_2 = k$, $k_1 = O(\alpha)$, and $k_{i'} = 0$ for all $i' \in \{3,\ldots,d\}$.
        \item $s_j = s$ and $s_{j'} = 0$ for all $j' \in [d] \backslash \{j\}$.
        \item The number of label predicates used in $\phi$ is $
        12$.   
    \end{enumerate}
\end{theorem}

\begin{proof}
First we apply  Lemma~\ref{lem:base-red} and obtain a tree $T_1$ and a $(2\alpha
+1,9)$-FO+${\sf id}$ formula $\psi$. Notice that Figure~\ref{fig:Base-tree} is an illustration of $T_1$. The FO+${
\sf id 
}$ formula $\psi$ is given below.

\begin{eqnarray}
\psi  &\equiv&\exists r\in {\sf R} \exists x_1\in {\sf P} \; \exists x_2 \in {\sf P}\; \ldots \; \exists x_{\alpha}\in {\sf P} \; \exists p_1 \; \ldots \; \exists p_{\alpha}  \nonumber\\ 
 && \forall z \in {\sf EN}\; \forall z_1, z_2 \in {\sf EP}\; \forall v_1, v_2 \in {\sf N} \; \forall y_1, y_2 \in {\sf I} \; \forall c_1,c_2 \in {\sf Q}
 _1\cup {\sf Q}_2 \cup {\sf Q}_3 \nonumber\\
 && \psi_{{\sf evalid}}\wedge (\neg \psi_{{\sf uvalid}} \vee \neg\psi_{{\sf id}}\vee \psi_{{\sf color}}) \label{eqn:FO1:s>0:i=2}
\end{eqnarray}
All the label predicates in the formula is explained in the proof of Lemma~\ref{lem:base-red}. In the formula $\psi$, the subformulas $\psi_{{\sf evalid}}, \psi_{{\sf uvalid}}$, and $\psi_{{\sf color}}$ are quantifier free FO formulas. But $\psi_{{\sf id}}$ contains the function sysmbol ${\sf id
}$. Recall that

$$\psi_{{\sf id}}\equiv {\sf id}(z_1)={\sf id}(y_1)\wedge {\sf id}(z_2)={\sf id}(y_2).$$

Let $k'=\frac{k-14}{4}$.  Now we replace each leaf node $t\in {\sf EP}\cup {\sf I}$ with a path as explained below. We create a path $\pi_t= a_1,\ldots,a_{k'},t$ on $k' +1$ vertices. 
Now, we replace node $t$ with path $\pi_t$. 
Let $T_2$ be the tree obtained after this process. 
Let $b_1,\ldots,b_{k'}$ be the base $2^{s}$ representation of the number ${\sf id}(t)$.
Now, let us define ${\sf id}(a_i)=b_i$ for all $i$. 
Moreover, we create a unary label predicate ${\sf P}_1$ and add all the vertices $a_1,\ldots,a_{k'}$ to ${\sf P}_1$. Since $2^{k' \cdot s} \geq n$, each number in $[n]$ can be uniquely represented as above. 
Now, ${\sf id}(z)={\sf id}(y)$ can be written as 

\begin{eqnarray}
\phi_{\sf yes}(z,y)\equiv \exists a_1,a_1'\in {\sf P}_1 \ldots \exists a_{k},a_{k'}' \in {\sf P}_1
&&
\left( \bigwedge_{i\in [k] }
 {\sf id}(a_i) = {\sf id}(a'_i) \right) \nonumber \\
&& \wedge  {\sf path}(a_1,\ldots,a_{k'},z)
\wedge
{\sf path}(a'_1,\ldots,a'_{k'},y)  \label{eqn:exuni1}
\end{eqnarray}

where, ${\sf path}(w_1,\ldots,w_q) \equiv (\bigwedge_{i\in [q-1]} {\sf adj}(w_i,w_{i+1}))$. Now, $\neg \psi_{{\sf id}}$ can be expressed as

\begin{eqnarray}
\neg \psi_{\sf id}\equiv \forall a_1,a_1'\in {\sf P}_1 \ldots \forall a_{k'},a_{k'}' \in {\sf P}_1
&&
\left( \bigvee_{i\in [k'] }
 {\sf id}(a_i) \neq {\sf id}(a'_i) \right) \nonumber \\
&& \vee  \neg {\sf path}(a_1,\ldots,a_{k'},z_1)
\vee 
\neg {\sf path}(a'_1,\ldots,a'_{k'},y_1) \vee \nonumber \\
\forall d_1,d_1'\in {\sf P}_1 \ldots \forall d_{k'},d_{k'}' \in {\sf P}_1
&&
\left( \bigvee_{i\in [k'] }
 {\sf id}(d_i) \neq {\sf id}(d'_i) \right) \nonumber \\
&& \vee  \neg {\sf path}(d_1,\ldots,d_{k'},z_2)
\vee 
\neg {\sf path}(d'_1,\ldots,d'_{k'},y_2) \nonumber 
\end{eqnarray}
The above formula can be rewritten as 
\begin{eqnarray}
\neg \psi_{\sf id} &\equiv&
 \forall a_1,a_1'\in {\sf P}_1 \ldots \forall a_{k'},a_{k'}' \in {\sf P}_1   \forall d_1,d_1'\in {\sf P}_1 \ldots \forall d_{k'},d_{k'}' \in {\sf P}_1 \nonumber \\
&& \psi_{\sf npath} \vee \bigvee_{i\in [k'] }
 {\sf id}(a_i) \neq {\sf id}(a'_i) \vee 
\bigvee_{i\in [k'] }
 {\sf id}(d_i) \neq {\sf id}(d'_i), \nonumber
\end{eqnarray}

where $\psi_{\sf npath}$ is the formula $$\neg {\sf path}(a_1,\ldots,a_{k'},z_1)
\vee 
\neg {\sf path}(a'_1,\ldots,a'_{k'},y_1)
\vee  \neg {\sf path}(d_1,\ldots,d_{k'},z_2)
\vee 
\neg {\sf path}(d'_1,\ldots,d'_{k'},y_2)
$$
Now, (\ref{eqn:FO1:s>0:i=2}) can be written as 

\begin{eqnarray}
\psi  &\equiv& \exists r\in {\sf R}\; \exists x_1\in {\sf P} \; \exists x_2 \in {\sf P}\; \ldots \; \exists x_{\alpha}\in {\sf P}\; \exists p_1 \ldots \exists p_{\alpha}  \nonumber\\ 
 && \forall z \in {\sf EN}\; \forall z_1, z_2 \in {\sf EP}\; \forall v_1, v_2 \in {\sf N} \; \forall y_1, y_2 \in {\sf I} \; \forall c_1,c_2 \in {\sf Q}
 _1\cup {\sf Q}_2 \cup {\sf Q}_3 \nonumber\\
 &&  \forall a_1,a_1'\in {\sf P}_1 \ldots \forall a_{k'},a_{k'}' \in {\sf P}_1   \forall d_1,d_1'\in {\sf P}_1 \ldots \forall d_{k'},d_{k'}' \in {\sf P}_1 \nonumber\\
 && \psi_{{\sf evalid}} \wedge \left( (\neg \psi_{{\sf uvalid}} \vee \psi_{{\sf npath}}\vee \psi_{{\sf color}}) \vee \bigvee_{i\in [k'] }
 {\sf id}(a_i) \neq {\sf id}(a'_i) \vee 
\bigvee_{i\in [k'] }
 {\sf id}(d_i) \neq {\sf id}(d'_i)\right)
 \label{eqn:FO2:s>0:i=2}
\end{eqnarray}

The length of formula in (\ref{eqn:FO2:s>0:i=2}) is $O(\alpha^2+k)$. 
In the above expression each ${\sf id}(a_i)$ is a number in $\{0,\ldots,2^s-1\}$. 
Now we explain how to encode ${\sf id}(a_i) = {\sf id}(a'_i)$. Notice that we can use $s$ set variables $W_1,\ldots,W_s$ in the fomula. Also, we know the value of ${\sf id}(a_i)$ and ${\sf id}(a_{i'})$. Let  $b_1,\ldots,b_s$ be the binary representation of ${\sf id}(a_i)$ and 
$b'_1,\ldots,b'_s$ be the binary representation of ${\sf id}(a'_i)$. 
Suppose we are able to force the set variables in such a way that for all $j\in [s]$, $a_i\in W_j$ if and only if $b_j=1$ and $a_i'\in W_j$ if and only if $b'_j=1$. Under that condition, ${\sf id}(a_i) \neq {\sf id}(a'_i)$ can be encoded as 

\begin{equation}
\label{eqn:setQIDcore}    
\neg \bigwedge_{j\in [s]} (a_i\in W_j) \Leftrightarrow (a'_i\in W_j)
\end{equation}

Now to satisfy the condition mentioned above, we create a formula $\psi_{\sf set}(W_1,\ldots W_s)$ such that $\psi_{\sf set}(W_1,\ldots W_s)$ is true if and only if the following is true. For any vertex $a\in {\sf P}_1$, let $b_1,\ldots,b_s$ be the binary representation of ${\sf id}(a)$. Then, the formula $\psi_{\sf set}(W_1,\ldots W_s)$ is true if and only if for all $a\in P_1$ and $j\in [s]$, $a\in W_j$ only when $b_j=1$. Once we have such a formula $\psi_{\sf set}(W_1,\ldots W_s)$, we can rewrite (\ref{eqn:FO2:s>0:i=2}) as follows.

\begin{eqnarray}
\psi  &\equiv& \exists W_1\ldots \exists W_s \exists r\in {\sf R} \exists x_1\in {\sf P}\; \exists x_2 \in {\sf P}\; \ldots \; \exists x_{\alpha}\in {\sf P} \; \exists p_1 \ldots \exists p_{\alpha}  \nonumber\\ 
 && \forall z \in {\sf EN}\; \forall z_1, z_2 \in {\sf EP}\; \forall v_1, v_2 \in {\sf N} \; \forall y_1, y_2 \in {\sf I} \; \forall c_1,c_2 \in {\sf Q}
 _1\cup {\sf Q}_2 \cup {\sf Q}_3 \nonumber\\
 &&  \forall a_1,a_1'\in {\sf P}_1 \ldots \forall a_{k'},a_{k'}' \in {\sf P}_1   \forall d_1,d_1'\in {\sf P}_1 \ldots \forall d_{k'},d_{k'}' \in {\sf P}_1 \nonumber\\
 && \psi_{\sf set}(W_1,\ldots, W_s) \wedge \psi_{{\sf evalid}} \wedge \Big( (\neg \psi_{{\sf uvalid}} \vee \neg\psi_{{\sf npath}}\vee \psi_{{\sf color}}) \vee \nonumber\\ 
 && \bigvee_{i\in [k'] }
 \left(\neg \bigwedge_{j\in [s]} (a_i\in W_j) \Leftrightarrow (a'_i\in W_j)\right)
 \vee 
\bigvee_{i\in [k'] }
\left(
\neg \bigwedge_{j\in [s]} (d_i\in W_j) \Leftrightarrow (d'_i\in W_j)
\right)\Big)
 \label{eqn:FO3:s>0:i=2}
\end{eqnarray}

Now, to define the formula $\psi_{\sf set}(W_1,\ldots,W_k)$, we create two label predicates ${\sf SP}_1$ and ${\sf SP}_2$. For each $a\in P_1$, we know that ${\sf id}(a)\in \{0,\ldots 2^s-1\}$. Let $c'={\sf id}(a)$. We replace $a$ with a path $\sigma_a=q_0,q_1,\ldots,q_{c'},a$. Let $T_3$ be the tree obtained after this process.  
We add $q_0$ to ${\sf SP}_1$ and for all $j\in [c']$, $q_j$ to ${\sf SP}_2$.  
Now, let us define ${\sf id}$ for the newly introduced vertices $q_0,q_1,\ldots,q_{c'}$. Let ${\sf id}(q_j)=j$ for all $j\in [c']$. Notice that ${\sf id}$ values of $q_0,\ldots,q_{c'},a$ are $0,1,2,\ldots,c'-1,c',c'$. To encode this, we will encode that 

\begin{itemize}
    \item[(i)] for all $u\in {\sf SP}_1$, ${\sf id}(u)=0$, 
    \item[(ii)] for all $u\in {\sf SP}_1\cup {\sf SP}_2$ and for all $v\in {\sf SP}_2$, ${\sf adj}(u,v)$ implies that for all $c\in [2^{s}-2] ({\sf id}(u)=c \Rightarrow {\sf id}(v)=c+1) \vee ({\sf id}(v)=c \Rightarrow {\sf id}(u)=c+1)$ \\
    \item[(iii)]  for all $u\in {\sf SP}_1\cup {\sf SP}_2$ and for all $v\in {\sf P}_1$, ${\sf adj}(u,v)$ implies that for all $c\in [2^{s}-1] ({\sf id}(u)=c \Rightarrow {\sf id}(v)=c)$
\end{itemize}

We encode (i) as follows. 

\begin{eqnarray*}
    \phi_{(i)}&\equiv& \forall u_1\in {\sf SP}_1, \bigwedge_{j\in [s]} \neg (u_1\in W_j)  
\end{eqnarray*}

To encode condition (ii), we need to encode that for a number $c\in [2^{s}-1]$, ${\sf id}(u)=c \Rightarrow {\sf id}(v)=c+1$. Fix $u$, and $c$. Let $d_1,\ldots,d_s$ be the binary representation of $c$.  
Let $D\subseteq [s]$
such that $i\in D$ if and only if $d_i=1$. Let


\begin{eqnarray*}
\phi_{\sf eq}(u,c)\equiv {\sf id}(u)=c &\equiv &  \bigwedge_{i\in D} u\in W_i \wedge \bigwedge_{j\in [s]\setminus D} \neg (u\in W_j) 
\end{eqnarray*}
Now, we can encode condition (ii) as

\begin{eqnarray*}
\phi_{(ii)}\equiv && \forall u_2\in {\sf SP}_1\cup {\sf SP}_2, \forall v_2\in  {\sf SP}_2 \\
&& {\sf adj}(u_2,v_2) \Rightarrow   \bigwedge_{c\in [2^s-1]} ((\phi_{\sf eq}(u_2,c)\Rightarrow \phi_{\sf eq}(v_2,c+1)) \vee (\phi_1(v_2,c)\Rightarrow \phi_1(u_2,c+1))) 
\end{eqnarray*}

Similarly, we can encode (iii) as

\begin{eqnarray*}
\phi_{(iii)}\equiv  \forall u_3\in {\sf SP}_1\cup {\sf SP}_2, \forall v_3\in  {\sf P}_1 &&  {\sf adj}(u_3,v_3) \Rightarrow   \bigwedge_{c\in [2^s-1]} (\phi_{\sf eq}(u_3,c)\Rightarrow \phi_{\sf eq}(v_3,c))  
\end{eqnarray*}

Thus, finally, we have $\psi_{\sf set}(W_1,\ldots,W_k)= \phi_{(i)} \wedge \phi_{(ii)} \wedge \phi_{(iii)}$, and its length is $O(2^s)$. Let $\phi_{\sf set}$ be the formula obtained from $\phi_{(i)} \wedge \phi_{(ii)} \wedge \phi_{(iii)}$ by removing the quantifiers. Then, 
\begin{eqnarray}
\label{eqn:psiset}
\psi_{\sf set}(W_1,\ldots,W_s) 
\equiv  \forall u_1\in {\sf SP}_1 \forall u_2,u_3\in {\sf SP}_1\cup {\sf SP}_2 \forall v_2 \in {\sf SP}_2  \forall v_3\in {\sf P}_1  \;\; \phi_{\sf set}
\end{eqnarray}

By substituting the above formula for $\psi_{\sf set}(W_1,\ldots,W_s)$ in 
(\ref{eqn:FO3:s>0:i=2}) we get an $((2\alpha+1,s),(4k'+14,0))$-MSO formula of length $O(\alpha^2+k+2^s)$.  
By the construction of the tree and formula $\psi$, property (i) of the theorem follows. Lemma~\ref{lem:base-red} introduces $9$ label predicates. We have also introduced $3$ additional label predicates ${\sf P}_1$, ${\sf SP}_1$ and ${\sf SP}_2$. This implies that the number of label predicates in $\psi$ is $12$. 
The number of vertices in $T_3$ is $O(V(T_1)2^s \log k)= O((3^{n/\alpha}+m)2^s\log k)$.
\end{proof}

Before proving Theorem~\ref{thm:low-MSO} for $s>0$ and $i>2$, let us prove an auxiliary lemma about encoding identifiers using an MSO formula.

\begin{lemma}
\label{lem:ID-level-d-FirstMSO}
Let $d\geq 1$ be an integer. 
Let $n$  be a positive integer and 
for any leaf node $x$, ${\sf id}(x) \in [n]$ (if ${\sf id}$ is defined on $x$). Let $k,s\geq 1$ be two integers such that $n\leq \exp^{(d+1)}(\frac{k-6}{2}\cdot s)$. 
One can construct a tree $T'$ (which is super graph of $T$) and 
two $((c_0,s_0),\ldots,(c_d,s_d))$-MSO formulas $\psi$ and $\psi'$ on two free variables  
with the following specifications. 
\begin{enumerate}[(a)]
    \item $c_i= 6$ for all $i\in \{0,1,\ldots,d-1\}$, $c_d=k$, $s_0=s$ and $s_i=0$ for all $i\in [d]$
    \item All the set variables are quantified with existential quantifier 
    in $\psi$ and all the set variables are quantified with universal quantifier in $\psi'$. 
    \item For any two nodes $z$ and $y$ in $T$, $T'\models \psi(z,y)$ if and only if ${\sf id}(z)={\sf id}(y)$. 
       \item For any two nodes $z$ and $y$ in $T$, $T'\models \psi'(z,y)$ if and only if ${\sf id}(z)={\sf id}(y)$.
        \item $|\psi|, |\psi'|\in O(d+k+2^s)$ 
    \item The number of label predicates used in $\psi$ and $\psi'$ are $4d+2$. 
    \item $V(T')=O(|V(T)|\cdot 2^s k \cdot 
    \ilog{n}{d+1})$. 
\end{enumerate}
\end{lemma}

\begin{proof}
We apply Lemma~\ref{lem:ID-level-d} and we get a tree $T_1$ which is a super graph of $T$ and a $(c_1,\ldots,c_d)$-FO+{\sf id} formula $\phi$ with the following properties.

\begin{enumerate}[(a)]
    \item $c_i=6$ for all $i\in [d]$. 
    \item If $d$ is odd, then $\phi(z,y)$ is of the form
    \begin{eqnarray*}
     \forall x_{1,1}, \ldots, \forall x_{1,6} && ( \phi_1(z,y, x_{1,1},\ldots,x_{1,6}) \vee \\
     (\exists x_{2,1}\ldots,\exists x_{2,6} && ( \phi_2(x_{1,1},x_{1,2},x_{2,1},\ldots,x_{2,6}) \wedge \\
     (\forall x_{3,1}, \ldots, \forall x_{3,6} &&  (\phi_3(x_{2,1},x_{2,2},x_{3,1},\ldots,x_{3,6})\vee \\ 
     \vdots && \\
     (\forall x_{d,1}, \ldots, \forall x_{d,6} && 
     (\phi_d(x_{d-1,1},x_{d-1,2},x_{d,1},\ldots,x_{d,6})\vee 
     \neg ({\sf id}(x_{d,1})={\sf id}(x_{d,2})))))))))
    \end{eqnarray*}
    where ${\sf id}(x_{d,1}), {\sf id}(x_{d,2}) \in [\lceil \log^{(d)} n \rceil],$ and $\phi_1,\ldots,\phi_{d}$ are quantifier free formulas. 
    \item If $d$ is even, then $\phi(z,y)$ is of the form
    \begin{eqnarray*}
     \forall x_{1,1}, \ldots, \forall x_{1,6} && ( \phi_1(z,y, x_{1,1},\ldots,x_{1,6}) \vee \\
     (\exists x_{2,1}\ldots,\exists x_{2,6} && ( \phi_2(x_{1,1},x_{1,2},x_{2,1},\ldots,x_{2,6}) \wedge \\
     (\forall x_{3,1}, \ldots, \forall x_{3,6} &&  (\phi_3(x_{2,1},x_{2,2},x_{3,1},\ldots,x_{3,6})\vee \\ 
     \vdots && \\
     (\exists x_{d,1}, \ldots, \exists x_{d,6} && 
     (\phi_d(x_{d-1,1},x_{d-1,2},x_{d,1},\ldots,x_{d,6})\wedge 
      ({\sf id}(x_{d,1})={\sf id}(x_{d,2})))))))))
    \end{eqnarray*}
where ${\sf id}(x_{d,1}), {\sf id}(x_{d,2}) \in [\lceil \log^{(d)} n \rceil]$ , and $\phi_1,\ldots,\phi_{d}$ are quantifier free formulas. 
    \item For any two leaf nodes $z$ and $y$ in $T$, $T_1\models \phi(z,y)$ if and only if ${\sf id}(z)={\sf id}(y)$
    \item The number of label predicates used in $\phi(z,y)$ is $4d$. 
    \item $V(T_1)=O(|V(T)|\cdot \ilog{n}{d})$
\end{enumerate}

Now, to encode ${\sf id}(x_{d,1})={\sf id}(x_{d,2})$ we use set variables like in Theorem~\ref{thm:low-MSO:s>0:i=2}. Let $k'=\frac{k-6}{2}$.  Now we replace each leaf node $t$ in $T_1$ (where ${\sf id}(t)$ is defined) with a path as explained below. We create a path $\pi_t= a_1,\ldots,a_{k'},t$ on $k' +1$ vertices. 
Now, we replace node $t$ with path $\pi_t$. 
Let $T_2$ be the tree obtained after this process. 
Let $b_1,\ldots,b_{k'}$ be the base $2^{s}$ representation of the number ${\sf id}(t)$.
Now, let us define ${\sf id}(a_i)=b_i$ for all $i$. 
Moreover, we create a unary label predicate ${\sf P}_1$ and add all the vertices $a_1,\ldots,a_{k'}$ to ${\sf P}_1$. Since $2^{k' \cdot s} \geq \log^{(d)} n$, each number in $[\log^{(d)} n]$ can be uniquely represented as above. 
Now, ${\sf id}(x_{d,1})={\sf id}(x_{d,2})$ can be written as 

\begin{eqnarray}
 \exists a_1,a_1'\in {\sf P}_1 \ldots \exists a_{k},a_{k'}' \in {\sf P}_1
\left( \bigwedge_{i\in [k] }
 {\sf id}(a_i) = {\sf id}(a'_i) \right) 
 \wedge  {\sf path}(a_1,\ldots,a_{k'},x_{d,1})
\wedge
{\sf path}(a'_1,\ldots,a'_{k'},x_{d,2}) \nonumber 
\end{eqnarray}

where, ${\sf path}(w_1,\ldots,w_q) \equiv (\bigwedge_{i\in [q-1]} {\sf adj}(w_i,w_{i+1}))$. As before, we want to encode the identifiers using these set variables. Now,  we introduce $s$ set variables $W_1,\ldots,W_s$. 
 Let  $b_1,\ldots,b_s$ be the binary representation of ${\sf id}(a_i)$ and 
$b'_1,\ldots,b'_s$ be the binary representation of ${\sf id}(a_i')$. 
Suppose we are able to force the set variables in such a way that for all $j\in [s]$, $a_i\in W_j$ if and only if $b_j=1$ and $x_{d,2}\in W_j$ if and only if $b'_j=1$. Under that condition, ${\sf id}(x_{a_i}) = {\sf id}({a_i'})$ can be encoded as 

\begin{equation*}
 \bigwedge_{j\in [s]} (a_i\in W_j) \Leftrightarrow (a'_i\in W_j)
\end{equation*}

In other words, we define a formula $\phi_{\sf yes}(x_{d,1},x_{d,2},W_1,\ldots,W_s)$ on two free vertex variables and $s$ free set variables as follows.  

\begin{eqnarray}
\phi_{\sf yes}(x_{d,1},x_{d,2},W_1,\ldots,W_s)&\equiv&    
 \exists a_1,a_1'\in {\sf P}_1 \ldots \exists a_{k},a_{k'}' \in {\sf P}_1 \nonumber \\ 
 && 
   {\sf path}(a_1,\ldots,a_{k'},x_{d,1})
\wedge
{\sf path}(a'_1,\ldots,a'_{k'},x_{d,2}) \wedge \nonumber\\
&&
\left( \bigwedge_{i\in [k] }
  \bigwedge_{j\in [s]} (a_i\in W_j) \Leftrightarrow (a'_i\in W_j) \right) 
\label{eqn:exuni1leveld2}
\end{eqnarray}

Assuming the set variables are properly selected, $\phi_{\sf yes}(x_{d,1},x_{d,2},W_1,\ldots,W_s)$ is true if and only if ${\sf id}(x_{d,1})={\sf id}(x_{d,2})$. Now substitute (\ref{eqn:exuni1leveld2}) in the formula $\phi(z,y)$  (see items (b) and (c)) and let the resulting formula be $\psi_1(z,y,W_1,\ldots,W_s)$. Notice that $\psi_1(z,y,W_1,\ldots,W_s)$ is a $(c_1,\ldots,c_{d-1},c_d+2k')$-MSO formula with $s$ free set variables and $2$ free vertex variables, where $c_i=6$ for all $i\in [d]$. Then, like in the case of Theorem~\ref{thm:low-MSO:s>0:i=2}, we add long paths to vertices for which {\sf id} is defined and define $\psi_{set}$ on free variables $W_1,\ldots,W_s$ and it looks like

\begin{eqnarray}
\label{eqn:psiset:aux}
\psi_{\sf set}(W_1,\ldots,W_s) 
\equiv  \forall u_1\in {\sf SP}_1 \forall u_2,u_3\in {\sf SP}_1\cup {\sf SP}_2 \forall v_2 \in {\sf SP}_2  \forall v_3\in {\sf P}_1  \;\; \phi_{\sf set}
\end{eqnarray}

Let $T'$ be the tree obtained after this process. 
Here, $\phi_{\sf set}$ is quantifier free formula and the only free variables in $\phi_{\sf set}$ are $W_1,\ldots,W_s,u_1,u_2,u_3$. 
Moreover, $\psi_{\sf set}(W_1,\ldots,W_s)$ is true if and only if for each vertex, its inclusion in these sets is according to the id of the vertex. Notice that $\psi_{\sf set}(W_1,\ldots,W_s)$ is a $(5)$-MSO formula. 

Now, the required formula $\psi(z,y)$ is defined as 
$$ \exists W_1,\ldots,\exists W_s \; 
\psi_{\sf set}(W_1,\ldots,W_s) \wedge  \psi_1(z,y,W_1,\ldots,W_s)
$$

Similarly, the required formula $\psi'$ is 

$$ \forall W_1,\ldots, \forall W_s \;\; 
\psi_{\sf set}(W_1,\ldots,W_s) \Rightarrow  \psi_1(z,y,W_1,\ldots,W_s)
$$
The definitions of $\psi_1, \psi_{\sf set}$ and $\phi$ implies that 
$\psi$ and $\psi'$ are $((c_0,s_0),\ldots,(c_d,s_d))$-MSO formulas. The size of the formula $\phi$ is $O(d)$. This and  (\ref{eqn:exuni1leveld2}) implies that $|\psi_1|=O(d+k)$. The length of $\psi_{\sf set}$ is $O(2^s)$. Therefore $|\psi|$ and $|\psi'|$ are upper bounded by $O(2^s+k+d)$. The fact that $V(T_1)=O(|V(T)|\cdot \ilog{n}{d})$, and the construction of the tree implies that $V(T')=O(|V(T)|\cdot 2^s k \cdot \ilog{n}{d})$
This completes the proof of the lemma. 
\end{proof}

Now, to prove Theorem~\ref{thm:low-MSO}, when $s>0$ and $i>2$, we use Lemma~\ref{lem:ID-level-d}.

\begin{theorem}
\label{thm:low-MSO:s>0:i>2}
    Let $d,i,j \in \mathbb{N}$ such that $j<i \leq d$ and $i>2$. Given an $n$-vertex graph $G$ and integers $k \geq 2$, $s \geq 1$, $\alpha \geq 1$ such that $n \leq \exp^{(i-1)}(\frac{k-6}{2}\cdot s)$, one can compute in $O(3^{n/\alpha})$ time a tree $T$ and an $S$-MSO formula $\psi$ for $S = ((k_1,s_1),\dots,(k_d,s_d))$ satisfying the following.

    \begin{enumerate}[(i)]
        \item $G$ is a $3$-colorable if and only if $T$ satisfies $\psi$.
        \item $V(T)=O((3^{n/\alpha}+m)\cdot 2^sk \cdot \ilog{n}{i-1})$ and $|\psi| = O(\alpha^2+k+2^s) $.
        \item $k_i = k$, $k_1 = O(\alpha)$, $k_{i'} = O(1)$ for all $i' \in \{2,\dots,i-1\}$, and $k_{i'} = 0$ for all $i' \in [d] \backslash [i]$.
        \item $s_j = s$ and $s_{j'} = 0$ for all $j' \in [d] \backslash \{j\}$.
        \item The number of label predicates used in $\psi$ is at most $
        4i+3$. 
    \end{enumerate}
\end{theorem}

\begin{proof}
Like before, we start applying  Lemma~\ref{lem:base-red}
and obtain a tree $T_1$ and an $(2\alpha+1,9)$-FO+${\sf id}$ formula $\psi$. 
The FO+${\sf id }$ formula $\psi$ is given below.

\begin{eqnarray}
\psi  &\equiv&\exists r\in {\sf R} \exists x_1\in {\sf P} \; \exists x_2 \in {\sf P}\; \ldots \; \exists x_{\alpha}\in {\sf P} \; \exists p_1 \; \ldots \; \exists p_{\alpha}  \nonumber\\ 
 && \forall z \in {\sf EN}\; \forall z_1, z_2 \in {\sf EP}\; \forall v_1, v_2 \in {\sf N} \; \forall y_1, y_2 \in {\sf I} \; \forall c_1,c_2 \in {\sf Q}
 _1\cup {\sf Q}_2 \cup {\sf Q}_3 \nonumber\\
 && \psi_{{\sf evalid}}\wedge (\neg \psi_{{\sf uvalid}} \vee \neg\psi_{{\sf id}}\vee \psi_{{\sf color}}) \label{eqn:FO1:s>0:i>2}
\end{eqnarray}
All the label predicates in the formula is explained in the proof of Lemma~\ref{lem:base-red}. In the formula $\psi$, the subformulas $\psi_{{\sf evalid}}, \psi_{{\sf uvalid}}$, and $\psi_{{\sf color}}$ are quantifier free FO formulas. But $\psi_{{\sf id}}$ contains the function sysmbol ${\sf id
}$. Recall that

$$\psi_{{\sf id}}\equiv {\sf id}(z_1)={\sf id}(y_1)\wedge {\sf id}(z_2)={\sf id}(y_2).$$

We know that  $i>2$ and ${\sf id}(z_1),{\sf id}(y_1), {\sf id}(z_2), {\sf id}(y_2) \in [n]$. 

\medskip
\noindent 
{\bf Case 1: $j>1$. }
We apply Lemma~\ref{lem:ID-level-d} where $d=j-2$ and we get a tree $T_2$ which is a super graph of $T_1$, and $(c_3,\ldots,c_j)$-FO+{\sf id} formula $\phi$ with the following properties. Since we want to substitute for $\neg\psi_{\sf id}$ in (\ref{eqn:FO1:s>0:i>2}), we negated the formulas in (b) and (c).

\begin{enumerate}[(a)]
    \item $c_q=6$ for all $q\in \{3,\ldots,j\}$. 
    \item If $j-2$ is odd, then $\neg \phi(z,y)$ is of the form
    \begin{eqnarray*}
     \exists x_{1,1}, \ldots, \exists x_{1,6} && ( \phi_1(z,y, x_{1,1},\ldots,x_{1,6}) \wedge \\
     (\forall x_{2,1}\ldots,\forall x_{2,6} && ( \phi_2(x_{1,1},x_{1,2},x_{2,1},\ldots,x_{2,6}) \vee \\
     (\exists x_{3,1}, \ldots, \exists x_{3,6} &&  (\phi_3(x_{2,1},x_{2,2},x_{3,1},\ldots,x_{3,6})\wedge \\ 
     \vdots && \\
     (\exists x_{j-2,1}, \ldots, \exists x_{j-2,6} && 
     (\phi_{j-2}(x_{j-3,1},x_{j-3,2},x_{j-2,1},\ldots,x_{j-2,6})\wedge 
     ({\sf id}(x_{j-2,1})={\sf id}(x_{j-2,2})))))))))
    \end{eqnarray*}
    where ${\sf id}(x_{j-2,1}), {\sf id}(x_{j-2,2}) \in [\lceil \log^{(j-2)} n \rceil],$ and $\phi_1,\ldots,\phi_{j-2}$ are quantifier free formulas. 
    \item If $j-2$ is even, then $\neg \phi(z,y)$ is of the form
    \begin{eqnarray*}
     \exists x_{1,1}, \ldots, \exists x_{1,6} && ( \phi_1(z,y, x_{1,1},\ldots,x_{1,6}) \wedge \\
     (\forall x_{2,1}\ldots,\forall x_{2,6} && ( \phi_2(x_{1,1},x_{1,2},x_{2,1},\ldots,x_{2,6}) \vee \\
     (\exists x_{3,1}, \ldots, \exists x_{3,6} &&  (\phi_3(x_{2,1},x_{2,2},x_{3,1},\ldots,x_{3,6})\wedge \\ 
     \vdots && \\
     (\forall x_{j-2,1}, \ldots, \forall x_{j-2,6} && 
     (\phi_{j-2}(x_{j-3,1},x_{j-3,2},x_{j-2,1},\ldots,x_{j-2,6})\vee 
      ({\sf id}(x_{j-2,1})\neq {\sf id}(x_{j-2,2})))))))))
    \end{eqnarray*}
where ${\sf id}(x_{j-2,1}), {\sf id}(x_{j-2,2}) \in [\lceil \log^{(j-2)} n \rceil]$, and $\phi_1,\ldots,\phi_{j-2}$ are quantifier free formulas. 
    \item For any two f nodes $z$ and $y$ in $T_1$, $T_2\models \neg \phi(z,y)$ if and only if ${\sf id}(z)\neq {\sf id}(y)$. 
    \item The number of label predicates used in $\phi(z,y)$ is $4j-8$. 
    \item $V(T_2)=O(|V(T_1)|\cdot \ilog{n}{j-2})$
\end{enumerate}

The formula $\neg \phi (z,y)$ contains exactly one subformula that uses the function sysmbol ${\sf id}$ and it is of the form  ${\sf id}(x_{j-2,1})\neq {\sf id}(x_{j-2,2})$ when $j$ is even 
and ${\sf id}(x_{j-2,1})= {\sf id}(x_{j-2,2})$ when $j$ is odd. Notice that in the base case when $j=2$, 
$\neg \phi (z,y)\equiv {\sf id}(z)\neq {\sf id}(y)$. Each id is a non-negative integer which is at most $\lceil \log^{(j-2)} n \rceil$. 
Since, 
$n \leq \exp^{(i-1)}(\frac{k-6}{2}\cdot s)$, 
we have 

$$\log^{j-2} n \leq \exp^{((i-1)-(j-2))}(\frac{k-6}{2}\cdot s)
=\exp^{(i-j+1)}(\frac{k-6}{2}\cdot s)
$$

Now to replace 
${\sf id}(x_{j-2,1})= {\sf id}(x_{j-2,2})$, we apply Lemma~\ref{lem:ID-level-d-FirstMSO} with 
$d=i-j$ and the id values are at most $\lceil \log^{(j-2)} n \rceil$. Thus, by Lemma~\ref{lem:ID-level-d-FirstMSO}, we get a tree $T_3$ (which is a super graph of $T_2$) and  $((c'_j,s'_j),\ldots,(c'_{i},s'_i))$-MSO formulas $\psi_{\sf sec}$ on two free variables  
with the following specifications. 
\begin{enumerate}[(i)]
    \item $c_q= 6$ for all $q\in \{j,\ldots,i-1\}$, $c_{i}=k$, $s'_j=s$ and $s'_q=0$ for all $q\in \{j+1,\ldots,i\}$. 
    \item All the set variables are quantified with existential quantifier 
    in $\psi_{\sf sec}$.
    \item For any two nodes $a$ and $b$ in $T_2$, $T_3\models \psi_{\sf sec}(a,b)$ if and only if ${\sf id}(a)={\sf id}(b)$. 
        \item $|\psi_{\sf sec}|\in O(i-j+k+2^s)$ 
    \item The number of label predicates used in $\psi_{\sf sec}$ and $\psi'_{\sf sec}$ are $4(i-j)+2$. 
    \item $V(T_3)=O(|V(T_2)|\cdot 2^s k \cdot 
    \ilog{\log^{j-2} n}{i-j+1})$. 
\end{enumerate}

We substitute $\psi_{\sf sec}(x_{j-2,1}, x_{j-2,2})$ for ${\sf id}(x_{j-2,1})= {\sf id}(x_{j-2,2}$ in $\neg \phi(z,y)$. In either case (i.e., when $j$ is odd or even), the resulting formula (let us call it $\neg \phi_{\sf sub}(z,y)$) is a 
$((c_3,s_3),\ldots,(c_{j-1},s_{j-1}),(c_j+c_j',s_j'),(c_{j+1}',s_{j+1'}))$-MSO formula. 
Finally, we substitute $\neg \psi_{\sf id}$ in (\ref{eqn:FO1:s>0:i>2}) with $\neg \phi_{\sf sub}(z_1,y_1)\vee \neg \phi_{\sf sub}(z_2,y_2)$ to obtain the required formula $\psi$. The required tree in the theorem is $T_3$.   

By Lemma~\ref{lem:base-red}, we know that $|V(T_1)| = O(3^{n/\alpha}+m)$. Since $V(T_2)=O(|V(T_1)|\cdot \ilog{n}{j-2})$ (see item (f) above), $V(T_2)=O((3^{n/\alpha}+m)\cdot \ilog{n}{j-2})$. So, by item (vi) above, $V(T_3)=O((3^{n/\alpha}+m)\cdot 2^sk \cdot \ilog{n}{i-1})$. This proved property (ii) of the theorem.

The length of the formula is (\ref{eqn:FO1:s>0:i>2}) is  $O(\alpha^2)$. The size of the formulas mentioned in items (b) and (c) above, are $O(j)$. The length of the formula $\psi_{\sf sec}$ is $O(i-j+k+2^s)$. Thus, the length of the final formula $\psi$ is $O(\alpha^2+i+k+2^s)=O(\alpha^2+k+2^s)$.

The number of label predicates used in (\ref{eqn:FO1:s>0:i>2}) is $9$. This, along with items $(e)$ and (v) implies that the total number of label predicates used in the final formula is $4i+3$.   

\medskip
\noindent
{\bf Case 2: $j=1$}. In this case we apply Lemma~\ref{lem:ID-level-d} where $d=i-2$. Then, use set variables to encode id. This part of the proof follows arguments similar to Theorem~\ref{thm:low-MSO:s>0:i=2}. This completes the proof of the theorem. 
\end{proof}

Theorems~\ref{thm:low-FO-Tree}, \ref{thm:low-MSO:s>0:i=2}, and \ref{thm:low-MSO:s>0:i>2} imply Theorem~\ref{thm:low-MSO}.
Now we give a proof sketch of Theorem~\ref{thm:low-MSO-TW1}.

\begin{proof}[Proof sketch of Theorem~\ref{thm:low-MSO-TW1}]
Let $s=\min\{k,t\}$. We apply Theorem~\ref{thm:low-MSO} with $k=k'$ and  get a tree $T_1$ and an $S$-MSO formula $\phi$ for $S = ((k_1,s_1),\dots,(k_d,s_d))$ in time $O(3^{n/\alpha})$, 
with the following properties. 
    \begin{enumerate}[(1)]
        \item $G$ is a $3$-colorable if and only if $T$ satisfies $\phi$.
        \item $|V(T)| = O(3^{n/\alpha})$ and 
        $|\phi| = O(2^{s}+k+\alpha^2)$. 
        \item $k_i = k$, $k_1 = O(\alpha)$, $k_{i'} = O(1)$ for all $i' \in \{2,\dots,i-1\}$, and $k_{i'} = 0$ for all $i' \in [d] \backslash [i]$.
        \item $s_j = s$ and $s_{j'} = 0$ for all $j' \in [d] \backslash \{j\}$.
        \item The number of label predicates used in $\phi$ is $O(1)$
    \end{enumerate}
Recall that, in the proof of  Theorem~\ref{thm:low-MSO}, the set variables $W_1,\ldots,W_s$ are used to encode id of nodes where each id is an integer in $\{0,\ldots,2^s-1\}$. Now instead of using set variables, we add $s$ new nodes $\{t_1,\ldots,t_s\}$ to $T$ and all those verices are labeled with ${\sf L}$. Recall the role of set variables in $\phi$. 
For two node $u$ and $v$ in $T$, ${\sf id}(u)={\sf id}(v)$ is encoded by 

$$\bigwedge_{i'\in [s]} (u\in W_{i'}) \Leftrightarrow (v\in W_{i'}).$$

Now, to get rid of set variables, we add edges between $\{t_1,\ldots,t_s\}$ and nodes in $T$ for which ${\sf id}$ is defined as follows. Let $x$ be a node in $T$ and $b_1,\ldots,b_s$ be the binary representation of ${\sf id}(x)$. 
Then, $x$ is adjacent to $t_r$ if and only if $b_r=1$. This completes the construction of the tree $T$. Clearly, the treewidth of $T$ is at most $s$.

Now we explain the change in the formula. 
Let $u$ and $v$ be two  nodes in $T$. 
Let $a_1,\ldots,a_s$ be the binary representation of ${\sf id}(u)$ and $b_1,\ldots,b_s$ be the binary representation of ${\sf id}(v)$. Then ${\sf id}(u)={\sf id}(v)$ is encoded by 

$$\bigwedge_{i'\in [s]} ({\sf adj}(u,t_{i'})) \Leftrightarrow ({\sf adj}(v,t_{i'}).$$

Now to covert the formula $\phi$ to an FO formula, we replace the set variable $W_r$ with vertex variable $w_r\in {\sf L}$ and subformulas of the form 
$(u\in W_{i'}) \Leftrightarrow (v\in W_{i'})$
with $({\sf adj}(u,t_{i'}) \Leftrightarrow ({\sf adj}(v,t_{i'})$. The resulting formula satisfies the properties mentioned in Theorem~\ref{thm:low-MSO-TW1}. Notice that we have introduced  
one new label predicate and hence the number of label predicates used in the new formula is $O(1)$ only. 
\end{proof}

Next, we prove Theorem~\ref{thm:low-MSO-TW22}. 

\begin{proof}[Proof of Theorem~\ref{thm:low-MSO-TW22}]
Apply Lemma~\ref{lem:base-red}
and obtain a tree $T$ and an $(2\alpha+1,9)$-FO+${\sf id}$ formula $\psi$. Notice that Figure~\ref{fig:Base-tree} is an illustration of $T$. The FO+${
\sf id 
}$ formula $\psi$ is given below. 

\begin{eqnarray}
\psi  &\equiv&\exists r\in {\sf R} \exists x_1\in {\sf P} \; \exists x_2 \in {\sf P}\; \ldots \; \exists x_{\alpha}\in {\sf P} \; \exists p_1 \; \ldots \; \exists p_{\alpha}  \nonumber\\ 
 && \forall z \in {\sf EN}\; \forall z_1, z_2 \in {\sf EP}\; \forall v_1, v_2 \in {\sf N} \; \forall y_1, y_2 \in {\sf I} \; \forall c_1,c_2 \in {\sf Q}
 _1\cup {\sf Q}_2 \cup {\sf Q}_3 \nonumber\\
 && \psi_{{\sf evalid}}\wedge (\neg \psi_{{\sf uvalid}} \vee \neg\psi_{{\sf id}}\vee \psi_{{\sf color}}) \label{eqn:MSO-TW22-1}
\end{eqnarray}

All the label predicates in the formula is explained in the proof of Lemma~\ref{lem:base-red}. In the formula $\psi$, the subformulas $\psi_{{\sf evalid}}, \psi_{{\sf uvalid}}$, and $\psi_{{\sf color}}$ are quantifier free FO formulas. But $\psi_{{\sf id}}$ contains the function sysmbol ${\sf id
}$. Recall that 

$$\psi_{{\sf id}}\equiv {\sf id}(z_1)={\sf id}(y_1)\wedge {\sf id}(z_2)={\sf id}(y_2).$$

Now we explain how to get rid of the function ${\sf id}$. Notice that ${\sf id}(y)\in [n]$ for any node $y$ and $i\geq 2$. We apply Lemma~\ref{lem:ID-level-d} where $d=i-2$ and we get a tree $T_1$ which is a super graph of $T$, and $(c_3,\ldots,c_i)$-FO+{\sf id} formula $\phi$ with the following properties. Since we want to substitute for $\neg\psi_{\sf id}$ in (\ref{eqn:MSO-TW22-1}), we negated the formulas in (b) and (c).

\begin{enumerate}[(a)]
    \item $c_q=6$ for all $q\in \{3,\ldots,i\}$. 
    \item If $i-2$ is odd, then $\neg \phi(z,y)$ is of the form
    \begin{eqnarray*}
     \exists x_{1,1}, \ldots, \exists x_{1,6} && ( \phi_1(z,y, x_{1,1},\ldots,x_{1,6}) \wedge \\
     (\forall x_{2,1}\ldots,\forall x_{2,6} && ( \phi_2(x_{1,1},x_{1,2},x_{2,1},\ldots,x_{2,6}) \vee \\
     (\exists x_{3,1}, \ldots, \exists x_{3,6} &&  (\phi_3(x_{2,1},x_{2,2},x_{3,1},\ldots,x_{3,6})\wedge \\ 
     \vdots && \\
     (\exists x_{i-2,1}, \ldots, \exists x_{i-2,6} && 
     (\phi_{i-2}(x_{i-3,1},x_{i-3,2},x_{i-2,1},\ldots,x_{i-2,6})\wedge 
     ({\sf id}(x_{i-2,1})={\sf id}(x_{i-2,2})))))))))
    \end{eqnarray*}
    where ${\sf id}(x_{i-2,1}), {\sf id}(x_{i-2,2}) \in [\lceil \log^{(i-2)} n \rceil],$ and $\phi_1,\ldots,\phi_{j-2}$ are quantifier free formulas. 
    \item If $i-2$ is even, then $\neg \phi(z,y)$ is of the form
    \begin{eqnarray*}
     \exists x_{1,1}, \ldots, \exists x_{1,6} && ( \phi_1(z,y, x_{1,1},\ldots,x_{1,6}) \wedge \\
     (\forall x_{2,1}\ldots,\forall x_{2,6} && ( \phi_2(x_{1,1},x_{1,2},x_{2,1},\ldots,x_{2,6}) \vee \\
     (\exists x_{3,1}, \ldots, \exists x_{3,6} &&  (\phi_3(x_{2,1},x_{2,2},x_{3,1},\ldots,x_{3,6})\wedge \\ 
     \vdots && \\
     (\forall x_{i-2,1}, \ldots, \forall x_{i-2,6} && 
     (\phi_{i-2}(x_{i-3,1},x_{i-3,2},x_{i-2,1},\ldots,x_{i-2,6})\vee 
      ({\sf id}(x_{i-2,1})\neq {\sf id}(x_{i-2,2})))))))))
    \end{eqnarray*}
where ${\sf id}(x_{i-2,1}), {\sf id}(x_{i-2,2}) \in [\lceil \log^{(i-2)} n \rceil]$, and $\phi_1,\ldots,\phi_{j-2}$ are quantifier free formulas. 
    \item For any two nodes $z$ and $y$ in $T$, $T_1\models \neg \phi(z,y)$ if and only if ${\sf id}(z)\neq {\sf id}(y)$. 
    \item The number of label predicates used in $\phi(z,y)$ is $4i-8$. 
    \item $V(T_1)=O(|V(T)|\cdot \ilog{n}{i-2})$
\end{enumerate}

The formula $\neg \phi (z,y)$ contains exactly one subformula that uses the function sysmbol ${\sf id}$ and it is of the form  ${\sf id}(x_{i-2,1})\neq {\sf id}(x_{i-2,2})$ when $i$ is even 
and ${\sf id}(x_{i-2,1})= {\sf id}(x_{i-2,2})$ when $j$ is odd. Notice that in the base case when $i=2$, $T_1=T$ and  
$\neg \phi (z,y)\equiv {\sf id}(z)\neq {\sf id}(y)$. That is, when $i=2$, we do not apply Lemma~\ref{lem:ID-level-d}. 
Each id is a non-negative integer which is at most $\lceil \log^{(i-2)} n \rceil$. 
Since, 
$n \leq \exp^{(i-1)}(\frac{k-6}{3}\cdot \log t)$, 
we have 

$$\log^{i-2} n \leq 2^{(\frac{k-6}{3}\cdot \log t)}
=t^{\frac{k-6}{2}}. 
$$

Now, we want to encode ${\sf id}(x_{i-2,1})={\sf id}(x_{i-2,2})$ where these id values are at most $\lceil \log^{(i-2)} n \rceil$. 
Let $k'=\frac{k-6}{3}$.  Now we replace each node $w$ in $T_1$ (where ${\sf id}(w)$ is defined) with a path as explained below. We create a path $\pi_w= a_1,\ldots,a_{k'},w$ on $k' +1$ vertices. 
Now, we replace node $w$ with path $\pi_w$. 
Let $T_2$ be the tree obtained after this process. 
Let $b_1,\ldots,b_{k'}$ be the base $t$ representation of the number ${\sf id}(w)$.
Now, let us define ${\sf id}(a_j)=b_j \in \{0,\ldots,t-1\}$ for all $j\in [k']$. Moreover, we create a unary label predicate ${\sf P}_1$ and add all the vertices $a_1,\ldots,a_{k'}$ to ${\sf P}_1$. Since $t^{k'} \geq \log^{(i-2)} n$, each number in $[\log^{(i-1)} n]$ can be uniquely represented as above. 
Now, ${\sf id}(x_{i-2,1})={\sf id}(x_{i-2,2})$ can be written as

\begin{eqnarray}
 \exists a_1,a_1'\in {\sf P}_1 \ldots \exists a_{k'},a_{k'}' \in {\sf P}_1  &&  {\sf path}(a_1,\ldots,a_{k'},x_{i-2,1})
\wedge
{\sf path}(a'_1,\ldots,a'_{k'},x_{i-2,2}) \wedge \nonumber \\
&& \left( \bigwedge_{j\in [k'] } 
 {\sf id}(a_j) = {\sf id}(a'_j) \right) 
 \label{eqn:tlogkid}
\end{eqnarray}

where, ${\sf path}(w_1,\ldots,w_q) \equiv (\bigwedge_{i\in [q-1]} {\sf adj}(w_i,w_{i+1}))$ and for all $j\in [k']$, ${\sf id}(a_j)\in \{0,1,\ldots,k'\}$. Now to encode ${\sf id}(a_j)$ we add $t$ nodes $h_1,\ldots,h_t$ to $T_2$ and we add the following edges. For each node $a$ in $T_2$ such that  ${\sf id}(a_j)\in \{0,1,\ldots,k'\}$, we add an edge between $a$ and $h_{b+1}$ where $b={\sf id}(a)$. Let $G'$ be the graph constructed after this process. Since $T_2$ is a tree $\mathbf{tw}(G')=t$. We also a create a label predicate ${\sf P}_2$ and add $h_1,\ldots,h_t$ to ${\sf P}_2$. Now   ${\sf id}(a_j) = {\sf id}(a'_j)$ if and only if both $a_j$ and $a'_j$ are adjacent to the same vertex in ${\sf P}_2$. So we can modify (\ref{eqn:tlogkid}) and get an encoding 
for ${\sf id}(x_{i-2,1})={\sf id}(x_{i-2,2})$ as follows.

\begin{eqnarray}
 \exists a_1,a_1'\in {\sf P}_1 \ldots \exists a_{k'},a_{k'}' \in {\sf P}_1\; \exists q_1\in {\sf P}_2\; \ldots \; q_{k'}\in {\sf P}_2  &&  {\sf path}(a_1,\ldots,a_{k'},x_{i-2,1})
\wedge \nonumber \\
&& {\sf path}(a'_1,\ldots,a'_{k'},x_{i-2,2}) \wedge \nonumber \\
&& \left( \bigwedge_{j\in [k'] } 
 {\sf adj}(a_j,q_j) \wedge {\sf adj}(a'_j,q_j) \right) 
 \label{eqn:tlogkid2}
\end{eqnarray}

Thus, after substituting (\ref{eqn:tlogkid2}) in $\neg \phi (z,y)$, 
we get a $(c_3,\ldots,c_{i-1},k)$-FO formula. Let us call this formula 
$\neg \phi_{\sf sub} (z,y)$. Finally, we substitute $\neg \psi_{\sf id}$ in (\ref{eqn:MSO-TW22-1}) with $\neg \phi_{\sf sub}(z_1,y_1)\vee \neg \phi_{\sf sub}(z_2,y_2)$ to obtain the required formula $\psi$. The required graph in the theorem is $G'$.
By Lemma~\ref{lem:base-red}, we know that $|V(T)| = O(3^{n/\alpha}+m)$. Since $V(T_1)=O(|V(T)|\cdot \ilog{n}{i-2})$ (see item (f) above), $V(G')=O((3^{n/\alpha}+m)\cdot \ilog{n}{i-2} \cdot k \cdot t)$.  

The length of the formula is (\ref{eqn:MSO-TW22-1}) is  $O(\alpha^2)$. 
The size of the formulas mentioned in items (b) and (c) above, are $O(i)$. The length of the formula in (\ref{eqn:tlogkid2}) is $O(k)$. 
 Thus, the length of the final formula $\psi$ is $O(\alpha^2+i+k)=O(\alpha^2+k)$. This proved property (ii) of the theorem. The number of label predicates used in (\ref{eqn:MSO-TW22-1}) is $9$. Item (e) above implies that the number of label predicates in $\phi(z,y)$ is $4i-8$. We have also introduced two more predicates ${\sf P}_1$ and ${\sf P}_2$. This implies that the number of label predicates in $\psi$ is $4i+3=O(1)$. This completes the proof of the theorem. 
\end{proof}

To prove the lower bound results( Theorems~\ref{thm:intro-LW1},\ref{thm:intro-LW2},\ref{thm:intro-LW3}, and \ref{thm:intro-LW4}) we need  reductions when $i=1$. The reductions explained below are used for the case when 
$i=1$. 

\begin{theorem}
\label{thm:red-SAT-MSO}
Given a {\sc $3$-CNF SAT} formula $\phi$ on $n$ variables and two integers $k$ and $s$ such that $n\leq k \cdot s$, one can construct a $(k,s)$-MSO formula $\psi$ in polynomial time with the following specifications. 

\begin{enumerate}[(i)]
\item For any graph $G$ with $|V(G)|\geq k$, $G\models \psi$ if and only if $\phi$ is satisfiable. 
\item $|\psi|=O(|\phi|+k^2+s)$. 
\end{enumerate}
\end{theorem}

\begin{proof}
Let $x_1,x_2,\ldots,x_n$ be the variables in $\phi$. We construct an MSO $\psi$ formula with $k$ vertex variables $v_1,\ldots,v_k$ and $s$ set variables $W_1,\ldots,W_s$. Notice that the number of atomic formulas of the form $v_i\in W_j$ for all $i\in [k]$ and $j\in [s]$ is $k\cdot s \geq n$. 
Let $\psi_1$ be the MSO formula on $k+s$ free variables obtained from $\phi$ by replacing each $x_{i,j}$ with $v_i\in W_j$. Then, we define 

Since $k\cdot s \geq n$, each variable in $\phi$ can be encoded as $v_i\in W_j$ for some $i$ and $j$, and $v_i\in W_j$ is true can be interpreted as the corresponding variable in $\phi$ is set to true.
For ease of presentation, we assume that $\phi$ contains $k\cdot s$ variables and the variable set is $\{x_{i,j}~:~i\in [k], j\in [s]\}$.

$$
\psi\equiv \exists v_1 \exists v_2 \ldots \exists v_k \exists W_1\ldots \exists W_s\;\; \psi_1 $$


Notice that $|\psi|\leq |\psi_1|+k^2+s= O(|\phi|+k^2+s)$. Next, we prove the correctness of the reduction. Let $G$ be a graph on at least $k$ vertices. Suppose $\phi$ is satisfiable and let $\beta$ be a satisfying assignment. Let $u_1,\ldots,u_k$ be arbitrarily chosen $k$ distinct vertices in $G$. Now we define $Z_1,\ldots,Z_k$ as follows. For each $i\in [k]$ and $j\in [s]$, $u_i\in Z_j$ if and only if $\beta(x_{i,j})$
is true. 
This implies that 
$\psi_1(u_1,\ldots,u_k,Z_1,\ldots,Z_s)$ is true. 
Since $u_1,\ldots,u_k$ are distinct vertices, $\bigwedge_{i,i'\in [k], \; i\neq i'} u_i \neq u_{i'}$ is true. Therefore, $G\models \psi$. 

Now we prove the reverse direction. Let $G$ be a graph on at least $k$ vertices and $G\models \psi$. Then, there exists an assignment to the variables $v_1,\ldots,v_k, W_1,\ldots,W_s$ such that $G\models \psi$. Let $u_1,\ldots,u_k, Z_1,\ldots, Z_s$ be an assignment to $v_1,\ldots,v_k, W_1,\ldots,W_s$ such that $\psi$ is true for this assignment. Now we construct a satisfying assignment $\beta$ for $\phi$ as follows. For each $i\in [k]$ and $j\in [s]$, 
$\beta(x_{i,j})$ is true if and only if $u_i\in Z_j$. 
Since $\psi_1(u_1,\ldots,u_k,Z_1,\ldots,Z_s)$ is true, $\beta$ is a satisfying assignment of $\phi$. 
\end{proof}

\begin{theorem}
\label{thm:red-SAT-FO}
Given a {\sc $3$-CNF SAT} formula $\phi$ on $n$ variables and integers $k\geq 2$ and $t\geq 1$ such that $n\leq \min \{k,t\} \cdot k$, one can construct a graph $G$ and a $(k+t)$-FO formula $\psi$ in time $O(2^{\sqrt{n}}+|\phi|)$ with the following specifications. 

\begin{enumerate}[(i)]
\item $G\models \psi$ if and only if $\phi$ is satisfiable. 
\item $V(G)=O(2^{\sqrt{n}})$ and $\mathbf{tw}(G) \leq t$. 
\item $|\psi|=O(|\phi|+k+t)$. 
\item The number of label predicates used in $\psi$ is $2$. 
\end{enumerate}
\end{theorem}

\begin{proof}
Let $s\leq \min \{k,t\}$ such that $n\geq s\cdot k$. This implies that $s\leq \sqrt{n}$. First, we explain the construction of $G$. $G$ is a bipartite graph with bipartition $A\uplus B$. The set $A$ has $s$ vertices. Let $A=\{a_1,\ldots,a_s\}$. The set $B$ has $2^s$ vertices. For each subset $A'\subseteq A$, there is exactly one vertex $b_{A'}$ in $B$ which is adjacent to all the vertices in $A'$ and not adjacent to any vertex in $A\setminus A'$. This completes the construction of $G$. Notice that 
$\mathbf{tw}(G)\leq s \leq t$ and 
$V(G)=2^{s}+s$. 
Next, we explain the construction of $\psi$. Let $s=\min\{k,t\}$. We create two label predicates ${\sf A}$ and ${\sf B}$, where ${\sf A}$ and ${\sf B}$ contain all the vertices in $A$ and $B$, respectively. The formula $\psi$ contains $s+k$ variables $w_1,\ldots,w_s, v_1,\ldots,v_k$ such that  $w_1,\ldots,w_s\in {\sf A}$ and $v_1,\ldots,v_k\in {\sf B}$. Since $k\cdot s \geq n$, each variable in $\phi$ can be encoded as ${\sf adj}(v_i,w_j)$ for some $i$ and $j$, and ${\sf adj}(v_i,w_j)$ is true can be interpreted as the corresponding variable in $\phi$ is set to true.
For the ease of presentation, we assume that $\phi$ contains $k\cdot s$ variables and the variable set is $\{x_{i,j}~:~i\in [k], j\in [s]\}$. Let $\psi_1$ be the FO formula on $k+s$ free variables obtained from $\phi$ by replacing each $x_{i,j}$ with ${\sf adj}(v_i,w_j)$. Then, we define 
$$
\psi\equiv \exists v_1\in {\sf B} \exists v_2\in {\sf B} \ldots \exists v_k\in {\sf B} \exists w_1 \in {\sf A}\ldots \exists w_s \in {\sf A}\;\; \psi_1
$$

Notice that $|\psi|=O(|\phi|+k+s)=O(|\phi|+k+t)$. Now we prove the correctness of the reduction. Suppose $\phi$ is satisfiable and let $\beta$ be a satisfying assignment. Now we need an assignments to $v_1,\ldots,v_k$ and $w_1,\ldots,w_s$ such that $\psi$ is true under the assignment. Let us assign $a_j$ to $w_j$ for all $j\in [s]$. For each $i\in [k]$, let $u_i$ be the vertex in $B$ such that the following holds. For each $j\in [s]$, $(u_i,a_j)\in E(G)$ if and only if $\beta(x_{i,j})$ is true.  
Now, for all $i\in [k]$, assign $u_i$ to $v_i$. Then, since $\beta$ is a satifying assignment to $\phi$,  $\psi_1(u_1,\ldots,u_k,a_1,\ldots,a_s)$ is true. This implies that $G\models \psi$.



Now we prove the reverse direction. Let $G$ be a graph on at least $k$ vertices and $G\models \psi$. Then, there exists an assignment to the variables $v_1,\ldots,v_k, w_1,\ldots,w_s$ such that $\psi_1$ is true under this assignment. Let $u_1,\ldots,u_k, c_1,\ldots, c_s$ be an assignment to $v_1,\ldots,v_k, w_1,\ldots,w_s$ such that $\psi_1$ is true for this assignment. Now we construct a satisfying assignment $\beta$ for $\phi$ as follows. For each $i\in [k]$ and $j\in [s]$, $\beta(x_{i,j})$ is true if and only if $(u_i,c_j)\in E(G)$.  Since $\psi_1(u_1,\ldots,u_k,c_1,\ldots,c_s)$ is true, $\beta$ is a satisfying assignment of $\phi$. 
\end{proof}

\subsection{Removing Label Predicates}
\label{subsec:labelpredicates}

\subsubsection{Removing Labels from Theorem~\ref{thm:low-MSO}}
\label{sub:Removelabel-thm1}
First, we explain how to remove the labels in Theorem~\ref{thm:low-MSO}. Let $T$ be the tree and $\ell$ be the number of labels used in the formula of Theorem~\ref{thm:low-MSO}. 
Let ${\sf L}_1,\ldots,{\sf L}_{\ell}$ be the label predicates. 
Notice that $T$ is a rooted tree. First, we subdivide each edge in $T$ twice. See Figure~\ref{fig:T} for an illustration. The tree on the left side of the figure is an example of $T$. For each edge $\{c,p\}\in E(T)$, where $p$ is the parent of $c$ we subdivide the edge twice and the node adjacent to $c$ is colored red and the node adjacent to $p$ is colored blue. In Figure~\ref{fig:T}, the tree on the right side of the figure, is the tree we get after subdividing each edge twice. Let the resulting tree be $T'$. Let $R$ be the set of red colored vertices and $B$ be the set of blue colored vertices.

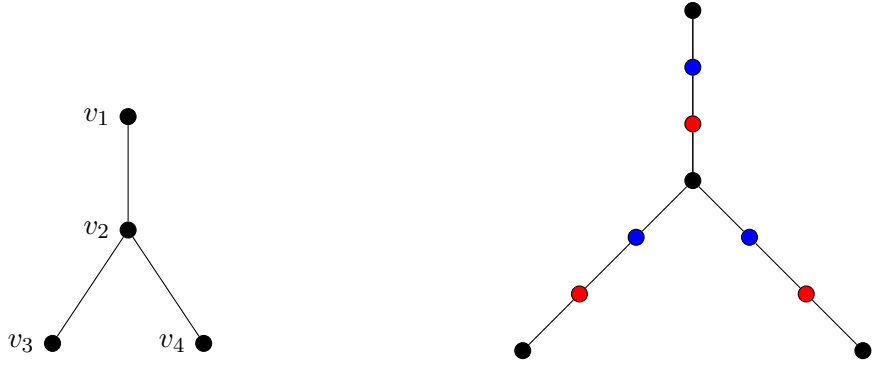
\begin{figure}
    \centering
\begin{subfigure}{.4\textwidth}
\begin{tikzpicture}[
  every node/.style={circle, draw, minimum size=6pt, inner sep=0pt},
  filled/.style={circle, fill=black, minimum size=6pt, inner sep=0pt},
  level distance=15mm,
  sibling distance=20mm
]

\node[filled] at (0,0) {} 
  child {node[filled] {}
    child {node[filled] {} 
    node[draw=none, left=2mm] {$v_3$}}
    child {node[filled] {}
    node[draw=none, left=2mm] {$v_4$}}
    node[draw=none, left=2mm] {$v_2$}
  }
  node[draw=none, left=2mm] {$v_1$};


\end{tikzpicture}
\end{subfigure}
\begin{subfigure}{.55\textwidth}
 \begin{tikzpicture}[scale=0.5,
  every node/.style={circle, draw, minimum size=6pt, inner sep=0pt},
  triangle/.style={regular polygon, regular polygon sides=3, draw, minimum size=12mm, inner sep=0pt},
  level distance=15mm,
  sibling distance=20mm
]

\node[fill=black] (n1) at (0,0) {}
  child {node[fill=blue] (n2) {}
    child {node[fill=red] (n3) {}
      child {node[fill=black] (n4) {}}
    }
  };



\path (n4) ++(-15mm,-15mm) node[fill=blue] (left1) {}
      ++(-15mm,-15mm) node[fill=red] (left2) {}
      ++(-15mm,-15mm) node[fill=black] (left3) {};



\path (n4) ++(15mm,-15mm) node[fill=blue] (right1) {}
      ++(+15mm,-15mm) node[fill=red] (right2) {}
      ++(+15mm,-15mm) node[fill=black] (right3) {};


\draw (n1) -- (n2) -- (n3) -- (n4);
\draw (n4) -- (left1) -- (left2) -- ++(-15mm,-15mm);
\draw (n4) -- (right1)--(right2) --(right3);

\end{tikzpicture}
   
\end{subfigure}

    \caption{Illustration of $T$ and $T'$}
    \label{fig:T}
\end{figure}

Next, for each label, we want to attach a {\em unique} copy of a tree. Towards that we create trees $T_1,\ldots,T_{\ell+1}$ as shown in Figure~\ref{fig:Ti}. Notice that each $T_i$ is a tree on $6(\ell+1)+1$ vertices and it is not isomorphic to a subgraph of $T'$. Moreover, for each $i\neq j$, $T_i$ and $T_j$ are not isomorphic. For each $i\in [\ell]$, $T_i$ corresponds to the label ${\sf L}_i$. Now for each vertex $v\in V(T)$, with label ${\sf L}_i$, we attach a copy of $T_i$ to $v$. For each $w\in R$, we attach a copy of $T_{\ell+1}$ to it. Let the resulting graph be $T''$.  See Figure~\ref{fig:Tdp}, for an illustation of the construction of $T''$ for the tree $T$ in Figure~\ref{fig:T}.

\begin{figure}
    \centering
\begin{tikzpicture}[level distance=15mm, sibling distance=25mm,
  every node/.style={circle, draw, minimum size=6pt, inner sep=0pt},
  level 2/.style={sibling distance=5mm}]

\node (root) {}
  child {node {} 
    child {node {}}
    child {node {}}
    child {node {}}
    child {node {}}
    child {node {}}
    node[draw=none, left=3mm] {$v_1$}
  }
  child {node {} 
    child {node {}}
    child {node {}}
    child {node {}}
    child {node {}}
    child {node {}}
    node[draw=none, left=2mm] {$v_2$}
    node[draw=none, right=10mm] {$\dots$}
  }
  child {node {} 
    child {node {}}
    child {node {}}
    child {node {}}
    child {node {}}
    child {node {}}
    node[draw=none, left=2mm] {$v_i$}
  }
  child {node {} 
  node[draw=none,left=2mm] {$v_{i+1}$}
  node[draw=none, right=10mm] {$\dots$}
  }
  child {node {} 
    node[draw=none, right=3mm] {$v_{(\ell+1)+5(\ell+1-i)}$}
  }
  ;

\end{tikzpicture}

    \caption{Graph $T_i$}
    \label{fig:Ti}
\end{figure}
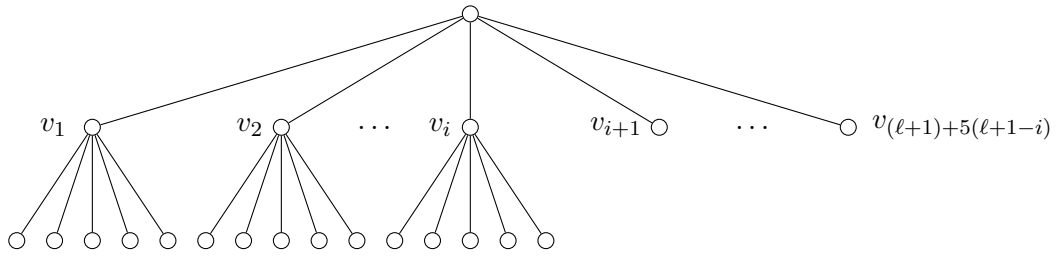

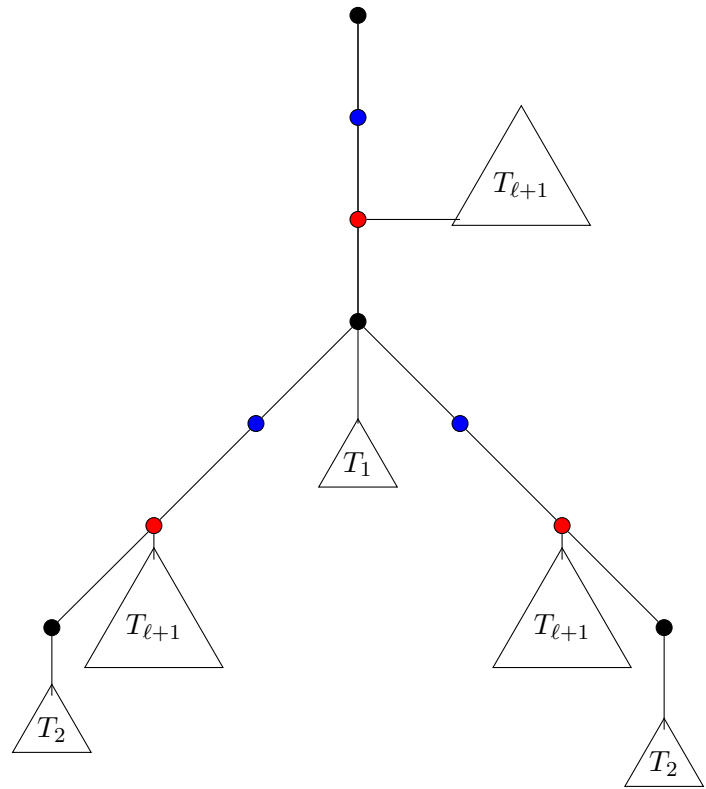
\begin{figure}
    \centering
\begin{tikzpicture}[scale=0.9, 
  every node/.style={circle, draw, minimum size=6pt, inner sep=0pt},
  triangle/.style={regular polygon, regular polygon sides=3, draw, minimum size=12mm, inner sep=0pt},
  level distance=15mm,
  sibling distance=20mm
]

\node[fill=black] (n1) at (0,0) {}
  child {node[fill=blue] (n2) {}
    child {node[fill=red] (n3) {}
      child {node[fill=black] (n4) {}}
    }
  };


\node[triangle] at ([shift={(24mm,5mm)}]n3) {$T_{\ell+1}$};

\path (n4) ++(-15mm,-15mm) node[fill=blue] (left1) {}
      ++(-15mm,-15mm) node[fill=red] (left2) {}
      ++(-15mm,-15mm) node[fill=black] (left3) {}
      ++(0mm,-15mm) node[triangle] {$T_2$};

\node[triangle] at ([shift={(0mm,-15mm)}]left2) {$T_{\ell+1}$};
\draw (left2)--++(0,-5mm);

\path (n4) ++(0,-21mm) node[triangle] {$T_1$};

\path (n4) ++(15mm,-15mm) node[fill=blue] (right1) {}
      ++(+15mm,-15mm) node[fill=red] (right2) {}
      ++(+15mm,-15mm) node[fill=black] (right3) {}
      ++(0,-20mm) node[triangle] {$T_{2}$};

\node[triangle] at ([shift={(0mm,-15mm)}]right2) {$T_{\ell+1}$};
\draw (right2)--++(0,-5mm);

\draw (n1) -- (n2) -- (n3) -- (n4);
\draw (n3) -- ++(15mm,0);
\draw (n4) -- (left1) -- (left2) -- ++(-15mm,-15mm)--++(0,-10mm);
\draw (n4) -- ++(0,-15mm);
\draw (n4) -- (right1)--(right2) --(right3)-- ++(0,-15mm);

\end{tikzpicture}
    \caption{Illustration of $T''$}
    \label{fig:Tdp}
\end{figure}

Now we explain the change we need to make to the  formula $\phi$. Notice that, we already mentioned that $\exists x \in {\sf L}_i \psi$ is same as $\exists x \;({\sf L}_i(x) \wedge \psi)$. Now, this can be represented as follows. Let $v_0,v_1,\ldots,v_{6(\ell+1)}$ be the vertices in the copy of $T_i$ that is attached to $x$, where $v_0$ is the root of the copy of $T_i$. 
Let $p_x$ be the parent of $x$ in $T$ and $a,b$ be the subdivision vertices such that $xa,ab,bp_x \in E(T'')$. Let $u_0,u_1,\ldots,u_{6(\ell+1)}$ be the vertices in the copy of $T_{\ell+1}$ that is attached to $a$, where $u_0$ is the root of the copy of $T_{\ell+1}$. 
Now we encode $\exists x \;({\sf L}_i(x) \wedge \psi)$ as 

\begin{eqnarray*}
\exists x \exists a \exists b \exists p_x \exists v_0 \ldots \exists v_{6(\ell+1)} \exists u_0 \ldots \exists u_{6(\ell+1)} && {\sf adj}(x,a) \wedge {\sf adj}(x,v_0)\wedge {\sf adj}(a,b) \wedge {\sf adj}(b,p_x) \wedge {\sf adj}(a,u_0) \wedge \\
&& {\sf iso}_i (v_0,\ldots,v_{6(\ell+1)}) \wedge {\sf iso}_{\ell+1} (u_0,\ldots,u_{6(\ell+1)}) \wedge \psi'
\end{eqnarray*}

Here, ${\sf iso}_j (v_0,\ldots,v_{6(\ell+1)})$ encodes that $T''[\{v_0,\ldots,v_{6(\ell+1)}\}]$ is isomorphic to $T_j$. The formula $\psi'$ is obtained from $\psi$ by replacing ${\sf adj}(x,p_x)$ with ${\sf adj}(b,p_x)$. Similarly,  
we encode $\forall x \;({\sf L}_i(x) \Rightarrow \psi)$ as

\begin{eqnarray*}
\forall x \forall a \forall b \forall p_x \forall v_0 \ldots \forall v_{6(\ell+1)} \forall u_0 \ldots \forall u_{6(\ell+1)} && ({\sf adj}(x,a) \wedge {\sf adj}(x,v_0)\wedge {\sf adj}(a,b) \wedge {\sf adj}(b,p_x) \wedge {\sf adj}(a,u_0) \wedge \\
&& {\sf iso}_i (v_0,\ldots,v_{6(\ell+1)}) \wedge {\sf iso}_{\ell+1} (u_0,\ldots,u_{6(\ell+1)}))\Rightarrow \psi'
\end{eqnarray*}

\subsubsection{Removing Labels from Theorems~\ref{thm:low-MSO-TW1} and \ref{thm:low-MSO-TW22}}

In both these theorems, the output graph $G'$ has the following property. There is a vertex subset $W\subseteq V(G')$ such that $G'\setminus W$ is a rooted tree $W$ is either an independent set or forms a path, and $|W|=k$. Moreover, all the vertices in $W$ is labeled with exactly one label predicate. Let us call, this label predicate to be ${\sf W}_{\ell}$ 
Now we replace all the label predicates used in $G'\setminus W$ like the way we did it in Section~\ref{sub:Removelabel-thm1}. Since $W$ is either an independent set or forms a path, there is $K_5$ (complete graph on $5$ vertices) as a subgraph in $G'$. So we attach a copy of $K_5$ to each vertex. In the formula $\phi$, when we select a variable $w\in {\sf W}_{\ell}$, we also select $5$ vertices and add a condition that one of them is connected $w$ and the $5$ vertices for a clique.

\subsection{Proofs of Theorems~\ref{thm:intro-LW1}, \ref{thm:intro-LW2}, \ref{thm:intro-LW3}, and \ref{thm:intro-LW4}}
\label{subsec:lw}


\begin{proposition}
\label{prop:ETH3col}
Assuming ETH, there is a constant $c$ such that any algorithm for \threecoloring\  takes time strictly more than $2^{cn} n^{O(1)}$, where $n$ is the number of vertices in the input graph. 
\end{proposition}

\begin{proof}[Proof of Theorem~\ref{thm:intro-LW1}]
Suppose there is an algorithm for \textsc{MSO Testing} on trees in $f(k_1,s_1,\dots,k_d,s_d) \cdot \exp^{(i)}(T(s_j, k_i)) \cdot (n+|\phi|)^{O(1)}$ time for a function $f(k_1,s_1,\dots,k_d,s_d)$ independent of $s_j$ and $k_i$. We want to prove that $T(x,y)=\Omega(xy)$. For the sake of contradiction, assume that $T(xy)=o(xy)$. Then, for any constant $a>0$, there exist $x_a,y_a>0$ such that $T(x,y)< a x y$ for all $x\geq x_a$ and $y\geq y_a$. Let $a<1/3$ and $\alpha>1$ be two constants which we fix later.


Let $G$ be an input instance of \threecoloring. Let $n=|V(G)|$. Fix two integers $s_j$ and $k_i$ such that $x_a\leq s_j=O(\log n)$, $k_i\geq \{y_a,20\}$, and $n \leq \exp^{(i-1)}(\frac{k_i-9}{2}\cdot s_j) \leq n+O(1)$. 
We apply Theorem~\ref{thm:low-MSO}, 
and  compute in  $O(3^{n/\alpha})$ 
    time a tree $T$ and an $S$-MSO formula $\phi$ for $S = ((k_1,s_1),\dots,(k_d,s_d))$ satisfying the following.
    \begin{enumerate}[(i)]
        \item $G$ is a $3$-colorable if and only if $T$ satisfies $\phi$.
        \item $|V(T)| = O(3^{n/\alpha})$ and 
        $|\phi| = O(2^{s_j}+k_i+\alpha^2)$.         \item $k_1 = O(\alpha)$, $k_{i'} = O(1)$ for all $i'\in [d]\setminus \{1,i\}$. 
        \item $s_{j'} = 0$ for all $j' \in [d] \backslash \{j\}$.
    \end{enumerate}

Now, applying the algorithm for \textsc{MSO Testing} on trees, 
\threecoloring\ can be solved in time 

\begin{eqnarray}
R(n)&=& O(3^{n/\alpha})+\exp^{(i)}(T(s_j, k_i)) \cdot (3^{n/\alpha}+2^{s_j}+k_i+\alpha^2)^{O(1)} \nonumber\\
&=& O(3^{n/\alpha})+\exp^{(i)}(a\cdot s_j\cdot k_i) \cdot (3^{n/\alpha}+2^{s_j}+k_i+\alpha^2)^{O(1)} \nonumber
\end{eqnarray}

Notice that, for $a<1/3$, $\exp^{(i-1)}(a \cdot k_i \cdot s_j)\leq 3a \exp^{(i-1)}(\frac{k_i}{3} \cdot s_j)$. Since $k_i\geq 20$, 
$$\exp^{(i-1)}(\frac{k_i}{3}\cdot s_j) \leq \exp^{(i-1)}(\frac{k_i-9}{2}\cdot s_j) \leq n+O(1). $$ 

Therefore, 
\begin{eqnarray*}
    \exp^{(i-1)}(a \cdot k_i \cdot s_j)&\leq&  3a(n+O(1)) \\
\exp^{(i)}(a \cdot k_i \cdot s_j)&\leq&  2^{3a\cdot n+O(1)} 
\end{eqnarray*}

Then, 

\begin{eqnarray}
R(n)&=& O(3^{n/\alpha})+2^{3a\cdot n + O(1)} \cdot (3^{n/\alpha}+2^{s_j}+k_i+\alpha^2)^{O(1)}\nonumber
\end{eqnarray}


Let $c$ be a constant mentioned in Proposition~\ref{prop:ETH3col}. Now, we choose $a$ and $\alpha$ in such a way that $R(n)$ is at most $2^{cn}$. Then, this is a contradiction to Proposition~\ref{prop:ETH3col}. 
\end{proof}


Next, we prove Theorem~\ref{thm:intro-LW2}. 

\begin{proof}[Proof of Theorem~\ref{thm:intro-LW2}]
Suppose there is an algorithm that solves \textsc{MSO Testing} on trees in time $f(k_1,s_1,\dots,k_d,s_d) \cdot 2^{T(s_d, k_i)} \cdot (n+|\phi|)^{O(1)}$ for a function $f(k_1,s_1,\dots,k_d,s_d)$ independent of $s_d$ and $k_i$. 
We want to prove that $T(x,y)=\Omega(xy)$. For the sake of contradiction, assume that $T(xy)=o(xy)$. Then, for any constant $a>0$, there exist $x_a,y_a>0$ such that $T(x,y)< a x y$ for all $x\geq x_a$ and $y\geq y_a$. Let $a<1$  be a constant which we fix later.  

Now, for any constant $c$ we design an algorithm for {\sc $3$-CNF SAT} that runs in time $2^{cn}$ time, where $n$ is the number of variables in the input $3$-CNF formula. This will contradicts ETH. Fix a constant $c$.   
Let $\phi$ be a {\sc $3$-CNF SAT} formula  on $n$ variables. 
Choose two integers $k_i$ and $s_d$ such that $n\leq k_i \cdot s_d \leq n+O(1)$. Now, we apply Theorem~\ref{thm:red-SAT-MSO}, and construct a $(k_i,s_d)$-MSO formula $\psi$ in polynomial time with the following specifications. 

\begin{enumerate}[(i)]
\item For any graph $G$ with $|V(G)|\geq k_i$, $G\models \psi$ if and only if $\phi$ is satisfiable. 
\item $|\psi|=O(|\phi|+k_i^2+s_d)$. 
\end{enumerate}
We may assume that $G$ is a tree on $k_i$ vertices. Also, $\psi$ is an $S$-MSO formula, for $S=((k_1,s_1),\ldots,(k_d,s_d))$, where $k_j=0$ for all $j\neq i$ and $s_j=0$ for all $j\neq d$. 
Now, applying the algorithm for \textsc{MSO Testing} on trees, 
{\sc $3$-CNF SAT} can be solved in time 

\begin{eqnarray}
R(n)&=& 2^{T(s_d, k_i)} \cdot (n+|\psi|)^{O(1)}\nonumber\\
&=&2^{a\cdot s_d\cdot k_i} (n+|\phi|)^{O(1)}\nonumber \\
&=&2^{a (n+O(1))} (n+|\psi|)^{O(1)} \nonumber
\end{eqnarray}
By fixing $a=c$, we get an algorithm for {\sc $3$-CNF SAT} that runs in time $2^{cn}$ time, which contradicts ETH. 
\end{proof}


Next, we prove Theorem~\ref{thm:intro-LW3}. 

\begin{proof}[Proof of Theorem~\ref{thm:intro-LW3}] Let $d,i,j \in \mathbb{N}$ such that $j \leq i \leq d$.
Suppose there is an 
algorithm that solves \textsc{FO Testing} in $f(k_1,\dots,k_d,t) \cdot \exp^{(i)}(T(k_j,t,k_i)) \cdot (n+|\phi|)^{O(1)}$ time for any function $f(k_1,\dots,k_d,t)$ that is independent of $t$, $k_j$, and $k_i$. We want to prove that $T(x,y,z) = \Omega(\min\{x,y\} \cdot z)$. For the sake of contradiction, assume that $T(x,y)=o(\min\{x,y\}\cdot z)$. Then, for any constant $a>0$, there exist $x_a,y_a,z_a>0$ such that $T(x,y)< a \min\{x,y\}\cdot z$ for all $x\geq x_a$ and $y\geq y_a$ and $z\geq z_a$. Let $a<1/3$ and and $\alpha>1$ be two constants which we fix later. 
We have two cases. 

\medskip
\noindent
{\bf Case 1: $i>1$.}
Let $G$ be an input instance of \threecoloring. Let $n=|V(G)|$. Fix two integers $k$ and $t$ such that  $x_a<k$, $y_a<t$, $\max\{z_a,20\}<k'$,  $\min\{k,t\} =O(\log n)$, and $n \leq \exp^{(i-1)}(\min\{k,t\}\frac{k'-9}{2}) \leq n+O(1)$.  
We apply Theorem~\ref{thm:low-MSO-TW1} and 
compute in $O(3^{n/\alpha})$ time a graph $G'$ and an $S$-FO formula $\phi$ for $S = (k_1,\dots,k_d)$ satisfying the following.
    \begin{enumerate}[(i)]
        \item $G$ is a $3$-colorable if and only if $G'$ satisfies $\phi$.
        \item $|V(G')| = O(3^{n/\alpha})$, $\mathbf{tw}(G') \leq t$, and 
        $|\phi| = O(2^{\min\{k,t\}} +k'+\alpha^2)$.
        \item $k_{i'} = O(1)$ for all $i'\in   [d]\setminus \{1,j,i\}$         \begin{itemize}
            \item if $1 = j < i$, then $k_1 = k_j = O(\alpha+k)$ and $k_i = k'$;
            \item if $1 < j < i$, then $k_1 = O(\alpha)$, $k_j = k$, and $k_i = k'$;
            \item if $1 < j = i$, then $k_1 = O(\alpha)$, $k_j = k_i = k+k'$.
        \end{itemize}
    \end{enumerate}

Now applying the 
algorithm for \textsc{FO Testing}, we get a solution for \threecoloring in time 

\begin{eqnarray*}
R(n)&=&O(3^{n/\alpha})+\exp^{(i)}(T(k,t,k')) \cdot (3^{n/\alpha}+2^{\min\{k,t\}} +k'+\alpha^2)^{O(1)}\\
&=&O(3^{n/\alpha})+\exp^{(i)}(a\min\{k,t\}\cdot k')) \cdot (3^{n/\alpha}+n)^{O(1)}
\end{eqnarray*}

Notice that, for $a<1/3$, $\exp^{(i-1)}(a \cdot \min\{k,t\} \cdot k)\leq 3a \exp^{(i-1)}(\frac{1}{3} \min\{k,t\} \cdot k')$. Since $k\geq 20$, 
$$\exp^{(i-1)}(\frac{1}{3} \min\{k,t\} \cdot k') \leq \exp^{(i-1)}(\min\{k,t\} \cdot \frac{k'-9}{2})\leq n+O(1). $$ 

Then, 
$$
R(n)= O(3^{n/\alpha})+2^{3a\cdot n + O(1)} \cdot (3^{n/\alpha}+n)^{O(1)}.
$$

Let $c$ be a constant mentioned in Proposition~\ref{prop:ETH3col}. Now, we choose $a$ and $\alpha$ in such a way that $R(n)$ is at most $2^{cn}$. Then, this is a contradiction to Proposition~\ref{prop:ETH3col}.

\medskip
\noindent
{\bf Case 2: $i=1$.}
In this case, for any constant $c$ we design an algorithm for {\sc $3$-CNF SAT} that runs in time $2^{cn}$ time, where $n$ is the number of variables in the input $3$-CNF formula. This will contradicts ETH. Fix a constant $c$.   
Let $\phi$ be a {\sc $3$-CNF SAT} formula  on $n$ variables. 
Choose two integers $k$ and $t$ such that
$n\leq \min\{k,t\} \cdot k \leq n+O(1)$. Now, we apply Theorem~\ref{thm:red-SAT-FO}, and construct a graph $G$ and a $(k+t)$-FO formula $\psi$ in time $O(2^{\sqrt{n}}+|\phi|)$ with the following specifications. 

\begin{enumerate}[(i)]
\item $G\models \psi$ if and only if $\phi$ is satisfiable. 
\item $V(G)=O(2^{\sqrt{n}})$ and $\mathbf{tw}(G) \leq t$. 
\item $|\psi|=O(|\phi|+k+t)$. 
\end{enumerate}

Since $t\leq k$, $\psi$ is a $(2k)$-FO formula. 
To solve $\psi$, we run the algorithm for \textsc{FO Testing} and its running time will be 

\begin{eqnarray*}
   R(n)&=&2^{T(k,t,k)} (|V(G)|+|\psi|)^{O(1)}\\
   &=& 2^{a\min\{k,t\}\cdot  k} (2^{\sqrt{n}}+n)^{O(1)}\\
   &=& 2^{an+O(1)} (2^{\sqrt{n}}+n)^{O(1)}
\end{eqnarray*}
By fixing $a=c$, we get an algorithm for {\sc $3$-CNF SAT} that runs in time $2^{cn}$ time, which contradicts ETH. 
\end{proof}


Next, we prove Theorem~\ref{thm:intro-LW4}. 

\begin{proof}[Proof of Theorem~\ref{thm:intro-LW4}]Let $d,i \in \mathbb{N}$ such that $i \leq d$. Suppose there is an algorithm solves \textsc{FO Testing} in $f(k_1,\dots,k_d,t) \cdot \exp^{(i)}(T(k_i,t)) \cdot (n+|\phi|)^{O(1)}$ time for any function $f(k_1,\dots,k_d,t)$ that is independent of $t$ and $k_i$. We need to prove that $T(x,y) = \Omega(x \log y)$. 
We have two cases.

\medskip
\noindent
{\bf Case 1: $i>1$.}
For the sake of contradiction, assume that $T(x,y)=o(x\log y)$. Then, for any constant $a>0$, there exist $x_a,y_a>0$ such that $T(x,y)< a x\log y$ for all $x\geq x_a$ and $y\geq y_a$.  
Let $a<1/3$ and and $\alpha>1$ be two constants which we fix later. 

Let $G$ be an input instance of \threecoloring. Let $n=|V(G)|$. Fix two integers $k_i$ and $t$ such that  $\max\{x_a,20\}<k_i$, $y_a<t$,  and $n \leq \exp^{(i-1)}(\frac{k_i-6}{2} \log t) \leq n+O(1)$.  We apply Theorem~\ref{thm:low-MSO-TW22} with $k=k_i$ and compute in $O(3^{n/\alpha})$ time a graph $G'$ and an $S$-FO formula $\phi$ for $S = (k_1,\dots,k_d)$ satisfying the following.
    \begin{enumerate}[(i)]
        \item $G$ is a $3$-colorable if and only if $G'$ satisfies $\phi$.
        \item $|V(G')| = O(3^{n/\alpha})$, $\mathbf{tw}(G') \leq t$, and $|\phi| = O(\alpha^{2}+k_i)$.
        \item $k_i = k$, $k_1 = O(\alpha)$, $k_{i'} = O(1)$ for all $i' \in \{2,\dots,i-1\}$, and $k_{i'} = 0$ for all $i' \in [d] \backslash [i]$.
    \end{enumerate}

Now applying the 
algorithm for \textsc{FO Testing}, we solve \threecoloring\ in time 

\begin{eqnarray*}
R(n)&=&O(3^{n/\alpha})+\exp^{(i)}(T(k_i,t) \cdot (3^{n/\alpha}+ k_i+\alpha^2)^{O(1)}\\
&=&O(3^{n/\alpha})+\exp^{(i)}(a\cdot k_i \log t)) \cdot (3^{n/\alpha})^{O(1)}
\end{eqnarray*}

Notice that, for $a<1/3$, $\exp^{(i-1)}(a \cdot k_i \log t)\leq 3a \exp^{(i-1)}(\frac{1}{3} k_i \log t)$. Since $k_i\geq 20$, 
$$\exp^{(i-1)}(\frac{k_i-6}{2}  \log t) \leq \exp^{(i-1)}(\frac{k_i}{3} \log t)\leq n+O(1). $$ 

Then, 
$$
R(n)= O(3^{n/\alpha})+2^{3a\cdot n + O(1)} \cdot (3^{n/\alpha}+n)^{O(1)}.
$$

Let $c$ be a constant mentioned in Proposition~\ref{prop:ETH3col}. Now, we choose $a$ and $\alpha$ in such a way that $R(n)$ is at most $2^{cn}$. Then, this is a contradiction to Proposition~\ref{prop:ETH3col}.

\medskip
\noindent
{\bf Case 2: $i=1$.} We know that the {\sc Clique} problem (finding a clique of size $k$ on an $n$ vertex graph) can not be solved in time $n^{o(k)}$ unless ETH fails. 
Let $G'$ be a graph on $n$ vertices. We can encode the existence of a clique of size $k$ using the following FO formula. 

$$\exists v_1 \ldots \exists v_k \bigwedge_{1\leq i'<j'\leq k} {\sf adj}(v_{i'},v_{j'})$$

So, if $T(x,y)=o(x\log y)$, then using the algorithm for \textsc{FO Testing}, we can solve {\sc Clique} in time $2^{o(k\log n)}$ (because treewidth of $G'$ is at most $n$). This contradicts ETH. 
\end{proof}

\section{Conclusion and future work} \label{sec-conclusion}
In this paper, we systematically study the time complexity of Courcelle's theorem, towards understanding its dependency on the MSO formula $\phi$ and the treewidth parameter $t$ in a fine-grained way.
We prove (almost) matching upper and lower bounds for the time complexity in terms of the quantifier structure of $\phi$.
We expect this work to be a starting point of the long-term research towards thoroughly understanding the time complexity of Courcelle's theorem.

Below we pose some open questions for future study.
First, the bound in Theorem~\ref{thm-MSO} (and also Theorem~\ref{thm-MSO*}) is only \textit{almost} tight because of the $\hat{O}(\cdot)$-notation.
It is thus natural to ask whether one can replace $\hat{O}(\cdot)$ with $O(\cdot)$ to make the bound \textit{exactly} tight.
Second, we only investigated the time complexity in terms of the \textit{quantifier} structure of $\phi$, while completely ignoring the quantifier-free part of $\phi$.
It might be interesting to study Courcelle's theorem in terms of even more fine-grained structural parameter of $\phi$.
Finally, we only considered treewidth in this paper.
Variants of \textsc{MSO Testing} have been also studied with other width parameters of graphs, e.g., clique-width.
Therefore, it is also worth studying fine-grained bounds for the complexity of (variants of) Courcelle's theorem with those parameters.
\bibliographystyle{alphaurl}
\bibliography{my_bib}

\appendix
\section{Undecidability of polynomial-time testable MSO properties} \label{apx-undecidable}
Consider the following problem: given an MSO formula $\phi$, decide whether there exists a graph satisfying $\phi$ or not.
This problem, which we call \textsc{MSO Realization}, is known to be undecidable~\cite{trakhtenbrot1950impossibility}.
Next, we consider our problem: given an MSO formula $\phi$, decide whether testing $\phi$ on graphs is polynomial-time solvable or NP-hard (assuming P$\neq$NP).
We call this problem \textsc{PTime MSO Testability}.
We show the undecidability of \textsc{PTime MSO Testability} by reducing \textsc{MSO Realization} to it as follows.
Let $\phi$ be an instance of \textsc{MSO Realization}.
We construct two MSO formulas $\phi_1$ and $\phi_2$, where $\phi_1$ expresses that ``$G$ is $3$-colorable and there exists $A \subseteq V(G)$ such that $G[A]$ is connected and satisfies $\phi$'' and $\phi_2$ expresses that ``$V(G)$ can be partitioned into $A$ and $B$ such that $G[A]$ is connected and satisfies $\phi$, and $G[B]$ is $3$-colorable''.

Assume there exists an algorithm $\mathcal{A}$ that solves \textsc{PTime MSO Testability}, and we run it on $\phi_1$ and $\phi_2$.
If $\mathcal{A}$ returns P (i.e., polynomial-time solvable) for both $\phi_1$ and $\phi_2$, we conclude that there does not exist a graph satisfying $\phi$; otherwise, we conclude that there exists a graph satisfying $\phi$.
To see our conclusion is correct, consider three cases.
First, assume there does not exist a connected graph satisfying $\phi$.
In this case, no graph can satisfy $\phi_1$ or $\phi_2$, and therefore, both $\phi_1$ and $\phi_2$ can be tested on graphs trivially in polynomial time by simply returning No.
Thus, $\mathcal{A}$ returns Yes for both $\phi_1$ and $\phi_2$, and our conclusion is correct.
Second, assume there exists a connected graph $G_0$ satisfying $\phi$ which is $3$-colorable.
In this case, testing $\phi_1$ is NP-hard.
Indeed, one can reduce an instance $G$ of \textsc{$3$-Coloring} to the task of testing $\phi_1$ on the graph that is the disjoint union of $G$ and $G_0$.
So $\mathcal{A}$ returns NP for $\phi_1$, and our conclusion is correct.
Finally, assume there does not exist a connected graph satisfying $\phi$ which is $3$-colorable, but there exists a connected graph $G_0$ satisfying $\phi$ that is not $3$-colorable.
In this case, testing $\phi_2$ is NP-hard.
Again, one can reduce an instance $G$ of \textsc{$3$-Coloring} to the task of testing $\phi_2$ on the disjoint union of $G$ and $G_0$.
So $\mathcal{A}$ returns NP for $\phi_2$, and our conclusion is correct.

\section{Other related work} \label{apx-literature}
A large body of work shows that many classic NP-hard problems admit dynamic programs on tree decompositions running in time \(d^{t}\cdot n^{O(1)}\) (for some fixed constant \(d\)) or \(2^{O(t\log t)}\cdot n^{O(1)}\), where \(t\) is the treewidth. A natural question is whether these running times, and in particular the dependence on treewidth \(t\), are optimal. In 2011, Lokshtanov, Marx, and Saurabh~\cite{DBLP:journals/talg/LokshtanovMS18} established conditional lower bounds showing that, for a range of concrete graph problems including \textsc{Vertex Cover}, \textsc{Dominating Set}, \textsc{$q$-Coloring}, and \textsc{Max Cut}, the base of the single-exponential algorithms is optimal: there is no \((d-\varepsilon)^{t}\cdot n^{O(1)}\) algorithm for any fixed \(\varepsilon>0\) unless Strong Exponential Time Hypothesis (SETH) fails.  

In a related result under the Exponential Time Hypothesis (ETH), Lokshtanov, Marx, and Saurabh~\cite{DBLP:journals/siamcomp/LokshtanovMS18} showed that for problems such as \textsc{Cycle Packing} and \textsc{Chromatic Number}, the known \(2^{O(t\log t)}\cdot n^{O(1)}\) algorithms are optimal with respect to the dependence on treewidth \(t\). Equivalently, unless ETH fails, none of these problems admits a \(2^{o(t\log t)}\cdot n^{O(1)}\) algorithm.

However, for several connectivity problems such as \textsc{Hamiltonian Cycle} and \textsc{Steiner Tree}, neither of the earlier works~\cite{DBLP:journals/talg/LokshtanovMS18,DBLP:journals/siamcomp/LokshtanovMS18} settled the dependence on treewidth: algorithms of the form \(2^{O(t\log t)}\cdot n^{O(1)}\) were known, but there were neither ETH lower bounds ruling out \(2^{o(t\log t)}\cdot n^{O(1)}\) nor faster algorithms. This gap was resolved by the seminal work of Cygan et al.~\cite{DBLP:journals/talg/CyganNPPRW22}, which introduced the cut-and-count technique and obtained randomized algorithms with running time \(d^{t}\cdot n^{O(1)}\) for a fixed constant \(d\). These developments catalyzed the field, and the past decade has produced numerous algorithms with optimal dependence on treewidth~\cite{DBLP:conf/stacs/HartmannM25,DBLP:conf/icalp/EsmerFMR24,DBLP:conf/soda/FockeMINSSW23}. 

Several works focus on \emph{classes} of problems and obtain SETH-tight algorithms. 
Prominent examples arise around \textsc{Dominating Set}, yielding optimal bounds for variants such as \textsc{$r$-Domination}~\cite{DBLP:conf/iwpec/BorradaileL16} and \textsc{$(\sigma,\rho)$-Domination}~\cite{DBLP:conf/soda/FockeMINSSW23}, among many others~\cite{DBLP:conf/stacs/HartmannM25,DBLP:conf/icalp/EsmerFMR24}. There has also been an effort to identify fragments of logic such that any problem expressible in the fragment admits an algorithm running in time \(2^{O(t)} n^{O(1)}\). Pilipczuk~\cite{DBLP:conf/mfcs/Pilipczuk11} introduced such a logic that essentially captures “most” problems solvable in time \(2^{O(t)} n^{O(1)}\).

Beyond problem-specific upper and lower bounds, there has been extensive work on identifying \emph{subclasses} of graphs of treewidth \(t\) for which the running time can be bounded by an \emph{elementary} function, in contrast to the non-elementary bounds implied by Courcelle’s theorem. These results show that, on suitably restricted graph classes, many MSO-expressible problems admit algorithms whose running time is bounded by an elementary function. 

Lampis~\cite{DBLP:journals/algorithmica/Lampis12} showed that by restricting to two proper subclasses of bounded-treewidth graphs, bounded vertex cover and bounded max-leaf number, one can obtain sharper meta-theorems. In particular, every FO property is decidable on both classes with singly-exponential dependence on the parameter, and MSO properties are decidable on bounded–vertex-cover graphs with doubly-exponential dependence. Matching lower bounds indicate these dependencies are near-optimal under standard assumptions~\cite{DBLP:journals/algorithmica/Lampis12}. Ganian generalized Lampis’s result from graphs of bounded vertex cover to graphs of bounded twin-cover~\cite{DBLP:conf/iwpec/Ganian11}. In a recent work, Lampis~\cite{DBLP:conf/icalp/Lampis23} showed that every FO-expressible property can be decided with an elementary dependence on the formula size when the input graph has bounded \emph{pathwidth} (as opposed to treewidth).

Gajarský and Hliněný showed the following~\cite{DBLP:journals/corr/abs-1204-5194}. In the universe of colored trees of fixed height, any MSO-definable decision problem with \(r\) quantifiers admits a finite family of kernels, each of size bounded by an elementary function of \(r\) and the number of colors. Consequently, for any graph class \(\mathcal{G}\) with a one-dimensional MSO interpretation in this universe (i.e., shrub-depth \(h\)), (i) MSO model checking runs in time with elementary dependence on the formula size, and (ii) FO and MSO have the same expressive power on \(\mathcal{G}\).  Lampis~\cite{DBLP:journals/corr/abs-1302-4266} complemented this line of work by showing that there is no MSO model-checking algorithm with elementary parameter dependence even on paths (equivalently, unary strings), unless \(\mathrm{E}=\mathrm{NE}\). Finally, it also showed that for MSO on colored trees of depth \(d\), assuming ETH, for every fixed \(d\ge 1\) at least \(d+1\) levels of exponentiation are necessary, implying that the \((d+1)\)-fold exponential algorithm of Gajarský and Hliněný is essentially optimal~\cite{DBLP:journals/corr/abs-1204-5194}.

Bergougnoux, Chekan, and Stamoulis~\cite{bergougnoux2025logicbasedalgorithmicmetatheoremtreedepth} introduced a new logic fragment that captures a broad family of NP-hard problems and shows that any problem expressible in this logic is solvable in time \(2^{O(k)} n^{O(1)}\) and space \(n^{O(1)}\) on graphs equipped with an elimination forest of depth \(k\) (treedepth).
This fragment extends fully-existential \(\mathsf{MSO}_2\) with predicates for generalized neighborhoods, connectivity/acyclicity, and clique checks, yielding a unified model-checking algorithm with these guarantees.

Building on an earlier result of Kreutzer~\cite{DBLP:conf/csl/Kreutzer09}, Kreutzer and Tazari~\cite{DBLP:conf/soda/KreutzerT10} showed that for graph classes with mild closure properties, the presence of graphs with sufficiently large treewidth (already polylogarithmic in \(n\)) precludes polynomial-time model checking for \(\mathsf{MSO}_2\) formulas. In this sense, bounded treewidth forms the effective boundary for tractable \(\mathsf{MSO}_2\) model checking (also see~\cite{DBLP:conf/lics/KreutzerT10}).

\end{document}